\documentclass[11pt,twoside]{article}
\usepackage[T1]{fontenc} 
\usepackage[title]{appendix}
\usepackage{titling}
\newcommand{\bs}{\begin{align}\begin{split}\nonumber}

\usepackage{multibib}
\newcommand{\weakconv}{\overset{\mathbb{P}}{\underset{\mathbb{W}}{\rightsquigarrow}}}
\newcommand{\rmnum}[1]{\romannumeral #1}
\newcommand{\Rmnum}[1]{\uppercase\expandafter{\romannumeral #1\relax}}
\usepackage{bookmark} \bookmarksetup{numbered}
\usepackage{amsmath, amsthm, amssymb, mathrsfs, upgreek, enumerate, mathtools,
  comment, natbib, subfigure, graphicx, mathabx, booktabs, threeparttable,
tikz, relsize, exscale, epstopdf, multirow, dsfont}
\usepackage{algorithm2e}
\RestyleAlgo{ruled}
\usepackage{fontawesome5}
\usepackage[headheight=105pt,margin=1.05in]{geometry}

\usepackage[section]{placeins}
\usepackage{caption}
\newcommand\fnote[1]{\captionsetup{font=small}\caption*{#1}}
\SetKwInput{KwData}{Input}
\newsavebox{\varmatrixbox}

\ExplSyntaxOn
\keys_define:nn {martin/varmatrix}
 {
  sep	.dim_set:N = \l_martin_varmatrix_sep_dim,
  delim .tl_set:N  = \l_martin_varmatrix_delim_tl,
  size	.tl_set:N  = \l_martin_varmatrix_size_tl,
  sep	.initial:n = 0.7\arraycolsep,
  size	.initial:n = \small,
 }

\NewDocumentEnvironment{varmatrix}{O{}}
 {
  \keys_set:nn {martin/varmatrix} { #1 }
  \begin{lrbox}{\varmatrixbox}
  \l_martin_varmatrix_size_tl
  \setlength{\arraycolsep}{\l_martin_varmatrix_sep_dim}
  $\begin{\l_martin_varmatrix_delim_tl matrix}
 }
 {
  \end{\l_martin_varmatrix_delim_tl matrix}$
  \end{lrbox}
  \vcenter{\box\varmatrixbox}
 }
\ExplSyntaxOff

\renewcommand{\(}{\left(}
\renewcommand{\)}{\right)}
\renewcommand{\[}{\left[}
\renewcommand{\]}{\right]}
\newcommand{\V}{\text{Vol}}

\usepackage{fancyhdr}
\usepackage{setspace} 
\usepackage{tcolorbox}

\usepackage{hyperref}
\usepackage{subfigure}
\usepackage{marginnote}
\hypersetup{
    colorlinks=true,
    citecolor=blue,
    filecolor=black,
    linkcolor=red,
    urlcolor=blue,
    bookmarksopen=true,
    pdfstartview=FitH
}
\newcommand{\ba}{\begin{array}}
\newcommand{\ea}{\end{array}}
\usepackage{scalerel,stackengine}
\stackMath
\newcommand\reallywidehat[1]{
\savestack{\tmpbox}{\stretchto{
  \scaleto{
    \scalerel*[\widthof{\ensuremath{#1}}]{\kern.1pt\mathchar"0362\kern.1pt}
    {\rule{0ex}{\textheight}}
  }{\textheight}
}{2.4ex}}
\stackon[-6.9pt]{#1}{\tmpbox}
}
\DeclareMathOperator*{\esssup}{ess\,sup}

\DeclareMathOperator\supp{supp}

\usepackage[utf8]{inputenc}

\DeclareFixedFont{\ttb}{T1}{txtt}{bx}{n}{12} 
\DeclareFixedFont{\ttm}{T1}{txtt}{m}{n}{12}  

\usepackage{color}
\definecolor{deepblue}{rgb}{0,0,0.5}
\definecolor{deepred}{rgb}{0.6,0,0}
\definecolor{deepgreen}{rgb}{0,0.5,0}

\usepackage{listings}

\newcommand\pythonstyle{\lstset{
language=Python,
basicstyle=\ttm\scriptsize,
morekeywords={self},		  
keywordstyle=\ttb\color{deepblue}\scriptsize,
emph={MyClass,__init__},	  
emphstyle=\ttb\color{deepred},	  
stringstyle=\color{deepgreen},
frame=tb,			  
showstringspaces=false
}}

\lstnewenvironment{python}[1][]
{
\pythonstyle
\lstset{#1}
}
{}

\newcommand\pythoninline[1]{{\pythonstyle\lstinline!#1!}}

\newtheorem{manualpropositioninner}{Assumption}
\newenvironment{manualproposition}[1]{%
  \renewcommand\themanualpropositioninner{#1}%
  \manualpropositioninner
}{\endmanualpropositioninner}

\def\indep{\perp\!\!\!\perp}
\def\arraystretch{1.8}

\newcommand{\bsnumber}{\begin{align}\begin{split}}
\newtheorem{theorem}{Theorem}[subsection]
\newtheorem*{theorem*}{Theorem}
\newtheorem{assump}{Assumption}[subsection]
\newtheorem{lemma}[theorem]{Lemma}
\newtheorem{cor}[theorem]{Corollary}
\newtheorem{proposition}[theorem]{Proposition}

\theoremstyle{definition}
\newtheorem{definition}{Definition}[subsection]

\newtheorem{remark}{Remark}[subsection]

\makeatletter
\def\thm@space@setup{
  \thm@preskip
15pt \thm@postskip=15pt 
}
\makeatother

\renewcommand{\qed}{\hfill \mbox{\raggedright \rule{0.08in}{0.08in}}} 
\renewenvironment{proof}[1][\proofname]{{\noindent\sc#1. }}{\qed\vspace{15pt}} 

\usepackage{tikz}
\usetikzlibrary{shapes.geometric, arrows}
\tikzstyle{basicnode} = [rectangle, rounded corners, minimum width=3cm, minimum height=1cm,text centered, draw=black]
\tikzstyle{arrow} = [thick,->,>=stealth]

\usetikzlibrary{calc}
\usetikzlibrary{arrows}
\usetikzlibrary{decorations.markings}

\newcommand*{\StrikeThruDistance}{0.15cm}%

\usepackage{lipsum}

\newcommand\blfootnote[1]{%
  \begingroup
  \renewcommand\thefootnote{}\footnote{#1}%
  \addtocounter{footnote}{-1}%
  \endgroup
}

\tikzset{strike thru arrow/.style={
    decoration={markings, mark=at position 0.5 with {
	\draw [blue, thick,-]
	    ++ (-\StrikeThruDistance,-\StrikeThruDistance)
	    -- ( \StrikeThruDistance, \StrikeThruDistance);}
    },
    postaction={decorate},
}}

\usepackage{versions, marginnote, todonotes}
\excludeversion{notes}
\includeversion{links}

\newcommand{\smalltodo}[2][] {\todo[caption={#2}, size=\scriptsize, fancyline, #1] {\begin{spacing}{.5}#2\end{spacing}}}

\newcommand{\kai}[2][]{\smalltodo[color=yellow!20,#1]{{\bf Kai:} #2}}

\title{\bf\sc\Large Statistical Inference of Optimal Allocations \MakeUppercase{\romannumeral 1}: Regularities and their Implications}

\thanksmarkseries{alph}
\author{Kai Feng\thanks{Institute of Economics, Tsinghua
University. Mingzhai Building, Tsinghua University, Beijing 100084,
China. Email: fengk22@mails.tsinghua.edu.cn.} \and Han
Hong\thanks{Department of Economics,  Stanford University. 579 Jane Stanford Way,
Stanford, California 94305, USA. Email: doubleh@stanford.edu.} \and
Denis Nekipelov\thanks{Departments of Economics and Computer
Science, University of Virginia. Monroe Hall,
Charlottesville, VA 22903, USA.  Email: dn4w@virginia.edu.}
}

\begin{document}
\maketitle

\begin{abstract}
  \noindent {\sc Abstract.}
In this paper, we develop a functional differentiability approach for solving
statistical optimal allocation problems.
We derive Hadamard differentiability of the value functions through
analyzing the properties of the sorting operator using tools from
geometric measure theory.
Building on our Hadamard differentiability results, we apply
the functional delta method to obtain the asymptotic properties of the value function process for the binary constrained optimal allocation problem and the plug-in ROC curve estimator.
Moreover, the convexity of the optimal allocation value functions facilitates demonstrating the degeneracy of first order derivatives
with respect to the policy.
We then present a double / debiased estimator for the value functions.
Importantly, the conditions that validate Hadamard differentiability
justify the margin assumption from the statistical classification literature
for the fast convergence rate of
plug-in methods.

  \vspace{15pt}

  \vspace{15pt}
\end{abstract}

\newpage

\section{Introduction}\label{introduction}

Characterizing individual heterogeneity is crucial
in the optimal allocation of scarce resources especially as a consequence of the rapid digitization of social-economic data.
Examples include 
medical treatment assignments when doctors adopt
different therapeutic approaches to different patients,
pretrial bails when judges decide where defendants will await trials,
college admissions when committees assign students to various educational programs.
In order to address the heterogeneity inherent in such problems,
where the potential effects of a treatment on
an outcome vary among individuals,
personalized allocation is more of a norm than an exception.
Personalized allocation is compelling in the presence of the inherent individual heterogeneity
in the potential effects of a treatment on an outcome.

Based on past experimental or observational data,
the core objective of optimal allocation problem is to find a
strategy that can perform well out of sample.
An effective approach advocated by the literature is to make use of information provided by
prediction algorithms. See for example \cite{QJEbail}, \cite{agrawal2019artificial} and \cite{babina2024artificial}.
A common decision framework is to form
a nonrandom treatment rule using an indicator function that
benchmarks
a first step regression estimate of individual level treatment effect  against
certain thresholds, representing a plug-in strategy
in a weighted classification problem.

Such a plug-in strategy is based on solving a population 
optimal allocation problem.
For a joint distribution of $\(X, Y\)$ where $Y = \(Y_{j}, j=0,\ldots,J\)$,
consider the \textit{social welfare potential function}\footnote{Appendix \ref{frechet differentiability section}  provides a mathematical justification 
for the term ``potential function''.}
\bsnumber\label{optimal allocation general problem}
\gamma\(\lambda, g\(\cdot\)\) \coloneqq
\max_{\phi \in \Phi} \mathbb{E}\[\sum_j \lambda_j \phi_j\(X\) Y_j\]
= \max_{\phi \in \Phi} \mathbb{E}\[\sum_j \lambda_j \phi_j\(X\) g_j\(X\)\].
\end{split}\end{align}
In the above, $\Phi$ is the set of multi-critical functions, i.e.
$\phi = \(\phi_{j}, j=0,\ldots,J\)$,
$\phi_j: \Omega\rightarrow \mathbb{R}$, with $\(\Omega, F, \mu\)$ being
a Borel probability space,  such that
$\phi_j \in L^\infty\(\mu\), 0 \leq \phi_j\leq 1, \forall
j\in \{0,\ldots,J\}$, and such that  $\sum_{j=0}^J
\phi_j\(x\) \equiv 1$ for all $x \in \Omega$;
$g\(\cdot\) = \(g_j\(\cdot\), j=0,\ldots,J\)$
and $g_j\(x\) = \mathbb{E}\(Y_j \vert X = x\)$.
Under very general conditions, the problem posed in
\eqref{optimal allocation general problem} admits an intuitive solution
where the $j$th-category receives all allocation
when $\mathbb{E}\(\lambda_j Y_j \vert X\)$ is the unique maximum.
If there exists some $l$ such that
$\mathbb{E}\(\lambda_j Y_j \vert X\) = \mathbb{E}\(\lambda_l Y_l \vert X\)$,
then 
all allocation can be divided among the maximal indexes.

In this paper, we are interested in differentiating $\gamma\(\lambda, g\)$ in \eqref{optimal allocation general problem} functionally
with respect to both $\lambda$ and $g\(\cdot\)$.
The solution 
can 
provide inference methods for the value of optimal treatment allocation in multiple applications.
In section \ref{constrained and roc}, 
we derive the asymptotic distributions and justify bootstrapping for binary optimal allocation subject to a resource constraint and for the receiver operating characteristic (ROC) curve in binary classification. 
Section \ref{fast convergence rates} provides asymptotic regret bounds and the asymptotic distribution of the value function for unconstrained multi-class optimal allocation.

Our paper attempts to synthesize several strands of literature.
It is widely-known 
in the plug-in classification literature that the error of the second step classification is smaller than that of
the first step regression 
\citep{devroye1996probabilistic, audibert2007fast}.
Our paper builds on seminal works by \cite{chernozhukov2018sorted} and \cite{chen2003estimation}
to provide conditions supporting $\sqrt{n}$ value function convergence rates and tight regret bounds in multiclass and constrained optimal allocations, covering
settings both for classical empirical processes and for recent advances in Neyman orthogonality and sample splitting \citep{chernozhukov2018double, chernozhukov2022locally}.
Subsequent works following \cite{manski2004statistical} have defined regret as the difference between the social
welfare potential and the expected utility achieved by an estimated decision rule. 

In the binary case of $J=1$, $\lambda = \(\lambda_0, \lambda_1\)$.
Consider the
finite dimensional derivatives of $\gamma\(\lambda, g\) =\gamma\(\lambda_0, \lambda_1, g\)$ with respect to $\lambda_0$.
Assuming $\mathbb{P}\{\lambda_1 g_1\(X\)=\lambda_0 g_0\(X\)\}=0$,  we can write
\bsnumber\label{binary indicator sorting operator}
\gamma\(\lambda_0,\lambda_1, g\) 
=\mathbb{E}\[1\(\lambda_0 g_0\(X\) > \lambda_1 g_1\(X\) \)\lambda_0 g_0\(X\)\]
+\mathbb{E}\[1\(\lambda_1 g_1\(X\) > \lambda_0 g_0\(X\) \)\lambda_1 g_1\(X\)\].
\end{split}\end{align}
Given the definition of $\gamma\(\lambda_0,\lambda_1, g\)$,  we expect based on intuitive reasoning that
\begin{align}\begin{split}\label{intuition for lambda finite dimensional}
&\frac{\partial \gamma\(\lambda_0, \lambda_1, g\)}{\partial \lambda_0}
= \mathbb{E}\[1\(\lambda_0 g_0\(X\) > \lambda_1 g_1\(X\) \)g_0\(X\)\].
\end{split}\end{align}
In a rigorous proof of \eqref{intuition for lambda finite dimensional}
using direct calculation, 
the terms of major difficulties 
are
\bsnumber\label{binary indicator sorting operator heuristic derivative}
\lim_{h \to 0} & \frac{
	\mathbb{E} \[1\(\(\lambda_0+h\) g_0\(X\) > \lambda_1 g_1\(X\)\)
	- 1\(\lambda_0 g_0\(X\) > \lambda_1 g_1\(X\)\) \lambda_0 g_0\(X\)\]
	}{
h} \quad \text{and} \\
\lim_{h \to 0} & \frac{
	\mathbb{E} \[1\(\lambda_1 g_1\(X\) > \(\lambda_0 +h \) g_0\(X\)\)
	- 1\(\lambda_1 g_1\(X\) > \lambda_0 g_0\(X\)\) \lambda_1 g_1\(X\)\]
	}{
h
	}.
\end{split}\end{align}
We will see later in section \ref{Weighted sorted effect} that under mild conditions,
the two terms in \eqref{binary indicator sorting operator heuristic derivative} cancel out because
they converge to the same integral on the level set
$\{x: \lambda_0 g_0\(x\)=\lambda_1 g_1\(x\)\}$ with opposing signs. 
As a consequence, the desired partial derivatives exist and take the displayed form 
\eqref{intuition for lambda finite dimensional}.

Our approach in section \ref{Weighted sorted effect} is greatly inspired by the \textit{integration on manifold} approach in \cite{chernozhukov2018sorted}.
In particular, we reapproach and generalize their results based on the concept of Hausdorff measure and powerful 
area and coarea integration formulas from geometric measure theory.
We show that the Hadamard derivative of $\gamma\(\lambda, g\)$ contains terms that are integrals under certain Hausdorff measure.
The Hadamard differentiability results are derived under mild primitive conditions instead of being directly assumed.
See for example \cite{sherman1993limiting} and \cite{luedtke2016optimal}.

In the general multi-class setting, the social welfare potential function
in \eqref{optimal allocation general problem} 
is not only subadditive but also positive homogeneous of degree one with respect to both $\lambda$ and $g\(\cdot\)$, respectively.
Consequently, while $\gamma\(\lambda, g\)$ is not necessarily jointly convex in $\(\lambda, g\)$,
$\gamma\(\lambda, g\(\cdot\)\)$ is convex in $g\(\cdot\)$ given $\lambda$ and is convex in $\lambda$ given $g\(\cdot\)$.
Interestingly, this simple observation seems to have gone 
unnoticed by the existing literature. 
As discussed in section \ref{frechet differentiability section}, 
profound results from geometric functional analysis imply that under nearly no assumptions, the social welfare potential 
function $\gamma\(\lambda, g\)$ is Fr\'{e}chet differentiable on a \textit{large} set and admits an application of the envelope  theorem.

The implied envelope theorem allows us to verify the degeneracy of the first order Hadamard derivative of a multi-discrete social welfare potential function. 
In section \ref{fast convergence rates},
we show that the primitive conditions in section \ref{Weighted sorted effect} can also be used to generalize fast convergence
rates 
for the regret of binary plug-in classifiers \citep{audibert2007fast} to multi-discrete allocation problems.
The same methodology can be combined with recent developments of double/debiased machine learning \citep{chernozhukov2018double,chernozhukov2022locally} to speed 
up the estimation of $\gamma\(\lambda, g\)$.

\section{Literature}\label{literature}

A multidisciplinary literature has studied
the construction of optimal decision rules from experimental and observational data in the forms of treatment allocation, policy learning and
individual treatment rule (ITR).
A large body of the literature focuses on empirical risk minimization (ERM).
See for example
\cite{manski2004statistical}, \cite{qian2011performance}, \cite{zhao2012estimating}, \cite{swaminathan2015batch}, \cite{zhou2017residual}, \cite{kitagawa2018should}, \cite{rai2018statistical}, \cite{luedtke2020performance}, \cite{athey2021policy}, \cite{mbakop2021model}, \cite{zhou2023offline}, \cite{ben2024policy} and \cite{ai2024data}, among others.
An important objective is to provide probabilistic bounds for
\bs
R\(\hat{\phi}, \phi^{*}\) & = \mathbb{E}_{X, Y}\[\sum_{j}\phi^{*}_{j}\(X\)Y_{j}\] -
\mathbb{E}_{X, Y}\[\sum_{j}\hat{\phi}_{j}\(X\)Y_{j}\], \\
\text{ for } \hat{\phi} & \in \mathop{\arg\max}\limits_{\phi \in \Phi_{0}}
 \sum_i \sum_{j} \phi_j\(X_i\)	 F_j\(X_{i}, D_{i}, Y_{i}\),
\end{split}\end{align}
where $\phi^{*}$ solves the population optimization program \eqref{optimal allocation general problem}.
The function $F\(X, D, Y\)$ usually takes the inverse propensity weighting (IPW) or doubly robust (DR) form
and $\Phi_{0} \subset \Phi$ is not too complex in an entropy sense.
 See for example \cite{kitagawa2018should} and \cite{athey2021policy}.

A statistical difficulty in the optimal allocation problem arises
from the indicator function in the population solution \citep{qian2011performance}. To address
such a problem, we build on a divergent literature motivated by very different purposes.
To restore monotonicity in conditional quantile estimation,
the functional derivative of the sorting operator in the univariate case 
is studied by \cite{chernozhukov2010quantile} and subsequently
extended to the multivariate cases
by \cite{chernozhukov2018sorted}
using calculus on manifold techniques. \cite{kim1990cube} hinted at a rudimentary form
of calculus on manifold in deriving the large sample properties of cube root consistent estimators.
\cite{sasaki2015quantile} incorporated Hausdorff measure tools in fluid mechanics to characterize the information content of quantile partial derivatives for general structural functions.

Optimal allocation problems appear in different but related settings.
Examples include asymptotically minimax optimal decision under the limits of experiments framework \citep{hirano2009asymptotics}, optimal decision under minimax regret \citep{stoye2009minimax, tetenov2012statistical, ben2024policy}, uniform confidence interval \citep{armstrong2023inference}, allocation with spillover effects \citep{kitagawa2023individualized, kitagawa2023should, viviano2025policy},
as well as binary constrained optimal allocation \citep{bhattacharya2012inferring, luedtke2016optimal} for which
our general functional Hadamard derivative results in section \ref{Weighted sorted effect} are applicable. 

Closely related is the receiver operating characteristic (ROC) curve widely used in binary classification. 
In the economics and finance literature, the plug-in ROC curves are used for example in \cite{berge2011evaluating}, \cite{vallee2019marketplace}, \cite{sadhwani2021deep}, \cite{feng2025statisticalMS} to evaluate prediction and decision making performance.
A previous literature studying the inference of the ROC curves, including \cite{hsieh1996nonparametric}, \cite{lloyd1998using}, \cite{li1999semiparametric}, \cite{hall2004nonparametric} and \cite{bertail2008bootstrapping},
did not account for the sampling error in the estimated propensity score.
An exception is \cite{luckett2021receiver} which
derives uniform asymptotics for the estimated ROC curve
using a support vector machine (SVM) classifier.
Unlike previous studies,
we derive the asymptotic properties of
the ROC curves obtained by plugging-in
the propensity score estimators,
which is consistent with the prevalent empirical usage of
predictive classifiers. 

The methodology in this paper also has an intriguing connection with the faster convergence phenomenon in the classification literature \citep{devroye1996probabilistic, audibert2007fast}.
Consider, especially, the margin assumption (MA) evoked by works including \cite{mammen1999smooth}, \cite{tsybakov2004optimal}, \cite{boucheron2005theory}, \cite{audibert2007fast}, \cite{kitagawa2018should}, \cite{luedtke2020performance}
and \cite{semenova2023debiased}:
\bsnumber\label{marginal assumption MA}
\mathbb{P}\left\{\vert p\(X\) - c\vert < t\right\} \leq Ct^{\alpha}
\end{split}\end{align}
for fixed $c$ and some constant $C > 0$ and $\alpha \geq 0$.
Our conditions in section \ref{Weighted sorted effect} directly imply the margin
assumption with $\alpha \geq 1$.
This connection can be further combined with the recent double/debiased machine learning approach \citep{chernozhukov2018double, chernozhukov2022locally} to form a transparent cross-validation and plug-in based method to complement
the optimal allocation literature.

\section{Hadamard differentiability}\label{Weighted sorted effect}

Following the seminal work by
\cite{chernozhukov2018sorted},	we consider \eqref{binary indicator sorting operator}
as a functional mapping 
and are interested in the infinitesimal change in  $\gamma\(\lambda, g\(\cdot\)\)$ induced by a
marginal change in $\(\lambda, g\(\cdot\)\)$. More precisely, we derive the
Hadamard derivative of
$\gamma\(\lambda, g\(\cdot\)\)$ with respect to $\(\lambda, g\(\cdot\)\)$, which justifies and generalizes the calculations in \eqref{binary indicator sorting operator heuristic derivative}.
We assume that the underlying distribution of $X$,
denoted as $\mu$, is absolutely continuous with density
function $f\(x\)$ with respect to $\mathcal{L}_{n}$, the $n$-dimensional Lebesgue measure. Then we write
\bs
\mathbb{E}\[\sum_{j=0}^J \lambda_j \phi_j\(X\) g_j\(X\)\]
= \sum_{j=0}^J \int \lambda_j \phi_j\(x\) g_j\(x\) d\mu\(x\)=
\sum_{j=0}^J \int \lambda_j \phi_j\(x\) g_j\(x\) f\(x\) d \mathcal{L}_{n} x.
\end{split}\end{align}
Typically, the optimizing policy function takes the form of
$\phi_j\(x\) = \prod_{l \neq j} 1\(\lambda_j g_j\(x\) > \lambda_l g_l\(x\) \)$.
For simplicity we abbreviate $\lambda_j g_j\(x\)$ as
$g_j\(x\)$. The main object of study is 
\bs
\int g_j\(x\) f\(x\) \prod_{l\neq j} 1\(g_j\(x\) > g_l\(x\)\) d \mathcal{L}_{n} x,
\end{split}\end{align}
which can be 
generalized to a form of the following sorting operator (as defined by \cite{chernozhukov2018sorted})
for vector-valued function $h\(x\)$ and vector-valued cutoff threshold $c$,
\bsnumber\label{sorting operator in general form}
\int g_j\(x\) f\(x\) \prod_{l} 1\(h_l\(x\) > c_{l}\) d
\mathcal{L}_{n} x.
\end{split}\end{align}
The number of terms in the product should be
less than $n$, the dimension of $X$.

\subsection{Hausdorff measure and integration formulas}

The following discussions are based on standard references by
\cite{federer2014geometric} and \cite{evans2018measure}.
In order to extend the analysis in \cite{chernozhukov2018sorted},
this section first 
defines extrinsically the 
notion of length, area and volume.
For this purpose, the Carath\'{e}odory criterion is used to
construct a measure space from a $\sigma$-subadditive set function (outer measure) on power set of a space.
More specifically, we introduce the series of Hausdorff's
outer measures
(See definition
\ref{basic outer measure} for outer measure) which
does not require any local parameterization or regularity
assumption on these sets.

\begin{definition}{(Hausdorff outer measure)}
Let $d$ be a metric on $O$.
For arbitrary subset
$E \subset O$, define its diameter as
$\mathrm{diam}\ E = \sup_{x_{1}, x_{2} \in E}
d\(x_{1}, x_{2}\)$ with
$\mathrm{diam}\ \emptyset = 0.$
Let
\begin{align}\begin{split}\nonumber
\alpha_{k} = \frac{\pi^{\frac{k}{2}}}
{\Gamma\(\frac{k}{2} + 1\)},
\quad k \in \mathbb{N}=\{0,1,\ldots\},
\end{split}\end{align}
where $\Gamma\(\cdot\)$ is the gamma function.	For arbitrary $E \subset O$,
and for $\delta > 0$,
define
\begin{align}\begin{split}\nonumber
\mathcal{H}^{*}_{k, \delta}\(E\) = \inf\left\{
\sum_{j \geq 1}\alpha_{k}\(
\frac{\mathrm{diam}\ B_{j}}{2}\)^{k} :
E \subset \bigcup_{j \geq 1}B_{j},
\mathrm{diam}\ B_{j} \leq \delta
\right\}.
\end{split}\end{align}
The $k$-dimensional Hausdorff outer measure
of $E$ is defined as
$\mathcal{H}^{*}_{k}\(E\) = \lim_{\delta \downarrow 0}
\mathcal{H}^{*}_{k, \delta}\(E\)$, and can be verified to
be a metric outer measure on the metric space $\(O,d\)$.
\end{definition}

\begin{definition}\label{def Hausdorff measure}
{(Hausdorff measure)}
For $k \in \mathbb{N}$, let $\mathcal{M}_{k} =
\mathcal{M}_{k}\(O\)$ be the $\sigma$-algebra
generated by $\mathcal{H}^{*}_{k}$ by the
Carath\'{e}odory criterion (See Theorem
\ref{Caratheodory's extension theorem}).
The restriction of $\mathcal{H}^{*}_{k}$ to
$\mathcal{M}_{k}$,
$\mathcal{H}_{k} =
\mathcal{H}^{*}_{k}\vert_{\mathcal{M}_{k}}$, is
called Hausdorff measure on $\mathcal{M}_{k}$.
\end{definition}

In $\mathbb{R}^{n}$,
  $\mathcal{H}^{*}_{n} = \mathcal{L}^{*}_{n}$, $\mathcal{H}_{n} = \mathcal{L}_{n}$,
  where $\mathcal{L}^{*}_n$ denotes the
  Lebesgue outer measure and $\mathcal{L}_{n}$
  denotes the Lebesgue measure on
  $\mathbb{R}^{n}$.
$\mathcal{H}^{*}_{0}\(E\)$ measures the cardinality of the set
$E$:
  \begin{align}\begin{split}\nonumber
  \mathcal{H}^{*}_{0}\(E\) = \mathcal{H}_{0}\(E\) = \left\{
  \begin{array}{ll}
    \text{number of elements in } E, &
    E \text{ is empty or finite,} \\
    + \infty, & E \text{ is infinite}.
  \end{array}\right.
  \end{split}\end{align}
We focus on the Hausdorff measures  $\mathcal{H}_{k}$
 defined on
Euclidean spaces $\mathbb{R}^{n}$ for $k \leq n$.

\begin{definition}\label{def J}
(Jacobian) For
$k \in \{1, 2, \ldots, n\}$, the Jacobian at $x$ is defined as
\begin{align}\begin{split}\nonumber
Jf\(x\) = \left\{\begin{array}{ll}
		   \sqrt{\det\(\nabla f\(x\)^{T} \nabla f\(x\)\)}, &
		   f: E \subset \mathbb{R}^{k} \to
		   \mathbb{R}^{n}, \\
		   \sqrt{\det\(\nabla f\(x\)\nabla f\(x\)^{T}\)}, &
		   f: E \subset \mathbb{R}^{n} \to
		   \mathbb{R}^{k}.
		 \end{array}\right.
\end{split}\end{align}
\end{definition}

The following Theorem \ref{area formula classic}
generalizes the conventional
integration change of variable formula for
invertible
mapping to allow for non-invertible mapping and an increase in dimension.

\begin{theorem}\label{area formula classic}
Area formula\\
Let $E \subset \mathbb{R}^{k}$ be an open set,
$k \in \{1, 2, \ldots, n\}$,
$\psi: E \mapsto \mathbb{R}^{n}$ be a
$C^1$ or Lipschitz function.
Then for all measurable $S \subset E$,
\begin{align}\begin{split}\label{area formula area}
\int_{S}J\psi\(x\)d\mathcal{L}_{k}x =
\int_{\psi\(S\)}\mathcal{H}_{0}\(S \cap
\psi^{-1}\(y\)\)d\mathcal{H}_{k}y.
\end{split}\end{align}
If a measurable function $f: S \mapsto \mathbb{R}$ is
nonnegative or if the left hand side
of \eqref{area formula integral} is finite, then the following equality holds,
\begin{align}\begin{split}\label{area formula integral}
\int_{S}f\(x\)J\psi\(x\)d\mathcal{L}_{k}x =
\int_{\psi\(S\)}\[\int_{S \cap \psi^{-1}\(y\)}
f\(x\)d\mathcal{H}_{0}x\]d\mathcal{H}_{k}y.
\end{split}\end{align}
\end{theorem}

Furthermore, the next \textit{curving Fubini-Tonelli}
theorem allows for slicing among curving surfaces and mapping from a higher
dimension to a lower dimension.

\begin{theorem}\label{coarea formula classic}
{Coarea formula}\\
 Let $E \subset \mathbb{R}^{n}$ be an open set,
$k \in \{1, 2, \ldots, n\}$,
$\varphi: E \mapsto \mathbb{R}^{k}$ be a $C^1$ or Lipschitz function,
then for all measurable $S \subset E$,
\begin{align}\begin{split}\label{coarea formula area}
\int_{S}J\varphi\(x\)d\mathcal{L}_{n}x =
\int_{\mathbb{R}^{k}}\mathcal{H}_{n - k}\(S \cap
\varphi^{-1}\(y\)\)d\mathcal{H}_{k}y.
\end{split}\end{align}
If a measurable function $f: S \mapsto \mathbb{R}$ is
nonnegative or if the left hand side
of \eqref{coarea formula integral} is finite, then the following equality holds,
\begin{align}\begin{split}\label{coarea formula integral}
\int_{S}f\(x\)J\varphi\(x\)d\mathcal{L}_{n}x =
\int_{\mathbb{R}^{k}}\[\int_{S \cap \varphi^{-1}\(y\)}
f\(x\)d\mathcal{H}_{n - k}x\]d\mathcal{H}_{k}y.
\end{split}\end{align}
\end{theorem}

\begin{remark}
  \eqref{area formula integral} and
  \eqref{coarea formula integral}
  are direct corollaries of
  \eqref{area formula area} and
  \eqref{coarea formula area}, respectively,
  by a standard simple function
  approximation of measurable functions.
  \eqref{area formula integral}, 
  \eqref{coarea formula integral}
  and the classical $\mathbb{R}^{n} \mapsto
  \mathbb{R}^{n}$ change of variables
  formula, 
  can be collectively written in a unified formula:
  \begin{align}\begin{split}\nonumber
  \int_{S}f\(x\)J\varphi\(x\)d\mathcal{L}_{n}x =
  \int_{\mathbb{R}^{k}}\[\int_{S \cap
  \varphi^{-1}\(y\)}f\(x\)
  d\mathcal{H}_{\max\{n - k, 0\}}x\]d\mathcal{H}_{\min\{n, k\}}y,
  \end{split}\end{align}
  where $\varphi: \mathbb{R}^{n} \supset E \to \mathbb{R}^{k}$, $n, k \in \{1, 2, \ldots\}$.
\end{remark}

\subsection{Hadamard derivatives of the sorting operator}
\label{sorting operator}

In this section,
we present
our mathematics results on the
Hadamard differentiability of the sorting
operator in \eqref{sorting operator in general form}, where we intend to compute
\bs
\lim_{t\to 0}
\frac{1}{t}
\[
\int \(g_j\(x\) + t G_j\(x\)\) f\(x\) 
1\(h\(x\) + t H\(x\) > c\) d \mathcal{L}_{n} x
-
\int g_j\(x\) f\(x\) 
1\(h\(x\) > c\) d \mathcal{L}_{n} x
\].
\end{split}\end{align}
Because the main term that needs to be accounted for is
\bs
\lim_{t\to 0}
\frac{1}{t}
\[
\int g_j\(x\) f\(x\) 
1\(h\(x\) + t H\(x\) > c\) d \mathcal{L}_{n} x
-
\int g_j\(x\) f\(x\) 
1\(h\(x\) > c\) d \mathcal{L}_{n} x
\],
\end{split}\end{align}
we focus on the following simplified form of the sorting operator with
vector-valued $c$ and $h\(x\)$: 
\begin{align}\begin{split}\label{simplied form of sorting operator}
F\(h,c\) \coloneqq \int_{h\(x\) > c}f\(x\)d\mathcal{L}_{n}x =
\int 1
\(h\(x\) > c\)f\(x\)d\mathcal{L}_{n}x.
\end{split}\end{align}
Without loss of generality, the
$g_j\(x\) f\(x\)$ term in \eqref{sorting operator in general form} is replaced by $f\(x\)$ in
\eqref{simplied form of sorting operator} above.

First, in the special case of the sorting operator where
$h\(x\)$ is a scalar function, consider the derivative with respect to $c$.
Let $h$ be  $C^{1}$ or Lipschitz, with
$\Vert \nabla h\(x\)\Vert > 0,
\mathcal{L}_{n} \; a.e. \; x.$
Suppose also that $f$ is an integrable function.
By the coarea formula Theorem \ref{coarea formula classic},
\begin{align}\begin{split}\nonumber
\int_{h\(x\) > c}
f\(x\)d\mathcal{L}_{n}x =
\int^{\infty}_{c}\[
\int_{h^{-1}\(c'\)}\frac{f\(x\)}
{\left\Vert\nabla h\(x\)\right\Vert}
d\mathcal{H}_{n - 1}x\]d\mathcal{L}_{1} c'.
\end{split}\end{align}
Further, 
by the Lebesgue differentiation theorem, the conventional derivative is given by
\begin{align}\begin{split}\nonumber
\frac{d}{dc}\(\int_{h\(x\) > c}
f\(x\)d\mathcal{L}_{n}x\) =
- \int_{h^{-1}\(c\)}\frac{f\(x\)}
{\left\Vert\nabla h\(x\)\right\Vert}d\mathcal{H}_{n - 1}x
\end{split}\end{align}
for $\mathcal{L}_{1} \; a.e. \; c$.
The right hand side above is also a derivative in the distributional (weak) sense.
\begin{definition}\label{first Hadamard}
(Hadamard differentiability) 
Let $\mathcal X$ and $\mathcal Y$ be normed spaces 
with norms
$\Vert\cdot\Vert_{\mathcal X}$ and $\Vert\cdot\Vert_{\mathcal Y}$.
Consider the map $\phi: \mathcal X_{\phi} \subset \mathcal X \to \mathcal Y$.
Then $\phi$ is called Hadamard differentiable
  at $\theta \in \mathcal X_{\phi}$ tangentially to
  $\mathcal X_{0} \subset \mathcal X$
  if there is a continuous linear map
  $\phi'_{\theta}: \mathcal X_{0} \to \mathcal Y$ such that
  \begin{align}\begin{split}\nonumber
  \left\Vert
  \frac{\phi\(\theta+ t_{n} \mathcal \mathcal{x}_{n}\) -
  \phi\(\theta\)
  }{t_{n}} - \phi'_{\theta}\(x\)
  \right\Vert_{\mathcal Y} \to 0, \; \text{for all} \;
  t_{n} \to 0 \; \text{and} \; \mathcal \mathcal{x}_{n} \to x \in \mathcal{X}_{0},
  \end{split}\end{align}
  as $n \to \infty$,
  where 
  $\theta + t_{n} \mathcal \mathcal{x}_{n} \in \mathcal X_{\phi}$ for all $n$.
\end{definition}

\begin{remark}
When $\mathcal X_0$ is a linear subspace,
by  \cite{fang2019inference} Proposition 2.1,
  $t_{n} \to 0$ can be replaced by $t_{n} \downarrow 0$
  in Definition \ref{first Hadamard} without
loss of generality. 
\end{remark}

To verify 
Hadamard differentiability,
we calculate the derivative
of the sorting operator, 
\begin{align}\begin{split}\nonumber
\lim_{n \to \infty}
\frac{1}{t_{n}}\int
\[1\(h\(x\) +
t_{n} \mathcal \mathcal{H}_{n}\(x\) > c\) -
1\(h\(x\) > c\)\]
f\(x\)d\mathcal{L}_{n}x,
\end{split}\end{align}
where $\mathcal \mathcal{H}_{n} \to H$ and $t_{n} \downarrow 0$
as $n \to \infty$.
This limit is closely related to
the
Hausdorff measure $\mathcal{H}_{k}$ along the {\it surface}
$\{x: h\(x\) = c\}$, which
defines the 
{\it volume} of $k$-dimensional surfaces
in $n$-dimensional spaces, $k \in  \{1, \ldots, n\}$.
The important insight of
\cite{chernozhukov2018sorted}
is a representation by an integral under
$\mathcal{H}_{n - 1}$.

Let $E \subset \mathbb{R}^{n}$ be an open set,
$k \in \{1, 2, \ldots, n\}$,
$h: E \mapsto \mathbb{R}^{k}$ be a $C^{1}$ 
function, $f: E \mapsto \mathbb{R}$
be a continuous function.
Let $C\(E, \mathbb{R}^{k}\)$ denote the space
of $\mathbb{R}^{k}$-valued continuous functions on $E$ equipped with
the sup-norm. 
We need the following technical
assumptions about $h$ and $f$.

\begin{assump}\label{compact assumption}
The function $f$ is continuous with compact support $K_{f} \subset E$.
\end{assump}

\begin{definition}\label{critical regular}
(Critical and regular point and value)
Let $n, k \in \{1, 2, \ldots\}$,
  $E \subset \mathbb{R}^{n}$ be an
  open set, $f: E \mapsto \mathbb{R}^{k}$
  be a $C^{1}$ function.
  A point $x \in E$ is
  a critical point if
  $J f\(x\) = 0$;
  otherwise $x$ is a regular point.
  If
  $c = f\(x\)$
  for some critical point $x$, then
  $c$ is a critical value;
  otherwise $c$
  is a regular value.
  For tractability, we only consider regular values $c \in f\(E\)$.
\end{definition}

\begin{assump}\label{regular value assumption}
The support $K_f$ consists of only regular points of the function $h$.
\end{assump}

Assumption
\ref{compact assumption}
guarantees that $K_f \cap \{x: h\(x\) = c'\}$
is close to
$K_f \cap \{x: h\(x\) = c\}$ when $c'$
is close to $c$.
The purpose of the
Assumption \ref{regular value assumption} is to
guarantee that $h$ is regular in a neighborhood of $K_f$ in order to apply the implicit function theorem.

\begin{theorem}\label{general k derivative}
Let $h: E \to \mathbb{R}^k$ be a $C^1$ function, $k\in \{1,2,\ldots,n\}$, $E$ an open subset of $\mathbb{R}^n$.
Let Assumption
 \ref{compact assumption} 
and Assumption \ref{regular value assumption} hold.
Let $\mathcal D \subset h\(E\)$ be a compact subset.
Then,
the map $F\(h,	c\):
C\(E, \mathbb{R}^{k}\) \to \mathbb{R}$
in \eqref{simplied form of sorting operator}
is Hadamard differentiable uniformly with respect to $c \in \mathcal{D}$ at $h$ tangentially
to $C\(E, \mathbb{R}^{k}\)$. 
When $k > 1$, the derivative is given by
\begin{align}\begin{split}\label{derivative k>1}
F'_{h, c}\(H, 0\) =  \sum_i \int_{y > \tau_{\neg i}\(c\)} \[
\int_{h^{-1}\(c'\(y, c_{i}; i\)\)}
\frac{
H_i\(x\) f\(x\)
}{
J h\(x\)
}
d \mathcal{H}_{n-k} x
\] d \mathcal{L}_{k-1} y,
\end{split}\end{align}
where $\tau_{\neg i}\(c\)$ denotes the $\mathbb{R}^k \to \mathbb{R}^{k-1}$ coordinate projection except
the $i$-th coordinate, $c'\(y, c_{i}; i\)$ is the $k$-dimensional vector such that $\tau_{\neg i}\(c'\(y, c_{i}; i\)\) = y$ and $\tau_{i}\(c'\(y, c_{i}; i\)\) = c_i$.
When 
$k=1$,	i.e. when
$h$ is a scalar-valued function,
the derivative 
is given by
\begin{align}\begin{split}\label{theorem 2.1 derivative k=1}
F'_{h, c}\(H, 0\) =
\int_{h^{-1}\(c\)}
\frac{H\(x\)f\(x\)}{\left\Vert
\nabla h\(x\)\right\Vert}
d\mathcal{H}_{n - 1}x.
\end{split}\end{align}
\end{theorem}

Embedded in the proof of Theorem \ref{general k derivative} 
is the continuity of the derivative of the sorting operator,
expressed in the following Proposition \ref{surface integral continuity}, which
handles the general case when $k \geq 1$.

\begin{proposition}\label{surface integral continuity}
Under the conditions in Theorem \ref{general k derivative},
the density of the random variable $h\(X\)$
is continuous on $h\(E\)$ and is given by 
\begin{align}\begin{split}\nonumber
\int_{h^{-1}\(c\)}\frac{f\(x\)}
{J h\(x\)}d\mathcal{H}_{n - k}x.
\end{split}\end{align}
\end{proposition}

We next consider a variant of Theorem \ref{general k derivative} under Lipschitz conditions. 

\begin{assump}\label{borel compact assumption}
The function $f$ is Borel and bounded with compact support $K_{f} \subset E$.
\end{assump}

\begin{assump}\label{Clarke regular point assumption}
The Clarke Jacobian $J_{c}h\(x\)$ (defined in \ref{Clarke Jacobian} in the Technical addendum) is of full rank for all $x \in E$.
For all $x$, there exists a neighborhood $B_{x}$ and a constant $\epsilon\(x\) > 0$, such that for all $x' \in B_{x}$, if $\nabla h\(x'\)$ exists then $Jh\(x'\) > \epsilon\(x\)$.
\end{assump}

\begin{assump}\label{a.e. continuous at boundary assumption}
For all $c \in h\(E\)$, $\mathcal{H}_{n - k}\left\{x: h\(x\) = c, h\(\cdot\) \text{ is not differentiable at } x\right\} = 0$ and $\mathcal{H}_{n - k}\left\{x: h\(x\) = c, \exists x_{\(k\)}\(\left\vert\det\(\frac{\partial}{\partial x_{\(k\)}}h\(\cdot\)\)\right\vert = 0\ \text{or}\ \frac{f\(\cdot\)}{\left\vert\det\(\frac{\partial}{\partial x_{\(k\)}}h\(\cdot\)\)\right\vert} \text{ is discontinuous at } x\)\right\} = 0$,
where the subscript $\(k\)$ denotes $k$-combination of the $n$ coordinates.
\end{assump}



The next theorem \ref{nonsmooth general k derivative}  is a nonsmooth version of Theorem \ref{general k derivative}.
\begin{theorem}\label{nonsmooth general k derivative}
Let $h: E \to \mathbb{R}^{k}$ be Lipschitz, $k \in \{1, 2, \ldots, n\}$, $E$ an open subset of $\mathbb{R}^{n}$.
Let Assumption \ref{borel compact assumption} to 
Assumption \ref{a.e. continuous at boundary assumption} hold. Let $\mathcal{D} \subset h\(E\)$ be a compact subset.
Then the map $F\(h, c\): C\(E, \mathbb{R}^{k}\) \to \mathbb{R}$ is Hadamard differentiable uniformly on $c \in \mathcal{D}$ at $h$ tangentially to $C\(E, \mathbb{R}^{k}\)$.
The derivatives are given by \eqref{derivative k>1} and \eqref{theorem 2.1 derivative k=1} respectively 
when $k > 1$ and when $k = 1$.
\end{theorem}

The assumptions employed in Theorem \ref{nonsmooth general k derivative} play several roles.
Assumption \ref{borel compact assumption} provides compactness and 
Assumption \ref{Clarke regular point assumption} guarantees $h$ is regular.
The compactness Assumption
\ref{borel compact assumption} ensures that the set $\{x: h\(x\)=c\} \cap K_f$ 
in the area of integration is not too large and
changes continuously with $c$, meaning that
each open covering of $\{x: h\(x\)=c\} \cap K_f$ covers $\{x: h\(x\)=c'\} \cap K_f$ for all $c'$ close enough to $c$.
Assumption \ref{Clarke regular point assumption} requires each $x \in E$ to be a regular point of $h\(\cdot\)$.
Under this assumption,
$h\(\cdot\)$ can be expressed by the implicit function theorem to facilitate
the change of variables.

In addition, Assumption \ref{a.e. continuous at boundary assumption} requires each
level set $\{x: h\(x\)=c\}$ to be sufficiently regular in the sense of
having
a small number of nonsmooth points and the integrand  to be
almost everywhere continuous. Under this assumption, 
derivatives can be taken almost everywhere on the level set.
Assumption \ref{a.e. continuous at boundary assumption} can be replaced by continuous differentiability of $h$ and continuity of $f$.

Alternative assumptions are possible.
Consider Theorem \ref{general k derivative}.
Therein, Assumption \ref{compact assumption}, which requires compactness of the support of $f\(\cdot\)$, can be replaced by
the next Assumption \ref{compact assumption '} which requires that local union of level sets of $h\(\cdot\)$  is compact.
\begin{manualproposition}{\ref{compact assumption}$'$}\label{compact assumption '}
For all $c \in h\(E\)$, there exists a neighborhood $\mathcal N_c$ of $c$ such that
\bs
\overline{
\bigcup_{c' \in \mathcal N_c}
h^{-1}\(c'\)} \subset E
\end{split}\end{align}
and is bounded. The function $f$ is continuous.
\end{manualproposition}
In addition, when $k = 1$ and when we only need Hadamard differentiability on a singleton $\{c\}$, 
where $c \in h\(E\)$, Assumption \ref{regular value assumption} can be replaced by the following
\begin{manualproposition}{\ref{regular value assumption}$'$}\label{regular value assumption '}
$K_f \cap h^{-1}\(c\)$ consists of only regular points of $h$.
\end{manualproposition}

\section{Binary treatment allocation and the ROC curve}\label{constrained and roc}

In this section, we derive uniform asymptotic distributions
for a binary treatment allocation value function under a resource constraint.
Using Hadamard differentiability results
from section \ref{sorting operator}, we directly and concisely analyze a plug-in
two-step estimator of the value function.
The functional delta methods alone (see for example
\cite{van2023weak} and \cite{fang2019inference}) are sufficient to obtain asymptotic results.
An important result of this section
is showing 
consistency 
of a computationally feasible bootstrap procedure for the plug-in
two-step ROC estimator commonly used in computer science and other fields.

\subsection{Binary constrained allocation}\label{binary constrained case}

In population, binary optimal constrained treatment allocation takes the form of 
\bsnumber\label{constrained optimized value function}
\max_{\phi \in \Phi} \left\{\mathbb{E} \[Y_{1}\phi\(X\) + Y_{0}\(1 - \phi\(X\)\)\]:
\mathbb{E} \[Z_{1}\phi\(X\) + Z_{0}\(1 - \phi\(X\)\)\] \leq \alpha\right\}
\end{split}\end{align}
where $\Phi$ is the set of critical functions with $J = 1$.
By the 
Neyman-Pearson lemma (see for example Theorem 3.2.1 in \cite{lehmann2022testing}),
solutions of \eqref{constrained optimized value function} satisfy 
\bs
\phi\(x\) = \left\{
\ba{cc}
1, & \text{when}\ \(g_{1}\(x\) - g_{0}\(x\)\) > k\(c_{1}\(x\) - c_{0}\(x\)\), \\
0, & \text{when}\ \(g_{1}\(x\) - g_{0}\(x\)\) < k\(c_{1}\(x\) - c_{0}\(x\)\),
\ea
\right.
\end{split}\end{align}
where $k$ is a constant to be determined, $g_{j}\(x\) = \mathbb{E}\(Y_{j}\vert X = x\)$,
$c_{j}\(x\) = \mathbb{E}\(Z_{j}\vert X = x\)$, $j=0,1$.

The constrained population programming in \eqref{constrained optimized value function} uses observed covariates to allocate a binary treatment among a population when a budget constraint limits the size of the treated group. 
In this context, $Y_{1}$ and $Y_{0}$ are the welfare measures with and without treatment; $Z_{1}$ and $Z_{0}$ are the respective costs of treatment and control. 
For a typical economic problem, the cost of treatment is usually no less than the cost of control. 
For example, \cite{bhattacharya2012inferring} study the efficient provision of anti-malaria bed net subsidies using a set
of features $X$ that include the presence of a child under 10, the wealth
per capita and the ownership of a bank account.
Here, $Y_{1}$ and $Y_{0}$ are indexed by whether a household is covered by an insecticide-treated net; $Z_{1}$ is the cost of a redeemed net coupon and $Z_{0}$ is taken to be zero. 
In \cite{elliott2013predicting}, $\(Y_1, Y_0, Z_1, Z_0\)$ represent the utility of a true positive, true negative, false positive and false negative, respectively, in a classification problem. 
In \cite{luedtke2016optimal},  the resource constraint takes the form of an upper bound on the fraction of the treated population. Their setting is equivalent to $Z_1=1$ and $Z_0=0$. 
In this paper, we therefore assume $Z_{1} \geq Z_{0}$. 
This section first considers full observability of $\(Y_1, Y_0, Z_1, Z_0, X\)$.  

\begin{assump}\label{variable assumption}
The random variables $\(Y_{1}, Y_{0}, Z_{1}, Z_{0}, X\) \sim Q$, $X \sim \mu$, where $\mu$ is absolutely continuous with respect to the $\dim\(X\)$-dimensional Lebesgue measure.
The density $\mu'$ is continuous
on an open set $E \subset \mathbb{R}^{\dim\(X\)}$ and $\supp{\mu'} \coloneqq K_{\mu'} \subset E$.
$\(Y_{1}, Y_{0}, Z_{1}, Z_{0}\)$ is bounded and $Z_{1} \geq Z_{0}$.
\end{assump}

Following empirical process notation conventions, let $Qf = \int fdQ$. We then 
express the value and constrain functions of
\eqref{constrained optimized value function}  as
\bsnumber\label{functional form binary constrained allocation}
\beta\(k\) & = Q\[\(y_{1} - y_{0}\)1\(\(g_{1}\(x\) - g_{0}\(x\)\) > k\(c_{1}\(x\) - c_{0}\(x\)\)\) + y_{0}\],\\
s.t. \;
\alpha\(k\) & = Q\[\(z_{1} - z_{0}\)1\(\(g_{1}\(x\) - g_{0}\(x\)\) > k\(c_{1}\(x\) - c_{0}\(x\)\)\) + z_{0}\] \leq \alpha,
\end{split}\end{align}
where level set $\left\{x: \(g_{1}\(x\) - g_{0}\(x\)\) = k\(c_{1}\(x\) - c_{0}\(x\)\)\right\}$ is of zero probability mass.
The Hadamard differentiability results 
in section \ref{sorting operator} allow us to derive the
asymptotic distribution of an estimator for 
\eqref{functional form binary constrained allocation} by the functional delta method under the following assumptions.
We largely adopt the styles in \cite{chernozhukov2018sorted}.

\begin{assump}\label{regular Delta assumption}
The 
limits
of $\(\hat{g}_{j}\(\cdot\), \hat{c}_{j}\(\cdot\), j=0,1\)$
satisfy
$\(g^{\star}_{j}\(\cdot\), c^{\star}_{j}\(\cdot\), j=0,1\) \in {C^{1}\(E,\mathbb{R}^4\)}$, $c^{\star}_{1}\(\cdot\) > c^{\star}_{0}\(\cdot\)$, and $0$ is a regular value of
$\Delta^{\star}\(\cdot; k\) \coloneqq g^{\star}_{1}\(\cdot\) - g^{\star}_{0}\(\cdot\) - k\(c^{\star}_{1}\(\cdot\) - c^{\star}_{0}\(\cdot\)\)$
for all $k \in \mathcal{D}$ where $\mathcal{D} \subset \mathbb{R}$ is a bounded closed interval.
\end{assump}

\begin{assump}\label{joint convergence assumption}
Let $\(g^{\star}_{j}\(\cdot\), c^{\star}_{j}\(\cdot\), j=0,1\) \in \mathcal{G}$
where $\mathcal{G} \subset {C\(E, \mathbb{R}^{4}\)}$ is uniformly bounded.
With outer probability tends to $1$, $\(\hat{g}_{j}\(\cdot\), \hat{c}_{j}\(\cdot\), j=0,1\) \in \mathcal{G}$.
The function class
\bs
\mathcal{F} \coloneqq \left\{1\(\(g'_{1}\(\cdot\) - g'_{0}\(\cdot\)\) > k\(c'_{1}\(\cdot\) - c'_{0}\(\cdot\)\)\):
\(g'_{0}\(\cdot\), g'_{1}\(\cdot\), c'_{0}\(\cdot\), c'_{1}\(\cdot\)\) \in \mathcal{G}, 
k \in \mathcal{D} \right\}
\end{split}\end{align}
is Donsker. The estimators and the empirical process converge jointly, i.e.
\bs
r_{n}\[
\(\ba{c}
\hat{g}_{j}\(\cdot\)_{j=0,1} \\
\hat{c}_{j}\(\cdot\)_{j=0,1}
\ea
\)^T -
\(\ba{c}
{g}_{j}^\star\(\cdot\)_{j=0,1} \\
{c}_{j}^\star\(\cdot\)_{j=0,1}
\ea
\)^T, \mathbb{Q}_{n} - Q\] \rightsquigarrow
\(\mathbb{H}^T, t_{Q}\mathbb{Q}\) \text{ in } {C\(E, \mathbb{R}^{4}\)} \times \ell^{\infty}\(\mathcal{F}\),
\end{split}\end{align}
where $r_{n} \to \infty$ as $n \to \infty$,
$\lim_{n \to \infty}\frac{r_{n}}{\sqrt{n}} = t_{Q} \in \[0, 1\]$ and
$\mathbb{H} = \(\mathbb{H}_{g_0}, \mathbb{H}_{g_1}, \mathbb{H}_{c_0}, \mathbb{H}_{c_1}\)^{T}$
is separable.
\end{assump}

\begin{theorem}\label{binary allocation asymptotic result}
Let Assumptions \ref{variable assumption} to \ref{joint convergence assumption} hold.
There exists a constant $\epsilon > 0$ such that 
\bsnumber\label{nonsingular manifold jacobian denominator in binary classification roc curve}
\inf_{k \in \mathcal{D}}
\int_{\Delta^{\star}\(x;k\)=0} \frac{\(c_{1}\(x\) -
c_{0}\(x\)\)\(c^{\star}_{1}\(x\) -
c^{\star}_{0}\(x\)\)\mu'\(x\)}{\Vert
\nabla\Delta^{\star}\(x;k\)\Vert}d\mathcal{H}_{n - 1}x >
\epsilon.
\end{split}\end{align}
Let $\Delta$ ($\hat{\Delta}, \Delta^{\star}$) to denote
$\(g_{j}\(\cdot\), c_{j}\(\cdot\),j=0,1\)$
(their estimators and the estimators' limits).	Also
let $k\(Q, \Delta, \alpha\)  = \inf\left\{k \in \mathbb{R}:
\alpha\(Q, \Delta, k\) \leq \alpha \right\}$ where
\bsnumber\label{double robust representing of constrained optimized value function under full observability}
&\beta\(Q, \Delta, \alpha\) \coloneqq \beta\(Q, \Delta, k\(Q, \Delta, \alpha\)\);\
\beta\(Q, \Delta, k\) \coloneqq Q\[\(y_{1} - y_{0}\)1\(\Delta\(x; k\) > 0\)+y_0\],\\
&\alpha\(Q, \Delta, k\) \coloneqq 
Q\[\(z_{1} - z_{0}\)1\(\Delta\(x; k\) > 0\)+z_0\].
\end{split}\end{align}
If we also define
\bs
f_{\beta}\(Q, \Delta^{\star}, k\) & \coloneqq
\frac{d\beta\(Q, \Delta^{\star}, k\)}{dk} = -\int_{\Delta^{\star}\(x; k\) = 0}
\frac{\(g_{1}\(x\) - g_{0}\(x\)\)\(c^{\star}_{1}\(x\) - c^{\star}_{0}\(x\)\)\mu'\(x\)}{\Vert\nabla\Delta^{\star}\(x; k\)\Vert}d\mathcal{H}_{n - 1}x, \\
f_{\alpha}\(Q, \Delta^{\star}, k\) & \coloneqq
\frac{d\alpha\(Q, \Delta^{\star}, k\)}{dk} = -\int_{\Delta^{\star}\(x; k\) = 0}
\frac{\(c_{1}\(x\) - c_{0}\(x\)\)\(c^{\star}_{1}\(x\) -
c^{\star}_{0}\(x\)\)\mu'\(x\)}{\Vert\nabla\Delta^{\star}\(x;
k\)\Vert}d\mathcal{H}_{n - 1}x,
\end{split}\end{align}
then as a process on compact $\Lambda'
\subsetneq \overset{\circ}{\Lambda}$ where
$\Lambda \coloneqq \biggl\{Q\[\(z_{1} -
z_{0}\)1\(\Delta^{\star}\(x; k\) > 0\) + z_0\], k\in \mathcal D
\biggr\},$
\begin{align}
& r_{n}\(\beta\(\mathbb{Q}_{n}, \hat{\Delta}, \alpha\) - \beta\(Q, \Delta^{\star}, \alpha\)\) \label{main asymptotic formula}\\
&\rightsquigarrow
t_{Q}\mathbb{Q}\biggl[\(y_{1} - y_{0} -\lambda\(Q, \Delta^\star, \alpha\) \(z_1 - z_0\)\)1\(\Delta^{\star}\(x; k\(Q, \Delta^{\star}, \alpha\)\) > 0\) 
+ y_{0} - \lambda\(Q, \Delta^\star, \alpha\) z_0 \biggr] \notag \\
& 
+ 
\int_{\Delta^{\star}\(x; k\(Q, \Delta^{\star}, \alpha\)\) = 0}
\frac{\Delta\(x; \lambda\(Q, \Delta^{\star}, \alpha\)\)\(\mathbb{H}_{g1} - \mathbb{H}_{g0}
-
k\(Q, \Delta^{\star}, \alpha\)
\(
\mathbb{H}_{c1} - \mathbb{H}_{c0}
\)
\)\mu'\(x\)}{\Vert\nabla\Delta^{\star}\(x; k\(Q, \Delta^{\star}, \alpha\)\)\Vert}
\notag
\end{align}
In the above, $\lambda\(Q, \Delta^\star, \alpha\) = \frac{f_{\beta}\(Q, \Delta^{\star}, k\(Q, \Delta^{\star}, \alpha\)\)}{f_{\alpha}\(Q, \Delta^{\star}, k\(Q, \Delta^{\star}, \alpha\)\)}$ and $\Delta\(x; \lambda\) = g_{1}\(x\) - g_{0}\(x\) - \lambda\(c_1\(x\) - c_0\(x\)\)$.
\end{theorem}
Theorem \ref{binary allocation asymptotic result} 
simplifies under correct specification when $\Delta^{\star}\(\cdot\) = \Delta\(\cdot\)$. 
In this case,
$k = \frac{f_{\beta}\(Q, \Delta^{\star}, k\)}{f_{\alpha}\(Q,
\Delta^{\star},k\)}$,
since on $\Delta^{\star}\(x; k\) = 0$,
$\Delta\(x; k\) =
\(g_{1}\(x\) - g_{0}\(x\)
-
k 
\(c_{1}\(x\) - c_{0}\(x\)\)
\)
=0$,
\bs
0 & = f_{\beta}\(Q, \Delta^{\star}, k\)
-k f_{\alpha}\(Q, \Delta^{\star}, k\) \\
& =
-\int_{\Delta^{\star}\(x; k\) = 0}
\frac{\(g_{1}\(x\) - g_{0}\(x\)
-
k 
\(c_{1}\(x\) - c_{0}\(x\)\)
\)\(c^{\star}_{1}\(x\) - c^{\star}_{0}\(x\)\)\mu'\(x\)}{\Vert\nabla\Delta^{\star}\(x; k\)\Vert}d\mathcal{H}_{n - 1}x.
\end{split}\end{align}
Then the asymptotic distribution in
\eqref{main asymptotic formula}
simplifies to a Gaussian process in $\alpha$: 
\bsnumber\label{main asymptotic formula under correct specification}
r_{n}\(\beta\(\mathbb{Q}_{n}, \hat{\Delta}, \alpha\) - \beta\(Q, \Delta^{\star}, \alpha\)\) \rightsquigarrow
t_{Q}\mathbb{Q} & \biggl[\(y_{1} - y_{0}\)1\(\Delta^{\star}\(x; k\(Q, \Delta^{\star}, \alpha\)\) > 0\) + y_{0} \\
-k\(Q, \Delta^{\star}, \alpha\) & \(
\(z_{1} - z_{0}\)1\(\Delta^{\star}\(x; k\(Q, \Delta^{\star}, \alpha\) > 0\)\) + z_{0}\)\biggr].
\end{split}\end{align}
In parametric models where $\Delta\(x\) = \Delta\(x; \theta\)$, $r_n = \sqrt{n}$ and $t_Q = 1$.
The covariance kernel for 
\eqref{main asymptotic formula under correct specification}
over $\(\alpha, \alpha'\)$
can be consistently estimated by its sample analog.
Appendix \ref{supplementary discussion on misspecified parametric models} analyzes the implications
of parametric model misspecification. 

An alternative model-based representation of
\eqref{constrained optimized value function}
is
\bsnumber\label{model based constrained optimized value function}
\max_{\phi \in \Phi} \left\{\mathbb{E} \[g_{1}\(X\) \phi\(X\) + g_{0}\(X\)\(1 - \phi\(X\)\)\]: \mathbb{E} \[c_{1}\(X\)\phi\(X\) + c_{0}\(X\)\(1 - \phi\(X\)\)\] \leq \alpha\right\}.
\end{split}\end{align}
Corresponding to \eqref{model based constrained optimized value function},
the measure $Q$ for $\(Y_{1}, Y_{0}, Z_{1}, Z_{0}, X\)$ in the functional expression of \eqref{double robust representing of constrained optimized value function under full observability} can be changed to the measure $\mu$ for $X$:
\bs
\alpha\(\mu, \Delta, k\) & \coloneqq 
\mu\[\(c_{1}\(x\) - c_{0}\(x\)\)1\(\Delta\(x; k\) >
0\)+c_0\(x\)\], \\
k\(\mu, \Delta, \alpha\) & \coloneqq   \inf\left\{k \in
\mathbb{R}: \alpha\(\mu, \Delta, k\) \leq \alpha \right\}, 
\\
\beta\(\mu, \Delta, \alpha\) & \coloneqq \beta\(\mu, \Delta, k\(\mu, \Delta, \alpha\)\),
\beta\(\mu, \Delta, k\)  \coloneqq
\mu\[\(g_{1}\(x\) - g_{0}\(x\)\)1\(\Delta\(x; k\) > 0\)+g_0\(x\)\].
\end{split}\end{align}
When the conditional data generating process of $y_1, y_0, z_1, z_0$ given $x$ is a 
parametric model
$f\(y_1, y_0, z_1, z_0 \vert x; \theta\)$, we denote $\hat \Delta\(x\) = \Delta\(x; \hat\theta\)$.
A plug-in estimator of 
\eqref{model based constrained optimized value function} is then
\bsnumber\label{estimate of model based representation of constrained optimized value function under full observability}
\alpha\(\hat{\mu}, \hat{\theta}, k\) & \coloneqq \alpha\(\hat{\mu}, \hat{\Delta}, k\),
\beta\(\hat{\mu}, \hat{\theta}, k\)  \coloneqq \beta\(\hat{\mu}, \hat{\Delta}, k\),	
\\
\beta\(\hat\mu, \hat\theta, \alpha\) & \coloneqq \beta\(\hat\mu, \hat\Delta, k\(\hat\mu, \hat\theta, \alpha\)\), \; 
k\(\hat\mu, \hat\theta, \alpha\)  \coloneqq \inf\left\{k \in \mathbb{R}: \alpha\(\hat\mu, \hat\theta, k\) \leq \alpha \right\}.
\end{split}\end{align}
We will call $\beta\(\mathbb{Q}_{n}, \hat{\Delta}, \alpha\)$ in \eqref{main asymptotic formula}
an \textit{orthogonal} estimator since its first order derivative cancels out under correct specification,
and $\beta\(\hat\mu, \hat\theta, \alpha\)$ in \eqref{estimate of model based representation of constrained optimized value function under full observability} a parametric model-based estimator.

Typically, $\hat\theta$ is the maximum likelihood estimator: 
$\hat\theta = \mathop{\arg\max}\limits_{\theta\in\Theta} \sum_{i=1}^n \log f\(y_{i1}, y_{i0}, z_{i1}, z_{i0} \vert x; \theta\).$
Let $\omega = \{y_1, y_0, z_1, z_0\}$ and denote the score function $s_{\theta}\(\omega \vert x; \theta\) =
\frac{\partial}{\partial \theta} \log f\(\omega \vert x, \theta\)$.
It is a standard result from maximum likelihood estimation that
\bsnumber\label{mle influence function}
\sqrt{n}\(\hat\theta - \theta^{\star}\) = - H\(\theta^{\star}\)^{-1} \frac{1}{\sqrt{n}} \sum_{i=1}^n s_{\theta}\(\omega_i \vert x_i; \theta^{\star}\)
+ o_{\mathbb{P}}\(1\),\ H\(\theta\) = Q\[\frac{\partial^2}{\partial \theta^2} \log f\(\omega \vert x, \theta\)\].
\end{split}\end{align}
See for example \cite{newey_mcfadden} section 3 or \cite{van2000asymptotic} chapter 5 for sufficient regularity conditions
for the asymptotic linear representation.

Under correct specification, the parametric model-based estimator has an efficiency advantage over the orthogonal estimator.
We have the following efficiency comparison proposition.

\begin{proposition} \label{parametric model based estimator more efficient than orthogonal estimator}
Assume that \eqref{mle influence function} holds.  Under the conditions in Theorem \ref{binary allocation asymptotic result},
when the parametric model is correctly specified,
the model based estimator $\beta\(\hat\mu, \hat{\theta}, \alpha\)$ in
\eqref{estimate of model based representation of constrained optimized value function under full observability}
is asymptotically more efficient than
the orthogonal estimator in the parametric case of
\eqref{main asymptotic formula under correct specification}.
\end{proposition}


The population value function $\beta\(\mu, \Delta, \alpha\)$ of the constrained optimization problem in \eqref{constrained optimized value function} 
satisfies concavity 
and monotonicity shape restrictions in the arguments as in the general framework developed by \cite{chen2021shape}. 
The orthogonal estimator $\beta\(\mathbb{Q}_{n}, \hat{\Delta}, \alpha\)$ is only monotonically increasing but not necessarily concave. 
When shape restrictions like concavity and monotonicity are applicable, 
shape-enforcing operators proposed by 
\cite{chen2021shape}  
can be used to improve estimation and inference \textit{ex post}. In particular, 
\cite{chen2021shape} demonstrate that an estimator transformed by the shape-enforcing operator is more accurate, and
confidence intervals transformed by the shape-enforcing operator have greater coverage and shorter length.

The double Legendre-Fenchel transform and the ``pooled adjacent violators algorithm'' as described in \cite{chen2021shape}
can be used to enforce concavity of the orthogonal estimator $\beta\(\mathbb{Q}_{n}, \hat{\Delta}, \alpha\)$.
An alternative approach of concavification is via the Lagrangian dual.
In particular, the model based estimator $\beta\(\hat\mu, \hat{\theta}, \alpha\)$ is equivalent to the Lagrangian dual 
\begin{align}\begin{split}\nonumber
\mathcal{L}\(\hat{\mu}, \hat{\theta}, \alpha\) = \min_{k \geq 0}\hat{\mu}\[\(\hat{g}_{1}\(x\) - \hat{g}_{0}\(x\)\)1\(\hat{\Delta}\(x; k\) > 0\) - k\(\(\hat{c}_{1}\(x\) - \hat{c}_{0}\(x\)\)1\(\hat{\Delta}\(x; k\) > 0\) - \alpha\)\], 
\end{split}\end{align}
where $\hat{g}_{j}\(\cdot\) = g_{j}\(\cdot, \hat{\theta}\)$, $\hat{c}_{j}\(\cdot\) = c_{j}\(\cdot, \hat{\theta}\)$ and $j = 0, 1$. 
Since the Lagrangian dual minimization problem is affine with respect to $\alpha$, $\beta\(\hat\mu, \hat{\theta}, \alpha\)$ 
can be considered an implicit concavification of $\beta\(\mathbb{Q}_{n}, \hat{\Delta}, \alpha\)$. 
Both the double Legendre-Fenchel transform and the Lagrangian dual $\mathcal{L}\(\hat{\mu}, \hat{\theta}, \alpha\)$ can be generalized to the problem of multi-constraints. 

Under misspecification, the asymptotic distribution of the parametric orthogonal estimator
$\beta\(\mathbb{Q}_{n}, \hat{\theta}, \alpha\)$ 
is described in Appendix \ref{supplementary discussion on misspecified parametric models},
where the integrals with respect to Hausdorff measure in \eqref{main asymptotic formula} do not vanish in general.

In causal inference,
$Y_j, Z_j, j \in \{0,1\}$ are not fully observed.
Only $Y=D Y_1 + \(1-D\)Y_0$ and $Z=D Z_1 + \(1-D\)Z_0$ are observed,
where $D \in \{0,1\}$ and unconfoundedness holds:
\bsnumber\label{binary unconfoundedness assumption}
\(Y_1, Y_0, Z_{1}, Z_{0}\) \indep D \vert X.
\end{split}\end{align}
The unconfoundedness assumption \eqref{binary unconfoundedness assumption} allows the model based constrained value function
\eqref{model based constrained optimized value function}
to be identified using
$g_1\(x\) = \mathbb{E}\(Y \vert D=1, x\)$,
$g_0\(x\) = \mathbb{E}\(Y \vert D=0, x\)$,
$c_1\(x\) = \mathbb{E}\(Z \vert D=1, x\)$,
$c_0\(x\) = \mathbb{E}\(Z \vert D=0, x\)$.
When the conditional data generating processes are parametric models of $f\(y, z \vert x, d; \theta\)$, denoted as
\bsnumber\label{joint conditional outcome and cost distribution 1}
f_1\(y, z \vert x; \theta\) \coloneqq
f\(y, z \vert d=1, x; \theta\),\quad  f_0\(y, z \vert x; \theta\) \coloneqq f\(y, z \vert d=0, x; \theta\),
\end{split}\end{align}
a parametric model based estimator is given in 
\eqref{estimate of model based representation of constrained optimized value function under full observability}
such that $f\(y_1 \vert x; \theta\)$, $f\(z_1 \vert x; \theta\)$, $f\(y_0 \vert x; \theta\)$
and $f\(z_0 \vert x; \theta\)$ are the marginal densities implied by
\eqref{joint conditional outcome and cost distribution 1},
and such that $\hat\theta$ is the conditional MLE:
\bs
\hat\theta = \mathop{\arg\max}\limits_{\theta\in\Theta} \sum_{i=1}^n
\log f\(y_{i}, z_{i} \vert d_i, x_i; \theta\)
=\sum_{i=1}^n \[d_i \log f_1\(y_{i}, z_{i} \vert x_i; \theta\)
+ \(1-d_i\) \log f_0\(y_{i}, z_{i} \vert x_i; \theta\)\].
\end{split}\end{align}
Under standard regularities, the asymptotic linear representation as in \eqref{mle influence function} continues to hold with
\bs
s_\theta\(y_i, z_i \vert d_i, x_i; \theta\) & 
 = d_i \frac{\partial}{\partial\theta}	\log f_1\(y_{i}, z_{i} \vert x_i; \theta\) + \(1-d_i\) \frac{\partial}{\partial\theta}	\log f_0\(y_{i}, z_{i} \vert x_i; \theta\), \\
H\(\theta\) & 
= Q\(
d \frac{\partial^{2}}{\partial\theta^{2}}  \log f_1\(y, z \vert x; \theta\) + \(1-d\) \frac{\partial^{2}}{\partial\theta^{2}}  \log f_0\(y, z \vert x; \theta\)\).
\end{split}\end{align}

In addition to the model based estimator \eqref{estimate of model based representation of constrained optimized value function under full observability},
an orthogonal doubly robust value function estimator under
unconfoundedness can be formed by re-expressing 
\eqref{double robust representing of constrained optimized value function under full observability}
as
\bsnumber\label{double robust representing of constrained optimized value function under partial observability}
\beta\(Q, \Delta, k\) & \coloneqq   Q\biggl[
\Gamma_{1}\(x, y, d\)
1\(\Delta\(x; k\) > 0\) 
+ 
\Gamma_{0}\(x, y, 1 - d\)
1\(\Delta\(x; k\) \leq 0\)
\biggr], \\
\alpha\(Q, \Delta, k\) & \coloneqq 
Q\biggl[
\Xi_{1}\(x, z, d\)
1\(\Delta\(x; k\) > 0\) 
+ 
\Xi_{0}\(x, z, 1 - d\)
1\(\Delta\(x; k\) \leq 0\)
\biggr],\\
\beta\(Q, \Delta, \alpha\) & \coloneqq	\beta\(Q, \Delta, k\(Q, \Delta, \alpha\)\), \ k\(Q, \Delta, \alpha\)  = \inf\left\{k \in \mathbb{R}: \alpha\(Q, \Delta, k\) \leq \alpha \right\}. 
\end{split}\end{align}
In the above, 
for $j \in \{0, 1\}$,  $p_{1}\(\cdot\) \coloneqq p\(\cdot\)$ and $p_{0}\(\cdot\) \coloneqq 1 - p\(\cdot\)$, we denote 
$\Gamma_{j}\(x, y, d\) \coloneqq g_{j}\(x\) + \frac{d}{p_{j}\(x\)}\(y - g_{j}\(x\)\)$ and $\Xi_{j}\(x, z, d\) = c_{j}\(x\) + \frac{d}{p_{j}\(x\)}\(z - c_{j}\(x\)\)$. 
Under correct specification, and assume that $\epsilon < p\(\cdot\) < 1-\epsilon$, for some constant $\epsilon > 0$, 
an adaption to Theorem \ref{binary allocation asymptotic result} shows that
\bs
\sqrt{n}\(\beta\(\mathbb{Q}_{n}, \hat{\Delta}, \alpha\) - \beta\(Q, \Delta^{\star}, \alpha\)\) \rightsquigarrow
\mathbb{Q}\(\zeta\(y, z, d, x; \alpha\)\) 
\end{split}\end{align}
where for $\phi^{\star}\(x; \alpha\) \coloneqq 1\(\Delta^{\star}\(x; k\(Q, \Delta^{\star}, \alpha\)\) > 0\)$, we write 
\bs
\zeta\(y,z,d,x; \alpha\) & \coloneqq 
\Gamma_{1}\(x, y, d\)
\phi^{\star}\(x; \alpha\)
+
\Gamma_{0}\(x, y, 1 - d\)
\(1 - \phi^{\star}\(x; \alpha\)\) \\ 
& - k\(Q, \Delta^{\star}, \alpha\)
\biggl(
\Xi_{1}\(x, z, d\)
\phi^{\star}\(x; \alpha\)
+ 
\Xi_{0}\(x, z, 1 - d\)
\(1 - \phi^{\star}\(x; \alpha\)\)
\biggr). 
\end{split}\end{align}

\begin{proposition}\label{treatment parametric model based estimator more efficient than orthogonal estimator}
Assume that $\epsilon < p\(\cdot\) < 1-\epsilon$ for $\epsilon > 0$ and that
both \eqref{mle influence function} and \eqref{binary unconfoundedness assumption}
hold.
Under the conditions in Theorem \ref{binary allocation asymptotic result},
when the parametric models $f_{1}\(y, z\vert x; \theta\)$ and $f_{0}\(y, z\vert x; \theta\)$ are correctly specified,
the model based estimator 
\eqref{estimate of model based representation of constrained optimized value function under full observability}
is asymptotically more efficient than
the parametric variant of the orthogonal estimator \eqref{double robust representing of constrained optimized value function under partial observability}
where $\hat \Delta\(x\) = \Delta\(x; \hat\theta\)$.
\end{proposition}

\subsection{The ROC curve}\label{theory for plug-in ROC curve}
A special case of \eqref{constrained optimized value function} is
when $Y_{1}=Y \in \{0, 1\}$, $Y_{0} = Z_{0} = 0$, $Z_{1} = 1 - Y_{1}$,
$g_{0}\(x\) = c_{0}\(x\) \equiv 0$, $g_{1}\(x\) = p\(x\)$, $c_{1}\(x\) = 1 - p\(x\)$ where $p\(x\) = \mathbb{E}\(Y \vert\ X = x\)$.
This special case, i.e.
\bs
\max_{\phi \in \Phi}\left\{\mathbb{E} \[Y \phi\(X\)\]: \mathbb{E} \[\(1 - Y\)\phi\(X\)\] \leq \alpha\right\},
\end{split}\end{align}
is closely related to the population receiver operating characteristic (ROC) curve defined by
\bsnumber\label{roc curve raw definition}
\beta\(\alpha\) \coloneqq \max_{\phi \in \Phi} \frac{\mathbb{E} \[Y\phi\(X\)\]}{\mathbb{E}Y} \quad \text{such that} \quad
\frac{\mathbb{E} \[\(1 - Y\)\phi\(X\)\]}{\mathbb{E}\(1 - Y\)} \leq \alpha.
\end{split}\end{align}
The objective in \eqref{roc curve raw definition} is to maximize the power $\beta$ 
while controlling the size of the test $\alpha$. 
This optimization problem is exactly the same as that of the \textit{Uniformly Most Powerful Tests} under simple null and alternative hypotheses. 
See for example section 3 of \cite{lehmann2022testing} for rigorous definitions and implications of these statistical testing concepts. 
In this paper, we adopt the more commonly employed terminology in the computer science, biostatistics and economics literatures to refer to $\beta$ and $\alpha$ as \textit{true positive rate} (TPR) and \textit{false positive rate} (FPR). 
In \cite{berge2011evaluating}, $Y_1 = Y$ is interpreted as the state of the economy ($Y = 0$ for expansion and $Y = 1$ for recession).
The ROC curve is also used in the finance loan literature 
\citep{vallee2019marketplace,sadhwani2021deep} to evaluate the tradeoff between earned interest and lost principal, where 
$Y=0$ indicates loan fulfillment and $Y=1$ indicates loan default.
The abnormal birth detection application of \cite{feng2025statisticalMS} uses $Y$ as the indicator of abnormal birth ($Y = 0$ for normal birth 
and $Y = 1$ for abnormal birth). 

The analog of \eqref{functional form binary constrained allocation} for the
constraint and value functions of \eqref{roc curve raw definition} is
\bsnumber\label{functional form roc curve}
\alpha\(k\)  = Q\[\(1-y\)1\(
p\(x\) > k 
\) \] / Q\(1-y\), \ 
\beta\(k\)  = Q\[y1\(
p\(x\) > k 
\) \] / Q\(y\).
\end{split}\end{align}
In order to adapt Theorem \ref{binary allocation asymptotic result} to the ROC curve in \eqref{roc curve raw definition} and \eqref{functional form roc curve}, 
define	similarly
to \eqref{double robust representing of constrained optimized value function under full observability},
\bsnumber\label{double robust representing of roc constrained optimized value function under full observability}
&\beta\(Q, p, \alpha\) \coloneqq \beta\(Q, p, k\(Q, p, \alpha\)\),
\beta\(Q, p, k\) = Q\[y1\(p\(x\) > k\)\] / Q\(y\),\\
&\alpha\(Q, p, k\) \coloneqq 
Q\[\(1-y\)1\(p\(x\) > k\)\] / Q\(1-y\),
\; k\(Q, p, \alpha\) = \inf\left\{k \in \mathbb{R}:
\alpha\(Q, p, k\) \leq \alpha \right\}.
\end{split}\end{align}
Different from \eqref{double robust representing of constrained optimized value function under full observability},
\eqref{double robust representing of roc constrained optimized value function under full observability} has
denominators
and $\Delta\(x; k\)$ is replaced by $p\(x\) - k$.
We also define
\bsnumber\label{beta alpha derivative wrt k}
f_{\beta}\(Q, p^{\star}, k\) & \coloneqq
 \frac{d\beta\(Q, p^{\star}, k\)}{dk} = - \frac{1}{Q y}
\int_{p^{\star}\(x\) = k}
\frac{p\(x\) 
\mu'\(x\)}{\Vert\nabla p^{\star}\(x\)\Vert}d\mathcal{H}_{n - 1}x, \\
f_{\alpha}\(Q, p^{\star}, k\) & \coloneqq
\frac{d\alpha\(Q, p^{\star}, k\)}{dk} = - \frac{1}{1- Qy}\int_{p^{\star}\(x\) = k}
\frac{\(1 - p\(x\)\)\mu'\(x\)}{\Vert\nabla p^{\star}\(x\)\Vert}d\mathcal{H}_{n - 1}x.
\end{split}\end{align}
The following
corollary 
adapts
Theorem \ref{binary allocation asymptotic result} to the commonly used orthogonal estimator of \eqref{roc curve raw definition}. 

\begin{cor}\label{roc asymptotic result}
Let the conditions in Theorem \ref{binary allocation asymptotic result} hold with $\(g_{0}\(x\), g_{1}\(x\), c_{0}\(x\), c_{1}\(x\)\) = \(0, p\(x\), 0, 1 - p\(x\)\)$ and $\mathbb{H}_p\(\cdot\)=\mathbb{H}_{g1}\(\cdot\)$.
On compact $\Lambda' \subsetneq \overset{\circ}{\Lambda}$, $\Lambda \coloneqq \left\{\frac{Q\[\(1 - y\)1\(p^{\star}\(x\) > k\)\]}{Q\(1 - y\)}; k \in \mathcal D\right\}$,
\begin{align}
& r_{n}\(\beta\(\mathbb{Q}_{n}, \hat{p}, \alpha\) - \beta\(Q, p^{\star}, \alpha\)\) \rightsquigarrow \frac{1}{Q y}\Biggl\{
t_{Q}\mathbb{Q}\[y\(1\(p^{\star}\(x\) > k\(Q, p^{\star}, \alpha\)\) - \beta\(Q, p^{\star}, \alpha\)\)\] \notag \\
& + 
\int_{p^{\star}\(x\)= k\(Q, p^{\star}, \alpha\)} \frac{p\(x\)\mathbb{H}_p\(x\)\mu'\(x\)}{\Vert\nabla p^{\star}\(x\)\Vert}d\mathcal{H}_{n - 1}x
\Biggl\} - \frac{1}{1-Q y} \frac{f_{\beta}\(Q, p^{\star}, k\(Q, p^{\star}, \alpha\)\)}{f_{\alpha}\(Q, p^{\star}, k\(Q, p^{\star}, \alpha\)\)} \Biggl\{ \label{roc main asymptotic distribution} \\
& t_{Q}\mathbb{Q}\[\(1 - y\)\(1\(p^{\star}\(x\) > k\(Q, p^{\star}, \alpha\)\) - \alpha\)\] + 
\int_{p^{\star}\(x\)=k\(Q, p^{\star}, \alpha\)} \frac{\(1-p\(x\)\)\mathbb{H}_p\(x\)\mu'\(x\)}{\Vert\nabla p^{\star}\(x\)\Vert}d\mathcal{H}_{n - 1}x
\Biggl\}. \notag
\end{align}
\end{cor}

Under correct specification, $p^\star\(\cdot\) = p\(\cdot\)$ and 
$\frac{
k
}{
1- k
} = \frac{Q y}{1-  Qy}
\frac{f_{\beta}\(Q, p^{\star}, k\)}{f_{\alpha}\(Q, p^{\star}, k\)}$.
This is
because, as  $\(1-k\) p\(x\) - k \(1-p\(x\)\) = 0$
on $p^\star\(x\) = p\(x\)= k$, 
\bs
k \(1-Qy\) f_{\alpha}\(Q, p^{\star}, k\)-
\(1-k\) \(Q y\) f_{\beta}\(Q, p^{\star}, k\)
= \int_{p^{\star}\(x\)=k}
\frac{
0 \cdot \mu'\(x\)}{\Vert\nabla p^{\star}\(x\)\Vert}d\mathcal{H}_{n - 1}x=0.
\end{split}\end{align}
Furthermore, the two integrals on the level set in Corollary \ref{roc asymptotic result} cancel out since 
\bs
&
\frac{1}{Q y}
\int_{p^{\star}\(x\)= k} \[
\frac{p\(x\)\mathbb{H}_{p}\(x\)\mu'\(x\)}{\Vert\nabla p^{\star}\(x\)\Vert}
- \frac{Qy}{1-Q y} \frac{f_{\beta}\(Q, p^{\star}, k\)}{f_{\alpha}\(Q, p^{\star}, k\)}
\frac{\(1-p\(x\)\)\mathbb{H}_{p}\(x\)\mu'\(x\)}{\Vert\nabla p^{\star}\(x\)\Vert}
\]d\mathcal{H}_{n - 1}x \\
&=\frac{1}{Q y}
\int_{p^{\star}\(x\)= k}\underbrace{\[
p\(x\)
- \frac{k}{1-k}
\(1-p\(x\)\)
\]}_{= 0 \; \text{on} \; p^\star\(x\) = p\(x\) = k}
\frac{\mathbb{H}\(x\)\mu'\(x\)}
{\Vert\nabla p^{\star}\(x\)\Vert}
d\mathcal{H}_{n - 1}x = 0.
\end{split}\end{align}
Under correct specification \eqref{roc main asymptotic distribution}
simplifies to 
\begin{align}
&r_{n}\(\beta\(\mathbb{Q}_{n}, \hat{p}, \alpha\) - \beta\(Q, p, \alpha\)\)
\label{roc curve asymptotics under correct specification}
\\
& \rightsquigarrow  \frac{t_{Q}}{Qy}\Biggl\{\mathbb{Q}\biggl[
\frac{y - k\(Q, p, \alpha\)}{1 - k\(Q, p, \alpha\)}
1\(p\(x\) > k\(Q, p, \alpha\)\) - y\beta\(Q, p, \alpha\)
+ \frac{k\(Q, p, \alpha\)}{1 - k\(Q, p, \alpha\)} \(1 - y\)\alpha
\biggr]
\Biggl\}. 
\notag
\end{align}
In parametric models of $p\(x\)$, where $r_n = \sqrt{n}$ and $t_Q = 1$,
\eqref{roc curve asymptotics under correct specification}
has a Gaussian 
covariance kernel 
\bsnumber\label{theoretical distribution of correctly specified ROC model}
\(\alpha, \alpha'\) \mapsto
\frac{1}{\(Q y\)^2} & \biggl(Q y \(\beta\(Q, p, \alpha \wedge \alpha'\)
-
\beta\(Q, p, \alpha\)\beta\(Q, p, \alpha'\)\) \\
&
+ \frac{
k\(Q, p, \alpha\)k\(Q, p, \alpha'\)
}{
\(
1- k\(Q, p, \alpha\)
\)\(
1- k\(Q, p, \alpha'\)
\)
} \(1- Qy\) \(\alpha \wedge \alpha' - \alpha \alpha'\)\biggr).\\
\end{split}\end{align}
Similar to \eqref{model based constrained optimized value function}
and \eqref{estimate of model based representation of constrained optimized value function under full observability},
a model based representation of the ROC curve \eqref{roc curve raw definition} is
\bs
\beta\(\alpha\) \coloneqq \max_{\phi \in \Phi} \frac{\mathbb{E} \[p\(X\)\phi\(X\)\]}{\mathbb{E} p\(X\)}\quad \text{such that}\quad \frac{\mathbb{E} \[\(1 - p\(X\)\)\phi\(X\)\]}{\mathbb{E}\(1 - p\(X\)\)} \leq \alpha.
\end{split}\end{align}
The corresponding estimator is
\begin{align}
&\alpha\(\hat{\mu}, \hat{\theta}, k\) \coloneqq 
\hat{\mu}\[\(1 - p\(x; \hat{\theta}\)\)1\(p\(x; \hat{\theta}\) > k\)\] / \hat{\mu}\(1 - p\(x; \hat{\theta}\)\),\nonumber \\
&k\(\hat{\mu}, \hat{\theta}, \alpha\)  = \inf\left\{k \in \mathbb{R}: \alpha\(\hat{\mu}, \hat{\theta}, k\) \leq \alpha \right\},
  \label{roc model-based functional expression}
\\
&\beta\(\hat{\mu}, \hat{\theta}, \alpha\) \coloneqq 
\beta\(\hat{\mu}, \hat{\theta}, k\(\hat{\mu}, \hat{\theta}, \alpha\)\), 
\beta\(\hat{\mu}, \hat{\theta}, k\) 
= \hat{\mu}\[p\(x; \hat{\theta}\)1\(p\(x; \hat{\theta}\) > k\)\] / \hat{\mu}\(p\(x; \hat{\theta}\)\).\nonumber
\end{align}

Typically, 
$\hat\theta = \arg\max_{\theta\in\Theta} \sum_{i=1}^n \[y_i \log p\(x_i; \theta\) + \(1-y_i\) \log \(1-p\(x_i; \theta\)\)\]$ is the maximum likelihood estimator.
When the parametric model is correctly specified,
the model based estimator is
more efficient than
the orthogonal estimator, 
as formalized in the next proposition and illustrated in a
synthetic data example in section \ref{simulation data CI and testing}.

\begin{proposition}\label{ROC parametric model based estimator more efficient than orthogonal estimator}
Assume that \eqref{mle influence function} holds.
Under the conditions in Corollary \ref{roc asymptotic result},
when the parametric model is correctly specified,
the model based estimator $\beta\(\hat\mu, \hat{\theta}, \alpha\)$ in
\eqref{roc model-based functional expression}
is asymptotically more efficient than
the orthogonal estimator in
\eqref{roc curve asymptotics under correct specification}.
\end{proposition}

Misspecified parametric models estimated by the orthogonal estimator
are discussed in Appendix \ref{supplementary discussion on misspecified parametric models}.
The next Corollary \ref{roc bootstrap} 
validates bootstrap consistency under general specification.
For correctly specified parametric models,
the first step model only needs to be trained once and does not need to be re-estimated for every bootstrapped sample.

Let $M = \(M_{1}, \ldots, M_{n}\)$ be the vector of bootstrap weights.
For a random element $\tilde{\mathbb{Z}}_{n} \coloneqq \tilde{\mathbb{Z}}_{n}\(\{\(X_{i}, Y_{i}\)\}^{n}_{i = 1}, M\)$ on normed space $\mathcal{X}$,
we say its bootstrap law is consistent for a tight process $\mathbb{G}$
and denote it as $\tilde{\mathbb{Z}}_{n} \weakconv \mathbb{G}$
if
in outer probability\footnote{
Under measurability, the consistency definition can be simplified to
$\sup_{l \in BL_{1}\(\mathcal{X}\)} \left\vert \mathbb{E}_{M}l\(\tilde{\mathbb{Z}}_{n}\) - \mathbb{E}l\(\mathbb{G}\)\right\vert = o_{\mathbb{P}}\(1\).$
}
\bsnumber\label{conditional weak convergence in probability}
\sup_{l \in BL_{1}\(\mathcal{X}\)} \left\vert \mathbb{E}_{M}l\(\tilde{\mathbb{Z}}_{n}\) - \mathbb{E}l\(\mathbb{G}\)\right\vert \to 0, \; \mathbb{E}_{M}l\(\tilde{\mathbb{Z}}_{n}\)^{*} - \mathbb{E}_{M}l\(\tilde{\mathbb{Z}}_{n}\)_{*} \to 0.
\end{split}\end{align}
In the above $\mathbb{E}_{M}$ denotes the conditional expectation with respect to $M$ given the data $\{\(X_{i}, Y_{i}\)\}^{n}_{i = 1}$, and $BL_{1}\(\mathcal{X}\)$ denotes the set of 
real-valued functions on $\mathcal{X}$ with Lipschitz constant and supremum norm both bounded by $1$.

\begin{cor}\label{roc bootstrap}
Assume the function class $\mathcal{G}$ in Assumption \ref{joint convergence assumption} is Donsker
and 
for $\tilde{\mathbb{Q}}_{n} \coloneqq \tilde{\mathbb{Q}}_{n}\(\{\(X_{i}, Y_{i}\)\}^{n}_{i = 1}, M\)$ and $M \sim multinomial\(n, \(\frac{1}{n}, \ldots, \frac{1}{n}\)\)$, 
$\sqrt{n}\(\tilde{p} - \hat{p}, \tilde{\mathbb{Q}}_{n} -\mathbb{Q}_{n}\) \weakconv \(\mathbb{H}_{p}, \mathbb{Q}\)$.
Under the conditions in Corollary \ref{roc asymptotic result} and correct specification, the orthogonal estimator satisfies
\bs
\sup_{l \in BL_{1}\(\ell^\infty\(\mathcal{D}\)\)} & \left\vert
\mathbb{E}_{M}l\(\sqrt{n}\(\beta\(\tilde{\mathbb{Q}}_{n}, \hat{p}, \alpha\) - \beta\(\mathbb{Q}_{n}, \hat{p}, \alpha\)\)\)
- \mathbb{E}l\(\beta'_{Q, p, \alpha}\(\mathbb{Q}, 0, 0\)\) \right\vert\to 0,
 \\
& \mathbb{E}_{M}l\(\sqrt{n}\(\beta\(\tilde{\mathbb{Q}}_{n}, \hat{p}, \alpha\) - \beta\(\mathbb{Q}_{n}, \hat{p}, \alpha\)\)\)^{*}
- \mathbb{E}_{M}l\(\sqrt{n}\(\beta\(\tilde{\mathbb{Q}}_{n}, \hat{p}, \alpha\) - \beta\(\mathbb{Q}_{n}, \hat{p}, \alpha\)\)\)_{*} \to 0
\end{split}\end{align}
in outer probability. 
Under general specification,
\begin{align}\begin{split}\nonumber
\sqrt{n}\(\beta\(\tilde{\mathbb{Q}}_{n}, \tilde{p}, \alpha\) - \beta\(\mathbb{Q}_{n}, \hat{p}, \alpha\)\)
\weakconv \beta'_{Q, p^{\star}, \alpha}\(\mathbb{Q}, \mathbb{H}_{p}, 0\).
\end{split}\end{align}
\end{cor}

Consider i.i.d. data $\{\(Y_i, W_i, X_i\), i=1,\ldots,n\}$, where $Y_i$ is the true outcome, $W_i$ is the human decision, 
and $X_i$ are observable features.
An estimator of the human TPR/FPR pair is 
\bs
\hat{\alpha}_{H} = \sum_{i=1}^n \[\(1-Y_i\) W_i\] / \sum_{i=1}^n \(1-Y_i\), \quad
\hat{\beta}_{H} = \sum_{i=1}^n	Y_i W_i / \sum_{i=1}^n Y_i.
\end{split}\end{align}
Human decision making quality can be evaluated by testing 
whether 
the population TPR/FPR pair of a human decision maker, denoted $\(\alpha_{H}, \beta_{H}\)$,
lies on, above or below a population ROC curve $\beta\(\alpha\)$ of a machine learning model.
In the above, 
\bs
\alpha_{H} = \mathbb{E}\[\(1-Y\) W\] / \mathbb{E}\(1-Y\), \quad
\beta_{H} = \mathbb{E}\[Y W\] / \mathbb{E}Y.
\end{split}\end{align}
The null and alternative hypotheses can be formulated as
\bsnumber\label{testing a person against a roc curve}
H_0: \beta\(\alpha_{H}\) = \beta_{H} \quad\text{versus}\quad H_A: \beta\(\alpha_{H}\) \neq \beta_{H},
\end{split}\end{align}
where $H_A$ can be further categorized into two one-sided directional alternatives $H_{A,h}: \beta\(\alpha_{H}\) < \beta_{H}$ 
and $H_{A,m}: \beta\(\alpha_{H}\) > \beta_{H}$, as in \cite{vuong1989likelihood}.
It follows from the Lindeberg-L\'{e}vy central limit theorem that
for $\theta_{H} = \(\alpha_{H}, \beta_{H}\)'$,
$\hat{\theta}_{H} = \(\hat{\alpha}_{H}, \hat{\beta}_{H}\)'$ and $p = Q Y$,
\begin{align}\begin{split}\nonumber
\sqrt{n}\(\hat{\theta}_{H} - \theta_{H}\)
\rightsquigarrow  N\(0, \Sigma\) \quad\text{where}\quad
\begingroup
\renewcommand*{\arraystretch}{1}
\Sigma = \begin{varmatrix}[delim = p, size = \normalsize, sep = 0.5\arraycolsep]
	   \frac{\alpha_H\(1-\alpha_H\)}{1 - p}   & 0 \\
	    0 & \frac{\beta_H\(1-\beta_H\)}{p}
	    \end{varmatrix}
\endgroup.
\end{split}\end{align}
When the 
ROC $\beta\(\alpha\)$ is fully known,\footnote{
For example, \cite{feng2025statisticalMS} assume a known ROC because  
the sample size of the data set used to train the machine learning
model is by orders of magnitude larger than the sample size for each 
individual decision maker.
}
the hypothesis \eqref{testing a person against a roc curve} can be tested using the statistics	$\beta\(\hat{\alpha}_{H}\) - \hat{\beta}_{H}$.
By the delta method, 
under the 
null hypothesis
$\beta\(\alpha_{H}\) = \beta_{H}$ in \eqref{testing a person against a roc curve},
\bs
\sqrt{n} \(\beta\(\hat{\alpha}_{H}\) - \hat{\beta}_{H}\) \rightsquigarrow
N\(0, \sigma^2\) \quad\text{where}\quad
\sigma^2 = \beta'\(\alpha_H\)^2\frac{\alpha_H\(1 - \alpha_H\)}{1 - p} + \frac{\beta_H\(1 - \beta_H\)}{p}.
\end{split}\end{align}

A more complex situation is when the machine ROC curve $\hat\beta\(\alpha\)$ and the human TPR/FPR pair $\(\hat{\alpha}_{H}, \hat{\beta}_{H}\)$ are estimated
from the same dataset.
In such a case,
a revised test statistic for the hypothesis
in \eqref{testing a person against a roc curve} is then
$\hat\beta\(\hat{\alpha}_{H}\) - \hat{\beta}_{H} =
\hat{\beta}\(\hat{k}\(\hat{\alpha}_{H}\)\) - \hat{\beta}_{H}$,
where $\hat\beta\(\cdot\)$ can be any one of orthogonal ROC $\beta\(\mathbb Q_n, \hat p\(\cdot\), \cdot\)$,
parametric model based ROC $\beta\(\hat\mu, \hat\theta, \cdot\)$ and parametric orthogonal ROC $\beta\(\mathbb Q_n, \hat \theta, \cdot\)$. 
Similarly, $\hat{k}\(\cdot\)$ can be the corresponding orthogonal $k\(\mathbb Q_n, \hat p\(\cdot\), \cdot\)$,
parametric model based $k\(\hat\mu, \hat\theta, \cdot\)$ and parametric orthogonal $k\(\mathbb Q_n, \hat \theta, \cdot\)$ in 
\eqref{double robust representing of roc constrained optimized value function under full observability} and 
\eqref{roc model-based functional expression}.
The following Hadamard differentiation result applies to these estimators.

\begin{lemma}\label{test Hadamard derivative}
Let $\(Y, W, X\) \sim S$. 
Denote,
\begin{align}\begin{split}\nonumber
& \alpha\(S, p, k\) \coloneqq S\[\(1 - y\)1\(p\(x\) > k\)\] / S\(1 - y\),\  k\(S, p, \alpha\) \coloneqq \inf\left\{k \in \mathbb{R}: \alpha\(S, p, k\) \leq \alpha \right\}, \\
&k\(S, p\) \coloneqq \inf\left\{k \in \mathbb{R}: \alpha\(S, p, k\) \leq \alpha_H \right\},\  
\alpha_{H} \coloneqq S\[\(1 - y\)w\] / S\(1 - y\), 
\\
& \beta\(S, p, k\) \coloneqq S\[y1\(p\(x\) > k\)\] / S\(y\), \; \beta\(S, p\) \coloneqq S\[y1\(p\(x\) > k\(S, p\)\)\] / S\(y\).
\end{split}\end{align}
Under the conditions in Corollary \ref{roc asymptotic result},
$\beta\(S, p\)$ is Hadamard differentiable at $\(S, p^{\star}\)$ tangentially to $\mathbf{S} \times C\(E\)$,
where $\mathbf{S} \subset \ell^{\infty}\(\mathcal{F}\)$ consists of elements that are uniformly continuous on $\mathcal{F}$ 
with respect to the $L^{2}\(S\)$ norm.
The Hadamard derivative is
\begin{align}
& \beta'_{S, p^{\star}}\(\mathcal{S}, H\) =
\frac{1}{S\(y\)}\biggl\{\mathcal{S}\[y\(1\(p^{\star}\(x\) > k^{\star}\) - \beta\(S, p^{\star}, k^{\star}\)\)\] \notag \\
& + \int_{p^{\star}\(x\) = k^{\star}}\frac{p\(x\)H\(x\)\mu'\(x\)}{\Vert\nabla p^{\star}\(x\)\Vert}d\mathcal{H}_{n - 1}x\biggl\}
- \frac{1}{S\(1 - y\)}\frac{f_{\beta}\(S, p^{\star}, k^{\star}\)}{f_{\alpha}\(S, p^{\star}, k^{\star}\)}\biggl\{ \label{testing roc doctor Hadamard derivative} \\
& \mathcal{S}\[\(1 - y\)\(1\(p^{\star}\(x\) > k^{\star}\) - \(w - \alpha_{H}\) - \alpha\(S, p^{\star}, k^{\star}\)\)\] + \int_{p^{\star}\(x\) = k^{\star}}\frac{\(1 - p\(x\)\)H\(x\)\mu'\(x\)}{\Vert\nabla p^{\star}\(x\)\Vert}d\mathcal{H}_{n - 1}x\biggr\}, \notag
\end{align}
where $k^{\star} = k\(S, p^{\star}, \alpha_{H}\)$, $f_{\alpha}\(S, p^{\star}, k\)$ and $f_{\beta}\(S, p^{\star}, k\)$ are defined in \eqref{beta alpha derivative wrt k}.
\end{lemma}

When $p\(x\)$ is modeled parametrically as $p\(x; \theta\)$ 
and $\sqrt{n}\(\hat{\theta} - \theta^{\star}\)$ has an asymptotic linear representation by $\kappa\(x, y\)$,
by the delta method, $H\(x\)$ will be determined by $\frac{\partial}{\partial\theta}p\(x; \theta^{\star}\)$
and $\kappa\(x_i, y_i\)$ giving rise to an asymptotic linear representation.
In each case, the asymptotic variance can be consistently estimated by sample analogs. Bootstrap also provides valid inference.
The asymptotic distribution and verification of bootstrapped confidence interval of
$\sqrt{n}\[\(\hat{\beta}\(\hat{\alpha}_{H}\) - \hat{\beta}_{H}\) - \(\beta\(\alpha_{H}\) - \beta_{H}\)\]$ is obtained by
combining the standard term  $\hat{\beta}_{H} - \beta_{H}$  with $\hat\beta\(\hat{\alpha}_{H}\) - \beta\(\alpha_{H}\)$ handled by
\eqref{testing roc doctor Hadamard derivative}.
Section \ref{simulation data CI and testing} showcases a synthetic data example.

\section{Fast convergence rates}\label{fast convergence rates}

In this section we build on the previous analysis to generalize the fast learning rates for plug-in classifiers in
\cite{devroye1996probabilistic} and \cite{audibert2007fast} to multiclassification and multi-discrete allocation problems.
The fast learning rate phenomena can also be found in recent works by \cite{semenova2023debiased} and \cite{liu2024inference}.
Following \cite{chernozhukov2018double}, \cite{chernozhukov2022locally} and \cite{foster2023orthogonal},
Neyman-orthogonalization and sampling splitting can be combined with the plug-in scheme to facilitate
statistical inference.

\subsection{Second order Hadamard differentiability}\label{second order Hadamard differentiability}

Our generalizations are motivated by the partial convexity of 
the social welfare potential function $\gamma\(\cdot, \cdot\)$. 
Consider, for a given $g\(\cdot\)$,
\bs
e\(\lambda'; \phi^{*}\(\cdot; \lambda, g\)\) = \mathbb{E} \[\sum_j \lambda_j' \phi_j^{*} g_j\] - \gamma\(\lambda', g\).
\end{split}\end{align}
By construction, $e\(\lambda'; \phi^{*}\(\cdot; \lambda, g\)\)$ achieves its maximum of zero
at $\lambda'=\lambda$.
When there is no ambiguity, we also denote $\phi^{*}\(\cdot\) \coloneqq \phi^{*}\(\cdot; \lambda, g\)$. 
Continuous convex function is locally Lipschitz and (locally) Lipschitz function is differentiable almost 
everywhere by the Rademacher theorem \ref{Rademacher theorem}.
Due to the convexity of $\gamma\(\lambda', g\)$ in $\lambda'$ given $g\(\cdot\)$, $\gamma\(\lambda', g\)$  is differentiable in $\lambda'$ given $g\(\cdot\)$ in a set $\Lambda$ such that its complement $\Lambda^c$ is a	Lebesgue null set.
Then for each $\lambda \in \Lambda$ and for each selection $\phi^{*}$ given this choice of
$\lambda$, $e\(\lambda'; \phi^{*}\)$ is differentiable in $\lambda'$ at $\lambda'=\lambda$.
Since zero must be in the set of supgradients
at a point of maximum of any function, there is $0 \in \nabla^+ e\(\lambda'; \phi^{*}\)\big\vert_{\lambda'=\lambda}$ and thus
$\nabla e\(\lambda'; \phi^{*}\)\big\vert_{\lambda'=\lambda} = 0$,
implying that
\bs
\nabla \gamma\(\lambda', g\)\Bigg\vert_{\lambda'=\lambda}
= \nabla \mathbb{E} \[\sum_j \lambda_j' \phi_j^{*} g_j\]\Bigg\vert_{\lambda'=\lambda}
= \(\mathbb{E} \phi_j^{*} g_j, j=0,\ldots,J\)^{T}.
\end{split}\end{align}

We expect such envelope theorem like statement remains true with
functional derivatives of
\bsnumber\label{auxiliary function}
e\(\lambda', g'; \phi^{*}\) = \mathbb{E}\[\sum_j \lambda_j' \phi_j^{*} g_j'\]-\gamma\(\lambda', g'\)
\end{split}\end{align}
at $\(\lambda', g'\) = \(\lambda, g\)$ (here, $\phi^{*}\(\cdot\) \coloneqq \phi^{*}\(\cdot; \lambda, g\) $).
Theorem
\ref{theorem envelope Theorem for social welfare potential function} shows 
Fr\'{e}chet differentiability and verifies \eqref{auxiliary function}.
Generically, under weak assumptions,
\begin{align}\begin{split}\nonumber
\gamma\(\lambda, g + \theta\) = \gamma\(\lambda, g\) +
\sum_{j}\mathbb{E}\lambda_{j}\phi^{*}_{j}\theta_{j}
+ o\(\Vert \theta\Vert\).
\end{split}\end{align}
By construction, the remainder term of $o\(\Vert \theta\Vert\)$ in the above expression is
\begin{align}\begin{split}\label{reminder term Frechet}
\mathbb{E}\[\sum_{j}\lambda_{j}\(\phi^{*}_j\(X; \lambda, g + \theta\) - \phi^{*}_j\(X; \lambda, g\)\)\(g_{j}\(X\) + \theta_{j}\)\].
\end{split}\end{align}
The magnitude of this term can be investigated by relating 
the Hadamard differentiability framework in section \ref{Weighted sorted effect} to the margin assumption (MA) in the plug-in classifier literature.

Consider the estimation of social welfare potential  function under full observability. 
Let $\mu$ denote the distribution of $X$ and  $Q$ the joint distribution of $(y,x).$
Define the following notation:
\begin{align}\begin{split}\nonumber
\gamma\(\lambda, g(\cdot),\mu\) & = \mu\[\sum_{j}\lambda_{j}\(\prod_{l \neq j}1\(\lambda_{j}g_{j}\(x\) > \lambda_{l}g_{l}\(x\)\)\)g_{j}\(x\)\], \\
\gamma\(\lambda, g(\cdot),Q\) & = Q\[\sum_{j}\lambda_{j}\(\prod_{l \neq j}1\(\lambda_{j}g_{j}\(x\) > \lambda_{l}g_{l}\(x\)\)\)y_{j}\].
\end{split}\end{align}

To contrast the application of
measure $\mu$ with that of measure
$Q$, we first consider the perturbation of $\gamma\(\lambda, g(\cdot),\mu\)$
to
$\gamma\(\lambda, g(\cdot)+t_n\,H_n(\cdot),\mu+t_n{\mathcal P}_n\)$ where $t_n$
approaches zero and $H_n(\cdot)$
and ${\mathcal P}_n$ provide directions of the perturbation of the conditional expectation function $g(\cdot)$ 
and 
of the measure $\mu\(\cdot\)$ over ${\mathcal X},$ respectively.
We consider the following decomposition:
\begin{align}\begin{split}\nonumber
\gamma\(\lambda, g + t_{n}H_{n}, \mu + t_{n}\mathcal{P}_{n}\) - \gamma\(\lambda, g, \mu\)
& = \underbrace{t_{n}\(\gamma\(\lambda, g + t_{n}H_{n}, \mathcal{P}_{n}\) - \gamma\(\lambda, g, 
\mathcal{P}_{n}\)\) }_{\(1\)} \\
&\hspace{-.5in} + \underbrace{\gamma\(\lambda, g + t_{n}H_{n}, \mu\) - \gamma\(\lambda, g, \mu\)}_{\(2\)}
 + \underbrace{\gamma\(\lambda, g, \mu + t_{n}\mathcal{P}_{n}\) - \gamma\(\lambda, g, \mu\)}_{\(3\)}.
\end{split}\end{align}

The first term $\(1\)$ is the second order difference ``empirical process'' term that should vanish under suitable continuity conditions.
The third term $\(3\)$ is a standard ``summation of random variables'' term. 
The second term $\(2\)$ can be further decomposed as
\begin{align}\begin{split}\nonumber
\(2\) & = \underbrace{\mu\[\sum_{j}\lambda_{j}\(\prod_{l \neq j}1\(\lambda_{j}\(g_{j} + t_{n}H_{nj}\) > \lambda_{l}\(g_{l} + t_{n}H_{nl}\)\) - \prod_{l \neq j}1\(\lambda_{j}g_{j} > \lambda_{l}g_{l}\)\)g_{j}\]}_{\textcircled{\footnotesize{1}}} \\
& + \underbrace{t_{n}\mu\[\sum_{j}\lambda_{j}\(\prod_{l \neq j}1\(\lambda_{j}g_{j} > \lambda_{l}g_{l}\)\)H_{nj}\]}_{\textcircled{\footnotesize{2}}} \\
& + \underbrace{t_{n}\mu\[\sum_{j}\lambda_{j}\(\prod_{l \neq j}1\(\lambda_{j}\(g_{j} + t_{n}H_{nj}\) > \lambda_{l}\(g_{l} + t_{n}H_{nl}\)\) - \prod_{l \neq j}1\(\lambda_{j}g_{j} > \lambda_{l}g_{l}\)\)H_{nj}\]}_{\textcircled{\footnotesize{3}}},
\end{split}\end{align}
where for brevity we omit the argument of functions $g(\cdot)$ and $H(\cdot)$. 
If we use $\gamma\(\lambda, g, Q\)$ instead of $\gamma\(\lambda, g, \mu\)$, then we will not have \textcircled{\footnotesize{2}} and \textcircled{\footnotesize{3}} terms.
In this sense, using $Y$ instead of $g$ provides a natural Neyman-orthogonalization of the estimator.
We show 
that under Hadamard differentiability, the term \textcircled{\footnotesize{1}} is $o\(t_n\)$, 
leading to a simplified form of the Hadamard derivative of $\gamma\(\lambda, g, Q\)$. 

\begin{lemma}\label{social welfare potential first order derivative lemma}
Consider the following class of functions indexed by $g \in \mathcal{G} \subset C\(E, \mathbb{R}^{J + 1}\)$,
\begin{align}\begin{split}\nonumber
\mathcal{F} \coloneqq \left\{\(x^{T}, y^{T}\)^{T} \mapsto \(y_{0}\(\prod_{l \neq 0}1\(\lambda_{0}g_{0}\(x\) > \lambda_{l}g_{l}\(x\)\)\), \ldots, y_{J}\(\prod_{l \neq J}1\(\lambda_{J}g_{J}\(x\) > \lambda_{l}g_{l}\(x\)\)\)\)^{T} \right\}
\end{split}\end{align}
where $y$ is bounded. Let $\mathbf{Q} \subset \ell^{\infty}\(\mathcal{F}\)$ consists of elements that are linear and uniformly continuous on $\mathcal{F}$ with respect to the $L^{2}\(Q\)$ norm.
Assume also that the conditions of Theorem \ref{general k derivative} hold for 
each $\Delta_{j} = \lambda_{j}g_{j} - \(\lambda_{1}g_{1}, \ldots, \lambda_{j - 1}g_{j - 1}, \lambda_{j + 1}g_{j + 1}, \ldots, \lambda_{J}g_{J}\)^{T}$, $j = 0, \ldots, J$ and for $\mu$ where $X \sim \mu$.
Then the Hadamard derivative of $\gamma\(\lambda, \cdot, \cdot\)$ at $\(g, Q\)$ tangentially to $  C\(E, \mathbb{R}^{J + 1}\) \times \mathbf{Q}$ is
\begin{align}\begin{split}\label{social welfare potential derivative}
\gamma'_{\lambda, g, Q}\(0, H, \mathcal{Q}\) = \mathcal{Q}\[\sum_{j}\lambda_{j}\(\prod_{l \neq j}1\(\lambda_{j}g_{j}\(x\) > \lambda_{l}g_{l}\(x\)\)\)y_{j}\].
\end{split}\end{align}
\end{lemma}

Compare Lemma \ref{social welfare potential first order derivative lemma} with Theorem \ref{theorem envelope Theorem for social welfare potential function}.
By \eqref{Frechet derivative} in Theorem \ref{theorem envelope Theorem for social welfare potential function},	
the first order influence of pertubing
the second variable $g$
will not come from the indicator term.
Equation \eqref{social welfare potential derivative} in Lemma \ref{social welfare potential first order derivative lemma} derives an analogous Hadamard differentiability result. 
Neyman-orthogonalization changes the remainder term \eqref{reminder term Frechet} to 
$\mathbb{E}\[\sum_{j}\lambda_{j}\(\phi^{*}_{j}\(X; \lambda, g + \theta\) - \phi^{*}_{j}\(X; \lambda, g\)\)g_{j}\(X\)\].$
Next, we   show under a strict MA in \eqref{marginal assumption MA} with $\alpha > 1$, this term is second order degenerate.

\begin{definition}{(Second order Hadamard differentiability)}
Let $\mathcal{X}$ and $\mathcal{Y}$ be two normed spaces equipped with norms $\Vert\cdot\Vert_{\mathcal{X}}$ and $\Vert\cdot\Vert_{\mathcal{Y}}$.
A map $\phi: \mathcal{X}_{\phi} \subset \mathcal{X} \to \mathcal{Y}$ is called second order Hadamard differentiable at $\theta \in \mathcal{X}_{\phi}$ tangentially to $\mathcal{X}_{0} \subset \mathcal{X}$ if
\begin{enumerate}[(i)]
  \item $\phi$ is (first order) Hadamard differentiable at $\theta$ tangentially to $\mathcal{X}_{0}$ and the derivative $\phi'_{\theta}: \mathcal{X}_{0} \to \mathcal{Y}$ is
well defined on $\mathcal{X}$;
  \item there exists a continuous bilinear map $\Phi''_{\theta}: \mathcal{X}_{0} \times \mathcal{X}_{0} \to \mathcal{Y}$ such that for $\phi''_{\theta}\(x\) \coloneqq \Phi''_{\theta}\(x, x\)$,
      \begin{align}\begin{split}\nonumber
      \left\Vert
      \frac{\phi\(\theta + t_{n}x_{n}\) - \phi\(\theta\) - t_{n}\phi'_{\theta}\(x_{n}\) - t^{2}_{n}\phi''_{\theta}\(x\)}
      {t^{2}_{n}}\right\Vert_{\mathcal{Y}} \to 0
      \end{split}\end{align}
      for all $t_{n} \to 0$ and $x_{n} \to x \in \mathcal{X}_{0}$, as $n \to \infty$, where $\theta + t_{n}x_{n} \in \mathcal{X}_{\phi}$ for all $n$.
\end{enumerate}
\end{definition}

In the following we adopt a convention for the optimal policy in \eqref{reminder term Frechet}
 such that all the weight among the maximal indexes is allocated to the smallest index member:
\begin{align}\begin{split}\label{tie breaking definition to ensure at least at 1}
\phi^{*}_j\(x; \lambda, g\) \coloneqq \prod_{l=j+1}^J 1\(\lambda_j g_j\(x\) \geq \lambda_l g_l\(x\)\)
\prod_{l=0}^{j-1} 1\(\lambda_j g_j\(x\) > \lambda_l g_l\(x\)\).
\end{split}\end{align}
We introduce this notation in particular to handle the case $\lambda_{j}\hat{g}_{j} = \lambda_{l}\hat{g}_{l}$ when $\hat{g}$ is estimated.
For $k\neq l, k,l=0,\ldots,J$, we also write
\begin{align}\begin{split}\label{decision difference}
1_{kl}\(x; \lambda, g\) \coloneqq 1\(k < l\)1\(\lambda_k g_k\(x\) \geq \lambda_l g_l\(x\)\)
+
1\(k > l\)1\(\lambda_k g_k\(x\) > \lambda_l
g_l\(x\)\).
\end{split}\end{align}

\begin{proposition}\label{first MA second Hadamard result}
Assume that $\lambda_{j}g_{j} - \lambda_{l}g_{l}$ satisfies MA in \eqref{marginal assumption MA} for all $j \neq l, j, l = 0, \ldots, J$ at $0$ with an $\alpha > 1$.
Under the conditions in Lemma \ref{social welfare potential first order derivative lemma}, where $\mathcal F$ is defined
using \eqref{tie breaking definition to ensure at least at 1} in the statement of the lemma,
the second order Hadamard derivative of $\gamma\(\lambda, g, Q\)$ at $g$ tangentially to $C\(E, \mathbb{R}^{J + 1}\)$ is $0$.

Furthermore, if $\mathcal{F}$ is a Donsker class of functions, $g, \hat{g} \in \mathcal{G}$,
and $r_{n}\(\hat{g} - g\) \rightsquigarrow \mathbb{H},$ where $\mathbb{H}$ is a separable process supported on $C\(E, \mathbb{R}^{J + 1}\)$, then for
$\liminf_{n \to \infty}\frac{r_{n}}{\sqrt[4]{n}} > 0$, 
\begin{align}\begin{split}\nonumber
\gamma\(\lambda, \hat{g}, \mathbb{Q}_{n}\) = \gamma\(\lambda, g, \mathbb{Q}_{n}\) + o_{\mathbb{P}^{*}}\(\frac{1}{\sqrt{n}}\).
\end{split}\end{align}
\end{proposition}

The main drawback of Hadamard differentiability (functional delta method) is that it does not fully utilize the convergence rate of the first step estimator.
Since Proposition \ref{surface integral continuity} guarantees $\alpha = 1$, the requirement of
$\alpha > 1$ in Proposition \ref{first MA second Hadamard result} is less satisfactory. 
In section \ref{inference on social welfare potential function} we overcome this drawback using sample splitting and Neyman orthogonality.
In the following we provide two additional results to show more connection between the theory in section \ref{Weighted sorted effect} and the MA.

Recall that the push-forward measure $T_{\#}\mu$ generated by a Borel measurable map $T$ from the measure $\mu$ is defined by $T_{\#}\mu\[A\] = \mu\[T^{-1}\(A\)\]$.
The original margin assumption in \eqref{marginal assumption MA} can be extended to the following:

\begin{definition}(Margin controlled random variable / distribution function)
Let $X$ be a random variable supported on $\[a, b\]$, where $a < b$, $a, b \in \mathbb{R}$.
If for all $c \in \[a, b\]$,
\begin{align}\begin{split}\nonumber
\mathbb{P}\left\{\left\vert X - c\right\vert < t\right\} \leq C\(c\)t^{\alpha},
\end{split}\end{align}
where $-\infty < C\(c\) < +\infty$ may depend on c, we call $X$ a margin controlled random variable of order $\alpha$.
Similarly, if $F$ is the distribution function of such a random variable, we call it a margin controlled distribution function.
\end{definition}
A definition of a margin controlled random variable
is to impose local H\"{o}lder-continuity on its distribution
function with a universal exponent
$\alpha$. 

\begin{lemma}\label{marginally controlled lemma}
If $X$ is a margin controlled random variable on $\[a, b\]$ of order $\alpha=1$, where $a < b$, $a, b \in \mathbb{R}$, then $X$ is absolutely continuous with respect to the Lebesgue measure. Consequently, $X$ admits a density function.
\end{lemma}

\begin{lemma}\label{density iff}
Let $f \in L^{1}\(\mathbb{R}^{n}, \mathcal{L}_{n}\)$ be a nonnegative function and $h: \Omega \to \mathbb{R}$ be a function in $C^{m}$ (or Sobolev space $W^{m, p}$, $p > m$), where $\Omega \subset \mathbb{R}^{n}$ is an open subset.
Let $X$ be a random variable with density function $f$.
Then $h\(X\)$ also admit a density, i.e. the push-forward probability $h_{\#}\(f\mathcal{L}_{n}\) \ll \mathcal{L}_{1}$, if and only if $Jh > 0 \; a.e.$ on the support of $f$.	In this case the density of $h\(X\)$ is
\begin{align}\begin{split}\nonumber
f_{h}\(y\) = \int_{h^{-1}\(y\)}\frac{f\(x\)}{Jh\(x\)}d\mathcal{H}_{n - 1}x.
\end{split}\end{align}
\end{lemma}

\begin{remark}
If one only needs the if part of Lemma \ref{density iff}, where $Jh > 0 \; a.e.$ implies a density function in the form of integration with respect to the Hausdorff measure,
then one does not need the $C^{m}$ assumption. Lipschitzness is sufficient.
Proposition \ref{surface integral continuity} tells more by confirming that with a bit more regularity conditions, $f_{h}$ is actually continuous.
\end{remark}

\subsection{Multi discrete allocation problem}

In this section we apply the fast convergence rate to
social welfare functions with multiple treatments, generalizing the arguments
in 
\cite{devroye1996probabilistic} and \cite{audibert2007fast}.

Recall that in a general discrete allocation problem, we define welfare under policy $\phi\(\cdot\)$ as
\bsnumber\label{multi welfare under policy}
\gamma\(\lambda, g, \phi\)
= \mathbb{E} \[\sum_{j=0}^J  \lambda_j\phi_j\(X\)Y_j\]
= \mathbb{E} \[\sum_{j=0}^J \lambda_j\phi_j\(X\)g_j\(X\)\],
\end{split}\end{align}
where
$g_j\(x\) = \mathbb{E}\(Y_j \vert X = x\)$ and $\phi_j\(x\) \geq 0, \forall j$, $\sum_{j=0}^J \phi_j\(x\) = 1$.
A special case of \eqref{multi welfare under policy} is the weighted multiclassification accuracy when
$Y_j \in \{0,1\}$, $\sum_{j=0}^J Y_j =1$, and when $\phi_j\(x\), j=0,\ldots,J$ is a multiclassification
allocation.\footnote{The binary treatment case in \cite{audibert2007fast} is a special of \eqref{multi welfare under policy} where $g_1\(x\)=p(x)$, the propensity score,
$g_0\(x\)=1-p\(x\)$, $\phi_1\(x\) = 1\(p(x)>\tau\)$,  $\phi_0\(x\) = 1\(p(x) \leq \tau\)$,  $\lambda_1=\tau$ and $\lambda_0=1-\tau$.}

In classification problems, $Y_j, j=0,\ldots,J$ are directly observed. 
In contrast, in optimal treatment allocation
problem, $Y_j$ is typically observed only when $d_j=1$, where $d_j$ is the indicator  of whether the $j$th treatment is taken,
such that $d_j \in \{0,1\}$ and $\sum_{j=0}^J d_j = 1$.
We assume the standard unfoundedness condition, i.e.,
$Y_{j}\indep D_{j} \vert X, \forall j=0,\ldots,J$.
Then $g_j\(x\) = \mathbb{E}\(Y_j \vert D_j =1, x\)$ can be estimated 
based on the subsample where $d_j = 1$.

The envelope-like Theorem \ref{theorem envelope Theorem for social welfare potential function} essentially
states that generically there is
no first order impact of estimating $\phi^*\(\cdot; \lambda, g\)$
in $\gamma\(\lambda, g\) = \mathbb{E} \[\sum_j \phi_j^*\(X; \lambda, g\) g_j\(X\)\]$
by $\phi^{*}\(\cdot; \lambda, \hat{g}\)$.
The fast convergence rate of regret to zero is a consequence of the second order effect
of estimating $\phi^*\(\cdot; \lambda, g\)$ by $\phi^{*}\(\cdot; \lambda, \hat{g}\)$.
We use the same convention as
\eqref{tie breaking definition to ensure at least at 1} 
in 
the following expressions for the optimal and feasible policies:
\bs
\phi_j^{*}\(x\) &= \prod_{l=j+1}^J 1\(\lambda_j g_j\(x\) \geq \lambda_l g_l\(x\)\)
\prod_{l=0}^{j-1} 1\(\lambda_j g_j\(x\) > \lambda_l g_l\(x\)\),\\
\hat\phi_j\(x\) &= \prod_{l=j+1}^J 1\(\lambda_j \hat g_j\(x\) \geq \lambda_l \hat g_l\(x\)\)
\prod_{l=0}^{j-1} 1\(\lambda_j \hat g_j\(x\) > \lambda_l \hat g_l\(x\)\).
\end{split}\end{align}
The regret of using the feasible policy for welfare
is:
\bs
\gamma\(\lambda, g, \phi^{*}\)
-
\gamma\(\lambda, g, \hat\phi\)
= \sum_{j=0}^{J} \mathbb{E}\[\lambda_j g_j\(X\) \(
\phi_j^{*}\(X\) - \hat \phi_j\(X\)
\)\].
\end{split}\end{align}
By definition, $\phi_j^{*}\(x\)$ is either $0$ or $1$. 
Furthermore, for each $x$, one and only one of $\phi_j^{*}\(x\)$ out of $j=0,\ldots,J$ takes
the value $1$.
The same is true for $\hat\phi_j\(x\)$ for each $x$.
Therefore we can write
\bsnumber\label{tsybakov 2nd order bound}
\gamma\(\lambda, g, \phi^{*}\)
& -\gamma\(\lambda, g, \hat\phi\) \\ 
&\leq \mathbb{E}\[\sum_{i\neq j} \left\vert\lambda_i g_i\(X\) - \lambda_j g_j\(X\) -
\(\lambda_i \hat g_i\(X\) - \lambda_j \hat g_j\(X\)\)
\right\vert \phi_i^{*}\(X\) \hat\phi_j\(X\)\] \\
&\leq \mathbb{E}\[\sum_{i\neq j}
\left\vert\lambda_i g_i\(X\) - \lambda_j g_j\(X\) -
\(\lambda_i \hat g_i\(X\) - \lambda_j \hat g_j\(X\)\)\right\vert
\sum_{k \neq l} \left\vert
1_{kl}^{*} -\hat 1_{kl}
\right\vert\],
\end{split}\end{align}
where the second inequality uses the simplified notation \eqref{decision difference}.
The first inequality of \eqref{tsybakov 2nd order bound} is due to the fact that when
$\phi_i^{*}\(x\) = \hat\phi_j\(x\) =1$,
\bs
0 \leq \lambda_i g_i\(x\) - \lambda_j g_j\(x\)
\leq \lambda_i g_i\(x\) - \lambda_j g_j\(x\) - \(
\lambda_i \hat g_i\(x\) - \lambda_j \hat g_j\(x\)
\).
\end{split}\end{align}
Following Example 1 in \cite{chen2003estimation}, for
$\delta_n = \sup_{x\in\mathcal X} \max_{j=0,\ldots,J} \vert
\lambda_j g_j\(x\) - \lambda_j \hat{g}_j\(x\)
\vert$,
\bs
\vert
1_{kl}^{*} -\hat 1_{kl}
\vert \leq
1\(
-2 \delta_n \leq \lambda_k g_k\(x\) - \lambda_l g_l\(x\)\leq 2 \delta_n
\).
\end{split}\end{align}
Then we can further bound the regret as
\bs
&\gamma\(\lambda, g, \phi^{*}\)
-\gamma\(\lambda, g, \hat\phi\)
\leq 2 \(J+1\)^2 \delta_n
\sum_{k\neq l} \mathbb{P}\(-2 \delta_n \leq \lambda_k g_k\(X\) -
\lambda_l g_l\(X\)\leq 2 \delta_n\).
\end{split}\end{align}
The conditions in section \ref{Weighted sorted effect} imply that
for 
$k\neq l$,
$\lambda_k g_k\(x\) - \lambda_l g_l\(x\)$
has a bounded density at zero, leading to
$\mathbb{P}\(-2 \delta_n \leq \lambda_k g_k\(X\) -
\lambda_l g_l\(X\)\leq 2 \delta_n\) \leq C \delta_n$ and
$\gamma\(\lambda, g, \phi^{*}\)
-\gamma\(\lambda, g, \hat\phi\)
\leq C \delta_n^2.$

Nonparametric estimators $\hat{g}_j\(x\)$ typically achieve the convergence rate $\delta_n = O_{\mathbb{P}}\(n^{-\frac{\beta}{2\beta+d}} \log n\)$ where $d$ is the dimension of $x$, and $\beta$ a smoothness parameter, implying that
\bs
\left\vert
\gamma\(\lambda, g, \phi^{*}\)
-\gamma\(\lambda, g, \hat\phi\)\right\vert
= O_{\mathbb{P}}\(
n^{-\frac{2\beta}{2\beta+d}} \log^2 n
\).
\end{split}\end{align}
Whenever $\beta > d/2$, the convergence rate of 
$\gamma\(\lambda, g, \phi^{*}\) -\gamma\(\lambda, g, \hat\phi\)$
becomes $O_{\mathbb{P}}\(
n^{-\frac{1}{2}}\).$

\subsection{Inference on social welfare potential function}\label{inference on social welfare potential function}

The unknown social welfare potential function $\beta_0 = \gamma\(\lambda, g\) = \mathbb{E}\[\sum_{j}\lambda_j \phi^{*}_{j}\(X\)g_{j}\(X\)\]$
can be estimated by its sample analog
$\frac{1}{n'} \sum_{i=1}^{n'} \sum_{j} \lambda_j \hat{\phi}_{j}\(X_{i}\)\hat g_j\(X_{i}\)$ for $n' = O\(n\)$.
On the one hand, under full observability, the Neyman-orthogonalized estimator
$\frac{1}{n'} \sum_{i=1}^{n'} \sum_{j} \lambda_j \hat{\phi}_{j}\(X_{i}\)Y_{ij}$ can be combined with
a sample splitting scheme for estimating $\hat{\phi}_{j}\(X_{i}\)$ to achieve  desirable
statistical properties.
On the other hand, for a typical treatment allocation problem, $Y_{ij}$ are not fully observable.
We only observe $Y_{i} = \sum_{j}D_{ij}Y_{ij}$ and need to use
the treatment status variables $D_{ij}$ to construct a Neyman-orthogonalized estimator.

\begin{theorem}\label{treatment effect allocation simple case}
Let $p, \hat{p}, g, \hat{g}: E \subset \mathbb{R}^{\dim\(X_i\)} \to \mathbb{R}^{J + 1}$ be uniformly bounded functions, where $E$ is an open set.
$Var\(Y_{i} \vert X_i\)$ is uniformly bounded.
Assume that there exists a constant $\epsilon > 0$ such that $p_{j}\(x\), \hat{p}_{j}\(x\) > \epsilon$ for all $x \in E$ and for all $j \in \{0, \ldots, J\}$.
Assume also that for some sequence $a_{n}$, $\lim_{n \to \infty}\frac{a_{n}}{\sqrt[4]{n}} = + \infty$,
there is $\esssup a_{n}\left\vert \hat{p} - p\right\vert = o_{\mathbb{P}}\(1\)$ and $\esssup a_{n}\left\vert \hat{g} - g\right\vert = o_{\mathbb{P}}\(1\)$, where the essential supremum is taken with respect to the distribution of $X$.
Further assume that 
$\lambda_{j}g_{j} - \lambda_{l}g_{l}$ satisfies the MA for all $j \neq l$, $j, l = 0, \ldots, J$ at $0$ with $\alpha = 1$.
Under unconfoundedness and using sample splitting scheme for estimating $\hat p\(x\)$ and $\hat g\(x\)$ in a treatment allocation model,
\bsnumber\label{influence function multiclass policy learning}
\sqrt{n'}\(\hat\beta - \beta_0\)
= \frac{1}{\sqrt{n'}}
\sum_{i=1}^{n'}
\(\sum_{j=0}^J \lambda_j \phi_j^{*}\(X_i\)
\[
	g_j\(X_i\) + \frac{D_{ij}}{p_j\(X_i\)} \(Y_i - g_j\(X_i\)\)
\] - \beta_0\)
+ o_\mathbb{P}\(1\)
\end{split}\end{align}
where
\bsnumber\label{DML social welfare potential treatment}
\hat\beta=
\frac{1}{n'} \sum_{i=1}^{n'}
\sum_{j=0}^J \lambda_j \hat \phi_j\(X_i\)
\[
	\hat g_j\(X_i\) + \frac{D_{ij}}{\hat p_j\(X_i\)} \(Y_i - \hat g_j\(X_i\)\)
\].
\end{split}\end{align}
Additionally, suppose $g$ and $\mu$ satisfy the conditions of Theorem \ref{general k derivative}.
Let $\mathcal{D} \subset \mathbb{R}^{J + 1}_{++}$ be a compact set where $\mathbb{R}^{J + 1}_{++} \coloneqq \left\{\lambda: \lambda \in \mathbb{R}^{J + 1}, \lambda_{j} > 0 , \forall j = 0, \ldots, J\right\}$
and let $\lambda_{j}g_{j} - \lambda_{l}g_{l}$ satisfy the MA with a uniform constant.
Then \eqref{influence function multiclass policy learning} holds uniformly on $\mathcal{D}$.
\end{theorem}

\begin{theorem}\label{classification simple case}
Under the conditions in Theorem \ref{treatment effect allocation simple case} and using sample splitting scheme
to estimate $\hat \phi\(x\)$ with full observability,  
\begin{align}
&\sqrt{n'}\(\hat\beta - \beta_0\)
= \frac{1}{\sqrt{n'}} \sum_{i=1}^{n'}
\(
\sum_{j=0}^J \lambda_j \phi_j^*\(X_i\)	Y_{ij} - \beta_0
\)   + o_\mathbb{P}\(1\),\label{influence function multiclass policy learning for classification}\\
&\text{where}\quad \hat\beta=
\frac{1}{n'} \sum_{i=1}^{n'}
\sum_{j=0}^J \lambda_j \hat \phi_j\(X_i\)  Y_{ij}.\notag
\end{align}
Additionally, suppose $g$ and $\mu$ satisfy the conditions of Theorem \ref{general k derivative}.
Let $\mathcal{D} \subset \mathbb{R}^{J + 1}_{++}$ be a compact set where $\mathbb{R}^{J + 1}_{++} \coloneqq \left\{\lambda: \lambda \in \mathbb{R}^{J + 1}, \lambda_{j} > 0, \forall j = 0, \ldots, J\right\}$
and let $\lambda_{j}g_{j} - \lambda_{l}g_{l}$ satisfy the MA with a uniform constant.
Then \eqref{influence function multiclass policy learning for classification} holds uniformly on $\mathcal{D}$.
\end{theorem}

We note that \cite{luedtke2016statistical} and \cite{luedtke2020performance} utilize a formula similar to \eqref{DML social welfare potential treatment} for the binary case under empirical process settings.
To apply broadly to machine learning methods,
Theorem \ref{treatment effect allocation simple case} and Theorem \ref{classification simple case}
combine the fast convergence rate of a plug-in classifier with the insighful techniques
of \cite{chernozhukov2018double}, \cite{chernozhukov2022locally} and \cite{foster2023orthogonal} based on Neyman orthogonalization and sample splitting.
Inspired by \cite{semenova2023debiased}, the asymptotic distributions of \eqref{influence function multiclass policy learning} and \eqref{influence function multiclass policy learning for classification} can be estimated through bootstrap.

\begin{cor}\label{multiplier bootstrap}
Let the multiplier $W_{n'i}$ be i.i.d. $Exponential(1)$ random variables,
which are independent of $\(Y, X, D\)$ or $\(Y, X\)$,
or follow the $n'$ times $n'$-dimensional Dirichlet distribution with parameter vector $\(1, \ldots, 1\)$ independent of $\(Y, X, D\)$ or $\(Y, X\)$.
Under the conditions in Theorem \ref{treatment effect allocation simple case} and Theorem \ref{classification simple case},
$\sqrt{n'}\(\tilde{\beta} - \hat{\beta}\) \weakconv \mathbb{G}$,
where the $\weakconv$ symbol is defined in \eqref{conditional weak convergence in probability}. 
On the one hand, in treatment allocation problems with partial
observability,
\begin{align}\begin{split}\nonumber
\tilde{\beta} & =
\frac{1}{n'} \sum_{i=1}^{n'}
W_{n'i}\(\sum_{j=0}^J \lambda_j \hat \phi_j\(X_i\)
\[
	\hat g_j\(X_i\) + \frac{D_{ij}}{\hat p_j\(X_i\)} \(Y_i - \hat g_j\(X_i\)\)
\]\)\ \text{and}\\
\mathbb{G} & = N\(0, Var\(\sum_{j=0}^J \lambda_j \phi_j^{*}\(X_i\)
\[
	g_j\(X_i\) + \frac{D_{ij}}{p_j\(X_i\)} \(Y_i - g_j\(X_i\)\)
\]\)\). \;
\end{split}\end{align}
On the other hand, in classification problems with full observability,
\begin{align}\begin{split}\nonumber
\tilde{\beta}  =
\frac{1}{n'} \sum_{i=1}^{n'} W_{n'i}
\(\sum_{j=0}^J \lambda_j \hat \phi_j\(X_i\)  Y_{ij}\) \ \ \text{and}\ \
\mathbb{G}  = N\(0, Var\(\sum_{j=0}^J \lambda_j \phi_j^*\(X_i\)	 Y_{ij}\)\).
\end{split}\end{align}
\end{cor}

\section{Simulation and empirical application}

\subsection{Confidence intervals and hypothesis testing of ROC curves}\label{simulation data CI and testing}


To evaluate the finite sample performance of our asymptotic results on the ROC curves, 
 we first study the simplified asymptotic distribution of the ROC curves under correct specification.
 Next, we consider a synthetic data example and illustrate an important application of comparing a machine generated ROC curve with a human decision maker.

In our first simulation, we generate data according to a logit model: 
\begin{align}\begin{split}\nonumber
Y_{i} = 1\(\frac{e^{X^{T}_{i}\theta_{0}}}{1 + e^{X^{T}_{i}\theta_{0}}} > U_{i}\), \theta^{T}_{0} = \(1, 0.8, 0.5\), U_{i} \sim Uniform\(0, 1\), X_{i} \sim Epanechnikov\(1\)^{\bigotimes 3}. 
\end{split}\end{align}
The Epanechnikov distribution with the parameter $c = 1$ has the density function
\begin{align}\begin{split}\nonumber
f\(x; c = 1\) = \frac{3}{4c}\max\left\{0, 1 - \(\frac{x}{c}\)^{2}\right\}\bigg\vert_{c = 1}.
\end{split}\end{align}

We examine the empirical coverage frequencies and average 
lengths of three types of confidence intervals (hereafter abbreviated as CI) for the ROC curve of the
correctly specified logit model.
The first $CI_{bootstrap}$ is formed via the orthogonal estimator using the maximum likelihood estimator (MLE) of the logit model and the bootstrapping method in Corollary \ref{roc bootstrap}. 
The second $CI_{theoretical}$ is formed via the MLE orthogonal estimator and the analytic asymptotic distribution in \eqref{theoretical distribution of correctly specified ROC model}.
The third $CI_{model \; based}$ is formed by bootstrapping the model based ROC curve described in Proposition \ref{ROC parametric model based estimator more efficient than orthogonal estimator}.
To calculate the $CI_{bootstrap}$ under
correctly specification we do not need to reestimate the logit model for each bootstrap sample. 
The constructions of the ROC estimators and three types of CIs are summarized in Algorithm \ref{algorithm three CI}. 

\begin{algorithm}
\caption{ROC estimators and CIs}\label{algorithm three CI}
\KwData{$\left\{ \(X_{i}, Y_{i}\)\right\}_{i = 1}^{n}$, 
$type \in \left\{bootstrap, theoretical, model \; based\right\}$, FPR level grid $G$, number of bootstrap repetitions $R$.}
\KwResult{ROC curve $\[\(\alpha, \hat{\beta}\(\alpha\)\): \alpha \in G\]$, pointwise CI $\[\(\beta_{l}\(\alpha\), \beta_{u}\(\alpha\)\): \alpha \in G\]$.}
Estimate propensity score model parameter $\hat{\theta}$; \\
\eIf{$type \neq model \;based$}{
\lForEach{$\alpha \in G$}{Estimate $\hat{\beta}\(\alpha\) = \beta\(\mathbb{Q}_{n}, \hat{\theta}, \alpha\)$ via \eqref{double robust representing of roc constrained optimized value function under full observability} 
}
\eIf{$type = bootstrap$}{
\For{$r \in \{1, 2, \ldots, R\}$}{
\lForEach{$\alpha \in G$}{Estimate $\tilde{\beta}_{r}\(\alpha\) = \beta\(\tilde{\mathbb{Q}}_{n, r}, \hat{\theta}, \alpha\)$
}}
}
{
Calculate CI via \eqref{theoretical distribution of correctly specified ROC model} \;
}
}{
\lForEach{$\alpha \in G$}{
Estimate $\hat{\beta}\(\alpha\) = \beta\(\hat{\mu}, \hat{\theta}, \alpha\)$ via \eqref{roc model-based functional expression}
}
\For{$r \in \left\{1, 2, \ldots, R\right\}$}{
Re-estimate propensity score parameter $\tilde{\theta}_{r}$ using the bootstrapped sample \; 
\lForEach{$\alpha \in G$}{
Estimate $\tilde{\beta}_{r}\(\alpha\) = \beta\(\tilde{\mu}_{r}, \tilde{\theta}_{r}, \alpha\)$
}}
}
\If{$type \neq theoretical$}{
\lForEach{$\alpha \in G$}{
Calculate the reverse percentile CI via 
\begin{align}\begin{split}\nonumber
\(\beta_{l}\(\alpha\), \beta_{u}\(\alpha\)\) \coloneqq 
\(2\hat{\beta}\(\alpha\) - \tilde{\beta}\(\alpha\)_{\(1 - \tau / 2\)},
2\hat{\beta}\(\alpha\) - \tilde{\beta}\(\alpha\)_{\(\tau / 2\)}\), 
\end{split}\end{align}
where $\tilde{\beta}\(\alpha\)_{\(\tau\)}$ is the $\tau$-quantile of $\left\{\tilde{\beta}_{r}\(\alpha\)\right\}_{r = 1}^{R}$
}
}
\end{algorithm}

We use $R = 1000$ bootstrap repetitions each in 1000 Monte Carlo simulations for sample sizes $n \in \{1000, 2000, 5000, 10000\}$.
We evaluate the pointwise CI at FPR level between $0.3$ and $0.7$ with stepsize $0.05$.
Since the analytic expression for the population ROC curve is
difficult to derive, we
calculate an approximation by simulating a large sample size of $5 \cdot 10^{6}$ from the data generating model.
As shown in Table \ref{coverage and average length table}, the coverage is very close to the nominal level $95\%$.
The average lengths of $CI_{bootstrap}$ and $CI_{theoretical}$ are nearly the same, while $CI_{model \; based}$ 
is saliently the shortest. 

\begin{table}[htbp]
  \centering
  \tabcolsep = 3pt
  \renewcommand{\arraystretch}{1.2}
  \caption{Coverage and average lengths of the confidence intervals}\label{coverage and average length table}
  \begin{tabular}{lccccccccc}
    \toprule
    \multicolumn{10}{l}{\textbf{Panel A}: $n = 1000$} \\
    \midrule
    FPR level & $0.30$ & $0.35$ & $0.40$ & $0.45$ & $0.50$ & $0.55$ & $0.60$ & $0.65$ & $0.70$ \\
    \hline
    \multirow{2}{*}{$CI_{bootstrap}$} & 0.938 & 0.929 & 0.924 & 0.931 & 0.929 & 0.932 & 0.931 & 0.928 & 0.924 \\
    & (0.128) & (0.124) & (0.119) & (0.114) & (0.108) & (0.101) & (0.0932) & (0.0857) & (0.0772) \\
    \multirow{2}{*}{$CI_{theoretical}$} & 0.954 & 0.947 & 0.95 & 0.949 & 0.952 & 0.952 & 0.939 & 0.946 & 0.946 \\
    & (0.127) & (0.123) & (0.118) & (0.113 ) & (0.106) & (0.0997) & (0.0925) & (0.0847) & (0.0764) \\
    \multirow{2}{*}{$CI_{model \; based}$} & 0.931 & 0.933 & 0.934 & 0.936 & 0.939 & 0.937 & 0.936 & 0.935 & 0.936 \\
    & (0.0987) & (0.0965) & (0.0929) & (0.088) & (0.0821) & (0.0754) & (0.0681) & (0.0603) & (0.0522) \\
    \midrule
    \multicolumn{10}{l}{\textbf{Panel B}: $n = 2000$} \\
    \midrule
    FPR level & $0.30$ & $0.35$ & $0.40$ & $0.45$ & $0.50$ & $0.55$ & $0.60$ & $0.65$ & $0.70$ \\
    \hline
    \multirow{2}{*}{$CI_{bootstrap}$} & 0.954 & 0.952 & 0.938 & 0.955 & 0.931 & 0.956 & 0.939 & 0.95 & 0.929 \\
    & (0.0906) & (0.0872) & (0.0841) & (0.0801) & (0.0759) & (0.0714) & (0.0661) & (0.0601) & 0.0545 \\
    \multirow{2}{*}{$CI_{theoretical}$} & 0.965 & 0.967 & 0.956 & 0.961 & 0.954 & 0.959 & 0.951 & 0.959 & 0.956 \\
    & (0.0897) & (0.0871) & (0.0837) & (0.0798) & (0.0754) & (0.0706) & (0.0654) & (0.0599) & (0.054) \\
    \multirow{2}{*}{$CI_{model \; based}$} & 0.947 & 0.95 & 0.952 & 0.953 & 0.955 & 0.955 & 0.956 & 0.958 & 0.958 \\
    & (0.0708) & (0.0692) & (0.0666) & (0.0631) & (0.0589) & (0.0541) & (0.0489) & (0.0433) & (0.0374) \\
    \midrule
    \multicolumn{10}{l}{\textbf{Panel C}: $n = 5000$} \\
    \midrule
    FPR level & $0.30$ & $0.35$ & $0.40$ & $0.45$ & $0.50$ & $0.55$ & $0.60$ & $0.65$ & $0.70$ \\
    \hline
    \multirow{2}{*}{$CI_{bootstrap}$} & 0.936 & 0.931 & 0.936 & 0.942 & 0.942 & 0.953 & 0.952 & 0.946 & 0.933 \\
    & (0.0571) & (0.0555) & (0.0532) & (0.0506) & (0.0478) & (0.045) & (0.0416) & (0.0381) & (0.0345) \\
    \multirow{2}{*}{$CI_{theoretical}$} & 0.944 & 0.938 & 0.943 & 0.943 & 0.953 & 0.96 & 0.956 & 0.951 & 0.95 \\
    & (0.0568) & (0.0551) & (0.053) & (0.0505) & (0.0477) & (0.0447) & (0.0414) & (0.0379) & (0.0342) \\
    \multirow{2}{*}{$CI_{model \; based}$} & 0.944 & 0.938 & 0.943 & 0.943 & 0.953 & 0.96 & 0.956 & 0.951 & 0.95 \\
    & (0.0453) & (0.0443) & (0.0427) & (0.0404) & (0.0377) & (0.0347) & (0.0314) & (0.0277) & (0.024) \\
    \midrule
    \multicolumn{10}{l}{\textbf{Panel D}: $n = 10000$} \\
    \midrule
    FPR level & $0.30$ & $0.35$ & $0.40$ & $0.45$ & $0.50$ & $0.55$ & $0.60$ & $0.65$ & $0.70$ \\
    \hline
    \multirow{2}{*}{$CI_{bootstrap}$} & 0.945 & 0.947 & 0.943 & 0.96 & 0.937 & 0.943 & 0.943 & 0.961 & 0.949 \\
    & (0.0404) & (0.039) & (0.0376) & (0.0359) & (0.0337) & (0.0317) & (0.0293) & (0.0269) & (0.0243) \\
    \multirow{2}{*}{$CI_{theoretical}$} & 0.946 & 0.953 & 0.952 & 0.964 & 0.95 & 0.956 & 0.951 & 0.96 & 0.949 \\
    & (0.0402) & (0.039) & (0.0375) & (0.0357) & (0.0338) & (0.0316) & (0.0293) & (0.0268) & (0.0242) \\
    \multirow{2}{*}{$CI_{model \; based}$} & 0.946 & 0.951 & 0.948 & 0.951 & 0.947 & 0.946 & 0.945 & 0.944 & 0.943 \\
    & (0.0321) & (0.0314) & (0.0303) & (0.0287) & (0.0268) & (0.0246) & (0.0222) & (0.0197) & (0.017) \\
    \bottomrule
  \end{tabular}
\end{table}

Next, we consider a problem of comparing algorithms with decision making of ``doctors'' using synthetic data.  
Consider i.i.d. observations $\{\(Y_i, W_i, X_i\), i = 1, \ldots, n\}$, such that $Y_i$ is the true outcome, $W_i$ is the doctor's diagnosis, 
and $X_i$ are observable features. 
To compare an algorithm with a doctor, 
we test the null hypothesis that the human TPR/FPR pair lies on the machine ROC curve against the alternative hypotheses that the human pair lies above or below the ROC curve.
We generate synthetic data according to the following model
\begin{align}\begin{split}\nonumber
& Y_{i} = 1\(\frac{e^{X^{T}_{i}\theta_{0}}}{1 + e^{X^{T}_{i}\theta_{0}}} > U_{i}\), \; \theta^{T}_{0} = \(1, 0.8, 0.5, 1\), \\
& W^{E}_{i} = 1\(\frac{e^{X^{T}_{i}\theta_{E}}}{1 + e^{X^{T}_{i}\theta_{E}}} > 0.5\), \; \theta_{E} = \theta_{0}, \;
W^{I}_{i} = 1\(\frac{e^{X^{T}_{i}\theta_{I}}}{1 + e^{X^{T}_{i}\theta_{I}}} > 0.5\), \; \theta_{I} = \(1, 0, 0.5, 0\), \\
& X_{i} \sim N\(0, I_{4}\), U_{i} \sim Uniform\(0, 1\),
\end{split}\end{align}
where $I_{4}$ is the $4$-dimensional identity matrix.
The $W^{E}_{i}$ decisions are made by an ideal ``experienced doctor'' who knows the data generating process, while the $W^{I}_{i}$ decisions are made by an ``inexperienced doctor'' who only utilizes $X_{1}$ and $X_{3}$.
Theorem \ref{general k derivative} is applicable to this model.

A human decision maker (e.g. doctor, judge and admission officer) is unlikely to make a large number of decisions.
We set the simulation sample size $n$ to	 $\{100, 250, 500, 1000, 2000, 5000, 10000\}$.
We first use Algorithm \ref{algorithm comparison CI} with $R = 1000$ to form a CI for the difference $\beta\(\alpha_{doc}\) - \beta_{doc}$. 
We then determine whether $0$ is within, above or below the interval.
In the second or third cases,
we reject the hypothesis of equal performance between the doctor and the machine model in favor of either the doctor 
or the model. 

\begin{algorithm}
\caption{ROC and TPR/FPR comparison CI}\label{algorithm comparison CI}
\KwData{$\left\{ \(X_{i}, W_{i}, Y_{i}\)\right\}_{i = 1}^{n}$, the number of bootstrap repetitions $R$.}
\KwResult{CI $\(\beta_{l}, \beta_{u}\)$.}
Estimate propensity score model parameter $\hat{\theta}$ \; 
Estimate $\hat{\beta}\(\hat{\alpha}_{doc}\) = \beta\(\mathbb{Q}_{n}, \hat{\theta}, \hat{\alpha}_{doc}\)$ via \eqref{double robust representing of roc constrained optimized value function under full observability}; \\
\For{$r \in \left\{1, 2, \ldots, R\right\}$}{
Re-estimate propensity score parameter $\tilde{\theta}_{r}$ using the bootstrapped sample \; 
Estimate $\tilde{\beta}_{r}\(\tilde{\alpha}_{doc}\) = \beta\(\tilde{\mathbb{Q}}_{n}, \tilde{\theta}_{r}, \tilde{\alpha}_{doc}\)$ and calculate $\(\tilde{\alpha
}_{doc}, \tilde{\beta}_{doc}\)$ \;
}
Calculate the reverse percentile CI via
\begin{align}\begin{split}\nonumber
\(2\(\hat{\beta}\(\hat{\alpha}_{doc}\) - \hat{\beta}_{doc}\) - \(\tilde{\beta}\(\tilde{\alpha}_{doc}\) - \tilde{\beta}_{doc}\)_{\(1 - \tau / 2\)},
2\(\hat{\beta}\(\hat{\alpha}_{doc}\) - \hat{\beta}_{doc}\) - \(\tilde{\beta}\(\tilde{\alpha}_{doc}\) - \tilde{\beta}_{doc}\)_{\(\tau / 2\)}\).
\end{split}\end{align}
\end{algorithm}

\begin{table}[!htp]
  \centering
  \tabcolsep = 3pt
  \renewcommand{\arraystretch}{1.2}
  \caption{ROC curve and doctor test rejection frequencies}
  \label{roc doctor test table}
  \begin{tabular}{llccccccc}
    \toprule
    \multicolumn{9}{l}{\textbf{Panel A}: nominal level 5\%} \\
    \midrule
    sample size $n$ & & $100$ & $250$ & $500$ & $1000$ & $2000$ & $5000$ & $10000$ \\
    \hline
    \multirow{3}{*}{experienced doctor} & reject in favor of model & 0.006 & 0.001 & 0.0 & 0.0 & 0.0 & 0.0 & 0.0 \\
    & reject in favor of doctor & 0.166 & 0.398 & 0.712 & 0.936 & 1.0 & 1.0 & 1.0 \\
    & fail to reject & 0.827 & 0.601 & 0.288 & 0.064 & 0.0 & 0.0 & 0.0 \\
    \hline
    \multirow{3}{*}{inexperienced doctor} & reject in favor of model & 0.212 & 0.344 & 0.525 & 0.79 & 0.981 & 1.0 & 1.0 \\
    & reject in favor of doctor & 0.004 & 0.001 & 0.0 & 0.0 & 0.0 & 0.0 & 0.0 \\
    & fail to reject & 0.783 & 0.654 & 0.475 & 0.21 & 0.019 & 0.0 & 0.0 \\
    \hline
    \multirow{3}{*}{model-like case} & reject in favor of model & 0.027 & 0.034 & 0.026 & 0.027 & 0.02 & 0.014 & 0.017 \\
    & reject in favor of doctor & 0.026 & 0.019 & 0.023 & 0.016 & 0.02 & 0.015 & 0.021 \\
    & fail to reject & 0.946 & 0.947 & 0.951 & 0.957 & 0.96 & 0.971 & 0.962 \\
    \midrule
    \multicolumn{9}{l}{\textbf{Panel B}: nominal level 10\%} \\
    \midrule
    sample size $n$ & & $100$ & $250$ & $500$ & $1000$ & $2000$ & $5000$ & $10000$ \\
    \hline
    \multirow{3}{*}{experienced doctor} & reject in favor of model & 0.008 & 0.004 & 0.0 & 0.0 & 0.0 & 0.0 & 0.0 \\
    & reject in favor of doctor & 0.262 & 0.524 & 0.806 & 0.97 & 1.0 & 1.0 & 1.0 \\
    & fail to reject & 0.728 & 0.471 & 0.194 & 0.03 & 0.0 & 0.0 & 0.0 \\
    \hline
    \multirow{3}{*}{inexperienced doctor} & reject in favor of model & 0.295 & 0.454 & 0.627 & 0.873 & 0.995 & 1.0 & 1.0 \\
    & reject in favor of doctor & 0.008 & 0.002 & 0.0 & 0.0 & 0.0 & 0.0 & 0.0 \\
    & fail to reject & 0.696 & 0.544 & 0.373 & 0.127 & 0.005 & 0.0 & 0.0 \\
    \hline
    \multirow{3}{*}{model-like case} & reject in favor of model & 0.073 & 0.067 & 0.055 & 0.053 & 0.052 & 0.041 & 0.032 \\
    & reject in favor of doctor & 0.049 & 0.04 & 0.037 & 0.036 & 0.044 & 0.035 & 0.053 \\
    & fail to reject & 0.874 & 0.891 & 0.908 & 0.911 & 0.904 & 0.924 & 0.915 \\
    \bottomrule
  \end{tabular}
\end{table}

Table \ref{roc doctor test table} shows the empirical rejection frequencies for nominal 5\% and 10\% tests.
In the first two entries of both panels, 
the ``experienced doctor'' and the ``inexperienced doctor'' are compared with a misspecified logit model using only $\(X_{1}, X_{2}, X_{3}\)$.
Additionally, in the last entry of both panels labelled ``model-like case'',
we compare the same ``experienced doctor'' with a correctly specified and estimated logit model using all of $\(X_{1}, X_{2}, X_{3}, X_{4}\)$.
When the sample size is only $100$, the test has low power but rarely rejects in the wrong direction. 
A high power level is achieved when the sample size reaches $500$ or $1000$. 
The ``model-like case'' represents the null hypothesis of equal performance between
 the ``experienced doctor'' and the model.
As shown in Table \ref{roc doctor test table}, the empirical size of the test is close to the nominal level. 
The simulation results validate the finite sample quality of our asymptotic approximation theory. 

\subsection{Real data application and simulation}\label{empirical analysis}

In this section, we apply the plug-in multi-class classification method in section \ref{fast convergence rates} to a voting dataset of the August 2006 primary
election in Michigan.
This dataset originates from the field experiment reported in \cite{gerber2008social}.\footnote{We download the dataset from \url{https://github.com/gsbDBI/ExperimentData/tree/master/Social}.}
The optimal voting incentive policy has been studied in the preprint version of \cite{zhou2023offline} using their decision-tree search method.\footnote{See the
preprint version of \cite{zhou2023offline} at \url{https://arxiv.org/pdf/1810.04778}.}

The voting dataset contains 180002 electorates from households spanning Michigan.
Following \cite{zhou2023offline}, we use $10$ features in the dataset: year of birth, sex, household size, city, turnout for 2000, 2002 and 2004 primary and general elections.
We experiment with applying both the multilayer perceptron (MLP) algorithm and the causal forest (CF) algorithm (developed in
\cite{wager2018estimation} and \cite{athey2019generalized}) to the voting data set.

There are four treatment groups and a control group in the voting field experiment.
Hereafter, the control group will be called \textbf{Control}, and the four treatment groups will be called \textbf{Civic Duty}, \textbf{Hawthorne}, \textbf{Self} and \textbf{Neighbors}.
See \cite{gerber2008social} for the specific meaning of those treatments.
Electorates are assigned to the five groups at random with probabilities $\(\frac{5}{9}, \frac{1}{9}, \frac{1}{9}, \frac{1}{9}, \frac{1}{9}\)$.

First, we generate a synthetic dataset similar to the voting experiment 
in order to validate the multiplier bootstrapping method described in Corollary \ref{multiplier bootstrap}. 
The data generating process is a MLP model with $1$ hidden layer of size $8$, ReLU activation and a softmax function (i.e. standard logistic function). 
The outcome is then generated according to
\begin{align}\begin{split}\nonumber
& \mathbb{P}\(Y_{ij} = 1 \vert X_{i}\) = \text{softmax} \circ \(W_{2} \circ \text{ReLU} \circ \(W_{1}\(X^{T}_{i}, j\)^{T} + b_{1}\) + b_{2}\), \quad j \in \{0, 1, 2, 3, 4\}, \\
& Y_{i} = \sum^{4}_{0}Y_{ij}D_{ij}, \quad D_{i} \sim multinomial\(1, \(\frac{5}{9}, \frac{1}{9}, \frac{1}{9}, \frac{1}{9}, \frac{1}{9}\)\),
\end{split}\end{align}
where
the treatment assignment probabilities are independent of the characteristics $X_i$.  We let $X_{i} \sim N\(\mu, \Sigma\)$, where $\mu$ and $\Sigma$ are specified in Appendix \ref{parameters}. 
The weights $W_{1}$, $W_{2}$, $b_{1}$ and $b_{2}$ are randomly drawn. 
The values of the weights matrices are also reported in Appendix \ref{parameters}. 

Using this data generating process, we first draw a large sample of size $10^{7}$ to compute a close approximation of the social welfare potential $\beta_{0}$
around $0.311$ .
An overparameterized MLP with $2$ hidden layers of size (48, 24) is used in estimation.
We then use
the multiplier bootstrap method ($1000$ repetitions) in section \ref{inference on social welfare potential function} to construct CI
based on a sample splitting scheme of $20000$ training and $5000$ testing observations.
To the best of our knowledge, there is no current guarantee of the almost sure uniform convergence rate for neural networks.
However, inspired by the theory in \cite{bos2022convergence}, it is reasonable to expect a fast uniform convergence rate of the MLP classifier on
 a large probability set of $X_{i}$.
In $1000$ simulations, our multiplier bootstrap CI has an empirical coverage rate of $0.939$, close to the nominal $0.95$ level.
The average length of the CI is around $0.0423$. 

Next, two machine learning methods are applied to this data.
The first method is a MLP with three hidden layers of size $\(64, 64, 64\)$ and ReLU activation (i.e. $\max\{0, x\}$). The second method is a CF with $500$ trees where
the minimum number of samples for each group in each leaf is equal to $600$. 
We use the Python package scikit-learn \citep{scikit-learn} for the MLP and the Python implement from the EconML package \citep{econml} for the CF.
The treatment group of each electorate is concatenated to the MLP input features as a $1$-dimensional variable taking value $0$ to $4$ (in the order listed above).

We follow \cite{zhou2023offline} and use $5$-fold cross validation. 
We report the estimated difference and CI of the 
difference between the plug-in policy and the control group based on 
\eqref{DML social welfare potential treatment}. 
 Note that \eqref{influence function multiclass policy learning} is Neyman orthogonal and amenable to a sample splitting scheme 
described in Algorithm \ref{cross validation multi} with $R = 1000$ and $K = 5$.

\begin{algorithm}
\caption{Policy gain estimation based on cross validation}\label{cross validation multi}
\KwData{$\left\{\(X_{i}, D_{i}, Y_{i}\)\right\}_{i = 1}^{n}$, the number of bootstrap iterations $R$, cross validation folds $K$, policy for comparison $\hat{\phi}_{c}\(\cdot\)$.}
\KwResult{gain $\hat{\beta}$, CI $\(\beta_{l}, \beta_{u}\)$.}
Split the sample $\left\{\(X_{i}, Y_{i}\)\right\}_{i = 1}^{n}$ to $K$ folds \; 
\For{$k \in \{1, 2, \ldots, K\}$}{
Estimate propensity score model $\hat{p}_{k}$ and conditional welfare $\hat{g}_{k}$ using the data from folds $\{1, 2, \ldots, k - 1, k + 1, \ldots, K\}$ \; 
Estimate $\hat{\beta}_{k}$ via \eqref{DML social welfare potential treatment} on the $k$-th fold of the data \; 
}
Estimate $\hat{\beta} = \frac{1}{K}\sum_{k = 1}^{K}\hat{\beta}_{k}$ \; 
\For{$r \in \left\{1, 2, \ldots, R\right\}$}{
Estimate 
\begin{align}\begin{split}\nonumber
\tilde{\beta}_{r} & =
\frac{1}{n} \sum_{k = 1}^{K}\sum_{i=1}^{n_{k}}
W_{n_{k}i}^r\(\sum_{j=0}^J \lambda_j \(\hat \phi_{kj}\(X_{ki}\) - \hat{\phi}_{cj}\(X_{ki}\)\)
\[
	\hat g_{kj}\(X_{ki}\) + \frac{D_{kij}}{\hat p_{kj}\(X_{ki}\)} \(Y_{ki} - \hat g_{kj}\(X_{ki}\)\)
\]\), 
\end{split}\end{align}
where $\left\{\(X_{ki}, D_{ki}, Y_{ki}\)\right\}_{i = 1}^{n_{k}}$ is the $k$-th fold of the data \; 
}
Calculate $\(\beta_{l}, \beta_{u}\)$ via the reverse percentile method. 
\end{algorithm}

Both the MLP and the CF assign nearly all electorates to the \textbf{Neighbors} group.
Meanwhile, the two models give value differences of $0.0865$ and $0.0837$ compared to \textbf{control}. 
The CIs for the true $\beta_{0}$ are reported in Table \ref{voting table}. 
These findings are similar to the value difference around $0.0864$ using the ERM policy of \cite{zhou2023offline}.
Their preprint version 
reported a value difference of $0.082$ for the same ERM policy.
The difference is due to the randomness from sample splitting and model training. 
The policy gain results need to be interpreted with caution. 
As shown in Table \ref{voting table}, 
the CIs for the policy gain between the MLP / CF and the ERM policy include $0$. 
Given that the performance of our proposed method is statistically on par with ERM, 
the method that we propose provides a viable computationally practical alternative approach to ERM policy learning.

\begin{table}[!htp]
  \centering
  \tabcolsep = 3pt
  \renewcommand{\arraystretch}{1.2}
  \caption{Policy gain estimation comparison}
  \label{voting table}
  \begin{tabular}{lccc}
    \toprule
     & ERM & MLP & CF \\ 
    \midrule
    v.s. \textbf{control} & 0.0864 $\(0.0769, 0.0954\)$ & 0.0865 $\(0.0753, 0.0965\)$ & 0.0837 $\(0.0748, 0.0937\)$ \\
    \hline 
    v.s. ERM & - & 0.000271 $\(-0.00232, 0.00302\)$ & -0.00261 $\(-0.00669, 0.00173\)$ \\
    \bottomrule
  \end{tabular}
  \fnote{\textit{Notes}: In this table, we first compare the ERM policy from \cite{zhou2023offline} and our method based on the MLP and the CF with \textbf{control}. 
  When comparing the ERM policy and \textbf{control}, we apply \eqref{DML social welfare potential treatment} but let $\hat{\phi}$ be the fixed ERM policy. 
  We use the CF algorithm to construct $\hat{g}$ to evaluate the comparison between the ERM policy and \textbf{control}.
  CIs are then constructed for the policy gain difference between the MLP / CF and the ERM. 
  99\% confidence intervals are reported in the parentheses.
  }
\end{table}

\section{Conclusion}
In this paper, we explore a functional differentiability approach for a class of statistical optimal allocation
problems.
Inspired by \cite{chernozhukov2018sorted}, Hadamard differentiability is facilitated by a study of the general properties of the sorting operator.
We provide a derivation that makes use of the concept of Hausdorff measure and the area and coarea formulas from geometric measure theory.
Based on our general Hadamard differentiability results, in section \ref{constrained and roc} we derive
the asymptotic properties of the plug-in estimators for both the value function process of a binary constrained optimal allocation problem and the ROC curve from the functional delta method.
When the first step propensity score model is correctly specified, a computationally feasible bootstrap procedure is validated for the plug-in ROC estimator.

Importantly, we build on the partial convexity of the social welfare potential function, i.e. the value function of the optimal allocation problem, 
to demonstrate the degeneracy of the first order derivative
of the social welfare potential function with respect to the policy. 
These intriguing results provide us with insights
when we
combine techniques from the literature of nonsmooth method of moment, plug-in classification and the recent development of double / debiased machine learning
to develop a debiased estimator of the social welfare potential function
in section \ref{fast convergence rates}.
Here, the conditions required for Hadamard differentiability validate the margin assumption, which leads to a faster convergence rate.

\section{Acknowledgments}

We thank the editor and two anonymous referees for insightful comments.
We also acknowledge helpful and encouraging conversations with
collaborators and colleagues
including Chunrong Ai, Timothy Armstrong, Haoge Chang, Qihui Chen, Xiaohong Chen, 
Victor Chernozhukov, Michael Fan, Yanqin Fan, Yue Fang, Yingjie Feng, 
Robin Han, Toru Kitagawa, Michael P. Leung, Hongjun Li, Jessie Li, Ruixuan Liu, Ye Luo, Thomas MaCurdy, Chen Qiu, Shuyang Sheng, Zhentao Shi, Kurt Sweat, Liangjun Su, Ke Tang, Guanyi Wang, Haitian Xie and Ping Yu, as well as participants at various seminars and conferences.
We are most grateful to Yuhao Xue for enlightening us
with his knowledge of geometry measure theory and to Yichuan Zhang for research assistance.
We acknowledge funding support from the National Science Foundation (SES 1658950 to Han Hong).

\FloatBarrier
\phantomsection
\addcontentsline{toc}{section}{References}
{
\bibliography{forarxiv}

@article{elliott2013predicting,
  title={Predicting binary outcomes},
  author={Elliott, Graham and Lieli, Robert P},
  journal={Journal of Econometrics},
  volume={174},
  number={1},
  pages={15--26},
  year={2013},
  publisher={Elsevier}
}

@article{sherman1993limiting,
  title={The limiting distribution of the maximum rank correlation estimator},
  author={Sherman, Robert P},
  journal={Econometrica: Journal of the Econometric Society},
  volume={61},
  number={1},
  pages={123--137},
  year={1993},
  publisher={JSTOR}
}

@article{vuong1989likelihood,
  title={Likelihood ratio tests for model selection and non-nested hypotheses},
  author={Vuong, Quang H},
  journal={Econometrica: Journal of the Econometric Society},
  volume={57},
  number={2},
  pages={307--333},
  year={1989},
  publisher={JSTOR}
}

@article{QJEbail,
  title={Human Decisions and Machine Predictions},
  author={Kleinberg, Jon and Lakkaraju, Himabindu and Leskovec, Jure and Ludwig, Jens and Mullainathan, Sendhil},
  journal={Quarterly Journal of Economics},
  volume={133},
  number={1},
  pages={237--293},
  year={2018},
  publisher={Oxford University Press}
}

@article{newey_mcfadden,
  title={Large sample estimation and hypothesis testing},
  author={Newey, Whitney K and McFadden, Daniel},
  journal={Handbook of Econometrics},
  volume={4},
  pages={2111--2245},
  year={1994},
  publisher={Elsevier}
}

@article{chen2003estimation,
  title={Estimation of semiparametric models when the criterion function is not smooth},
  author={Chen, Xiaohong and Linton, Oliver and Van Keilegom, Ingrid},
  journal={Econometrica},
  volume={71},
  number={5},
  pages={1591--1608},
  year={2003},
  publisher={Wiley Online Library}
}

@article{chernozhukov2018sorted,
  title={The sorted effects method: discovering heterogeneous effects beyond their averages},
  author={Chernozhukov, Victor and Fern{\'a}ndez-Val, Iv{\'a}n and Luo, Ye},
  journal={Econometrica},
  volume={86},
  number={6},
  pages={1911--1938},
  year={2018},
  publisher={Wiley Online Library}
}

@article{fang2019inference,
  title={Inference on directionally differentiable functions},
  author={Fang, Zheng and Santos, Andres},
  journal={Review of Economic Studies},
  volume={86},
  number={1},
  pages={377--412},
  year={2019},
  publisher={Oxford University Press}
}

@article{chernozhukov2010quantile,
  title={Quantile and probability curves without crossing},
  author={Chernozhukov, Victor and Fern{\'a}ndez-Val, Iv{\'a}n and Galichon, Alfred},
  journal={Econometrica},
  volume={78},
  number={3},
  pages={1093--1125},
  year={2010},
  publisher={Wiley Online Library}
}

@article{chen2019inference,
  title={Inference on functionals under first order degeneracy},
  author={Chen, Qihui and Fang, Zheng},
  journal={Journal of Econometrics},
  volume={210},
  number={2},
  pages={459--481},
  year={2019},
  publisher={Elsevier}
}

@book{van2000asymptotic,
  title={Asymptotic statistics},
  author={van\;der\;Vaart, A.W.},
  year={2000},
  series={Cambridge Series in Statistical and Probabilistic Mathematics},
  volume={3},
  publisher={Cambridge university press}
}

@book{munkres2000topology,
  title={Topology},
  author={Munkres, James R},
  year={2000},
  publisher={Prentice Hall, Upper Saddle River, NJ}
}

@article{figalli2008simple,
  title={A simple proof of the Morse-Sard theorem in Sobolev spaces},
  author={Figalli, Alessio},
  journal={Proceedings of the American Mathematical Society},
  volume={136},
  number={10},
  pages={3675--3681},
  year={2008}
}

@book{evans2018measure,
  title={Measure theory and fine properties of functions, revised edition},
  author={Evans, Lawrence C and Garzepy, Ronald F},
  year={2015},
  series={Textbooks in Mathematics},
  publisher={Chapman and Hall/CRC}
}

@book{falconer1986geometry,
  title={The geometry of fractal sets},
  author={Falconer, Kenneth J},
  volume={85},
  series={Cambridge Tracts in Mathematics},
  year={1986},
  publisher={Cambridge University Press}
}

@book{resnick2008extreme,
  title={Extreme values, regular variation, and point processes},
  author={Resnick, Sidney I},
  volume={4},
  series={Springer Series in Operations Research and Financial Engineering},
  year={1987},
  publisher={Springer Science \& Business Media}
}

@article{sasaki2015quantile,
  title={What do quantile regressions identify for general structural functions?},
  author={Sasaki, Yuya},
  journal={Econometric Theory},
  volume={31},
  number={5},
  pages={1102--1116},
  year={2015},
  publisher={Cambridge University Press}
}

@article{kim1990cube,
  title={Cube root asymptotics},
  author={Kim, Jeankyung and Pollard, David},
  journal={Annals of Statistics},
  volume={18},
  number={1},
  pages={191--219},
  year={1990},
  publisher={JSTOR}
}

@article{vallee2019marketplace,
  title={Marketplace lending: A new banking paradigm?},
  author={Vallee, Boris and Zeng, Yao},
  journal={Review of Financial Studies},
  volume={32},
  number={5},
  pages={1939--1982},
  year={2019},
  publisher={Oxford University Press}
}

@article{hall2004nonparametric,
  title={Nonparametric confidence intervals for receiver operating characteristic curves},
  author={Hall, Peter and Hyndman, Rob J and Fan, Yanan},
  journal={Biometrika},
  volume={91},
  number={3},
  pages={743--750},
  year={2004},
  publisher={Oxford University Press}
}

@article{li1999semiparametric,
  title={Semiparametric inference for a quantile comparison function with applications to receiver operating characteristic curves},
  author={Li, Gang and Tiwari, Ram C and Wells, Martin T},
  journal={Biometrika},
  volume={86},
  number={3},
  pages={487--502},
  year={1999},
  publisher={Oxford University Press}
}

@article{lloyd1998using,
  title={Using smoothed receiver operating characteristic curves to summarize and compare diagnostic systems},
  author={Lloyd, Chris J},
  journal={Journal of the American Statistical Association},
  volume={93},
  number={444},
  pages={1356--1364},
  year={1998},
  publisher={Taylor \& Francis}
}

@article{hsieh1996nonparametric,
  title={Nonparametric and semiparametric estimation of the receiver operating characteristic curve},
  author={Hsieh, Fushing and Turnbull, Bruce W},
  journal={Annals of Statistics},
  volume={24},
  number={1},
  pages={25--40},
  year={1996},
  publisher={Institute of Mathematical Statistics}
}

@article{bertail2008bootstrapping,
  title={On bootstrapping the ROC curve},
  author={Bertail, Patrice and Cl{\'e}men{\c{c}}con, St{\'e}phan and Vayatis, Nicolas},
  journal={Advances in Neural Information Processing Systems},
  volume={21},
  pages={137--144},
  year={2008}
}

@unpublished{xuguanglulecturenote2018,
  title={Lecture Notes for Mathematical Analysis},
  author={Xuguang Lu},
  month={September},
  year={2019},
  note={Tsinghua University}, 
}

@unpublished{robertjerrard,
  title={Lecture note for Geometric Measure Theory},
  author={Robert L. Jerrard},
  year={2013},
  note={University of Toronto, http://www.math.toronto.edu/rjerrard/1501/gmt.html},
}

@book{van2023weak,
  title={Weak Convergence and Empirical Processes: With Applications to Statistics},
  author={van der Vaart, AW and Wellner, Jon A},
  edition={2},
  year={2023},
  series={Springer Series in Statistics},
  publisher={Springer Nature}
}

@article{athey2021policy,
  title={Policy learning with observational data},
  author={Athey, Susan and Wager, Stefan},
  journal={Econometrica},
  volume={89},
  number={1},
  pages={133--161},
  year={2021},
  publisher={Wiley Online Library}
}

@article{chernozhukov2018double,
  title={Double/debiased machine learning for treatment and structural parameters},
  author={Chernozhukov, Victor and Chetverikov, Denis and Demirer, Mert and Duflo, Esther and Hansen, Christian and Newey, Whitney and Robins, James},
  year={2018},
  volume={21},
  number={1},
  pages={C1-C68},
  journal={The Econometrics Journal},
  publisher={Oxford University Press Oxford, UK}
}

@article{hirano2009asymptotics,
  title={Asymptotics for statistical treatment rules},
  author={Hirano, Keisuke and Porter, Jack R},
  journal={Econometrica},
  volume={77},
  number={5},
  pages={1683--1701},
  year={2009},
  publisher={Wiley Online Library}
}

@book{lehmann2022testing,
  title={Testing statistical hypotheses},
  author={Lehmann, Erich Leo and Romano, Joseph P},
  series={Springer Texts in Statistics},
  edition={4},
  year={2022},
  publisher={Springer}
}

@book{villani2009optimal,
  title={Optimal transport: old and new},
  author={Villani, C{\'e}dric},
  volume={338},
  series={Grundlehren der mathematischen Wissenschaften},
  year={2009},
  publisher={Springer}
}

@article{bolte2021conservative,
  title={Conservative set valued fields, automatic differentiation, stochastic gradient methods and deep learning},
  author={Bolte, J{\'e}r{\^o}me and Pauwels, Edouard},
  journal={Mathematical Programming},
  volume={188},
  pages={19--51},
  year={2021},
  publisher={Springer}
}

@book{guide2006infinite,
  title={Infinite dimensional analysis, A Hitchhiker’s Guide},
  author={Aliprantis, Charalambos D and Border, Kim C},
  year={2006},
  publisher={Springer}
}

@book{rockafellar2009variational,
  title={Variational analysis},
  author={Rockafellar, R Tyrrell and Wets, Roger J-B},
  volume={317},
  series={Grundlehren der mathematischen Wissenschaften}, 
  year={2009},
  publisher={Springer Science \& Business Media}
}

@book{niculescu2018convex,
  title={Convex Functions and Their Applications: A Contemporary Approach},
  author={Niculescu, Constantin P and Persson, Lars-Erik},
  year={2018},
  edition={2},
  series={CMS/CAIMS Books in Mathematics}, 
  volume={14},
  publisher={Springer}
}

@article{lindenstrauss2003frechet,
  title={On Fr{\'e}chet differentiability of Lipschitz maps between Banach spaces},
  author={Lindenstrauss, Joram and Preiss, David},
  journal={Annals of Mathematics},
  volume={157},
  pages={257--288},
  year={2003},
  publisher={JSTOR}
}

@book{dellacherie1979probabilities,
  title={Probabilities and potential, A},
  author={Dellacherie, Claude and Meyer, Paul-Andr\'{e}},
  year={1978},
  publisher={North-Holland},
  volume={29},
  series={North-Holland Mathematics Studies}
}

@article{stegall1978duality,
  title={The duality between Asplund spaces and spaces with the Radon-Nikodym property},
  author={Stegall, Charles},
  journal={Israel Journal of Mathematics},
  volume={29},
  pages={408--412},
  year={1978},
  publisher={Springer}
}

@article{namioka1975banach,
  title={Banach spaces which are Asplund spaces},
  author={Namioka, Isaac and Phelps, Robert R},
  journal={Duke Mathematical Journal},
  volume={42},
  number={4},
  pages={735--750},
  year={1975}
}

@article{averbuh1968razlivcnye,
  title={The various definitions of the derivative in linear topological spaces},
  author={Averbuh, VI and Smoljanov, OG},
  journal={Russian Mathematical Surveys},
  volume={23},
  number={4},
  pages={67--113},
  year={1968}
}

@article{asplund1968frechet,
  title={Fr{\'e}chet differentiability of convex functions},
  author={Asplund, Edgar},
  journal={Acta Mathematica},
  volume={121},
  pages={31--47},
  year={1968},
  publisher={Springer}
}

@article{milgrom2002envelope,
  title={Envelope theorems for arbitrary choice sets},
  author={Milgrom, Paul and Segal, Ilya},
  journal={Econometrica},
  volume={70},
  number={2},
  pages={583--601},
  year={2002},
  publisher={Wiley Online Library}
}

@article{preiss1984frechet,
  title={Fr{\'e}chet differentiation of convex functions in a Banach space with a separable dual},
  author={Preiss, D and Zaj{\'\i}{\v{c}}ek, L},
  journal={Proceedings of the American Mathematical Society},
  volume={91},
  number={2},
  pages={202--204},
  year={1984}
}

@book{lindenstrauss2012frechet,
  title={Fr{\'e}chet Differentiability of Lipschitz Functions and Porous Sets in Banach Spaces},
  author={Lindenstrauss, Joram and Preiss, David and Ti\v{s}er, Jaroslav},
  volume={179},
  series={Annals of Mathematics Studies}, 
  year={2012},
  publisher={Princeton University Press}
}

@article{manski2004statistical,
  title={Statistical treatment rules for heterogeneous populations},
  author={Manski, Charles F},
  journal={Econometrica},
  volume={72},
  number={4},
  pages={1221--1246},
  year={2004},
  publisher={Wiley Online Library}
}

@article{audibert2007fast,
  title={Fast learning rates for plug-in classifiers},
  author={Audibert, Jean-Yves and Tsybakov, Alexandre B},
  journal={Annals of Statistics},
  volume={35},
  number={2},
  pages={608--633},
  year={2007}
}

@article{chernozhukov2022locally,
  title={Locally robust semiparametric estimation},
  author={Chernozhukov, Victor and Escanciano, Juan Carlos and Ichimura, Hidehiko and Newey, Whitney K and Robins, James M},
  journal={Econometrica},
  volume={90},
  number={4},
  pages={1501--1535},
  year={2022},
  publisher={Wiley Online Library}
}

@book{devroye1996probabilistic,
  title={A Probabilistic Theory of Pattern Recognition},
  author={Devroye, Luc and Gy{\"o}rfi, L{\'a}szl{\'o} and Lugosi, G{\'a}bor},
  series={Stochastic Modelling and Applied Probability},
  year={1996},
  publisher={Springer New York},
  volume={31}
}

@article{tsybakov2004optimal,
  title={Optimal aggregation of classifiers in statistical learning},
  author={Tsybakov, Alexander B},
  journal={Annals of Statistics},
  volume={32},
  number={1},
  pages={135--166},
  year={2004},
  publisher={Institute of Mathematical Statistics}
}

@article{mammen1999smooth,
  title={Smooth discrimination analysis},
  author={Mammen, Enno and Tsybakov, Alexandre B},
  journal={Annals of Statistics},
  volume={27},
  number={6},
  pages={1808--1829},
  year={1999},
  publisher={Institute of Mathematical Statistics}
}

@article{kitagawa2018should,
  title={Who should be treated? empirical welfare maximization methods for treatment choice},
  author={Kitagawa, Toru and Tetenov, Aleksey},
  journal={Econometrica},
  volume={86},
  number={2},
  pages={591--616},
  year={2018},
  publisher={Wiley Online Library}
}

@book{federer2014geometric,
  title={Geometric Measure Theory},
  author={Federer, Herbert},
  publisher={Springer},
  volume={153},
  series={Grundlehren der mathematischen Wissenschaften},
  year={1969}
}

@article{clarke1976inverse,
  title={On the inverse function theorem},
  author={Clarke, Francis},
  journal={Pacific Journal of Mathematics},
  volume={64},
  number={1},
  pages={97--102},
  year={1976},
  publisher={Mathematical Sciences Publishers}
}

@article{hiriart1979tangent,
  title={Tangent cones, generalized gradients and mathematical programming in Banach spaces},
  author={Hiriart-Urruty, Jean-Baptiste},
  journal={Mathematics of Operations Research},
  volume={4},
  number={1},
  pages={79--97},
  year={1979},
  publisher={INFORMS}
}

@article{federer1959curvature,
  title={Curvature measures},
  author={Federer, Herbert},
  journal={Transactions of the American Mathematical Society},
  volume={93},
  number={3},
  pages={418--491},
  year={1959}
}

@article{debreu1970economies,
  title={Economies with a finite set of equilibria},
  author={Debreu, Gerard},
  journal={Econometrica: Journal of the Econometric Society},
  volume={38},
  number={3},
  pages={387--392},
  year={1970},
  publisher={JSTOR}
}

@article{clarke1975generalized,
  title={Generalized gradients and applications},
  author={Clarke, Frank H},
  journal={Transactions of the American Mathematical Society},
  volume={205},
  pages={247--262},
  year={1975}
}

@article{preiss1990differentiability,
  title={Differentiability of Lipschitz functions on Banach spaces},
  author={Preiss, David},
  journal={Journal of Functional Analysis},
  volume={91},
  number={2},
  pages={312--345},
  year={1990},
  publisher={Elsevier}
}

@book{grigoryan2009heat,
  title={Heat kernel and analysis on manifolds},
  author={Grigoryan, Alexander},
  volume={47},
  series={AMS/IP Studies in Advanced Mathematics},
  year={2009},
  publisher={American Mathematical Society}
}

@article{federer1944surface,
  title={Surface area. II},
  author={Federer, Herbert},
  journal={Transactions of the American Mathematical Society},
  volume={55},
  number={3},
  pages={438--456},
  year={1944}
}

@article{chernozhukov2015valid,
  title={Valid post-selection and post-regularization inference: An elementary, general approach},
  author={Chernozhukov, Victor and Hansen, Christian and Spindler, Martin},
  journal={Annual Review of Economics},
  volume={7},
  number={1},
  pages={649--688},
  year={2015},
  publisher={Annual Reviews}
}

@article{smale1965infinite,
  title={An Infinite Dimensional Version of Sard's Theorem},
  author={Smale, S},
  journal={American Journal of Mathematics},
  volume={87},
  number={4},
  pages={861--866},
  year={1965}
}

@article{luedtke2016optimal,
  title={Optimal individualized treatments in resource-limited settings},
  author={Luedtke, Alexander R and van der Laan, Mark J},
  journal={International Journal of Biostatistics},
  volume={12},
  number={1},
  pages={283--303},
  year={2016},
  publisher={De Gruyter}
}

@article{bhattacharya2012inferring,
  title={Inferring welfare maximizing treatment assignment under budget constraints},
  author={Bhattacharya, Debopam and Dupas, Pascaline},
  journal={Journal of Econometrics},
  volume={167},
  number={1},
  pages={168--196},
  year={2012},
  publisher={Elsevier}
}

@article{armstrong2023inference,
  title={Inference on optimal treatment assignments},
  author={Armstrong, Timothy B and Shen, Shu},
  journal={The Japanese Economic Review},
  volume={74},
  number={4},
  pages={471--500},
  year={2023},
  publisher={Springer}
}

@article{mbakop2021model,
  title={Model selection for treatment choice: Penalized welfare maximization},
  author={Mbakop, Eric and Tabord-Meehan, Max},
  journal={Econometrica},
  volume={89},
  number={2},
  pages={825--848},
  year={2021},
  publisher={Wiley Online Library}
}

@article{zhou2023offline,
  title={Offline multi-action policy learning: Generalization and optimization},
  author={Zhou, Zhengyuan and Athey, Susan and Wager, Stefan},
  journal={Operations Research},
  volume={71},
  number={1},
  pages={148--183},
  year={2023},
  publisher={INFORMS}
}

@article{zhao2012estimating,
  title={Estimating individualized treatment rules using outcome weighted learning},
  author={Zhao, Yingqi and Zeng, Donglin and Rush, A John and Kosorok, Michael R},
  journal={Journal of the American Statistical Association},
  volume={107},
  number={499},
  pages={1106--1118},
  year={2012},
  publisher={Taylor \& Francis}
}

@article{zhou2017residual,
  title={Residual weighted learning for estimating individualized treatment rules},
  author={Zhou, Xin and Mayer-Hamblett, Nicole and Khan, Umer and Kosorok, Michael R},
  journal={Journal of the American Statistical Association},
  volume={112},
  number={517},
  pages={169--187},
  year={2017},
  publisher={Taylor \& Francis}
}

@article{qian2011performance,
  title={Performance guarantees for individualized treatment rules},
  author={Qian, Min and Murphy, Susan A},
  journal={Annals of Statistics},
  volume={39},
  number={2},
  pages={1180},
  year={2011},
  publisher={NIH Public Access}
}

@article{swaminathan2015batch,
  title={Batch learning from logged bandit feedback through counterfactual risk minimization},
  author={Swaminathan, Adith and Joachims, Thorsten},
  journal={Journal of Machine Learning Research},
  volume={16},
  number={1},
  pages={1731--1755},
  year={2015},
  publisher={JMLR. org}
}

@article{boucheron2005theory,
  title={Theory of classification: A survey of some recent advances},
  author={Boucheron, St{\'e}phane and Bousquet, Olivier and Lugosi, G{\'a}bor},
  journal={ESAIM: Probability and Statistics},
  volume={9},
  pages={323--375},
  year={2005},
  publisher={EDP Sciences}
}

@article{stoye2009minimax,
  title={Minimax regret treatment choice with finite samples},
  author={Stoye, J{\"o}rg},
  journal={Journal of Econometrics},
  volume={151},
  number={1},
  pages={70--81},
  year={2009},
  publisher={Elsevier}
}

@article{tetenov2012statistical,
  title={Statistical treatment choice based on asymmetric minimax regret criteria},
  author={Tetenov, Aleksey},
  journal={Journal of Econometrics},
  volume={166},
  number={1},
  pages={157--165},
  year={2012},
  publisher={Elsevier}
}

@article{ben2024policy,
  title={Policy Learning with Asymmetric Counterfactual Utilities},
  author={Ben-Michael, Eli and Imai, Kosuke and Jiang, Zhichao},
  journal={Journal of the American Statistical Association},
  volume = {119},
  number = {548},
  pages={3045--3058},
  year={2024},
  publisher={Taylor \& Francis}
}

@article{kitagawa2023individualized,
  title={Individualized Treatment Allocation in Sequential Network Games},
  author={Kitagawa, Toru and Wang, Guanyi},
  journal={arXiv preprint arXiv:2302.05747},
  year={2023}
}

@article{kitagawa2023should,
  title={Who should get vaccinated? Individualized allocation of vaccines over SIR network},
  author={Kitagawa, Toru and Wang, Guanyi},
  journal={Journal of Econometrics},
  volume={232},
  number={1},
  pages={109--131},
  year={2023},
  publisher={Elsevier}
}

@article{luckett2021receiver,
  title={Receiver operating characteristic curves and confidence bands for support vector machines},
  author={Luckett, Daniel J and Laber, Eric B and El-Kamary, Samer S and Fan, Cheng and Jhaveri, Ravi and Perou, Charles M and Shebl, Fatma M and Kosorok, Michael R},
  journal={Biometrics},
  volume={77},
  number={4},
  pages={1422--1430},
  year={2021},
  publisher={Oxford University Press}
}

@article{agrawal2019artificial,
  title={Artificial intelligence: the ambiguous labor market impact of automating prediction},
  author={Agrawal, Ajay and Gans, Joshua S and Goldfarb, Avi},
  journal={Journal of Economic Perspectives},
  volume={33},
  number={2}, 
  pages={31--50},
  year={2019},
  publisher={American Economic Association 2014 Broadway, Suite 305, Nashville, TN 37203-2418}
}

@article{babina2024artificial,
  title={Artificial intelligence, firm growth, and product innovation},
  author={Babina, Tania and Fedyk, Anastassia and He, Alex and Hodson, James},
  journal={Journal of Financial Economics},
  volume={151},
  pages={103745},
  year={2024},
  publisher={Elsevier}
}

@article{luedtke2020performance,
  title={Performance guarantees for policy learning},
  author={Luedtke, Alex and Chambaz, Antoine},
  journal={Annales de l'Institut Henri Poincar\`{e}, Probabilit\`{e}s et Statistiques},
  volume={56},
  number={3},
  pages={2162--2188},
  year={2020},
  organization={NIH Public Access}
}

@article{luedtke2016statistical,
  title={Statistical inference for the mean outcome under a possibly non-unique optimal treatment strategy},
  author={Luedtke, Alexander R and Van Der Laan, Mark J},
  journal={Annals of Statistics},
  volume={44},
  number={2},
  pages={713},
  year={2016},
  publisher={NIH Public Access}
}

@unpublished{rai2018statistical,
  title={Statistical inference for treatment assignment policies},
  author={Rai, Yoshiyasu},
  year={2018},
  note={Unpublished Results}
}

@article{semenova2023debiased,
  title={Debiased machine learning of set-identified linear models},
  author={Semenova, Vira},
  journal={Journal of Econometrics},
  volume={235},
  number={2},
  pages={1725--1746},
  year={2023},
  publisher={Elsevier}
}

@article{foster2023orthogonal,
  title={Orthogonal statistical learning},
  author={Foster, Dylan J and Syrgkanis, Vasilis},
  journal={Annals of Statistics},
  volume={51},
  number={3},
  pages={879--908},
  year={2023},
  publisher={Institute of Mathematical Statistics}
}

@article{liu2024inference,
  title={Inference for an Algorithmic Fairness-Accuracy Frontier},
  author={Liu, Yiqi and Molinari, Francesca},
  journal={arXiv preprint arXiv:2402.08879},
  year={2024}
}

@book{gantmakher1977theory,
  title={The theory of matrices},
  author={Gantmakher, Feliks Ruvimovich},
  volume={1},
  year={1977},
  publisher={Chelsea Publishing Company}
}

@book{saks1937theory,
  title={Theory of the Integral},
  author={Saks, Stanis{\l}aw},
  year={1937}, 
  publisher={Hafner Publishing Company}
}

@article{hagood2006recovering,
  title={Recovering a function from a Dini derivative},
  author={Hagood, John W and Thomson, Brian S},
  journal={American Mathematical Monthly},
  volume={113},
  number={1},
  pages={34--46},
  year={2006},
  publisher={Taylor \& Francis}
}

@article{de2001morse,
  title={The Morse--Sard theorem in Sobolev spaces},
  author={de Pascale, Luigi},
  journal={Indiana University mathematics journal},
  pages={1371--1386},
  volume={50},
  number={3},
  year={2001},
  publisher={JSTOR}
}

@article{gerber2008social,
  title={Social pressure and voter turnout: Evidence from a large-scale field experiment},
  author={Gerber, Alan S and Green, Donald P and Larimer, Christopher W},
  journal={American Political Science Review},
  volume={102},
  number={1},
  pages={33--48},
  year={2008},
  publisher={Cambridge University Press}
}

@article{athey2019generalized,
  title={Generalized random forests},
  author={Athey, Susan and Tibshirani, Julie and Wager, Stefan},
  journal={Annals of Statistics},
  volume={47},
  number={2},
  pages={1148--1178},
  year={2019}
}

@article{wager2018estimation,
  title={Estimation and inference of heterogeneous treatment effects using random forests},
  author={Wager, Stefan and Athey, Susan},
  journal={Journal of the American Statistical Association},
  volume={113},
  number={523},
  pages={1228--1242},
  year={2018},
  publisher={Taylor \& Francis}
}

@misc{econml,
  author={Battocchi, Keith and Dillon, Eleanor and Hei, Maggie and Lewis, Greg and Oka, Paul and Oprescu, Miruna and Syrgkanis, Vasilis},
  title={{EconML}: {A Python Package for ML-Based Heterogeneous Treatment Effects Estimation}},
  howpublished={https://github.com/py-why/EconML},
  date={November 2, 2024},
  note={Version 0.15.1, last accessed date November 2, 2024},
  year={2019}
}

@article{scikit-learn,
  title={Scikit-learn: Machine Learning in {P}ython},
  author={Pedregosa, Fabian and Varoquaux, Gaël and Gramfort, Alexandre and Michel, Vincent 
          and Thirion, Bertrand and Grisel, Olivier and Blondel, Mathieu and Prettenhofer, Peter 
          and Weiss, Ron and Dubourg, Vincent and Vanderplas, Jake and Passos, Alexandre and
          Cournapeau, David and Brucher, Matthieu and Perrot, Matthieu and Duchesnay, Édouard},
  journal={Journal of Machine Learning Research},
  volume={12},
  pages={2825--2830},
  year={2011}
}

@article{bos2022convergence,
  title={Convergence rates of deep ReLU networks for multiclass classification},
  author={Bos, Thijs and Schmidt-Hieber, Johannes},
  journal={Electronic Journal of Statistics},
  volume={16},
  number={1},
  pages={2724--2773},
  year={2022},
  publisher={The Institute of Mathematical Statistics and the Bernoulli Society}
}

@article{ai2024data,
  title={Data-driven Policy Learning for a Continuous Treatment},
  author={Ai, Chunrong and Fang, Yue and Xie, Haitian},
  journal={arXiv preprint arXiv:2402.02535},
  year={2024}
}

@article{chen2021shape,
  title={Shape-enforcing operators for generic point and interval estimators of functions},
  author={Chen, Xi and Chernozhukov, Victor and Fern{\'a}ndez-Val, Iv{\'a}n and Kostyshak, Scott and Luo, Ye},
  journal={Journal of Machine Learning Research},
  volume={22},
  number={1},
  pages={10034--10075},
  year={2021}
}

@article{sadhwani2021deep,
  title={Deep learning for mortgage risk},
  author={Sadhwani, Apaar and Giesecke, Kay and Sirignano, Justin},
  journal={Journal of Financial Econometrics},
  volume={19},
  number={2},
  pages={313--368},
  year={2021},
  publisher={Oxford University Press}
}

@article{berge2011evaluating,
  title={Evaluating the classification of economic activity into recessions and expansions},
  author={Berge, Travis J and Jord{\`a}, {\`O}scar},
  journal={American Economic Journal: Macroeconomics},
  volume={3},
  number={2},
  pages={246--277},
  year={2011},
  publisher={American Economic Association}
}

@article{feng2025statisticalMS,
  title={Statistical tests for replacing human decision makers with algorithms},
  author={Feng, Kai and Hong, Han and Tang, Ke and Wang, Jingyuan},
  journal={Management Science},
  volume={71},
  number={11},
  pages={9145--9170},
  year={2025},
  publisher={INFORMS}
}

@article{donald2014estimation,
  title={Estimation and inference for distribution functions and quantile functions in treatment effect models},
  author={Donald, Stephen G and Hsu, Yu-Chin},
  journal={Journal of Econometrics},
  volume={178},
  number={3},
  pages={383--397},
  year={2014},
  publisher={Elsevier}
}

@article{viviano2025policy,
  title={Policy targeting under network interference},
  author={Viviano, Davide},
  journal={Review of Economic Studies},
  volume={92},
  number={2},
  pages={1257--1292},
  year={2025},
  publisher={Oxford University Press UK}
}
\bibliographystyle{aer}
}

\newpage
\setstretch{1.25}
\renewcommand{\thesection}{A}
{\centering{\section{Appendix for manuscript}\label{appendixmanuscript}}}

\setcounter{equation}{0}

\counterwithin*{equation}{section}
\renewcommand\theequation{\thesection.\arabic{equation}}

\setcounter{figure}{0}
\counterwithin*{figure}{section}
\renewcommand\thefigure{\thesection.\arabic{figure}}

\subsection{Proofs of Theorems \ref{general k derivative}, Proposition \ref{surface integral continuity} and Theorem \ref{nonsmooth general k derivative}}
\begin{proof}
We will prove these three results in several parts. Before proceeding to the more general
case of $k > 1$ in Theorem \ref{general k derivative} (\textbf{part} 2), we first prove a pointwise version of the special case when $k = 1$ (\textbf{part} 1).
Proposition \ref{surface integral continuity} can be proved by a similar derivation in passing (\textbf{part} 1 \textbf{step} 5).
We then demonstrate the uniformity of convergence of Theorem \ref{general k derivative} (\textbf{part} 3).
We finalize by elaborating the proof techniques to handle the nonsmoothness in Theorem \ref{nonsmooth general k derivative} (\textbf{part} 4). \\

\noindent \textbf{part} 1 First consider the special case of $k=1$ in Theorem \ref{general k derivative}. \\
\textbf{step} 1
We claim that
there exists an $\eta >0$ small enough such that for all $c'$ with
$\vert c' - c\vert < \eta$, we can get change of variable
formulas simultaneously.  Consider
\bs
\Psi\(x_1', \ldots, x_n', c'\) \coloneqq h\(x_1',\ldots,x_n'\) - c'.
\end{split}\end{align}
Without loss of generality, we may assume that $\nabla_{x_1} \Psi\(x_1, \ldots, x_n, c\)$ is full rank. Then by the implict function theorem (see for example Theorem \ref{Nonsmooth implicit function theorem}, $C^{1}$ case), there exists
an open set $B_{x_1} \times B_{x_2,\ldots,x_n,c} \subset U \subset \mathbb{R}^{n+1}$, where $U$ is a neighborhood
of $\(x_1,\ldots,x_n,c\)$, such that for some positive vectors $\alpha$, $\beta$,
\bs
B_{x_1} &= \left\{x_1' \in \mathbb{R}: \vert x_1' - x_1\vert < \alpha\right\}\\
B_{x_2,\ldots,x_n,c} &= \left\{\(x_2',\ldots,x_n',c'\)	\in \mathbb{R}^n: \vert \(x_2',\ldots,x_n',c'\)- \(x_2,\ldots,x_n,c\)\vert < \beta\right\},
\end{split}\end{align}
and a $C^1$ implicit function $\xi\(\cdot\)$ defined on
$B_{x_2,\ldots,x_n,c}$ such that
\bs
\Psi\(x_1', \ldots, x_n', c'\) = 0 \Leftrightarrow  \(x_1'\)
= \xi\(x_2',\ldots,x_n', c'\)
\end{split}\end{align}
for all $\(x_1',\ldots,x_n', c'\) \in B_{x_1}\times B_{x_2,\ldots,x_n,c}$. For all
$x \in h^{-1}\(c\) \cap K_f$, we can find intervals like
$B_{x_1}\times B_{x_2,\ldots,x_n,c}$ in the above. The implicit function theorem may not be
always about the first dimension. However, an implicit function for some $x_m', m \in \{1,\ldots,n\}$ always exists by the regular point assumption \ref{regular value assumption '}. By the
compactness of $h^{-1}\(c\) \cap K_f$, there exists a finite open cover denoted as
$
\left\{B_j\right\} \coloneqq \left\{B_{x_{j, 1}} \times B_{x_{j,2},\ldots, x_{j, n}}\right\}$
of  $h^{-1}\(c\) \cap K_f$. Now, we claim that there exists $\eta > 0$, such that for all $c'$ satisfying $\vert c'-c\vert < \eta$, $h^{-1}\(c'\) \cap K_f \subset \bigcup_j B_j$.  This claim is proven by contradiction.
Suppose $\bigcup_j B_j$ will not cover $h^{-1}\(c'\) \cap K_f$ for some $c'\neq c$, such that $\vert c'-c\vert < \eta$	where $\eta$ is arbitrarily chosen. In
other words, $\forall \eta > 0$, $\exists c'$,	$\vert c'-c\vert < \eta$, such
that $h^{-1}\(c'\) \cap K_f \not\subset \bigcup_j B_j$. Then there exists
a sequence $\left\{\(x_i, c_{i}\)\right\}_{i=1}^\infty$ such that
\bs
x_i \in h^{-1}\(c_i\) \cap K_f, x_i \notin \bigcup_j B_j.
\end{split}\end{align}
By the compactness of $h^{-1}\([a,b]\) \cap K_f$ and the Bolzano-Weierstrass
theorem, there exists a convergent subsequence
$\{x_{i_l}\}_{l=1}^\infty$ such that $\lim_{l\rightarrow \infty} h\(x_{i_l}\)
= c$ and thus $\lim_{l\rightarrow\infty} x_{i_l} \in h^{-1}\(c\) \subset \bigcup_j
B_j$. Therefore, $x_{i_l} \in \bigcup_{j} B_j$ for all sufficiently large $l$,
which is a contradiction. \\
\textbf{step} 2 To calculate the change of variables
formula analytically, consider
\bs
\psi_{B_j}\(x_{j,2}',\ldots,x_{j,n}', c'\)
= \(\xi_{B_j}\(x_{j,2}',\ldots,x_{j,n}', c'\), x_{j,2}',\ldots,x_{j,n}'\)^T,
\end{split}\end{align}
where $\(x_{j,2}',\ldots,x_{j,n}'\)$ are used by the implicit function theorem to obtain the local implicit function
$\xi_{B_j}\(\cdot\)$ for $x_{j,1}'$. Then, by the generalized matrix determinant
lemma (for clarity the subscript $j$ is omitted), we have 
\begin{align}
J\psi_{B_j}\(x'_{2, \ldots,n}, c'\) & = \det\(
\(
\(
\[
\frac{\partial h }{\partial  x_1}
\]
^{-1}
\nabla_{x_{2, \ldots, n}} h
\)^{T},
I_{n - 1}\)
\(\begin{array}{c}
\[
\frac{\partial h }{\partial  x_1}
\]
^{-1}
\nabla_{x_{2,\ldots,n}}h\\
I_{n - 1}
\end{array}\)\)^{\frac{1}{2}} \notag \\
& = \det\(I_{n - 1} +
\(
\[
\frac{\partial h }{\partial  x_1}
\]
^{-1}
\nabla_{x_{2,\ldots,n}} h
\)^{T}
\(
\[
\frac{\partial h }{\partial  x_1}
\]
^{-1}
\nabla_{x_{2,\ldots,n}} h
\)
\)^{\frac{1}{2}} \notag \\
& = \(\det\(I_{n - 1}\)
\det\(I_{1} + \(
\[
\frac{\partial h }{\partial  x_1}
\]
^{-1}
\nabla_{x_{2,\ldots,n}} h\)I_{n - 1}
\(
\[
\frac{\partial h }{\partial  x_1}
\]
^{-1}
\nabla_{x_{2,\ldots,n}} h
\)^{T}\)\)
^{\frac{1}{2}} \notag \\
& = \(\det\(\[
\frac{\partial h }{\partial  x_1}
\]\[
\frac{\partial h }{\partial  x_1}
\]^{T} +
\(\nabla_{x_{2,\ldots,n}} h\)
\(\nabla_{x_{2,\ldots,n}} h\)^T
\)
\)
^{\frac{1}{2}}
\left\vert\det\(\[
\frac{\partial h }{\partial  x_1}
\]^{-1}\)\right\vert \notag
\\
& = \(\det\(
\(
\nabla h
\)
\(
\nabla h
\)^T
\)\)^{\frac{1}{2}}
\left\vert\det\(\[
\frac{\partial h }{\partial  x_1}
\]^{-1}\)\right\vert \notag \\
& = Jh \left\vert\det\(\[
\frac{\partial h }{\partial  x_1}
\]^{-1}\)\right\vert. \label{long jacobian calculus}
\end{align}
In the above, for simplicity, we omit the point at which the derivatives
are calculated, where
\bs
\frac{\partial h }{\partial  x_1}
&= \frac{\partial}{\partial x_1}  h\(
\xi_{B_j}\(x'_{2,\ldots,n}, c'\)
, x_2', \ldots, x_n'\)\ \text{and}\\
\nabla_{x_{2,\ldots,n}} h
&=
\nabla_{x_{2,\ldots,n}} h\(
\xi_{B_j}\(x'_{2,\ldots,n}, c'\), x_2', \ldots, x_n'\).
\end{split}\end{align}
\textbf{step} 3
Consider a sequence $t_n\downarrow 0$.
First let $H_n \equiv H$ for all $n$. By definition of Hadamard differentiability, we need to calculate
\bs
\lim_{n\to\infty} \frac{
1
}{
t_n
}\[
\int 1\(h\(x\)+t_n H\(x\) > c\) f\(x\) d \mathcal L_n x
- \int 1\(h\(x\) > c\) f\(x\) d \mathcal L_n x
\].
\end{split}\end{align}
By coarea formula Theorem \ref{coarea formula classic},
\bs
&\frac{
1
}{
t_n
}\[
\int 1\(h\(x\)+t_n H\(x\) > c\) f\(x\) d \mathcal L_n x
- \int 1\(h\(x\) > c\) f\(x\) d \mathcal L_n x
\]\\
&=
\frac{
1
}{
t_n
}\int \[
\int_{h^{-1}\(c'\) \cap K_f}
\frac{
\(
1\(c'+t_n H\(x\) > c\)	- 1\(c' > c\)
\)
f\(x\)
}{
J h\(x\)
}
d \mathcal H_{n-1} x
\] d \mathcal L_1 c'\\
&=
\int_{c-t_n M}^{c+t_n M} \[\frac{
1
}{
t_n
}
\int_{h^{-1}\(c'\) \cap K_f}
\frac{
\(
1\(c'+t_n H\(x\) > c\)	- 1\(c' > c\)
\)
f\(x\)
}{
J h\(x\)
}
d \mathcal H_{n-1} x
\] d \mathcal L_1 c',
\end{split}\end{align}
where $M = \max_{x\in K_f} \vert H\(x\) \vert < +\infty$.
Then, we can apply the area formula Theorem \ref{area formula classic} to calculate, for each $j$,
\bs
&\frac{1}{t_n} \int
\[\int_{h^{-1}\(c'\) \cap B_j \cap K_f}
\frac{
\[
1\(c' + t_n H\(x\) > c\) - 1\(c' > c\)
\] f\(x\)
}{
J h\(x\)
}
d \mathcal H_{n-1} x
\]
d \mathcal L_1 c' \\
&=\frac{1}{t_n} \int
\[
\int_{B_{x_{j, 2}, \ldots, x_{j, n}}}
\frac{
\[
1\(c' + t_n H\(
\psi_{B_j}\(x_{2, \ldots, n}, c'\)
\) > c\)
-1\(c' > c\)
\]
f\(\psi_{B_j}\(
x_{2,\ldots,n}, c'
\)
\)
}{
\left\vert
\text{det}\(
\frac{\partial}{\partial x_1} h\(
\psi_{B_j}\(x_{2, \ldots, n}, c'\)
\)
\)
\right\vert
}
d \mathcal L_{n-1} x
\]
d \mathcal L_1 c'. \\
\end{split}\end{align}
Next by the Fubini-Tonelli theorem, the above is equal to
\bs
\int_{B_{x_{j, 2}, \ldots, x_{j, n}}}
\[
\frac{1}{t_n} \int
\frac{
\[
1\(c' + t_n H\(
\psi_{B_j}\(x_{2, \ldots, n}, c'\)
\) > c\)
-1\(c' > c\)
\]
f\(\psi_{B_j}\(
x_{2,\ldots,n}, c'
\)
\)
}{
\left\vert
\text{det}\(
\frac{\partial}{\partial x_1} h\(
\psi_{B_j}\(x_{2, \ldots, n}, c'\)
\)
\)
\right\vert
}
d \mathcal L_1 c'
\]
d \mathcal L_{n-1} x.
\end{split}\end{align}
Since
\bs
&\frac{1}{t_n}
\int
\[
1\(c' + t_n H\(
\psi_{B_j}\(x_{2, \ldots, n}, c'\)
\) > c\)
-1\(c' > c\)
\]
d\mathcal L_1 c'\\
&=
\frac{1}{t_n}
\int_{c-t_n M}^{c+t_n M}
\[
1\(c' + t_n H\(
\psi_{B_j}\(x_{2, \ldots, n}, c'\)
\) > c\)
-1\(c' > c\)
\]
d\mathcal L_1 c'
\to H\( \psi_{B_j}\(x_{2, \ldots, n}, c\)\),
\end{split}\end{align}
by the dominated convergence theorem (DCT).
With the help of the partition of unity theorem
\ref{partition of unity},
we can extend the above local convergence to the entire level set $h^{-1}\(c\)$,
since for $n$ large enough, we have
\begin{align}\begin{split}\nonumber
& \frac{
1
}{
t_n
}\int \[
\int_{h^{-1}\(c'\) \cap K_f}
\frac{
\(
1\(c'+t_n H\(x\) > c\)	- 1\(c' > c\)
\)
f\(x\)
}{
J h\(x\)
}
d \mathcal H_{n-1} x
\] d \mathcal L_1 c'\\
& =
\frac{
1
}{
t_n
}\int \[\sum_{l}
\int_{h^{-1}\(c'\) \cap K_f}
\frac{
\(
1\(c'+t_n H\(x\) > c\)	- 1\(c' > c\)
\)
f\(x\)\rho_{l}\(x\)
}{
J h\(x\)
}
d \mathcal H_{n-1} x
\] d \mathcal L_1 c',
\end{split}\end{align}
where $\rho_{l}\(\cdot\)$ is the partition of unity. \\
\noindent Then by the area formula again,
\bs
\lim_{n\rightarrow\infty} \frac{1}{t_n} \int
\[
1\(h\(x\) + t_n H\(x\) > c\)
-1\(h\(x\) > c\)
\] f\(x\)
d \mathcal L_n x = \int_{h^{-1}\(c\)}
\frac{
H\(x\) f\(x\)
}{
J h\(x\)
}
d \mathcal H_{n-1} x.
\end{split}\end{align}
\textbf{step} 4 Now consider functions $H\(x\)+\zeta$, for arbitrary $\zeta > 0$. Since by assumption
$\sup_{x\in E} \vert H_n\(x\) - H\(x\) \vert \to 0$, for all $\zeta > 0$, for sufficiently large $n$,
$H\(x\)+\zeta > H_n\(x\)$ for all $x$.	Therefore,
\bs
\limsup_{n\rightarrow\infty}
\frac{
F\(h+t_n H_n,  c\)-
F\(h,  c\)
}{
t_n
}
\leq &
\lim_{n\rightarrow\infty}
\frac{
F\(h+t_n\(H + \zeta\),	c\)-
F\(h,  c\)
}{
t_n
}\\
= &\int_{h^{-1}\(c\)}
\frac{
\(H\(x\) + \zeta\) f\(x\)
}{
J h\(x\)
}
d \mathcal H_{n-1} x.
\end{split}\end{align}
Similarly,
\bs
\liminf_{n\rightarrow\infty}
\frac{
F\(h+t_n H_n,  c\)-
F\(h,  c\)
}{
t_n
}
\geq &
\lim_{n\rightarrow\infty}
\frac{
F\(h+t_n\(H - \zeta\),	c\)-
F\(h,  c\)
}{
t_n
}\\
= &\int_{h^{-1}\(c\)}
\frac{
\(H\(x\) - \zeta\) f\(x\)
}{
J h\(x\)
}
d \mathcal H_{n-1} x.
\end{split}\end{align}
By the H\"{o}lder inequality (integration by Hausdorff measure is also in the Lebesgue sense):
\bs
\int_{h^{-1}\(c\)}
\frac{
\left\vert
H\(x\) - H'\(x\)
\right\vert
f\(x\)
}{
J h\(x\)
}
d \mathcal H_{n-1} x
\leq \sup_{x\in E} \vert H\(x\) - H'\(x\) \vert
\int_{h^{-1}\(c\)}
\left\vert
\frac{
f\(x\)
}{
J h\(x\)
}
\right\vert
d \mathcal H_{n-1} x.
\end{split}\end{align}
Now by the arbitrariness of $\zeta$,
\bs
\lim_{n\rightarrow\infty}
\frac{
F\(h+t_n H_n,  c\)-
F\(h,  c\)
}{
t_n
}
= \int_{h^{-1}\(c\)}
\frac{
H\(x\)	f\(x\)
}{
J h\(x\)
}
d \mathcal H_{n-1} x.
\end{split}\end{align}
\textbf{step} 5
Now we can also prove Proposition \ref{surface integral continuity}.
Note that
the above proof techniques of \textbf{step 1} to \textbf{step} 3, including the area formula, partition of unity, and generalized
matrix determinant lemma, all work for $k > 1$.
As a result, the proof is essentially a reconstruction of the proof of the special case of $k=1$. 
Here, we only need to point out
that
\bs
\lim_{c' \to c}
\int_{B_{x_{j, 2}, \ldots, x_{j, n}}} &
\frac{
f\(\psi_{B_j}\(
x_{2,\ldots,n}, c'
\)
\)
}{
\left\vert
\text{det}\(
\frac{\partial}{\partial x_1} h\(
\psi_{B_j}\(x_{2, \ldots, n}, c'\)
\)
\)
\right\vert
}
d \mathcal L_{n-1} x \\
& =
\int_{B_{x_{j, 2}, \ldots, x_{j, n}}}
\frac{
f\(\psi_{B_j}\(
x_{2,\ldots,n}, c
\)
\)
}{
\left\vert
\text{det}\(
\frac{\partial}{\partial x_1} h\(
\psi_{B_j}\(x_{2, \ldots, n}, c\)
\)
\)
\right\vert
}
d \mathcal L_{n-1} x.
\end{split}\end{align}
\textbf{part} 2 For $k > 1$,  we still have, by the coarea formula
\bs
F\(h,c\) = \int_{h\(x\) > c} f\(x\) d \mathcal L_n x =
\int_{c' > c} \[
\int_{h^{-1}\(c'\)} \frac{f\(x\)}{J h\(x\)} d \mathcal H_{n-k}
\]
d \mathcal L_k c'.
\end{split}\end{align}
So we have that
\bs
&\frac{1}{t_n} \int
\[
1\(
h\(x\) + t_n H_n\(x\) > c
\) - 1\(h\(x\) > c\)
\]
f\(x\)
d \mathcal L_n x\\
&=\frac{1}{t_n} \int
\[
\int_{h^{-1}\(c'\) \cap K_f}
\frac{
\[
1\(c'+ t_n H_n\(x\) > c\) - 1\(c' > c\)
\] f\(x\)
}{
J h\(x\)
}
d \mathcal H_{n-k} x
\]
d \mathcal L_k c'.
\end{split}\end{align}
By the telescoping identity
 of higher order expansion,
\bs
\prod_{i=1}^k a_i' - \prod_{i=1}^k a_i
= \sum_{i=1}^k \(a_i' - a_i\) \prod_{l\neq i} a_l
+ \sum_{i\neq j} \(a_i'-a_i\)\(a_j'-a_j\) \prod_{l\neq i, l\neq j} a_l + \cdots
+ \prod_{i=1}^k \(a_i'-a_i\),
\end{split}\end{align}
using a proof procedure similar to that of the special case of $k=1$, 
the first
order term of the difference becomes
\bs
&\sum_i \frac{1}{t_n}
\int
\[
\int_{h^{-1}\(c'\) \cap K_f}
\frac{
\[
1\(c_i'+ t_n H_{n,i}\(x\) > c_i\) - 1\(c'_i > c_i\)
\]
\prod_{l\neq i} 1\(c_l' > c_l\)
f\(x\)
}{
J h\(x\)
}
d \mathcal H_{n-k} x
\]
d \mathcal L_k c' \\
&=\sum_i \sum_{l_i}
\int_{y > \tau_{\neg i}\(c\)}
\int_{B_{l_i}}
d \mathcal L_{n-k} x  d L_{k-1} y
\\
&\[
\int
\frac{
\[
1\(c_i'+ t_n H_{n,i}\(
\psi_{B_{l_i}}\(x, c'\)
\) > c_i\) - 1\(c'_i > c_i\)
\]
f\(\psi_{B_{l_i}}\(x, c'\)\)
\rho_{B_{l_{i}}}\(\psi_{B_{l_i}}\(x, c'\)\)
}{
\left\vert
\det\(
\frac{\partial}{\partial x_{l_i}}
h\(
\psi_{B_{l_i}}\(x, c'\)\)
\)
\right\vert
}
\frac{1}{t_n}
d \mathcal L_1 c_i'\].
\end{split}\end{align}
where $\rho_{B_{l_{i}}}\(\psi_{B_{l_i}}\(x, c'\)\)$ is the partition of unity, and $x_{l_i}$ denotes
that $x_{l_i}$	is picked for locally implicit function for change of variables by the
implicit function theorem.
We also note that in the above $y = \tau_{\neg i}\(c'\)$, and $c'_i=\tau_i\(c'\)$ and we use the fact that
\bs
1\(c'+ t_n H_n\(x\) > c\)  = \prod_{i=1}^k 1\(c_i'+ t_n H_{n,i}\(x\) > c_i\)\ \text{and}\
1\(c' > c\) = \prod_{i=1}^k 1\(c'_i > c_i\).
\end{split}\end{align}
Then by DCT and the area formula again, we obtain the limit of the first order
difference as
\bs
\sum_i \int \[
\int_{h^{-1}\(c'\(y, c_{i}, i\)\)}
\frac{
H_i\(x\) f\(x\)
}{
J h\(x\)
}
d \mathcal H_{n-k} x
\] d \mathcal L_{k-1} y.
\end{split}\end{align}
The convergence of the first order term also implies that the second and higher order terms of
the difference should vanish eventually. \\
\textbf{part} 3 By \textbf{part} 1 and \textbf{part} 2, we only need to consider the $k=1$ case since the
$k > 1$ case is similar. Let $\{c_n\}_{n=1}^\infty$ be a sequence such that  $\lim_{n\rightarrow\infty} c_n = c$. Consider
\bs
\lim_{n\rightarrow\infty} &
\frac{1}{t_n}
\[
F\(h+t_n H_n, c_n\) -F\(h, c_n\)
\] \\
& =
\lim_{n\rightarrow\infty}
\frac{1}{t_n}
\int
\[
1\(h\(x\) + t_n H_n\(x\) > c_n\) -
1\(h\(x\)  > c_n\)
\]
f\(x\)
d \mathcal L_n x.
\end{split}\end{align}
The proof in \textbf{part} 1 shows that we need to calculate, for
$M = \sup_{x\in K_f} \vert H\(x\) \vert$,
\bs
&\frac{1}{t_n}
\int_{c - t_n M - \vert c - c_n\vert}^{c + t_n M + \vert c - c_n\vert}
\[
1\(c'+ t_n H\(\psi_{B_j}\(\cdot, c'\)\) > c_n\)
- 1\(c' > c_n\)\]
d \mathcal L_1 c'. 
\end{split}\end{align}
Note that
\bs
\frac{1}{t_n} \int
\[
1\(c' + t_n A + c - c_n > c\) - 1\(c'+c - c_n > c\)
\]
d \mathcal L_1 c' = A,
\end{split}\end{align}
we can obtain
\bs
\lim_{n\rightarrow\infty}
\frac{1}{t_n}
\int_{c - t_n M - \vert c - c_n\vert}^{c + t_n M + \vert c - c_n\vert}
\[
1\(c'+ t_n H\(\psi_{B_j}\(\cdot, c'\)\) > c_n\)
- 1\(c' > c_n\)
\]
d \mathcal L_1 c' = H\(\psi_{B_j}\(\cdot, c\)\)
\end{split}\end{align}
by simple upper and lower bound of $H\(\psi_{B_j}\(\cdot, c'\)\)$ on
$\[
c - t_n M - \vert c - c_n\vert,
c + t_n M + \vert c - c_n\vert\]$. So we obtain, 
that the limit in the calculation of the Hadamard derivative actually converges continuously (see for example Definition \ref{continuous convergence}).
By Lemma \ref{UConvergence and CConvergence}, we get that the Hadamard derivative can be taken
in a uniform sense. In other words, let $\mathcal D \subset h\(E\)$  , then we have
\bs
\lim_{n\rightarrow\infty}
\sup_{c \in \mathcal D}
\left\vert
\frac{1}{t_n}
\int \[
1\(h\(x\)+ t_n H_{n}\(x\) > c\)
- 1\(h\(x\) > c\)
\]
f\(x\) d \mathcal L_n x
-
\int_{h^{-1}\(c\)} \frac{
H\(x\) f\(x\)
}{
J h\(x\)
}d \mathcal H_{n-1} x
\right\vert = 0.
\end{split}\end{align}
This finishes the proof of Theorem \ref{general k derivative}. \\
\textbf{part} 4 We only need to point out changes under lower regularity of Theorem \ref{nonsmooth general k derivative}. \\
\textbf{step} 1
Let $x \in h^{-1}\(\mathcal{D}\)$.
By Assumption \ref{Clarke regular point assumption}, $J_{c}h\(x\)$ is of full rank, so we can change the $C^{1}$ implicit function theorem to the Lipschitz version, see for example Theorem \ref{Nonsmooth implicit function theorem} (Lipchitz case) in the Technical addendum.
For simplicity of notations, consider the case when $k = 1$.
Without loss of generality, we may also assume the first coordinate is picked to be expressed by the implicit function, i.e., $x'_{1} = \xi\(x'_{2}, \ldots, x'_{n}, c'\)$.
Compared to \textbf{step} 2 of \textbf{part} 1, we omit the subscript like $j$ and $B_{j}$, since the reader should have been clear that the discussion here is purely local.

The implicit function theorem says that locally we have
$h\(\xi\(x'_{2}, \ldots, x'_{n}, c'\), x'_{2}, \ldots, x'_{n}\) = 0$.
Fix $c'$, $\xi\(\cdot, c\)$ locally characterized the level set $\{x: h\(x\) = c'\}$, i.e.,
$x' \in \{x': h\(x'\) = c'\} \Leftrightarrow x'_{1} = \xi\(x'_{2}, \ldots, x'_{n}, c'\)$.
By Assumption \ref{a.e. continuous at boundary assumption},
$\mathcal{H}_{n - 1}\left\{x': h\(x'\) = c', h\(\cdot\) \text{is not differentiable at } x'\right\} = 0$,
thus for $L_{n - 1}$ almost every $\(x'_{2}, \ldots, x'_{n}\)$ in a neighborhood,
we can perform differential calculus to $\xi\(\cdot, c'\)$ and $\psi\(\cdot, c'\)$ by the chain rule and Rademacher Theorem (see for example Theorem \ref{Rademacher theorem}) to see \eqref{long jacobian calculus} holds $a.e.$ \\
\textbf{step} 2 Now, following the main idea from \textbf{part} 1 \textbf{step} 3, we need to calculate
\begin{align}\begin{split}\nonumber
\lim_{t_{n} \downarrow 0, c_{n} \to c}
\frac{1}{t_n} \int
\frac{
\[
1\(c' + t_n H_{n}\(
\psi_{B}\(x_{2, \ldots, n}, c'\)
\) > c_{n}\)
-1\(c' > c_{n}\)
\]
f\(\psi_{B}\(
x_{2,\ldots,n}, c'
\)
\)
}{
\left\vert
\text{det}\(
\frac{\partial}{\partial x_1} h\(
\psi_{B}\(x_{2, \ldots, n}, c'\)
\)
\)
\right\vert
}
d \mathcal L_1 c',
\end{split}\end{align}
where $B$ is the neighborhood where we apply the Lipschitz implicit function theorem.
Consider $\(x_{2}, \ldots, x_{n}\)$ where $\frac{f\(\psi_{B}\(\cdot, c\)\)}
{\left\vert\det\(\frac{\partial}{\partial x_1} h\(\psi_{B}\(\cdot, c\)\)\)\right\vert}$ is continuous, then we have
\begin{align}\begin{split}\nonumber
\int
\frac{
\[
1\(c' + t_n A > c_{n}\)
-1\(c' > c_{n}\)
\]
f\(\psi_{B}\(
x_{2,\ldots,n}, c'
\)
\)
}{
\left\vert
\det\(
\frac{\partial}{\partial x_1} h\(
\psi_{B}\(x_{2, \ldots, n}, c'\)
\)
\)
\right\vert
}
d \mathcal L_1 c'
= t_{n} A\(\frac{f\(\psi_{B}\(
x_{2,\ldots,n}, c
\)
\)}{\left\vert
\det\(
\frac{\partial}{\partial x_1} h\(
\psi_{B}\(x_{2, \ldots, n}, c\)
\)
\)
\right\vert
} + o\(1\)\).
\end{split}\end{align}
For $m \in \{1, 2, \ldots\}$,
define $D\(m\)$ as the set of $\(x_{2}, \ldots, x_{n}\) \in B_{x_{2}, \ldots, x_{n}}$ such that
\begin{align}\begin{split}\nonumber
\lim_{t_{n} \downarrow 0, c_{n} \to c}
\frac{1}{t_n} & \int
\frac{
\[
1\(c' + t_n \(H\(
\psi_{B}\(x_{2, \ldots, n}, c\)\) + \frac{1}{m}
\) > c_{n}\)
-1\(c' > c_{n}\)
\]
f\(\psi_{B}\(
x_{2,\ldots,n}, c'
\)
\)
}{
\left\vert
\det\(
\frac{\partial}{\partial x_1} h\(
\psi_{B}\(x_{2, \ldots, n}, c'\)
\)
\)
\right\vert
}
d \mathcal L_1 c' \\
& = \(H\(
\psi_{B}\(x_{2, \ldots, n}, c\)\) + \frac{1}{m}
\) \frac{f\(\psi_{B}\(
x_{2,\ldots,n}, c
\)
\)}{\left\vert
\det\(
\frac{\partial}{\partial x_1} h\(
\psi_{B}\(x_{2, \ldots, n}, c\)
\)
\)
\right\vert
}.
\end{split}\end{align}
By Assumption \ref{a.e. continuous at boundary assumption},
$\mathcal{L}_{n - 1}\left\{B_{x_{2}, \ldots, x_{n}} \backslash D\(m\)\right\} = 0$ and further
$\mathcal{L}_{n - 1}\left\{\bigcup_{m}B_{x_{2}, \ldots, x_{n}} \backslash D\(m\)\right\} = 0$,
which imply
\begin{align}\begin{split}\nonumber
\lim_{t_{n} \downarrow 0, c_{n} \to c}
& \frac{1}{t_n} \int
\frac{
\[
1\(c' + t_n H_{n}\(
\psi_{B}\(x_{2, \ldots, n}, c'\)
\) > c_{n}\)
-1\(c' > c_{n}\)
\]
f\(\psi_{B}\(
x_{2,\ldots,n}, c'
\)
\)
}{
\left\vert
\text{det}\(
\frac{\partial}{\partial x_1} h\(
\psi_{B}\(x_{2, \ldots, n}, c'\)
\)
\)
\right\vert
}
d \mathcal L_1 c' \\
& = H\(
\psi_{B}\(x_{2, \ldots, n}, c\)\)
\frac{f\(\psi_{B}\(
x_{2,\ldots,n}, c
\)
\)}{\left\vert
\det\(
\frac{\partial}{\partial x_1} h\(
\psi_{B}\(x_{2, \ldots, n}, c\)
\)
\)
\right\vert
}
\end{split}\end{align}
for $\mathcal{L}_{n - 1}$ almost every $\(x_{2}, \ldots, x_{n}\) \in B_{x_{2}, \ldots, x_{n}}$.
By Binet-Cauchy formula (see for example chapter 1 of \cite{gantmakher1977theory}), Assumption \ref{Clarke regular point assumption} implies that locally $\det\(\frac{\partial}{\partial x_{1}}h\(\cdot\)\)$ is bounded away from $0$.
Since $f$ is bounded by Assumption \ref{borel compact assumption}, through DCT, we have
\begin{align}\begin{split}\nonumber
\lim_{t_{n} \downarrow 0, c_{n} \to c} \frac{1}{t_n} \int
\[
1\(h\(x\) + t_n H_{n}\(x\) > c\)
-1\(h\(x\) > c\)
\] f\(x\)
d \mathcal L_n x = \int_{h^{-1}\(c\)}
\frac{
H\(x\) f\(x\)
}{
J h\(x\)
}
d \mathcal H_{n-1} x.
\end{split}\end{align}
This finishes the proof of Theorem \ref{nonsmooth general k derivative}.
\end{proof}

\subsection{Proofs of section \ref{constrained and roc}}

\begin{proof}[Proof of Theorem \ref{binary allocation asymptotic result}]
We have the following decomposition
\begin{align}\begin{split}\nonumber
\alpha\(Q + t_{n}\mathcal{Q}_{n}, \Delta^{\star} + t_{n}H_{n}, k\) - \alpha\(Q, \Delta^{\star}, k\)
& = \underbrace{\alpha\(Q + t_{n}\mathcal{Q}_{n}, \Delta^{\star}, k\) - \alpha\(Q, \Delta^{\star}, k\)}_{\(1\)} \\
& + \underbrace{\alpha\(Q, \Delta^{\star} + t_{n}H_{n}, k\) - \alpha\(Q, \Delta^{\star}, k\)}_{\(2\)} \\
& + \underbrace{t_{n}\(\alpha\(\mathcal{Q}_{n}, \Delta^{\star} + t_{n}H_{n}, k\) - \alpha\(\mathcal{Q}_{n}, \Delta^{\star}, k\)\)}_{\(3\)}.
\end{split}\end{align}
The $\(3\)$ term can be further decomposed and bounded as
\begin{align}\begin{split}\label{alpha residual bound}
\vert \alpha\(\mathcal{Q}_{n}, \Delta^{\star} + t_{n}H_{n}, k\) - \alpha\(\mathcal{Q}_{n}, \Delta^{\star}, k\)\vert
& \leq \vert \alpha\(\mathcal{Q}_{n}, \Delta^{\star} + t_{n}H_{n}, k\) - \alpha\(\mathcal{Q}, \Delta^{\star} + t_{n}H_{n}, k\)\vert \\
& + \vert \alpha\(\mathcal{Q}_{n}, \Delta^{\star}, k\) - \alpha\(\mathcal{Q}, \Delta^{\star}, k\)\vert \\ 
& + \vert \alpha\(\mathcal{Q}, \Delta^{\star} + t_{n}H_{n}, k\) - \alpha\(\mathcal{Q}, \Delta^{\star}, k\)\vert. 
\end{split}\end{align}
Since we assume Donskerness of $\mathcal{F}$ in Assumption \ref{joint convergence assumption},
we can derive Hadamard differentiability tangential to $\mathbf{Q} \times C\(E\)$ of $\alpha\(\cdot, \cdot, \cdot\)$ and $\beta\(\cdot, \cdot, \cdot\)$ with respect to the first two variables,
such that $\mathbf{Q} \subset \ell^{\infty}\(\mathcal{F}\)$ consists of elements that are uniformly continuous on $\mathcal{F}$ with respect to the $L^{2}\(Q\)$ norm.
Under this setting, the first two terms in \eqref{alpha residual bound} converge to $0$ since $\Vert\mathcal{Q}_{n} - \mathcal{Q}\Vert_{\ell^{\infty}\(\mathcal{F}\)} \to 0$.
Note that under Assumption \ref{variable assumption} and Assumption \ref{regular Delta assumption},
Proposition \ref{surface integral continuity} shows that $\Delta^{\star}\(X\)$ has a continuous density.
By the $L^{2}\(Q\)$ uniform continuity of $\mathcal{Q}$, the third term also converges to $0$.

Under Assumption \ref{variable assumption} and Assumption \ref{regular Delta assumption},
we can also evoke Theorem \ref{general k derivative};
also by the chain rule of Hadamard derivative (Lemma 3.10.3 in \cite{van2023weak}),
\bs
& \alpha'_{Q, \Delta^{\star}, k}\(\mathcal{Q}, H, 0\) =
\mathcal{Q}\[\(z_{1} - z_{0}\)1\(\Delta^{\star}\(x; k\) > 0\) + z_{0}\] \\
& + \Biggl[\int_{\Delta^{\star}\(x; k\) = 0} \frac{\(c_{1}\(x\) - c_{0}\(x\)\)\(H_{1}\(x\) - H_{0}\(x\)\)\mu'\(x\)}{\Vert\nabla\Delta^{\star}\(x, k\)\Vert}d\mathcal{H}_{n - 1}x \\
& \quad\quad\quad\quad\quad\quad\quad\quad\quad\quad\quad\quad\quad\quad - \int_{\Delta^{\star}\(x; k\) = 0} \frac{k \(c_{1}\(x\) - c_{0}\(x\)\)\(H_{3}\(x\) - H_{2}\(x\)\)\mu'\(x\)}{\Vert\nabla\Delta^{\star}\(x, k\)\Vert}d\mathcal{H}_{n - 1}x\Biggl].
\end{split}\end{align}
By the condition in \eqref{nonsingular manifold jacobian denominator in binary classification roc curve}, 
the Hadamard differentiability of the inverse map in Lemma 3.10.24 of \cite{van2023weak} and the chain rule again imply that
\bs
k'_{Q, \Delta^{\star}, \alpha}\(\mathcal{Q}, H, 0\) = \frac{\alpha'_{Q, \Delta^{\star}, k}\(\mathcal{Q}, H, 0\)}{f_{\alpha}\(Q, \Delta^{\star}, k\(Q, \Delta^{\star}, \alpha\)\)},
\end{split}\end{align}

By exactly the same calculation, we can also get
\bs
& \beta'_{Q, \Delta^{\star}, k}\(\mathcal{Q}, H, 0\) =
\mathcal{Q}\[\(y_{1} - y_{0}\)1\(\Delta^{\star}\(x; k\) > 0\) + y_{0}\] \\
& + \Biggl[\int_{\Delta^{\star}\(x; k\) = 0} \frac{\(g_{1}\(x\) - g_{0}\(x\)\)\(H_{1}\(x\) - H_{0}\(x\)\)\mu'\(x\)}{\Vert\nabla\Delta^{\star}\(x, k\)\Vert}d\mathcal{H}_{n - 1}x \\
& \quad\quad\quad\quad\quad\quad\quad\quad\quad\quad\quad\quad\quad\quad - \int_{\Delta^{\star}\(x; k\) = 0} \frac{k\(g_{1}\(x\) - g_{0}\(x\)\)\(H_{3}\(x\) - H_{2}\(x\)\)\mu'\(x\)}{\Vert\nabla\Delta^{\star}\(x, k\)\Vert}d\mathcal{H}_{n - 1}x\Biggl],
\end{split}\end{align}
By the chain rule the third time, there is
\bs
\beta'_{Q, \Delta^{\star}, \alpha}\(\mathcal{Q}, H, 0\) = \beta'_{Q, \Delta^{\star}, k}\(\mathcal{Q}, H, 0\)\bigg\vert_{k = k\(Q, \Delta^{\star}, \alpha\)} -
\frac{f_{\beta}\(Q, \Delta^{\star}, k\(Q, \Delta^{\star}, \alpha\)\)}{f_{\alpha}\(Q, \Delta^{\star}, k\(Q, \Delta^{\star}, \alpha\)\)} \alpha'_{Q, \Delta^{\star}, k}\(\mathcal{Q}, H, 0\)\bigg\vert_{k = k\(Q, \Delta^{\star}, \alpha\)}.
\end{split}\end{align}
Now, \eqref{main asymptotic formula} follows from Theorem 3.10.4 in \cite{van2023weak}.
\end{proof}

\begin{proof}[Proof of Proposition \ref{parametric model based estimator more efficient than orthogonal estimator}]
Define
\begin{align*}
&\zeta\(\omega, x; \theta, \alpha\)
=  \(y_1 - y_0
-k\(Q, \theta, \alpha\)
\(z_{1} - z_{0}\) \)
1\(\Delta\(x; k\(\mu, \theta, \alpha\)\) > 0\) + y_0
-k\(Q, \theta, \alpha\) z_0. \notag
\end{align*}
We also define, using the score function notation
$s_{\theta_1}\(\omega \vert x; \theta_1\) =
\frac{\partial}{\partial \theta_1} \log f\(\omega \vert x, \theta_1\)$,
\bsnumber\label{jacobian function in the asymptotics of parametric model based estimator}
\mathcal J\(\theta^{\star},\alpha\) \equiv \frac{\partial}{\partial \theta_1} \mu\(\psi\(x; \theta_1, \theta_2, \alpha\)\)\bigg\vert_{\theta_1 = \theta_2 = \theta^{\star}}
= \underbrace{Q \[\zeta\(\omega, x; \theta^{\star}, \alpha\) s_{\theta_1}\(\omega \vert x; \theta_1\)\]\bigg\vert_{\theta_1= \theta^{\star}}}_{\text{under correct specification}}.
\end{split}\end{align}
In the above, we let
\begin{align}\begin{split}\label{psi function in the asymptotics of parametric model based estimator}
& \psi\(x; \theta_1, \theta_2, \alpha\) \coloneqq  \int \zeta\(\omega, x; \theta_2, \alpha\) f\(\omega \vert x; \theta_1\) d\omega \\
&= \(g_{1}\(x;\theta_1\) - g_{0}\(x; \theta_1\)
-k\(Q, \theta_2, \alpha\)
\(c_{1}\(x; \theta_1\) - c_{0}\(x; \theta_1\)\)
\)1\(\Delta\(x; k\(\mu, \theta_2, \alpha\)\) > 0\) \\
& + g_{0}\(x;\theta_1\)
-k\(Q, \theta_2, \alpha\) c_{0}\(x; \theta_1\).
\end{split}\end{align}
Under correct parametric specification, by Theorem \ref{binary allocation asymptotic result} and the above definitions the model based estimator satisfies
\bsnumber\label{main asymptotic formula under parametric specification of model based estimator}
& \sqrt{n} \(\beta\(\hat\mu, \hat{\theta}, \alpha\) - \beta\(\mu, \theta^{\star}, \alpha\)\)
\rightsquigarrow
\mathbb{Q}\biggl[
\psi\(x; \theta^{\star}, \theta^{\star}, \alpha\)  - \mathcal J\(\theta^{\star},\alpha\) H\(\theta^{\star}\)^{-1} s_{\theta_1}\(\omega \vert x;
\theta^{\star}\)
\biggr].
\end{split}\end{align}
We can also rewrite
\eqref{main asymptotic formula under correct specification} as
\bsnumber\label{rewrite orthogonal main asymptotic formula under correct parametric specification}
& \sqrt{n} \(\beta\(\mathbb{Q}_{n}, \hat{\theta}, \alpha\) - \beta\(Q, \theta^{\star}, \alpha\)\) \rightsquigarrow
\mathbb{Q}\biggl[ \zeta\(\omega, x; \theta^{\star}, \alpha\) \biggr].
\end{split}\end{align}
Under correct parametric specification
\bsnumber\label{information matrix inequality for MLE}
H\(\theta^{\star}\) = Q\[\frac{\partial^2}{\partial \theta_1^2} \log f\(\omega \vert x, \theta^{\star}\)\]
= -V\(\theta^{\star}\) = -Q\[
s_{\theta_1}\(\omega \vert x; \theta^{\star}\)
s^{T}_{\theta_1}\(\omega \vert x; \theta^{\star}\)\]
\end{split}\end{align}
Then
\eqref{main asymptotic formula under parametric specification of model based estimator}	 becomes
\bsnumber\label{main asymptotic formula under correct parametric specification of model based estimator}
& \sqrt{n} \(\beta\(\hat\mu, \hat{\theta}, \alpha\) - \beta\(\mu, \theta^{\star}, \alpha\)\) \rightsquigarrow
\mathbb{Q}\[
\psi\(x; \theta^{\star}, \theta^{\star}, \alpha\)  + \mathcal J\(\theta^{\star},\alpha\) V\(\theta^{\star}\)^{-1} s_{\theta_1}\(\omega \vert x; \theta^{\star}\)
\].
\end{split}\end{align}
The covariance kernel induced by \eqref{main asymptotic formula under correct parametric specification of model based estimator} is smaller
than the covariance kernel induced by
\eqref{rewrite orthogonal main asymptotic formula under correct parametric specification}, in the sense that their difference is seminegative
definitive. For this purpose it suffices to show that
\bs
Cov\biggl(
\zeta\(\omega, x; \theta^{\star}, \alpha\) -  \psi\(x; \theta^{\star}, \theta^{\star}, \alpha\) & - \mathcal J\(\theta^{\star},\alpha\) V\(\theta^{\star}\)^{-1} s_{\theta_1}\(\omega \vert x; \theta^{\star}\), \\
& \psi\(x; \theta^{\star}, \theta^{\star}, \alpha'\)  + \mathcal J\(\theta^{\star}, \alpha'\) V\(\theta^{\star}\)^{-1} s_{\theta_1}\(\omega \vert x; \theta^{\star}\)
\biggr) = 0.
\end{split}\end{align}
for all $\alpha$ and $\alpha'$.	 We decompose this into
\bs
&Cov\biggl(
\zeta\(\omega, x; \theta^{\star}, \alpha\) - \psi\(x; \theta^{\star}, \theta^{\star}, \alpha\) - \mathcal J\(\theta^{\star}, \alpha\) V\(\theta^{\star}\)^{-1} s_{\theta_1}\(\omega \vert x; \theta^{\star}\), \psi\(x; \theta^{\star}, \theta^{\star}, \alpha'\)\biggr) \\
& -Cov\biggl(
\psi\(x; \theta^{\star}, \theta^{\star}, \alpha\),
\mathcal J\(\theta^{\star}, \alpha'\) V\(\theta^{\star}\)^{-1} s_{\theta_1}\(\omega \vert x; \theta^{\star}\)
\biggr)\\
& +Cov\biggl(
\zeta\(\omega, x; \theta^{\star}, \alpha\)  - \mathcal J\(\theta^{\star},\alpha\) V\(\theta^{\star}\)^{-1} s_{\theta_1}\(\omega \vert x; \theta^{\star}\),
\mathcal J\(\theta^{\star}, \alpha'\) V\(\theta^{\star}\)^{-1} s_{\theta_1}\(\omega \vert x; \theta^{\star}\)\biggr).
\end{split}\end{align}
The first two terms are zero by the law of iterate expectation.
The last term is zero by direct calculation
using \eqref{information matrix inequality for MLE}, because
$\mathcal J\(\theta^{\star}, \alpha\) V\(\theta^{\star}\)^{-1} s_{\theta_1}\(\omega \vert x; \theta^{\star}\)$ is the least square projection of
$\zeta\(\omega, x; \theta^{\star}, \alpha\)$ to the linear space spanned by the conditional score functions
$s_{\theta_1}\(\omega \vert x; \theta^{\star}\)$.
The efficiency comparison continues to hold when $\hat\theta$ is estimated by
any moment condition $s_{\theta_1}\(\omega \vert x; \theta_1\)$ such that  $\mathbb{E}\[s_{\theta_1}\(\omega \vert x; \theta^{\star}\) \vert x\] = 0$
and information matrix equality holds, such as a method of moment estimator with optimally chosen instruments.
\end{proof}

\begin{proof}[Proof of Proposition \ref{treatment parametric model based estimator more efficient than orthogonal estimator}]
Under correct parametric specification, let
\bs
&\zeta\(y, z, d,x;\theta^{\star}, \alpha\) = \psi\(x; \theta^{\star}, \theta^{\star}, \alpha\) + \eta\(y, z, d,x;\theta^{\star}, \alpha\)  -
\beta\(Q, \theta^{\star}, \alpha\) - k\(Q, \theta^{\star}, \alpha\)  \alpha
\end{split}\end{align}
where
\bs
\eta\(y, z, d,x;\theta, \alpha\) & =
\biggl[ \frac{d}{p\(x, \theta\)} \(y - g_1\(x\)
- k\(Q, \theta, \alpha\)
\(
z - c_1\(x\)
\)
\)
1\(\Delta\(x; k\(Q, \theta, \alpha\)\) > 0\) \\
& + \frac{1-d}{1-p\(x, \theta\)} \(y - g_0\(x\)
-
k\(Q, \theta, \alpha\)
\(
z - c_0\(x\)
\)
\)  1\(\Delta\(x; k\(Q, \theta, \alpha\)\) \leq 0\)
\biggr].
\end{split}\end{align}
Recall \eqref{psi function in the asymptotics of parametric model based estimator} and (the first equality in)
\eqref{jacobian function in the asymptotics of parametric model based estimator}, \eqref{main asymptotic formula under parametric specification of model based estimator} continues to hold:
\bs
& \sqrt{n} \(\beta\(\hat\mu, \hat{\theta}, \alpha\) - \beta\(\mu, \theta^{\star}, \alpha\)\) \rightsquigarrow
\mathbb{Q}\biggl[
\psi\(x; \theta^{\star}, \theta^{\star}, \alpha\) - \mathcal J\(\theta^{\star}, \alpha\) H\(\theta^{\star}\)^{-1} s_{\theta}\(y, z \vert d, x; \theta^{\star}\)
\biggr].
\end{split}\end{align}
We still have
$H\(\theta^{\star}\) = -V\(\theta^{\star}\)$. 
Efficiency ranking follows from showing
\bs
Cov\biggl(
\eta\(y, z, d, x;\theta^{\star}, \alpha\) & - \mathcal J\(\theta^{\star},\alpha\) V\(\theta^{\star}\)^{-1} s_{\theta}\(y, z \vert d, x; \theta^{\star}\), \\
& \psi\(x; \theta^{\star}, \theta^{\star}, \alpha'\)  + \mathcal J\(\theta^{\star},\alpha'\) V\(\theta^\star\)^{-1} s_{\theta}\(y, z \vert d, x; \theta^{\star}\)
\biggr) = 0.
\end{split}\end{align}
for all $\alpha$ and $\alpha'$.
By the law of iterated expectation,
\bs
Cov\biggl(
\eta\(y, z, d, x;\theta^{\star}, \alpha\) - \mathcal J\(\theta^{\star}, \alpha\) V\(\theta^{\star}\)^{-1} s_{\theta}\(y, z \vert d, x; \theta^{\star}\),
\psi\(x; \theta^{\star}, \theta^{\star}, \alpha'\)
\biggr) = 0.
\end{split}\end{align}
Finally, direct verification of
\bs
&Cov\biggl(
\eta\(y, z, d, x;\theta^{\star}, \alpha\) - \mathcal J\(\theta^{\star}, \alpha\) V\(\theta^{\star}\)^{-1} s_{\theta}\(y, z \vert d, x; \theta^{\star}\),
\mathcal J\(\theta^{\star}, \alpha'\) V\(\theta^{\star}\)^{-1} s_{\theta}\(y, z \vert d, x; \theta^{\star}\)
\biggr) = 0,
\end{split}\end{align}
follows from checking that under \eqref{binary unconfoundedness assumption} and using the definition in
\eqref{jacobian function in the asymptotics of parametric model based estimator},
\bs
\mathcal J\(\theta^{\star}, \alpha\) =
\frac{\partial}{\partial \theta_1} \mu\(\psi\(x; \theta_1, \theta_2, \alpha\)\)\bigg\vert_{\theta_1 = \theta_2 = \theta^{\star}}
= Cov\(\eta\(y, z, d, x; \alpha\),	 s_{\theta}\(y, z \vert d, x; \theta^{\star}\)\).
\end{split}\end{align}

In the case when $p\(x;\theta\)$ is misspecified but
$g_1\(x;\theta\), g_0\(x;\theta\), c_1\(x;\theta\), c_0\(x;\theta\)$ are all
correctly specified, 
an adaption to Theorem \ref{binary allocation asymptotic result} show that
\bsnumber\label{main asymptotic formula with partial observability under partial correct specification and misspecified pscore}
& \sqrt{n}\(\beta\(\mathbb{Q}_{n}, \hat{\theta}, \alpha\) - \beta\(Q, \theta^{\star}, \alpha\)\) \rightsquigarrow
\mathbb{Q}\(\zeta_1\(y, z, d, x;\theta^{\star}, \alpha\)
+ \zeta_2\(y, d, x;\theta^{\star}, \alpha\)\).
\end{split}\end{align}
Denote the parametric function 
\bs
\Delta\(x; k\(Q, \theta^{\star}, \alpha\)\) = \Delta\(x; k\(Q, \theta^{\star}, \alpha\);\theta^{\star}\)
\coloneqq \(g_{1}\(x; \theta^{\star}\) - g_{0}\(x; \theta^{\star}\)\)
- k\(Q, \theta^{\star}, \alpha\)\(\(c_{1}\(x; \theta^{\star}\) - c_{0}\(x; \theta^{\star}\)\)\).
\end{split}\end{align}
In \eqref{main asymptotic formula with partial observability under partial correct specification and misspecified pscore}, we define
\bs
&\zeta_1\(y, z, d,x;\theta^{\star}, \alpha\) = \biggl[\(w_{1}\(x; Q, \theta^{\star}, \alpha\) + \frac{d}{p\(x;\theta^{\star}\)} \(\omega\(Q, \theta^{\star}, \alpha\) -
w_{1}\(x; Q, \theta^{\star}, \alpha\)
\) \)
1\(\Delta\(x; k\(Q, \theta^{\star}, \alpha\)\) > 0\) \\
& +
\(
w_{0}\(x; Q, \theta^{\star}, \alpha\)
+ \frac{1-d}{1-p\(x; \theta^{\star}\)} \(
\omega\(Q, \theta^{\star}, \alpha\) -
w_{0}\(x; Q, \theta^{\star}, \alpha\)
\) \)
1\(\Delta\(x; k\(Q, \theta^{\star}, \alpha\)\) \leq 0\)
\biggr],
\end{split}\end{align}
where $\omega\(Q, \theta^{\star}, \alpha\) = y - k\(Q, \theta^{\star}, \alpha\)	 z$,
\bs
&w_{1}\(x; Q, \theta^{\star}, \alpha\) = \mathbb{E}\[y -   k\(Q, \theta^{\star}, \alpha\)  z \vert x, d=1, \theta^{\star}\]
= g_1\(x; \theta^{\star}\) - k\(Q, \theta^{\star}, \alpha\) c_1\(x; \theta^{\star}\), \\
&w_{0}\(x; Q, \theta^{\star}, \alpha\) = \mathbb{E}\[y -   k\(Q, \theta^{\star}, \alpha\)  z \vert x, d=0, \theta^{\star}\]
=g_0\(x; \theta^{\star}\) - k\(Q, \theta^{\star}, \alpha\) c_0\(x; \theta^{\star}\).
\end{split}\end{align}
We also define
\bs
&\zeta_2\(y, z, d,x;\theta^{\star}, \alpha\)
= - \mu\biggl(\(1-\frac{p\(x\)}{p\(x;\theta^{\star}\)}\)
\(\frac{g_1\(x;\theta^{\star}\)}{\partial\theta} - k\(Q, \theta^{\star}, \alpha\)
\frac{c_1\(x;\theta^{\star}\)}{\partial\theta} \)
1\(\Delta\(x; k\(Q, \theta^{\star}, \alpha\); \theta^\star\) > 0\) \\
& + \(\frac{p\(x\)-p\(x;\theta^{\star}\)}{1-p\(x;\theta^{\star}\)}\)
\(\frac{g_0\(x;\theta^{\star}\)}{\partial\theta} - k\(Q, \theta^{\star}, \alpha\)
\frac{c_0\(x;\theta^{\star}\)}{\partial\theta} \)
1\(\Delta\(x; k\(Q, \theta^{\star}, \alpha\); \theta^\star\) \leq 0\)\biggr)\\
& \times
 H\(\theta^{\star}\)^{-1} s_\theta\(y, z \vert d, x; \theta^{\star}\).
\end{split}\end{align}
When $f_{1}\(y, z\vert x; \theta\)$ and $f_{0}\(y, z\vert x; \theta\)$ are correctly specified, $H\(\theta^{\star}\) = -V\(\theta^{\star}\)$,
\bs
&\zeta_1\(y,z, d,x;\theta^{\star}, \alpha\) =
\psi\(x; \theta^{\star}, \theta^{\star}, \alpha\) + \zeta_3\(y, z, d,x;\theta^{\star}, \alpha\),\ \ \text{where}\\
&\zeta_3\(y, z, d,x;\theta^{\star}, \alpha\) = \biggl[\frac{d}{p\(x;\theta^{\star}\)} \(\omega\(Q, \theta^{\star}, \alpha\) -
w_{1}\(x; Q, \theta^{\star}, \alpha\)
\)
1\(\Delta\(x; k\(Q, \theta^{\star}, \alpha\)\) > 0\) \\
& +
 \frac{1-d}{1-p\(x; \theta^{\star}\)} \(
\omega\(Q, \theta^{\star}, \alpha\) -
w_{0}\(x; Q, \theta^{\star}, \alpha\)
\)
1\(\Delta\(x; k\(Q, \theta^{\star}, \alpha\)\) \leq 0\)
\biggr].
\end{split}\end{align}
The efficiency claim can be verified by showing that
\bs
Cov\biggl(
\zeta_3\(y, z, d, x;\theta^{\star}, \alpha\) +
\zeta_2\(y, z, d, x;\theta^{\star}, \alpha\)
& - \mathcal J\(\theta^{\star},\alpha\) V\(\theta^{\star}\)^{-1} s_{\theta}\(y, z \vert d, x; \theta^{\star}\), \\
& \psi\(x; \theta^{\star}, \theta^{\star}, \alpha'\) + \mathcal J\(\theta^{\star},\alpha'\) V\(\theta^{\star}\)^{-1} s_{\theta}\(y, z \vert d, x; \theta^{\star}\)
\biggr) = 0.
\end{split}\end{align}
for all $\alpha$ and $\alpha'$.	 The law of iterated expectation implies that
\bs
&Cov\biggl(
\zeta_3\(y, z, d, x;\theta^{\star}, \alpha\) +
\zeta_2\(y, z, d, x;\theta^{\star}, \alpha\)
 - \mathcal J\(\theta^{\star},\alpha\) V\(\theta^{\star}\)^{-1} s_{\theta}\(y, z \vert d, x; \theta^{\star}\),
\psi\(x; \theta^{\star}, \theta^{\star}, \alpha'\)
\biggr) = 0.
\end{split}\end{align}
Furthermore, the equality
\bs
Cov\biggl(
\zeta_3\(y, z, d, x;\theta^{\star}, \alpha\) +
\zeta_2\(y, z, d, x;\theta^{\star}, \alpha\)
& - \mathcal J\(\theta^{\star}, \alpha\) V\(\theta^{\star}\)^{-1} s_{\theta}\(y, z \vert d, x; \theta^{\star}\), \\
& \mathcal J\(\theta^{\star},\alpha'\) V\(\theta^{\star}\)^{-1} s_{\theta}\(y, z \vert d, x; \theta^{\star}\)
\biggr) = 0,
\end{split}\end{align}
follows from
\bs
&Cov\biggl(
\zeta_3\(y, z, d, x;\theta^{\star}, \alpha\) +
\zeta_2\(y, z, d, x;\theta^{\star}, \alpha\),
s_{\theta}\(y, z \vert d, x; \theta^{\star}\)\biggr) = \mathcal J\(\theta^{\star}, \alpha\).
\end{split}\end{align}
This in turn follows from verifying that
\bs
& Cov\biggl(
\zeta_3\(y, z, d, x;\theta^{\star}, \alpha\),
s_{\theta}\(y, z \vert d, x; \theta^{\star}\)\biggr) \\
& =
\mu\biggl(\frac{p\(x\)}{p\(x;\theta^{\star}\)}
\(\frac{g_1\(x;\theta^{\star}\)}{\partial\theta} - k\(Q, \theta^{\star}, \alpha\)
\frac{c_1\(x;\theta^{\star}\)}{\partial\theta} \)
1\(\Delta\(x; k\(Q, \theta^{\star}, \alpha\)\) > 0\) \\
& + \(\frac{1-p\(x\)}{1-p\(x;\theta^{\star}\)}\)
\(\frac{g_0\(x;\theta^{\star}\)}{\partial\theta} - k\(Q, \theta^{\star}, \alpha\)
\frac{c_0\(x;\theta^{\star}\)}{\partial\theta} \)
1\(\Delta\(x; k\(Q, \theta^{\star}, \alpha\)\) \leq 0\)\biggr).
\end{split}\end{align}
\end{proof}

\begin{proof}[Proof of Corollary \ref{roc asymptotic result}]
The proof is similar to the proof of Theorem \ref{binary allocation asymptotic result}, we only need to point out that by chain rule
\bs
&\alpha'_{Q, p^{\star}, k}\(\mathcal{Q}, H, 0\) \\
& = \frac{1}{Q\(1-y\)}\[\mathcal{Q}\[\(1 - y\)\(1\(p^{\star}\(x\) > k\) - \alpha\(Q, p^{\star}, k\)\)\] + \int_{p^{\star}\(x\) = k}\frac{\(1 - p\(x\)\)H\(x\)\mu'\(x\)}{\Vert\nabla p^{\star}\(x\)\Vert}d\mathcal{H}_{n - 1}x\],
\end{split}\end{align}
so we get an additional $\alpha$ term.
\end{proof}

\begin{proof}[Proof of Proposition \ref{ROC parametric model based estimator more efficient than orthogonal estimator}]
The proof is similar to the proof of Proposition \ref{parametric model based estimator more efficient than orthogonal estimator}.
The score function and Hessian in \eqref{mle influence function} now take the form of
\bs
&s_{\theta_1}\(y \vert x; \theta_1\) = 
\frac{y - p\(x; \theta_1\)}{ p\(x; \theta_1\) \(1-p\(x; \theta_1\)\)} \frac{\partial}{\partial \theta_1} p\(x; \theta_1\) \quad \text{and} \quad H\(\theta_1\) = \mu\[\frac{\partial}{\partial\theta_1} s_{\theta_1}\(y \vert x; \theta_1\)
\].
\end{split}\end{align}
Similar to 
\eqref{psi function in the asymptotics of parametric model based estimator},
we define
\bs
&\zeta\(y, x; \theta, \alpha\) \coloneqq \frac{y -
k\(\mu, \theta, \alpha\)}{
1 - k\(\mu, \theta, \alpha\)
}
1\(p\(x; \theta\) >  k\(\mu, \theta, \alpha\)\)
- y\beta\(\mu, \theta, \alpha\)
+ \frac{k\(\mu, \theta, \alpha\)}{1 - k\(\mu, \theta, \alpha\)} \(1 - y\)\alpha, \\
&\psi\(x; \theta_1, \theta_2, \alpha\)
\coloneqq \int	\zeta\(y, x; \theta_{2}, \alpha\) f\(y \vert x; \theta_1\) dy \\
& = \frac{p\(x; \theta_1\)
- k\(\mu, \theta_2, \alpha\)
}{
1 - k\(\mu, \theta_2, \alpha\)
} 1\(p\(x; \theta_2\) >	 k\(\mu, \theta_2, \alpha\)\)
- p\(x; \theta_{1}\)\beta\(\mu, \theta_{2}, \alpha\) \\
& + \frac{k\(\mu, \theta_{2}, \alpha\)}{1 - k\(\mu, \theta_{2}, \alpha\)} \(1 - p\(x; \theta_{1}\)\)\alpha.
\end{split}\end{align}
Under general specification, the model-based ROC satisfies
\bsnumber\label{main asymptotic formula under general parametric specification of model based roc}
& \sqrt{n} \(\beta\(\hat\mu, \hat{\theta}, \alpha\) - \beta\(\mu, \theta^{\star}, \alpha\)\)\\
 &\rightsquigarrow
\frac{1}{\mu\(p\(x; \theta^{\star}\)\)}
\mathbb{Q}\biggl[
\psi\(x; \theta^{\star}, \theta^{\star}, \alpha\) - \mathcal J\(\theta^{\star},\alpha\) H\(\theta^{\star}\)^{-1} s_{\theta_1}\(y \vert x; \theta^{\star}\)
\biggr].
\end{split}\end{align}
First rewrite
\eqref{roc curve asymptotics under correct specification} as
\bs
\sqrt{n}\(
\beta\(\mathbb Q_n, \hat\theta, \alpha\) - \beta\(Q, \theta^{\star}, \alpha\)\)
\rightsquigarrow \frac{1}{Qy} Q\[\zeta\(y, x; \theta^{\star}, \alpha\]
\).
\end{split}\end{align}
Under correct specification, $\mu\(p\(x; \theta^\star\)\) = Q y$, and
\bs
-H\(\theta^{\star}\) = V\(\theta^{\star}\) =
\mathbb{E} \[s_{\theta_1}\(y \vert x; \theta^{\star}\) s^{T}_{\theta_1}\(y \vert x; \theta^{\star}\)\]
=\mu\[
\frac{1}{ p\(x; \theta^{\star}\) \(1-p\(x; \theta^{\star}\)\)} \frac{\partial}{\partial \theta_1} p\(x; \theta^{\star}\)
\frac{\partial}{\partial \theta^{T}_1} p\(x; \theta^{\star}\)
\].
\end{split}\end{align}
Then \eqref{main asymptotic formula under general parametric specification of model based roc} becomes
\bs
& \sqrt{n} \(\beta\(\hat\mu, \hat{\theta}, \alpha\) - \beta\(\mu, \theta^{\star}, \alpha\)\)
\rightsquigarrow
\frac{1}{Q y}
\mathbb{Q}\biggl[
\psi\(x; \theta^{\star}, \theta^{\star}, \alpha\) + \mathcal J\(\theta^{\star}, \alpha\) V\(\theta^{\star}\)^{-1} s_{\theta_1}\(y \vert x; \theta^{\star}\)
\biggr].
\end{split}\end{align}
The remaining arguments are identical to those following equation
\eqref{main asymptotic formula under correct parametric specification of model based estimator}.
\end{proof}

\begin{proof}[Proof of Corollary \ref{roc bootstrap}]
For the second part of this result, note that by assumption and Theorem 2.10.8 in \cite{van2023weak}, the functional class of pointwise products $\mathcal{G} \cdot \mathcal{F}$ is also Donsker.
Since we already derive the Hadamard derivative of $\beta\(\cdot, \cdot, \alpha\)$ in the proof of Corollary \ref{roc asymptotic result},
the bootstrap consistency is just a direct application of the conditional delta method Theorem 3.10.11 in \cite{van2023weak}.

For the first part, note that by Assumption \ref{joint convergence assumption}, we can evoke Theorem 3.7.1 in \cite{van2023weak} to get
\bs
\sup_{l \in BL_{1}\(\ell^{\infty}\(E\) \times \ell^{\infty}\(\mathcal{F}\)\)} & \left\vert \mathbb{E}_{M}l\(\sqrt{n}\(
\begingroup
\renewcommand*{\arraystretch}{1}
\begin{varmatrix}[delim = p, size = \normalsize, sep = 0.5\arraycolsep]
\hat{p} \\
\tilde{\mathbb{Q}}_{n}
\end{varmatrix}
\endgroup
-
\begingroup
\renewcommand*{\arraystretch}{1}
\begin{varmatrix}[delim = p, size = \normalsize, sep = 0.5\arraycolsep]
\hat{p} \\
\mathbb{Q}_{n}
\end{varmatrix}
\endgroup
\)\)
- \mathbb{E}l\(
\begingroup
\renewcommand*{\arraystretch}{1}
\begin{varmatrix}[delim = p, size = \normalsize, sep = 0.5\arraycolsep]
0 \\
\mathbb{Q}
\end{varmatrix}
\endgroup
\)
\right\vert \to 0, \\
& \mathbb{E}_{M}l\(\sqrt{n}\(
\begingroup
\renewcommand*{\arraystretch}{1}
\begin{varmatrix}[delim = p, size = \normalsize, sep = 0.5\arraycolsep]
\hat{p} \\
\tilde{\mathbb{Q}}_{n}
\end{varmatrix}
\endgroup
-
\begingroup
\renewcommand*{\arraystretch}{1}
\begin{varmatrix}[delim = p, size = \normalsize, sep = 0.5\arraycolsep]
\hat{p} \\
\mathbb{Q}_{n}
\end{varmatrix}
\endgroup
\)\)^{*} -
\mathbb{E}_{M}l\(\sqrt{n}\(
\begingroup
\renewcommand*{\arraystretch}{1}
\begin{varmatrix}[delim = p, size = \normalsize, sep = 0.5\arraycolsep]
\hat{p} \\
\tilde{\mathbb{Q}}_{n}
\end{varmatrix}
\endgroup
-
\begingroup
\renewcommand*{\arraystretch}{1}
\begin{varmatrix}[delim = p, size = \normalsize, sep = 0.5\arraycolsep]
\hat{p} \\
\mathbb{Q}_{n}
\end{varmatrix}
\endgroup
\)\)_{*} \to 0
\end{split}\end{align}
in outer probability. The results then follow from
Theorems 3.10.4 and 3.10.11 in \cite{van2023weak}. 
In fact, by definition of conditional weak convergence, we have 
\begin{align}\begin{split}\nonumber
\sup_{l \in BL_{1}\(\ell^{\infty}\(E\) \times \ell^{\infty}\(\mathcal{F}\)\)} \left\vert
\mathbb{E}_{M}\ell\(\sqrt{n}\beta'_{Q, p^{\star}, \alpha}\(\tilde{\mathbb{Q}}_{n} - \mathbb{Q}_{n}, \hat{p} - \hat{p}, 0\)\) 
- \mathbb{E}\ell\(\beta'_{Q, p^{\star}, \alpha}\(\mathbb{Q}, 0, 0\)\)
\right\vert  \to 0
\end{split}\end{align}
in outer probability. 
We also have that 
\begin{align}\begin{split}\nonumber
& \sup_{l \in BL_{1}\(\ell^{\infty}\(E\) \times \ell^{\infty}\(\mathcal{F}\)\)} \left\vert
\mathbb{E}_{M}\ell\(\sqrt{n}\(\beta\(\tilde{\mathbb{Q}}_{n}, \hat{p}, \alpha\) - \beta\(\mathbb{Q}, \hat{p}, \alpha\)\)\)
- \mathbb{E}_{M}\ell\(\sqrt{n}\beta'_{Q, p^{\star}, \alpha}\(\tilde{\mathbb{Q}}_{n} - \mathbb{Q}_{n}, 0, 0\)\)
\right\vert \\ 
& \leq \epsilon + 2\mathbb{P}_{M}\(\left\vert \sqrt{n}\(\beta\(\tilde{\mathbb{Q}}_{n}, \hat{p}, \alpha\) - \beta\(\mathbb{Q}_{n}, \hat{p}, \alpha\)\) - 
\sqrt{n}\beta'_{Q, p^{\star}, \alpha}\(\tilde{\mathbb{Q}}_{n} - \mathbb{Q}_{n}, 0, 0\)\right\vert^{*} > \epsilon\), 
\end{split}\end{align}
for all $\epsilon > 0$. 
By unconditional convergence and correct specification of the model, 
\begin{align}\begin{split}\nonumber
\sqrt{n}\(\beta\(\tilde{\mathbb{Q}}_{n}, \hat{p}, \alpha\) - \beta\(Q, p^{\star}, \alpha\)\) 
- \sqrt{n}\beta'_{Q, p^{\star}, \alpha}\(\tilde{\mathbb{Q}}_{n} - Q, 0, 0\) & = o_{\mathbb{P}^{*}}\(1\), \\
\sqrt{n}\(\beta\(\mathbb{Q}_{n}, \hat{p}, \alpha\) - \beta\(Q, p^{\star}, \alpha\)\) 
- \sqrt{n}\beta'_{Q, p^{\star}, \alpha}\(\mathbb{Q}_{n} - Q, 0, 0\) & = o_{\mathbb{P}^{*}}\(1\). \\
\end{split}\end{align}
Subtracting the second equation from the first equation gives 
\begin{align}\begin{split}\nonumber
\sqrt{n}\(\beta\(\tilde{\mathbb{Q}}_{n}, \hat{p}, \alpha\) - \beta\(\mathbb{Q}_{n}, \hat{p}, \alpha\)\) - 
\sqrt{n}\beta'_{Q, p^{\star}, \alpha}\(\tilde{\mathbb{Q}}_{n} - \mathbb{Q}_{n}, 0, 0\) = o_{\mathbb{P}^{*}}\(1\), 
\end{split}\end{align}
i.e., the left hand side converges unconditionally to $0$ in outer probability. 
Thus, the conditional probability $\mathbb{P}_{M}\(\left\vert \sqrt{n}\(\beta\(\tilde{\mathbb{Q}}_{n}, \hat{p}, \alpha\) - \beta\(\mathbb{Q}_{n}, \hat{p}, \alpha\)\) - 
\sqrt{n}\beta'_{Q, p^{\star}, \alpha}\(\tilde{\mathbb{Q}}_{n} - \mathbb{Q}_{n}, 0, 0\)\right\vert^{*} > \epsilon\)$ converges in outer mean. 
By Fubini inequality of outer expectation Theorem 1.2.6 in \cite{van2023weak} and the basic property of minimal measurable majorant Lemma 1.2.2 (\rmnum{6}) in \cite{van2023weak}, we get the conditional convergence in outer probability. 
\end{proof}

\begin{proof}[Proof of Lemma \ref{test Hadamard derivative}]
The only difference between Lemma \ref{test Hadamard derivative} and Corollary \ref{roc asymptotic result} is that here we have to take into account the randomness of $\hat{\alpha}_{H}$.
So we can define
\begin{align}\begin{split}\nonumber
\overline{\alpha}\(S, p, k\) & \coloneqq
\frac{S\[1\(p\(x\) > k\)\(1 - y\) - \(w\(1 - y\) - \alpha_{H}\(1 - y\)\)\]}{S\(1 - y\)} \\
& = \frac{S\[1\(p\(x\) > k\)\(1 - y\) - \(w - \alpha_{H}\)\(1 - y\)\]}{S\(1 - y\)}.
\end{split}\end{align}
then use delta method for inverse map to the $\overline{\alpha}\(\cdot, \cdot, \cdot\)$ at $S$, $p^{\star}$ and $k^{\star} = \inf\left\{k: \overline{\alpha}\(S, p^{\star}, k\) \leq \alpha_{H}\right\}$.
Note that for the underlying distribution $S$, $\alpha\(S, p, k\) = \overline{\alpha
}\(S, p, k\)$, since $\alpha_{H} = \frac{S\[w\(1 - y\)\]}{S\(1 - y\)}$.

We have
\begin{align}\begin{split}\nonumber
\overline{\alpha}'_{S, p^{\star}, k^{\star}}\(\mathcal{S}, H, 0\) & =
\frac{1}{S\(1 - y\)}\biggl[\mathcal{S}\[\(1 - y\)\(1\(p^{\star}\(x\) > k^{\star}\) - \(w - \alpha_{H}\) - \overline{\alpha}\(S, p^{\star}, k^{\star}\)\)\] \\
& + \int_{p^{\star}\(x\) = k^{\star}}\frac{\(1 - p\(x\)\)H\(x\)\mu'\(x\)}{\Vert\nabla p^{\star}\(x\)\Vert}d\mathcal{H}_{n - 1}x \biggr].
\end{split}\end{align}
So by the functional delta method for inverse map, there is
\begin{align}\begin{split}\nonumber
c'_{S, p^{\star}, \alpha_{H}}\(\mathcal{S}, H, 0\) = -\frac{\overline{\alpha}'_{S, p^{\star}, k^{\star}}}{f_{\alpha}\(S, p^{\star}, k^{\star}\)},
\end{split}\end{align}
where $f_{\alpha}\(S, p^{\star}, k^{\star}\)$ is defined the same as \eqref{beta alpha derivative wrt k}.

The Hadamard derivative of $\beta\(S, p, k\)$ is 
\begin{align}\begin{split}\nonumber
\beta'_{S, p^{\star}, k^{\star}}\(\mathcal{S}, H, K\) & =
\frac{1}{S\(y\)}\biggl[\mathcal{S}\[y\(1\(p^{\star}\(x\) > k^{\star}\) - \beta\(S, p^{\star}, k^{\star}\)\)\] \\
& + \int_{p^{\star}\(x\) = k^{\star}}\frac{p\(x\)H\(x\)\mu'\(x\)}{\Vert\nabla p^{\star}\(x\)\Vert}d\mathcal{H}_{n - 1}x
- \int_{p^{\star}\(x\) = k^{\star}}\frac{p\(x\)K\mu'\(x\)}{\Vert\nabla p^{\star}\(x\)\Vert}d\mathcal{H}_{n - 1}x
\biggr].
\end{split}\end{align}
By the chain rule of Hadamard derivative (see for example Lemma 3.10.3 in \cite{van2023weak}),
we get
\begin{align}\begin{split}\nonumber
\beta'_{S, p^{\star}}\(\mathcal{S}, H\) & =
\frac{1}{S\(y\)}\biggl\{\mathcal{S}\[y\(1\(p^{\star}\(x\) > k^{\star}\) - \beta\(S, p^{\star}, k^{\star}\)\)\] \\
& + \int_{p^{\star}\(x\) = k^{\star}}\frac{p\(x\)H\(x\)\mu'\(x\)}{\Vert\nabla p^{\star}\(x\)\Vert}d\mathcal{H}_{n - 1}x \\
& - \frac{f_{\beta}}{f_{\alpha}}\frac{S\(y\)}{S\(1 - y\)}
\biggl[\mathcal{S}\[\(1 - y\)\(1\(p^{\star}\(x\) > k^{\star}\) - \(w - \alpha_{H}\) - \overline{\alpha}\(S, p^{\star}, k^{\star}\)\)\] \\
& + \int_{p^{\star}\(x\) = k^{\star}}\frac{\(1 - p\(x\)\)H\(x\)\mu'\(x\)}{\Vert\nabla p^{\star}\(x\)\Vert}d\mathcal{H}_{n - 1}x
\biggr]\biggr\}.
\end{split}\end{align}
\end{proof}

\subsection{Proofs of section \ref{fast convergence rates}}

\begin{proof}[Proof of Lemma \ref{social welfare potential first order derivative lemma}]
Consider the decomposition
\begin{align}\begin{split}\nonumber
\gamma\(\lambda, g + t_{n}H_{n}, Q + t_{n}\mathcal{Q}_{n}\) - \gamma\(\lambda, Q, \mu\)
& = \underbrace{t_{n}\(\gamma\(\lambda, g + t_{n}H_{n}, \mathcal{Q}_{n}\) - \gamma\(\lambda, g, \mathcal{Q}_{n}\)\) }_{\(1\)} \\
& + \underbrace{\gamma\(\lambda, g + t_{n}H_{n}, Q\) - \gamma\(\lambda, g, Q\)}_{\(2\)} \\
& + \underbrace{\gamma\(\lambda, g, Q + t_{n}\mathcal{Q}_{n}\) - \gamma\(\lambda, g, Q\)}_{\(3\)}.
\end{split}\end{align}
Note that \eqref{social welfare potential derivative} is the limit of $\frac{\(3\)}{t_{n}}$ as $n \to \infty$, so we only need to show that $\lim_{n \to \infty}\frac{\(1\)}{t_{n}} = \lim_{n \to \infty}\frac{\(2\)}{t_{n}} = 0$.
For the $\(1\)$ term, we further decompose
\begin{align}\begin{split}\nonumber
\left\vert \gamma\(\lambda, g + t_{n}H_{n}, \mathcal{Q}_{n}\) - \gamma\(\lambda, g, \mathcal{Q}_{n}\)\right\vert
& \leq \left\vert \gamma\(\lambda, g + t_{n}H_{n}, \mathcal{Q}_{n}\) - \gamma\(\lambda, g + t_{n}H_{n}, \mathcal{Q}\)\right\vert \\
& + \left\vert \gamma\(\lambda, g, \mathcal{Q}_{n}\) - \gamma\(\lambda, g, \mathcal{Q}\)\right\vert
+ \left\vert \gamma\(\lambda, g + t_{n}H_{n}, \mathcal{Q}\) - \gamma\(\lambda, g, \mathcal{Q}\)\right\vert.
\end{split}\end{align}
The first two terms converge to $0$ since $\left\Vert \mathcal{Q}_{n} - \mathcal{Q}\right\Vert_{\ell^{\infty}\(\mathcal{F}\)} \to 0$.
The third term also converges to $0$ by the $L^{2}\(Q\)$ uniform continuity of $\mathcal{Q}$, since by Proposition \ref{surface integral continuity}, $g\(X\)$ has a continuous density.

By Theorem \ref{general k derivative} and the chain rule of Hadamard derivative, we have
\begin{align}\begin{split}\label{social welfare potential boundary derivative}
\lim_{n \to \infty}\frac{\(2\)}{t_{n}} = \sum_{j}\left\{\sum_{l \neq j}\int_{y > \tau_{\neg l}\(0\)}\[\int_{\Delta^{-1}_{j}\(c'\(y, 0, l\)\)}\frac{\(\lambda_{j}H_{j}\(x\) - \lambda_{l}H_{l}\(x\)\)f\(x\)\lambda_{j}g_{j}\(x\)}{J\Delta_{j}\(x\)}
d\mathcal{H}_{n - k}x\]d\mathcal{L}_{k - 1}y\right\},
\end{split}\end{align}
where $f$ is the density function of $X$ and $\Delta_{j} = \lambda_{j}g_{j} - \(\lambda_{0}g_{0}, \ldots, \lambda_{j - 1}g_{j - 1}, \lambda_{j + 1}g_{j + 1}, \ldots, \lambda_{J}g_{J}\)^{T}$.
Consider a Borel set $C \subset \mathbb{R}$, note that $J\Delta_{j}\(x\) > 0 \; a.e.$ on the support of $X$, there is
\begin{align}\begin{split}\label{another expression}
& \int_{c_{l} \in C}\[\int_{y > \tau_{\neg l}\(0\)}\[\int_{\Delta^{-1}_{j}\(c'\(y, c_{l}, l\)\)}\frac{\(\lambda_{j}H_{j}\(x\) - \lambda_{l}H_{l}\(x\)\)f\(x\)\lambda_{j}g_{j}\(x\)}{J\Delta_{j}\(x\)}
d\mathcal{H}_{n - k}x\]d\mathcal{L}_{k - 1}y\]dc_{l} \\
& = \int_{\tau_{\neg l}\(y\) > \tau_{\neg l}\(0\), \tau_{l}\(y\) \in C}\[\int_{\Delta^{-1}_{j}\(y\)}\frac{\(\lambda_{j}H_{j}\(x\) - \lambda_{l}H_{l}\(x\)\)f\(x\)\lambda_{j}g_{j}\(x\)}{J\Delta_{j}\(x\)}d\mathcal{H}_{n - k}x\]d\mathcal{L}_{k}y \\
& = \int\(\prod_{i \neq j, l}1\(\lambda_{j}g_{j}\(x\) > \lambda_{i}g_{i}\(x\)\)\)1\(\lambda_{j}g_{j}\(x\) - \lambda_{l}g_{l}\(x\) \in C\)\(\lambda_{j}H_{j}\(x\) - \lambda_{l}H_{l}\(x\)\)f\(x\)\lambda_{j}g_{j}\(x\)d\mathcal{L}_{n}x \\
& = \int_{y \in C}\[\int_{\Delta^{-1}_{j, l}\(y\)}\frac{\(\lambda_{j}H_{j}\(x\) - \lambda_{l}H_{l}\(x\)\)f\(x\)\lambda_{j}g_{j}\(x\)}{J\Delta_{j, l}\(x\)}\(\prod_{i \neq j, l}1\(\lambda_{j}g_{j}\(x\) > \lambda_{i}g_{i}\(x\)\)\)d\mathcal{H}_{n - 1}x\]d\mathcal{L}_{1}y,
\end{split}\end{align}
where $\Delta_{j, l} = \lambda_{j}g_{j} - \lambda_{l}g_{l}$.
Here we use the coarea formula two times.
By the arbitrariness of $C$, the right hand side of \eqref{another expression} and \eqref{derivative k>1} (without the summation) coincide.
So, switch the role of $j$ and $l$ in the summation of \eqref{social welfare potential boundary derivative}, we get $\lim_{n \to \infty}\frac{\(2\)}{t_{n}} = 0$ by the fact that $\lambda_{j}g_{j}\(x\) = \lambda_{l}g_{l}\(x\)$ on $\Delta^{-1}_{j, l}\(0\)$ (and $\Delta^{-1}_{l, j}\(0\)$).
\end{proof}

\begin{proof}[Proof of Proposition \ref{first MA second Hadamard result}]
Note that we have the following decomposition:
\begin{align}\begin{split}\nonumber
\gamma\(\lambda, \hat{g}, \mathbb{Q}_{n}\) - \gamma\(\lambda, g, \mathbb{Q}_{n}\)
& = \underbrace{\gamma\(\lambda, \hat{g}, \mathbb{Q}_{n} - Q\) - \gamma\(\lambda, g, \mathbb{Q}_{n} - Q\)}_{\(\rmnum{1}\)} \\
& + \underbrace{\gamma\(\lambda, \hat{g}, Q\) - \gamma\(\lambda, g, Q\)}_{\(\rmnum{2}\)}.
\end{split}\end{align}
Since Donskerness implies stochastic equicontinuity (see for example Theorem 1.5.7 and section 2.1.2 of \cite{van2023weak}), the $\(\rmnum{1}\)$ term is of $o_{\mathbb{P}^{*}}\(\frac{1}{\sqrt{n}}\)$.

For the $\(\rmnum{2}\)$ term, consider the second order Hadamard derivative of $\gamma\(\lambda, \cdot, Q\)$ at $g$, i.e.
\begin{align}\begin{split}\nonumber
\lim_{n \to \infty}\frac{\gamma\(\lambda, g + t_{n}H_{n}, Q\) - \gamma\(\lambda, g, Q\)}{t^{2}_{n}}
\end{split}\end{align}
since the first order Hadamard derivative is $0$.
Here we use the observation that Lemma \ref{social welfare potential first order derivative lemma} still holds when we adopt the allocation rule in \eqref{tie breaking definition to ensure at least at 1}.
Note that we have the following bound:
\begin{align}\begin{split}\nonumber
& \gamma\(\lambda, g, Q\) - \gamma\(\lambda, g', Q\) \\
& = \mathbb{E}\[\sum_{i\neq j} \(\lambda_i g_i\(X\) - \lambda_j g_j\(X\)\) \phi_i\(X; g\) \phi_j\(X; g'\)\] \\
&\leq \mathbb{E}\[\sum_{i\neq j} \left\vert\lambda_i g_i\(X\) - \lambda_j g_j\(X\) -
\(\lambda_i g'_i\(X\) - \lambda_j g'_j\(X\)\)
\right\vert \phi_i\(X; g\) \phi_j\(X; g'\)\] \\
& =  \mathbb{E}\[\sum_{i\neq j} \left\vert\lambda_i g_i\(X\) - \lambda_j g_j\(X\) -
\(\lambda_i g'_i\(X\) - \lambda_j g'_j\(X\)\)
\right\vert 1\(\phi_i\(X; g\) = \phi_j\(X; g'\) = 1\)\] \\
&\leq \mathbb{E}\[\sum_{i\neq j}
\left\vert\lambda_i g_i\(X\) - \lambda_j g_j\(X\) -
\(\lambda_i g'_i\(X\) - \lambda_j g'_j\(X\)\)
\right\vert
1\(\phi\(X; g\)\neq \phi\(X; g'\)\)\] \\
&\leq \mathbb{E}\[\sum_{i\neq j}
\left\vert\lambda_i g_i\(X\) - \lambda_j g_j\(X\) -
\(\lambda_i g'_i\(X\) - \lambda_j g'_j\(X\)\)\right\vert
\sum_{k \neq l} \left\vert
1_{kl}\(X; \lambda, g\) - 1_{kl}\(X; \lambda, g'\)
\right\vert\]
\end{split}\end{align}
For $g' = g + t_{n}H_{n}$, following \cite{chen2003estimation},
\begin{align}\begin{split}\nonumber
\left\vert 1_{kl}\(x; \lambda, g\) - 1_{kl}\(x; \lambda, g + t_{n}H_{n}\)\right\vert
\leq 1\(-2\delta_{n} \leq \lambda_{k}g_{k}\(x\) - \lambda_{l}g_{l}\(x\) \leq 2\delta_{n}\),
\end{split}\end{align}
where $\delta_{n} = \sup_{x \in K_{f}}\max _{0, \ldots, J}\vert t_{n}H_{nj}\(x\)\vert$, $K_{f}$ is the compact support of $X$.
Obviously, $\delta_{n} \leq M\vert t_{n}\vert$ for some constant $M > 0$.
So by the MA, we have
\begin{align}\begin{split}\nonumber
\gamma\(\lambda, g + t_{n}H_{n}, Q\) - \gamma\(\lambda, g, Q\) \leq C\vert t_{n}\vert^{1 + \alpha}
\end{split}\end{align}
for some constant $C > 0$.
Therefore, the second order Hadamard derivative we would like to calculate here is $0$.
Now, the delta method for second order Hadamard differentiable functional in Theorem 2.1 of \cite{chen2019inference} implies the $\(\rmnum{2}\)$ term is also of $o_{\mathbb{P}^{*}}\(\frac{1}{\sqrt{n}}\)$.
\end{proof}

\begin{proof}[Proof of Lemma \ref{marginally controlled lemma}]
Consider the distribution function $F$ of $X$, By MA at each point, $\mathbb{P}\left\{X = c\right\} = 0$, so $F$ is continuous.
We then have
\begin{align}\begin{split}\nonumber
F\(x + h\) - F\(x\) = \mathbb{P}\left\{x < X \leq x + h\right\} = \mathbb{P}\left\{x < X < x + h\right\} \leq C\(x\)h
\end{split}\end{align}
for all $x \in \[a, b\]$ and for all $h > 0$.
Therefore, for all $x \in \[a, b\]$, $D^{+}F\(x\)$ exists and is finite.
Now, we can invoke Theorem \ref{Dini derivative Lusin} to see that $F$ has Lusin property (N) on $\[a, b\]$.
Since $F$ is continuous, monotone (thus of bounded variation) and has Lusin property (N), by Theorem \ref{Lusin + BV = AC}, $F$ is absolutely continuous.
\end{proof}

\begin{proof}[Proof of Lemma \ref{density iff}]
Denote the support of $f$ as $K_{f}$.
If $Jh > 0 \; a.e.$ on $K_{f}$, by the coarea formula, we have with $B \in \mathcal{B}\(\mathbb{R}\)$,
\begin{align}\begin{split}\nonumber
\int_{h^{-1}\(B\)}f\(x\)d\mathcal{L}_{n}x
= \int_{B}\[\int_{h^{-1}\(y\)}\frac{f\(x\)}{Jh\(x\)}d\mathcal{H}_{n - 1}x\]dy,
\end{split}\end{align}
where $\mathcal{B}\(\mathbb{R}\)$ is the Borel $\sigma$-algebra on $\mathbb{R}$.
So
\begin{align}\begin{split}\nonumber
\int_{h^{-1}\(y\)}\frac{f\(x\)}{Jh\(x\)}d\mathcal{H}_{n - 1}x
\end{split}\end{align}
is the density of $h\(X\)$.
Conversely, assume there is a Borel set $B \subset K_{f}$ such that $\mathcal{L}_{n}\(B\) > 0$ and $Jh = 0$ on $B$.
On the one hand, by Morse-Sard theorem (see Theorem \ref{morse-sard theorem}), $\mathcal{L}_{1}\(h\(B\)\) \leq \mathcal{L}_{1}\(h\(Crit\(h\)\)\) = 0$ (see Corollary \ref{Sard type lemma} for $Crit\(\cdot\)$).
On the other hand, by definition
\begin{align}\begin{split}\nonumber
h_{\#}\(f\mathcal{L}_{n}\)\(h\(B\)\) = \int_{h^{-1}\(h\(B\)\)}f\(x\)d\mathcal{L}_{n}x > 0.
\end{split}\end{align}
Hence $h_{\#}\(f\mathcal{L}_{n}\)$ is not absolutely continuous with respect to $\mathcal{L}_{1}$.
\end{proof}

\begin{proof}[Proof of Theorem \ref{treatment effect allocation simple case}]
In the following, the notation $C$ should be understood as \textit{some constant}, and they are not necessarily equal to each other.
With notation $\mathbb{E}$ and $\mathbb{P}$, the expectation and probability are taken with respect to $\(X_{i}, Y_{i}, D_{i}\)$ ($o_{\mathbb{P}}\(1\)$ means convergence in probability with respect to the sample used to estimate $\hat{g}$ and $\hat{p}$).
The result will follow if we can show that $\Delta = o_\mathbb{P}\(1\)$, where
\bs
\Delta = \frac{1}{\sqrt{n'}}
\sum_{i=1}^{n'} \sum_{j=0}^J
\Biggl(\lambda_j
\hat \phi_j\(X_i\)
& \[
	\hat g_j\(X_i\) + \frac{D_{ij}}{\hat p_j\(X_i\)} \(Y_i - \hat g_j\(X_i\)\)
\] \\
& - \lambda_j
\phi_j^{*}\(X_i\)
\[
	g_j\(X_i\) + \frac{D_{ij}}{p_j\(X_i\)} \(Y_i - g_j\(X_i\)\)
\]\Biggl).
\end{split}\end{align}
For this purpose we decompose $\Delta = \Delta_1 + \Delta_2$, where
\bs
\Delta_1 =
\frac{1}{\sqrt{n'}}
\sum_{i=1}^{n'} \sum_{j=0}^J
\lambda_j
\hat \phi_j\(X_i\)
\[
	\hat g_j\(X_i\) + \frac{D_{ij}}{\hat p_j\(X_i\)} \(Y_i - \hat g_j\(X_i\)\)
	- g_j\(X_i\) - \frac{D_{ij}}{p_j\(X_i\)} \(Y_i - g_j\(X_i\)\)
\]
\end{split}\end{align}
and
\bs
\Delta_2 =
\frac{1}{\sqrt{n'}}
\sum_{i=1}^{n'} \sum_{j=0}^J
\lambda_j
\(
\hat \phi_j\(X_i\) - \phi_j^{*}\(X_i\)\)
\[
	g_j\(X_i\) + \frac{D_{ij}}{p_j\(X_i\)} \(Y_i - g_j\(X_i\)\)
\].
\end{split}\end{align}
Also define a linearized approximation of $\Delta_1$ as
\bs
\Delta_3 =
\frac{1}{\sqrt{n'}}
\sum_{i=1}^{n'}	 \sum_{j=0}^J
\lambda_j
\hat \phi_j\(X_i\)
\[
	\(1- \frac{D_{ij}}{p_j\(X_i\)}\)\(\hat g_j\(X_i\)-  g_j\(X_i\)\)
	- \frac{D_{ij}}{p_j^2\(X_i\)} \(Y_i - g_j\(X_i\)\)
	\(\hat p_j\(X_i\) - p_j\(X_i\)\) \].
\end{split}\end{align}
Then we can write, when $p\(x\)$ and $\hat{p}\(x\)$ is bounded away from zero,
\bs
\left\vert \Delta_1 - \Delta_3 \right\vert \leq C
\frac{1}{\sqrt{n'}}
\sum_{i=1}^{n'} \sum_{j=0}^J
\(\left\vert \hat g_j\(X_i\) - g_j\(X_i\)\right\vert^2
+ \left\vert \hat p_j\(X_i\) - p_j\(X
_i\)\right\vert^2\).
\end{split}\end{align}
By the essential supremum convergence rate assumption, $\vert \Delta_1 - \Delta_3 \vert = o_{\mathbb{P}}\(1\)$.
To see that $\Delta_3 = o_{\mathbb{P}}\(1\)$, note that by the split sample scheme, conditional on the
estimate $\hat p\(\cdot\)$ and $\hat g\(\cdot\)$,
$\mathbb{E} \Delta_3 = 0$,
\bs
Var\(\Delta_3\)
& \leq C\sum_{j=0}^J \lambda_j^2
\Biggl(\mathbb{E}\[
	\frac{1-p_j\(X_i\)}{p_j\(X_i\)} \(\hat g\(X_i\) - g\(X_i\)\)^2 \hat \phi_j\(X_i\)^2\] \\
& + \mathbb{E}\[\hat \phi_j\(X_i\)^2 Var\(Y_{ij} \vert X_i, D_{ij}=1\)
\frac{1}{p_j^3\(X_i\)} \(\hat p\(X_i\) - p\(X_i\)\)^2\]
\Biggr) \\
&\leq C\( \mathbb{E} \(\hat g\(X_i\) - g\(X_i\)\)^2
+ \mathbb{E} \(\hat p\(X_i\) - p\(X_i\)\)^2\).
\end{split}\end{align}
Next we decompose $\Delta_2 = \Delta_2^1 +  \Delta_2^2$, where
\bs
\Delta_2^1 & =
\frac{1}{\sqrt{n'}}
\sum_{i=1}^{n'} \sum_{j=0}^J\lambda_j
\(
\hat \phi_j\(X_i\) - \phi_j^{*}\(X_i\)\)
	\frac{D_{ij}}{p_j\(X_i\)} \(Y_i - g_j\(X_i\)\), \\
\Delta_2^2 & =
\frac{1}{\sqrt{n'}}
\sum_{i=1}^{n'} \sum_{j=0}^J\lambda_j
\(
\hat \phi_j\(X_i\) - \phi_j^{*}\(X_i\)\)  g_j\(X_i\).
\end{split}\end{align}
To see $\Delta_2^1 = o_\mathbb{P}\(1\)$, note that $\mathbb{E} \Delta_2^1 = 0$. It suffices to show
that $Var\(\Delta_2^1\)	 = o_\mathbb{P}\(1\)$ by bounding
\bs
Var\(\Delta_2^1\)
&\leq \sum_{j=0}^J  \lambda_j^2 \mathbb{E}
\[\(
\hat \phi_j\(X_i\) - \phi_j^{*}\(X_i\)\)^2 \frac{1}{p_j\(X_i\)}	 Var\(Y_j \vert X_i, D_{ij}=1\)\]\\
& \leq C
\sum_{j=0}^J \mathbb{E}\(
\hat \phi_j\(X_i\) - \phi_j^{*}\(X_i\)\)^2 \\
& \leq C \sum_{j=0}^J \sum_{p\neq q}  \mathbb{E} 1\(\phi^{*}\(X_i\)\neq \hat\phi\(X_i\)\)\\
&\leq C \sum_{j=0}^J \sum_{p\neq q} \sum_{k\neq l} \mathbb{E}\left\vert
1_{kl}^{*}\(X_i\) -\hat 1_{kl}\(X_i\)\right\vert\\
&\leq C
\sum_{k\neq l} \mathbb{P}\(-2 \delta_n \leq \lambda_k g_k\(X\)
- \lambda_l g_l\(X\)\leq 2 \delta_n\)
\end{split}\end{align}
where $\delta_{n} = C\max_{j = 0, \ldots, J}\esssup \left\vert \hat{g}_{j} - g_{j}\right\vert$.
By the MA, the right hand side of the last inequality is $O\(\delta_{n}\)$.

Since $\Delta^{2}_{2} < 0$, we have
\bs
-\Delta^{2}_{2} &=
\frac{1}{\sqrt{n'}}
\sum_{i=1}^{n'}	 \sum_{j=0}^J \lambda_j
\(
\phi_j^{*}\(X_i\) - \hat\phi_j\(X_i\)\)	 g_j\(X_i\)
\\
&
=
\frac{1}{\sqrt{n'}}
\sum_{i=1}^{n'}
\sum_{p\neq q} \(\lambda_p g_p\(X_i\) - \lambda_q g_q\(X_i\)\)
\phi_p^{*}\(X_i\) \hat\phi_q\(X_i\)\\
&\leq
\frac{1}{\sqrt{n'}}
\sum_{i=1}^{n'}
\sum_{p\neq q} \left\vert\lambda_p g_p\(X_i\) - \lambda_q g_q\(X_i\) -
\(\lambda_p \hat g_p\(X_i\) - \lambda_q \hat g_q\(X_i\)\)
\right\vert \phi_p^{*}\(X_i\) \hat\phi_q\(X_i\)\\
&\leq
\frac{1}{\sqrt{n'}}
\sum_{i=1}^{n'}
\sum_{p\neq q}\[
\left\vert\lambda_p g_p\(X_i\) - \lambda_q g_q\(X_i\) -
\(\lambda_p \hat g_p\(X_i\) - \lambda_q \hat g_q\(X_i\)\)
\right\vert
 1\(\phi^{*}\(X_i\)\neq \hat\phi\(X_i\)\)\] \\
&\leq
\frac{1}{\sqrt{n'}}
\sum_{i=1}^{n'}
\sum_{p\neq q}\[
\left\vert\lambda_p g_p\(X_i\) - \lambda_q g_q\(X_i\) -
\(\lambda_p \hat g_p\(X_i\) - \lambda_q \hat g_q\(X_i\)\)\right\vert
\sum_{k \neq l}\left\vert
1_{kl}^{*}\(X_i\) -\hat 1_{kl}\(X_i\)\right\vert\] \\
&\leq 2 \(J+1\)^2 \delta_n
\frac{1}{\sqrt{n'}}
\sum_{i=1}^{n'}
\sum_{k\neq l}
\vert
1_{kl}^{*}\(X_i\) -\hat 1_{kl}\(X_i\)
\vert.
\end{split}\end{align}
To show that $\Delta_2^2 =o_{\mathbb{P}}\(1\)$, it suffices to check that
\bs
- \mathbb{E} \Delta_2^2 &\leq \mathbb{E} C \delta_n \sqrt{n'} \sum_{k\neq l} \left\vert
1_{kl}^{*}\(X_i\) -\hat 1_{kl}\(X_i\)
\right\vert \\
&\leq C \sqrt{n'} \delta_n
\sum_{k\neq l} \mathbb{P}\(-2 \delta_n \leq \lambda_k g_k\(X\)
- \lambda_l g_l\(X\)\leq 2 \delta_n\)\\
&= C \sqrt{n'} \delta_n^2 = o_{\mathbb{P}}\(1\).
\end{split}\end{align}

If the conditions of Theorem \ref{general k derivative} hold, by Proposition \ref{surface integral continuity}, the vector random variable $g\(X_{i}\)$ will have a continuous and bounded density.
By coarea formula, for $\(\lambda_{k}, \lambda_{l}\) \neq \(0, 0\)$, we have
\begin{align}\begin{split}\label{density of summation}
f_{\lambda_{k}g_{k} - \lambda_{l}g_{l}}\(y\)
= \int_{\lambda_{k}g_{k} - \lambda_{l}g_{l} = y}\frac{f_{g}\(g\)}{\sqrt{\lambda^{2}_{k} + \lambda^{2}_{l}}}d\mathcal{H}_{J}g,
\end{split}\end{align}
where $f_{g}$ and $f_{\lambda_{k}g_{k} - \lambda_{l}g_{l}}$ are the density functions of $g\(X\)$ and $\lambda_{k}g_{k}\(X\) - \lambda_{l}g_{l}\(X\)$, respectively.
Since $\mathcal{D}$ is a compact subset of $\mathbb{R}^{J + 1}_{++}$ and $g$ is bounded, the right hand side of \eqref{density of summation} is uniformly bounded on $\mathcal{D}$.
It is direct to check that all probability bounds in the above are uniform on $\mathcal{D}$.
\end{proof}

\begin{proof}[Proof of Theorem \ref{classification simple case}]
Same as the proof of Theorem \ref{treatment effect allocation simple case}, the notation $C$
should be understood as \textit{some constant}, and they are not necessarily equal to each other.
The result will follow if we can show that $\Delta = o_{\mathbb{P}}\(1\)$, where
\bs
\Delta = \frac{1}{\sqrt{n'}} \sum_{i=1}^{n'}
\sum_{j=0}^J \lambda_j \(\hat \phi_j\(X_i\)  - \phi_j^*\(X_i\)\) Y_{ij} = \Delta_1 + \Delta_2
\end{split}\end{align}
with
\bs
\Delta_1 =
\frac{1}{\sqrt{n'}}
\sum_{i=1}^{n'} \sum_{j=0}^J\lambda_j
\(
\hat \phi_j\(X_i\) - \phi_j^{*}\(X_i\)\)  g_j\(X_i\),
\end{split}\end{align}
and
\bs
\Delta_2 = \frac{1}{\sqrt{n'}}
\sum_{i=1}^{n'} \sum_{j=0}^J\lambda_j
\(
\hat \phi_j\(X_i\) - \phi_j^{*}\(X_i\)\)  \(Y_{ij} - g_j\(X_i\)\).
\end{split}\end{align}
Since $\mathbb E \Delta_2 = 0$, it suffices to show $\Delta_2 = o_\mathbb{P}\(1\)$ by bounding,
\bs
Var\(\Delta_2\) & \leq
C\sum_{j=0}^J \mathbb E \(\hat \phi\(X_i\) - \phi_j^*\(X_i\)\)^2 Var\(Y_j \vert X_i\) \\
& \leq C \sum_{j=0}^J \mathbb E\(\hat \phi_j\(X_i\) - \phi_j^*\(X_i\)\)^2 \\
& \leq	C \sum_{j=0}^J \sum_{p\neq q}  \mathbb{E} 1\(\phi^{*}\(X_i\)\neq \hat\phi\(X_i\)\)\\
& \leq C \sum_{j=0}^J \sum_{p\neq q} \sum_{k\neq l} \mathbb{E}\left\vert
1_{kl}^{*}\(X_i\) -\hat 1_{kl}\(X_i\)\right\vert\\
& \leq C
\sum_{k\neq l} \mathbb{P}\(-2 \delta_n \leq \lambda_k g_k^{*}\(X\)
- \lambda_l g_l^{*}\(X\)\leq 2 \delta_n\)
\end{split}\end{align}
where $\delta_{n} = C\max_{j = 0, \ldots, J}\esssup \left\vert \hat{g}_{j} - g_{j}\right\vert$.

We can also bound $\vert \Delta_1\vert$ by 
\bs
-\Delta_1 
&\leq
\frac{1}{\sqrt{n'}}
\sum_{i=1}^{n'}
\sum_{p\neq q} \left\vert\lambda_p g_p\(X_i\) - \lambda_q g_q\(X_i\) -
\(\lambda_p \hat g_p\(X_i\) - \lambda_q \hat g_q\(X_i\)\)
\right\vert \phi_p^{*}\(X_i\) \hat\phi_q\(X_i\)\\
&\leq
\frac{1}{\sqrt{n'}}
\sum_{i=1}^{n'}
\sum_{p\neq q}\[
\left\vert\lambda_p g_p\(X_i\) - \lambda_q g_q\(X_i\) -
\(\lambda_p \hat g_p\(X_i\) - \lambda_q \hat g_q\(X_i\)\)
\right\vert
 1\(\phi^{*}\(X_i\)\neq \hat\phi\(X_i\)\)\] \\
&\leq
\frac{1}{\sqrt{n'}}
\sum_{i=1}^{n'}
\sum_{p\neq q}\[
\left\vert\lambda_p g_p\(X_i\) - \lambda_q g_q\(X_i\) -
\(\lambda_p \hat g_p\(X_i\) - \lambda_q \hat g_q\(X_i\)\)\right\vert
\sum_{k \neq l}\left\vert
1_{kl}^{*}\(X_i\) -\hat 1_{kl}\(X_i\)\right\vert\] \\
&\leq 2 \(J+1\)^2 \delta_n
\frac{1}{\sqrt{n'}}
\sum_{i=1}^{n'}
\sum_{k\neq l}
\vert
1_{kl}^{*}\(X_i\) -\hat 1_{kl}\(X_i\)
\vert.
\end{split}\end{align}
Given that $-\Delta_{1} \geq 0$,
to show that $\Delta_{1} = o_{\mathbb{P}}\(1\)$, it suffices to check that
\bs
- \mathbb{E} \Delta_{1} &\leq \mathbb{E} C \delta_n \sqrt{n'} \sum_{k \neq l} \left\vert
1_{kl}^{*}\(X_i\) -\hat 1_{kl}\(X_i\)
\right\vert \\
& \leq C \sqrt{n'} \delta_n
\sum_{k\neq l} \mathbb{P}\(-2 \delta_n \leq \lambda_k g_k\(X\)
- \lambda_l g_l\(X\)\leq 2 \delta_n\) \\
& = C \sqrt{n'} \delta_n^2 = o_{\mathbb{P}}\(1\).
\end{split}\end{align}
\end{proof}

\begin{proof}[Proof of Corollary \ref{multiplier bootstrap}]
Without loss of generality, we consider $W_{n'i} \sim Exponential(1)$.
First note that for
\begin{align}\begin{split}\nonumber
\Delta = \frac{1}{\sqrt{n'}}
\sum_{i=1}^{n'} \(W_{n'i} - 1\)\sum_{j=0}^J
\Biggl(\lambda_j
\hat \phi_j\(X_i\)
& \[
	\hat g_j\(X_i\) + \frac{D_{ij}}{\hat p_j\(X_i\)} \(Y_i - \hat g_j\(X_i\)\)
\] \\
& - \lambda_j
\phi_j^{*}\(X_i\)
\[
	g_j\(X_i\) + \frac{D_{ij}}{p_j\(X_i\)} \(Y_i - g_j\(X_i\)\)
\]\Biggl),
\end{split}\end{align}
and
\begin{align}\begin{split}\nonumber
\Delta = \frac{1}{\sqrt{n'}} \sum_{i=1}^{n'}
\(W_{n'i} - 1\)\sum_{j=0}^J \lambda_j \(\hat \phi_j\(X_i\)  - \phi_j^*\(X_i\)\) Y_{ij},
\end{split}\end{align}
we can use the same bounding strategy as in the proof of Theorem \ref{treatment effect allocation simple case} and Theorem \ref{classification simple case} to get $\Delta = o_{\mathbb{P}}\(1\)$,
since $W_{n'i}$ are independent of $\(Y, X, D\)$ or $\(Y, X\)$ and $\mathbb{E}W_{n'i} = Var\(W_{n'i}\) = 1$.
By Theorem 3.7.13 in \cite{van2023weak} conditional weak convergence in probability of exchangeable bootstrap, we have
\begin{align}\begin{split}\nonumber
\frac{1}{\sqrt{n'}}
\sum_{i=1}^{n'} \(\frac{W_{n'i}}{\sum^{n'}_{i = 1}W_{n'i}} - 1\)\sum_{j=0}^J
\lambda_j
\phi_j^{*}\(X_i\)
\[
	g_j\(X_i\) + \frac{D_{ij}}{p_j\(X_i\)} \(Y_i - g_j\(X_i\)\)
\] & \weakconv \mathbb{G}, \\
\frac{1}{\sqrt{n'}} \sum_{i=1}^{n'}
\(\frac{W_{n'i}}{\sum^{n'}_{i = 1}W_{n'i}} - 1\)\sum_{j=0}^J \lambda_j \phi_j^*\(X_i\) Y_{ij}  & \weakconv \mathbb{G}.
\end{split}\end{align}
See also Example 3.7.9 of \cite{van2023weak}.
The difference
\begin{align}\begin{split}\nonumber
& \frac{1}{\sqrt{n'}}
\sum_{i=1}^{n'} \(\frac{W_{n'i}}{\sum^{n'}_{i = 1}W_{n'i}} - W_{n'i}\)\sum_{j=0}^J
\lambda_j
\phi_j^{*}\(X_i\)
\[
	g_j\(X_i\) + \frac{D_{ij}}{p_j\(X_i\)} \(Y_i - g_j\(X_i\)\)
\] \\
& = \(\frac{1}{\sum^{n'}_{i = 1}W_{n'i}} - 1\)\frac{1}{\sqrt{n'}}
\sum_{i=1}^{n'} W_{n'i}\sum_{j=0}^J
\lambda_j
\phi_j^{*}\(X_i\)
\[
	g_j\(X_i\) + \frac{D_{ij}}{p_j\(X_i\)} \(Y_i - g_j\(X_i\)\)
\] \\
& = o_{\mathbb{P}}\(1\)O_{\mathbb{P}}\(1\) = o_{\mathbb{P}}\(1\).
\end{split}\end{align}
A similar argument is true for the difference
\begin{align}\begin{split}\nonumber
\frac{1}{\sqrt{n'}}
\sum_{i=1}^{n'} \(\frac{W_{n'i}}{\sum^{n'}_{i = 1}W_{n'i}} - W_{n'i}\)\sum_{j=0}^J
\lambda_j \phi_j^*\(X_i\) Y_{ij}.
\end{split}\end{align}
By Fubini Theorem unconditional convergence implies conditional convergence in the $L^{1}$ sense and thus in probability.
Therefore, we get that all residual terms conditional convergence to $0$ in probability, i.e. $\sqrt{n'}\(\tilde{\beta} - \hat{\beta}\) \weakconv \mathbb{G}$.
\end{proof}

\subsection{Supplementary results for section \ref{constrained and roc}}\label{supplementary discussion on misspecified parametric models}

To describe the asymptotic
distribution in parametric misspecification case, assume
\begin{align}\begin{split}\label{influence function for theta}
\sqrt{n}\(\hat\theta - \theta^{\star}\) =
\frac{1}{\sqrt{n}} \sum_{i = 1}^n m_i +
o_\mathbb{P}\(1\), \quad\text{where}\quad
m_i = m\(\omega_i, x_i\).
\end{split}\end{align}
The asymptotic distribution of a misspecified parametric model is then implied by
\eqref{main asymptotic formula} to be
\bsnumber\label{main asymptotic formula under mis parametric specification}
& \sqrt{n}\(\beta\(\mathbb{Q}_{n}, \hat{\theta}, \alpha\) - \beta\(Q, \theta^{\star}, \alpha\)\) \rightsquigarrow
\mathbb{Q}\biggl[\(y_{1} - y_{0}\)1\(\Delta\(x; k\(Q, \theta^{\star}, \alpha\); \theta^{\star}\) > 0\) + y_{0} \\
& - 
\frac{f_{\beta}\(Q, \theta^{\star}, k\(Q, \theta^{\star}, \alpha\)\)}{f_{\alpha}\(Q, \theta^{\star}, k\(Q, \theta^{\star}, \alpha\)\)}
\(
\(z_{1} - z_{0}\)1\(\Delta\(x; k\(Q, \theta^{\star}, \alpha\) > 0\); \theta^{\star}\) + z_{0}\)
+ \mathcal G\(Q, \theta^{\star}, \alpha\) m
\biggr],
\end{split}\end{align}
where
\bs
&\mathcal G\(Q, \theta^{\star}, \alpha\) =
\int_{\Delta\(x; k\(Q, \theta^{\star}, \alpha\);\theta^{\star}\) = 0}\frac{\(g_{1}\(x\) - g_{0}\(x\)\)\(
\(
\frac{\partial}{\partial\theta} \(g_1\(x; \theta^{\star}\) - g_0\(x; \theta^{\star}\)\)
\)
\)\mu'\(x\)}{\Vert\nabla\Delta\(x; k\(Q, \theta^{\star}, \alpha\);\theta^{\star}\)\Vert}d\mathcal{H}_{n - 1}x \\
& 
- \int_{\Delta\(x; k\(Q, \theta^{\star}, \alpha\); \theta^{\star}\) = 0}\frac{
k\(Q, \theta^{\star}, \alpha\)
\(g_{1}\(x\) - g_{0}\(x\)\)
\(
\frac{\partial}{\partial\theta} \(c_1\(x; \theta^{\star}\) - c_0\(x; \theta^{\star}\)\)
\)\mu'\(x\)}{\Vert\nabla\Delta\(x; k\(Q, \theta^{\star}, \alpha\);\theta^{\star}\)\Vert}d\mathcal{H}_{n - 1}x \\
& - \frac{f_{\beta}\(Q, \theta^{\star}, k\(Q, \theta^{\star}, \alpha\)\)}{f_{\alpha}\(Q, \theta^{\star}, k\(Q, \theta^{\star}, \alpha\)\)}
  \Biggl[
\int_{\Delta\(x; k\(Q, \theta^{\star}, \alpha\);\theta^{\star}\) = 0}\frac{\(c_{1}\(x\) - c_{0}\(x\)\)\(
\frac{\partial}{\partial\theta} \(g_1\(x; \theta^{\star}\) - g_0\(x; \theta^{\star}\)\)
\)\mu'\(x\)}{\Vert\nabla\Delta\(x; k\(Q, \theta^{\star}, \alpha\);\theta^{\star}\)\Vert}d\mathcal{H}_{n - 1}x \\
& 
- \int_{\Delta\(x; k\(Q, \theta^{\star}, \alpha\);\theta^{\star}\) = 0}\frac{
k\(Q, \theta^{\star}, \alpha\)
\(c_{1}\(x\) - c_{0}\(x\)\)\(
\frac{\partial}{\partial\theta} \(c_1\(x; \theta^{\star}\) - c_0\(x; \theta^{\star}\)\)
\)\mu'\(x\)}{\Vert\nabla\Delta\(x; k\(Q, \theta^{\star}, \alpha\);\theta^{\star}\)\Vert}d\mathcal{H}_{n - 1}x
\Biggl].
\end{split}\end{align}
In the above, as in Theorem \ref{binary allocation asymptotic result},
\bs
f_{\beta}\(Q, \theta^{\star}, k\) &
= -\int_{\Delta\(x; k; \theta^{\star}\) = 0}
\frac{\(g_{1}\(x\) - g_{0}\(x\)\)\(c_{1}\(x;\theta^{\star}\) - c_{0}\(x;\theta^{\star}\)\)\mu'\(x\)}{\Vert\nabla\Delta\(x; k;\theta^{\star}\)\Vert}d\mathcal{H}_{n - 1}x, \\
f_{\alpha}\(Q, \theta^{\star}, k\) &
= -\int_{\Delta\(x; k;\theta^{\star}\) = 0}
\frac{\(c_{1}\(x\) - c_{0}\(x\)\)\(c_{1}\(x;\theta^{\star}\) - c^{\star}_{0}\(x; \theta^{\star}\)\)\mu'\(x\)}{\Vert\nabla\Delta\(x; k;\theta^{\star}\)\Vert}d\mathcal{H}_{n - 1}x.
\end{split}\end{align}
On the one hand, Theorem \ref{binary allocation asymptotic result} and an adaption of Corollary \ref{roc bootstrap} validates
the consistency
of bootstrap inference for misspecified parametric models.
On the other hand, under mild conditions, the influence function in
\eqref{main asymptotic formula under mis parametric specification} can be estimated using a kernel function $\kappa\(\cdot\) \geq 0$, $\int \kappa\(u\) du = 1$, and a bandwidth parameter $h$ such that $h\rightarrow 0$ and $n h / \log n \rightarrow \infty$ as $n\rightarrow \infty$. Define
\bs
\hat f_{\beta}\(\hat\mu, \hat\theta, k\)
&=
\frac1n \sum_{i=1}^n \frac{1}{h}
\kappa\(
\frac{
\Delta\(x_i, k, \hat\theta\) 
}{h}
\)
\(y_{i1} - y_{i0}\)\(
c_1\(x_i, \hat\theta\) -
c_0\(x_i, \hat\theta\)
\),\\
\hat f_{\alpha}\(\hat\mu, \hat\theta, k\) & =
\frac1n \sum_{i=1}^n \frac{1}{h}
\kappa\(
\frac{
\Delta\(x_i, k, \hat\theta\)
}{h}
\)
\(z_{i1} - z_{i0}\)\(
c_1\(x_i, \hat\theta\) -
c_0\(x_i, \hat\theta\)
\),
\end{split}\end{align}
and
\bs
\hat{\mathcal G}_1\(\hat\mu, \hat\theta, \alpha\) & =
\frac1n \sum_{i=1}^n \frac{1}{h}
\kappa\(
\frac{
\Delta\(x_i, k\(\hat\mu, \hat\theta,\alpha\), \hat\theta\)
}{h}
\)
\(y_{i1} - y_{i0}\)\\
& 
\[
\frac{\partial}{\partial\theta} \(g_1\(x; \hat\theta\) - g_0\(x; \hat\theta\)\)
-k\(\hat\mu, \hat\theta, \alpha\)
\frac{\partial}{\partial\theta} \(c_1\(x; \hat\theta\) - c_0\(x; \hat\theta\)\) \],\\
\hat{\mathcal G}_2\(\hat\mu, \hat\theta, \alpha\) & =
\frac1n \sum_{i=1}^n \frac{1}{h}
\kappa\(
\frac{
\Delta\(x_i, k\(\hat\mu, \hat\theta,\alpha\), \hat\theta\)
}{h}
\)
\(z_{i1} - z_{i0}\)\\
& 
\[
\frac{\partial}{\partial\theta} \(g_1\(x; \hat\theta\) - g_0\(x; \hat\theta\)\)
-k\(\hat\mu, \hat\theta, \alpha\)
\frac{\partial}{\partial\theta} \(c_1\(x; \hat\theta\) - c_0\(x; \hat\theta\)\) \],\\
\end{split}\end{align}
where
\bs
\Delta\(x_i, k\(\hat\mu, \hat\theta,\alpha\), \hat\theta\)
= g_1\(x; \hat\theta\) - g_0\(x; \hat\theta\)
- k\(\hat\mu, \hat\theta, \alpha\)
\(
c_1\(x; \hat\theta\) - c_0\(x; \hat\theta\)
\).
\end{split}\end{align}
Then the influence function in \eqref{main asymptotic formula under mis parametric specification} can be consistently estimated by
\bs
\psi\(y_{i1}, y_{i0}, z_{i1}, z_{i0}, x_i, \hat \mu, \hat\theta, \alpha\)
& =
\(y_{i1} - y_{i0}\)1\(\Delta\(x_i; k\(\hat\mu, \hat\theta, \alpha\); \hat\theta\) > 0\) + y_{i0} \\
& -
\frac{\hat f_{\beta}\(\hat\mu, \hat\theta, k\(\hat\mu, \hat\theta, \alpha\)\)}{\hat f_{\alpha}\(\hat\mu, \hat\theta, k\(\hat\mu, \hat\theta, \alpha\)\)}
\(
\(z_{i1} - z_{i0}\)1\(\Delta\(x; k\(\hat\mu, \hat\theta, \alpha\) > 0\);\hat\theta\) + z_{i0}\)\\
& + \[
\hat{\mathcal G}_1\(\hat\mu, \hat\theta, \alpha\)  -
\frac{\hat f_{\beta}\(\hat\mu, \hat\theta, k\(\hat\mu, \hat\theta, \alpha\)\)}{\hat f_{\alpha}\(\hat\mu, \hat\theta, k\(\hat\mu, \hat\theta, \alpha\)\)}
\hat{\mathcal G}_2\(\hat\mu, \hat\theta, \alpha\)
\]\hat{m}_i,
\end{split}\end{align}
where $\hat{m}_i$ is a uniformly consistent estimate of the influence function $m_i$ in
\eqref{influence function for theta}.

Still assume \eqref{influence function for theta}, it then follows from
Corollary \ref{roc asymptotic result}  that
\bs
& \sqrt{n}\(\beta\(\mathbb{Q}_{n}, \hat{\theta}, \alpha\) - \beta\(Q, \theta^{\star}, \alpha\)\) \rightsquigarrow
\frac{1}{Q y} \mathbb{Q}\biggl[y 1\(\Delta\(x; k\(Q, \theta^{\star}, \alpha\); \theta^{\star}\) > 0\)  \\
&\hspace{.01in}-
\frac{\(Q y\) f_{\beta}\(Q, \theta^{\star}, k\(Q, \theta^{\star}, \alpha\)\)}{\(1-Q y\) f_{\alpha}\(Q, \theta^{\star}, k\(Q, \theta^{\star}, \alpha\)\)}
\(
\(1-y\)1\(\Delta\(x; k\(Q, \theta^{\star}, \alpha\) ;\theta^{\star}\) > 0\)\)
+ \mathcal G\(Q, \theta^{\star}, \alpha\) m
\biggr],
\end{split}\end{align}
where
\bs
&\mathcal G\(Q, \theta^{\star}, \alpha\) =
\int_{p\(x; \theta^{\star}\)=k\(Q, \theta^{\star}, \alpha\)}\frac{
p\(x\)
\(
\frac{\partial}{\partial\theta} p\(x; \theta^{\star}\)
\)\mu'\(x\)}{
\Vert\nabla p\(x; \theta^{\star}\)\Vert}d\mathcal{H}_{n - 1}x \\
& - \frac{\(Q y\) f_{\beta}\(Q, \theta^{\star}, k\(Q, \theta^{\star}, \alpha\)\)}{\(1-Qy\) f_{\alpha}\(Q, \theta^{\star}, k\(Q, \theta^{\star}, \alpha\)\)}
 \int_{p\(x; \theta^{\star}\) = k\(Q, \theta^{\star}, \alpha\)}\frac{
\(1 - p\(x\)\)\(
\frac{\partial}{\partial\theta} p\(x; \theta^{\star}\)
\)\mu'\(x\)}{
\Vert\nabla p\(x; \theta^{\star}\)\Vert}d\mathcal{H}_{n - 1}x
\end{split}\end{align}
In the above, as in Corollary \ref{roc asymptotic result},
\bs
f_{\beta}\(Q, \theta^{\star}, k\) &
= - \frac{1}{Q y} \int_{p\(x;\theta^{\star}\) = k}
\frac{p\(x\) 
\mu'\(x\)}{\Vert\nabla p\(x; \theta^{\star}\)\Vert}d\mathcal{H}_{n - 1}x, \\
f_{\alpha}\(Q, \theta^{\star}, k\) &
= -\frac{1}{1-Qy}\int_{p\(x; \theta^{\star}\) = k}
\frac{\(1-p\(x\)\)
\mu'\(x\)}{\Vert\nabla p\(x; \theta^{\star}\)\Vert}d\mathcal{H}_{n - 1}x.
\end{split}\end{align}

Then the influence function 
can be consistently estimated by
\bs
\psi\(y_{i}, x_i, \hat \mu, \hat\theta, \alpha\)
& =
\frac{1}{\hat Q y} \biggl[ y_{i} 1\(p\(x_i; \hat\theta\) > k\(\hat\mu, \hat\theta, \alpha\)\)  \\
& -
\frac{\hat Qy \hat f_{\beta}\(\hat\mu, \hat\theta, k\(\hat\mu, \hat\theta, \alpha\)\)}{\(1-\hat Qy\) \hat f_{\alpha}\(\hat\mu, \hat\theta, k\(\hat\mu, \hat\theta, \alpha\)\)}
\(
\(1-y_i\)1\(p\(x; \hat\theta\) > k\(\hat\mu, \hat\theta, \alpha\)\) \)\\
& + \[
\hat{\mathcal G}_1\(\hat\mu, \hat\theta, \alpha\)  -
\frac{\hat Qy \hat f_{\beta}\(\hat\mu, \hat\theta, k\(\hat\mu, \hat\theta, \alpha\)\)}{ \(1-\hat Qy\) \hat f_{\alpha}\(\hat\mu, \hat\theta, k\(\hat\mu, \hat\theta, \alpha\)\)}
\hat{\mathcal G}_2\(\hat\mu, \hat\theta, \alpha\)
\]\hat{m}_i\biggr],
\end{split}\end{align}
where $\hat{m}_i$ is still a uniformly consistent estimate of the influence function $m_i$ in
\eqref{influence function for theta}.
In the above $\hat Q y = \frac1n \sum_{i=1}^n y_i$,
\bs
\hat f_{\beta}\(\hat\mu, \hat\theta, k\)
&=
- \frac{1}{\hat Q y} \frac1n \sum_{i=1}^n \frac{1}{h}
\kappa\(
\frac{
p\(x_i, \hat\theta\) - k
}{h}
\)
y_i, \\ 
\hat f_{\alpha}\(\hat\mu, \hat\theta, k\) & =
- \frac{1}{1-\hat Q y}
\frac1n \sum_{i=1}^n \frac{1}{h}
\kappa\(
\frac{
p\(x_i, \hat\theta\) -k
}{h}
\)
\(1 - y_i\), 
\end{split}\end{align}
and
\bs
\hat{\mathcal G}_1\(\hat\mu, \hat\theta, \alpha\) =&
\frac1n \sum_{i=1}^n \frac{1}{h}
\kappa\(
\frac{
p\(x_i, \hat\theta\) - k\(\hat\mu, \hat\theta,\alpha\)
}{h}
\)
y_{i}
\frac{\partial}{\partial\theta} p\(x; \hat\theta\)\\ 
\hat{\mathcal G}_2\(\hat\mu, \hat\theta, \alpha\) =&
\frac1n \sum_{i=1}^n \frac{1}{h}
\kappa\(
\frac{
p\(x_i, \hat\theta\) -	k\(\hat\mu, \hat\theta,\alpha\)
}{h}
\)
\(1 - y_{i}\)
\frac{\partial}{\partial\theta} p\(x; \hat\theta\).
\end{split}\end{align}
See for example \cite{donald2014estimation} for uniformly consistent kernel estimation.

\subsection{Sub (super) and Fr\'{e}chet differentiability}\label{frechet differentiability section}

In this section, we provide a mathematical justification of defining
the value function of the optimal treatment choice (policy / individual treatment rule) as a \textit{social welfare potential function}.
Using the seminal convexity results of \cite{bolte2021conservative},
we show that
this value function satisfies the precise definition of a potential function.
We further demonstrate that the convexity of the social welfare potential function leads to a surprising generic Fr\'{e}chet differentiability property.

\paragraph{Social welfare potential function}
Given the convexity 
in $\lambda$ given $g\(\cdot\)$,
a subgradient (see definition \ref{sub sup differentiability}) of $\gamma\(\lambda, g\)$ in $\lambda$ is given by
\bs
p\(\lambda\)
= \(
\mathbb{E} \phi^*_0\(X; \lambda, g\) g_0\(X\), \ldots, \mathbb{E} \phi^*_J\(X; \lambda, g\) g_J\(X\)
\)^T.
\end{split}\end{align}
where $\phi^{*}\(\cdot; \lambda, g\)$ is one of the optimal allocation under $\(\lambda, g\)$.
To verify by direct calcuation,
\bs
\gamma\(\lambda, g\) & + \left\langle p\(\lambda\), \lambda' - \lambda \right\rangle \\
&= \sum_j \mathbb{E} \lambda_j \phi_j^*\(X; \lambda, g\) g_j\(X\)
+ \sum_j \mathbb{E} \(\lambda_j'-\lambda_j\) \phi_j^*\(X; \lambda, g\) g_j\(X\) \\
&= \sum_j \mathbb{E} \lambda_j' \phi_j^*\(X; \lambda, g\) g_j\(X\)
\leq
\sum_j \mathbb{E} \lambda_j' \phi_j^*\(X; \lambda', g\) g_j\(X\)
= \gamma\(\lambda', g\).
\end{split}\end{align}
The rest of this subsection rigorously justifies calling $\gamma\(\lambda, g\)$ a social welfare {\it potential} function, with $p\(\lambda\)$ belonging to the vector field generated by
$\gamma\(\lambda, g\)$.
Recall that in a set-value map $D: \mathbb{R}^{J+1}\rightrightarrows \mathbb{R}^{J+1}$, the image of a point under the map is a set.
A function is then understood as a singleton valued set-valued map. 
The following definition is needed:

\begin{definition}{(Conservative set-valued fields, \cite{bolte2021conservative})}\label{definition 1}
	Let $D: \mathbb{R}^p \rightrightarrows \mathbb{R}^p$ be a set-valued map. $D$ is a conservative (set-valued)
	field whenever it has closed graph, nonempty compact values, and for any
	absolutely continuous loop $\ell: \[0,1\] \to \mathbb{R}^p$, such that $\ell\(0\)=\ell\(1\)$, the Aumann integral of $t \rightrightarrows \left\langle\dot{\ell}\(t\), D\(\ell\(t\)\)\right\rangle$ is $\{0\}$, namely,
	\bsnumber\label{aumann integral}
	\int^{1}_{0} & \left\langle \dot{\ell}\(t\), D\(\ell\(t\)\)\right\rangle dt \\
    & \coloneqq
    \left\{\int_{0}^{1} \omega(t)dt: \omega\(t\): [0, 1] \to \mathbb{R}
	\text{ is a measurable selection of}
	\left\langle\dot{\ell}\(t\), D\(\ell\(t\)\)\right\rangle\right\} = \{0\}.
\end{split}\end{align}
\end{definition}
It is shown in \cite{bolte2021conservative} that \eqref{aumann integral} is equivalent to
requiring
\bs
\int_0^1 \max_{v \in D\(\ell\(t\)\)} \left\langle\dot{\ell}\(t\), v\right\rangle dt = 0,
\end{split}\end{align}
where the integral is understood in the Lebesgue sense. This is possible by Theorem 18.19 and Theorem 18.20 in \cite{guide2006infinite}. See also Lemma 1 in
\cite{bolte2021conservative}.
Based on Definition \ref{definition 1}, we introduce the following Definition
\ref{definition 2}:
\begin{definition}{(Potential functions of conservative fields, \cite{bolte2021conservative})}\label{definition 2}
	A function $f: \mathbb{R}^p \to \mathbb{R}$ is called a potential function of a
	conservative field whenever there exists a conservative field $D: \mathbb{R}^p \rightrightarrows  \mathbb{R}^p$
    such that
    \bs
    f\(x\) = f\(0\) + \int^{1}_{0} \left\langle \dot{\mathtt{p}}\(t\), D\(\mathtt{p}\(t\)\)\right\rangle dt
    \end{split}\end{align}
    for all absolutely continuous path $\mathtt{p}: \[0, 1\] \to \mathbb{R}^{p}$ with $\mathtt{p}\(0\) = 0$ and $\mathtt{p}\(1\) = x$.
	\end{definition}
\cite{bolte2021conservative} show that a potential function must be locally Lipschitz and their Theorem 1 states that if 
$f$ is a potential function of the conservative field $D$ then $D$ coincides
with the gradient of $f$ on Lebesgue almost every point where the gradient exists. Therefore,
we expect that certain 
nonsmooth functions will generate conservative fields
by some generalized approach of taking derivatives. In addition, a real-valued convex (or concave) function is locally Lipschitz and thus Lebesgue almost everywhere differentiable by the
Rademacher's theorem.
See for example 
Theorem \ref{Rademacher theorem} in the Technical addendum \ref{technical addendum}.

\begin{theorem}\label{theorem 1}
{\cite{bolte2021conservative} Corollary 2 and Proposition 2} \\
    Let $f:\mathbb{R}^p \to \mathbb{R}$ be a convex (or concave) function,
    then $\nabla^- f\(\nabla^+ f\) = \partial f$ is a conservative field that admits $f$ as its potential,
    where $\partial f$ is the Clarke subgradient defined in Definition \ref{definition 3}.
    Consequently,  $\gamma\(\lambda, g\)$ is a potential function in the sense of Definition \ref{definition 2}.
\end{theorem}

\begin{proof}[Proof of Theorem \ref{theorem 1}]
	By Proposition 1.2 in \cite{clarke1975generalized}, for convex (concave) functions $f$,
	$\nabla^- f \(\nabla^+ f\) = \partial f$.   Furthermore, $\nabla^- f\(x\) \(\nabla^+ f\(x\)\) = \partial_g f\(x\)$ by Proposition 8.12 in \cite{rockafellar2009variational},
where $\partial_g f\(x\)$ is the general subgradient as in definition 8.3 in
	\cite{rockafellar2009variational} (see also Theorem 9.61 therein).
	Theorem 10.49 in
	\cite{rockafellar2009variational} confirms that $f$ has chain rule for the Clark subgradient.
Finally, by Corollary 2 in \cite{bolte2021conservative}, $\partial f$ is a conservative field that admits $f$.
\end{proof}

\paragraph{Envelope like theorem for social welfare potential function}
\label{Envelope like theorem for social welfare potential function}

\begin{definition}{(Fr\'{e}chet differentiability)}
Let $\mathcal{X}$ and $\mathcal{Y}$ be normed spaces equipped with norm $\Vert\cdot\Vert_{\mathcal{X}}$ and $\Vert\cdot\Vert_{\mathcal{Y}}$, $E \subset \mathcal{X}$ be an open set. Consider the map $\xi: E \to Y$. Then $\xi$ is called Fr\'{e}chet differentiable at $\theta \in E$, if there is a continuous linear map $\xi'_{\theta}: \mathcal{X} \to \mathcal{Y}$ such that:
\bs
\left\Vert \xi\(x\) - \xi\(\theta\) - \xi'_{\theta}\(x - \theta\)
\right\Vert_{\mathcal{Y}} = o\(\Vert x - \theta\Vert_{\mathcal{X}}\).
\end{split}\end{align}
\end{definition}

Before we describe the next result we recall several topological and measure theoretic concepts. In a topological space $\mathcal X$, if a subset $S \in \mathcal X$ satisfies
$\overset{\circ}{\bar S} =\emptyset$, i.e. the closure of $S$ has empty interior, then $S$ is called
a nowhere dense subset of $\mathcal X$. If $\mathcal Z \subset \mathcal X$ is a countable union of
nowhere dense subsets of $\mathcal X$, then $\mathcal Z$ is called a meager subset of 
$\mathcal X$, or
of the first category in $\mathcal X$. A subset $A \subset \mathcal X$ is called residually
many (or just residual) in $\mathcal X$ if $\mathcal X \setminus A$ is meager in $\mathcal X$.
A set $B \subset \mathcal X$ is called a $G_\delta$ set of $\mathcal X$ if $B$ is a
countable intersection of open sets; it is call a $F_\sigma$ set if it is a countable union
of closed sets.

\cite{lindenstrauss2003frechet} originate the concept of $\Gamma$-null sets to
describe the {\it magnitude} of sets. 
Let $T = \[0,1\]^{\mathbb{N}}$, where
$\mathbb{N}$ is the set of natural numbers, be endowed with the product topology and product
Lebesgue measure $\mathcal{L}_{\infty}$.
Let $\mathcal X$ be a Banach space and
$\Gamma\(\mathcal X\)=\{\gamma: T \to \mathcal X\}$
be the space of continuous mappings having continuous partial derivatives $D_j\gamma$.	  Equip
$\Gamma\(\mathcal X\)$ with a topology generated by
$\vert\vert \gamma \vert\vert_0 = \sup_{t\in T} \vert\vert \gamma\(t\)\vert\vert$, and all
$\vert\vert \gamma \vert\vert_k = \sup_{t\in T} \vert\vert D_k \gamma\(t\)\vert\vert$,
$k \geq 1$. Equivalently, the same topology is generated by
all $\vert\vert\gamma\vert\vert_{\leq k} = \max_{0 \leq i \leq k} \vert\vert \gamma\vert\vert_i, k \in \mathbb{N}$.

\begin{definition}{($\Gamma$-null)}
	A Borel set $N \subset \mathcal X$ is called $\Gamma$-null if $\mathcal{L}_{\infty}\{t \in T: \gamma\(t\)
	\in N\} = 0$ for residually many $\gamma \in \Gamma\(\mathcal X\)$. If a set $A$ is
	contained in such a $N$, then $A$ is also called $\Gamma$-null.
\end{definition}

Now, we can state our result. In the following theorem statement, \ref{envelope theorem 1} and
\ref{envelope theorem 2} follow from the convexity of $\gamma\(\lambda, g\)$ in $\lambda$ given each $g$,
and the convexity of $\gamma\(\lambda, g\)$ in $g$ given each $\lambda$, while
\ref{envelope theorem 3}  accounts for the non joint-convexity
of $\gamma\(\lambda, g\)$ in $\(\lambda, g\)$.

\begin{theorem}
\label{theorem envelope Theorem for social welfare potential function}
	Let $\gamma\(\lambda, g\)$ be a real valued social welfare potential function, where $\lambda\in \mathbb{R}^{J+1}$, $g \in \mathcal X$ where $\mathcal X$ is a Banach space. Assume that $\gamma\(\lambda, g\)$ is continuous at a point $\(\bar{\lambda}, \bar{g}\) \in \mathbb{R}^{J+1} \otimes \mathcal X$,	then we have:
	\begin{enumerate}
		\item\label{envelope theorem 1} If every separable subspace $\mathcal Y$ of $\mathcal X$ has a separable 
			dual space $\mathcal Y^*$,
	then		for any given $\lambda \in
\mathbb{R}^{J+1}$, $\gamma\(\lambda, g\)$ is Fr\'{e}chet
			differentiable in $g$ on a dense
			$G_\delta$ subset of  $\mathcal X$.
		\item \label{envelope theorem 2} If the 
			dual space $\mathcal X^*$ of $\mathcal X$ is separable, then
for any given $\lambda \in   \mathbb{R}^{J+1}$,
			$\gamma\(\lambda, g\)$ is $\Gamma$-almost everywhere Fr\'{e}chet differentiable
			with respect to $g \in \mathcal X$. In other words,
for any given $\lambda \in   \mathbb{R}^{J+1}$,
$\gamma\(\lambda, g\)$ is not Fr\'{e}chet differentiable in $g$ on a $\Gamma$-null set of $\mathcal X$.
\item\label{envelope theorem 3} In addition, If every separable subspace $\mathcal Y$ of $\mathcal X$ has a separable 
			dual space $\mathcal Y^*$, and $\mathcal{X}$ can be continuously
	embedded into $L^1\(\Omega, \mu\)$ (in the sense that
for all $x \in \mathcal X$,
			$\Vert x\Vert_{L^1\(\Omega, \mu\)} \leq C \Vert x \Vert_{\mathcal X}$ for a constant $C$), 
   then $\gamma\(\lambda, g\)$ is locally jointly Lipschitz in $\(\lambda, g\)$, and is also jointly Fr\'{e}chet differentiable with respect to $\(\lambda, g\(\cdot\)\)$ on a dense subset of $\mathbb{R}^{J+1} \otimes \mathcal X$.
	\end{enumerate}
When $\gamma\(\lambda, g\)$ is Fr\'{e}chet differentiable (totally or partially as in the theorem),
the Fr\'{e}chet derivative with respect to $\(\lambda, g\)$ can be calculated to be
\bsnumber\label{Frechet derivative}
\(\(\mathbb{E}\phi^{*}_{0}g_{0}, \ldots, \mathbb{E}\phi^{*}_{J}g_{J}\),
\(\mathbb{E}\lambda_{0}\phi^{*}_{0}, \ldots, \mathbb{E}\lambda_{J}\phi^{*}_{J}\)\)^{T}
\end{split}\end{align}
where 
\bs
\phi^{*}\(\cdot\) \coloneqq \phi^{*}\(\cdot; \lambda, g\) 
= \mathop{\arg\max}\limits_{\phi\(\cdot\) \in \Phi} \mathbb{E} \[\sum_j \lambda_j \phi_j\(X\) g_j\(X\)\]
\end{split}\end{align}
in the $\mu$ almost surely sense, where $\mu$ is the distribution of $X$.
\end{theorem}

\begin{proof}[Proof of Theorem \ref{theorem envelope Theorem for social welfare potential function}] \\
\textbf{step} 1 First, it is direct that $\gamma\(\lambda, g\)$ satisfies subadditivity and is also positive homogeneous of degree one with respect to both $\lambda$ and $g\(\cdot\)$, separately,
\bs
\gamma\(\lambda, g + g'\) \leq \gamma\(\lambda, g\) + \gamma\(\lambda, g'\), & \quad \gamma\(\lambda + \lambda', g\) \leq \gamma\(\lambda, g\) + \gamma\(\lambda', g\), \\
\gamma\(\alpha\lambda, g\) = \alpha\gamma\(\lambda, g\), & \quad
\gamma\(\lambda, \alpha g\) = \alpha\gamma\(\lambda, g\), \ \text{for all} \ \alpha \geq 0,
\end{split}\end{align}
implying that $\gamma\(\lambda, g\)$ is convex in $g\(\cdot\)$ given $\lambda$ and is convex in $\lambda$ given $g\(\cdot\)$.
Note that when $\mathcal{X}$ is infinite-dimensional, a real-valued convex function is not necessarily continuous (there is always a noncontinuous linear functional on $\mathcal{X}$).
Fortunately, a real-valued convex function on a Banach space is either continuous at evey point or discontinuous at every point of its domain (see for example Proposition 3.1.11 of \cite{niculescu2018convex}). 
So by assumption, $\gamma\(\lambda, g\)$ is a convex continuous function.
Now, for condition 1, we can directly evoke Theroem 2 in \cite{stegall1978duality}; for condition 2, we can directly evoke Corollary 3.11 in \cite{lindenstrauss2003frechet}. \\
    \textbf{step} 2 From now, without loss of generality, we may assume that $\mathbb{R}^{J + 1} \times \mathcal{X}$ is equipped with the norm $\(\left\Vert r\right\Vert^{2} + \left\Vert x\right\Vert^{2}_{\mathcal{X}}\)^{\frac{1}{2}}$. Let $\mathcal{X}$ be a Banach space such that every separable subspace $\mathcal{Y}$ has a separable dual space $\mathcal{Y}^{*}$, $\mathcal{Z} \subset \mathbb{R}^{J + 1} \times \mathcal{X}$ be a separable subspace of $\mathbb{R}^{J + 1} \times \mathcal{X}$.
    Then, $\mathcal{X}' = \left\{x: \exists\(r, x\) \in \mathcal{Z}\right\}$ is a subspace of $\mathcal{X}$, $R' = \left\{r: \exists\(r, x\) \in \mathcal{Z}\right\}$ is a subspace of $\mathbb{R}^{J + 1}$, and $\mathcal{Z} \subset R' \times \mathcal{X}'$. The separability of $\mathcal{Z}$ implies the separability of $\mathcal{X}'$.
    Note that the dual space of $R' \times \mathcal{X}'$ isometrically isomorphic to $\(R'\)^{*} \times \(\mathcal{X}'\)^{*}$. Obviously, $\mathbb{R}^{J + 1}$ is a separable Banach space with separable dual, so by the separability of $\(\mathcal{X}'\)^{*}$, $\(R' \times \mathcal{X}'\)^{*}$ is separable. Let $z^{*} \in \mathcal{Z}^{*}$, by Hahn-Banach theorem, there exists an element $\(r, x\)^{*} \in \(R', \mathcal{X}'\)^{*}$ such that $\(r, x\)^{*} \vert_{\mathcal{Z}} = z^{*}$ and $\left\Vert z^{*}\right\Vert = \left\Vert \(r, x\)^{*}\right\Vert$. Therefore, with some misuse of language, we can write that $\mathcal{Z}^{*} \subset \(R' \times \mathcal{X}'\)^{*}$. Therefore, the separability of $\(R' \times \mathcal{X}'\)^{*}$ implies the separability of $\mathcal{Z}^{*}$. In another word, if $\mathcal{Z}$ is a separable subspace of $\mathbb{R}^{J + 1} \times \mathcal{X}$, then $\mathcal{Z}^{*}$ is separable. \\
\textbf{step} 3
    We need to prove that $\gamma\(\lambda, g\)$ is locally Lipschitz continuous.
    Note that
    \bs
    \gamma\(\lambda, g\) = \sup_{\phi \in \Phi} \mathbb{E}\[\sum_{j}\lambda_{j}\phi_{j}g_{j}\] \leq
    \sup_{\phi \in \Phi} \mathbb{E}\[\sum_{j}\left\vert\lambda_{j}\right\vert \left\vert g_{j}\right\vert\phi_{j}\] \leq
    \(\sum_{j}\left\vert\lambda_{j}\right\vert\)\mathbb{E}\left\vert g_{i}\right\vert.
    \end{split}\end{align}
    By the continuous embedding assumption $\mathbb{E}\left\vert g_{j}\right\vert \leq C\left\Vert g_{j}\right\Vert_{\mathcal{X}}$ for some constant $C$.
    Consider an open neighborhood $B_{\lambda, g, r} \coloneqq B\(\lambda, r\) \times B\(g, r\)$ for arbitrarily chosen $\(\lambda, g\) \in \mathbb{R}^{J + 1} \times \mathcal{X}$,
    where open balls $B\(\lambda, r\)$ and $B\(g, r\)$ are taken in $\mathbb{R}^{J + 1}$ and $\mathcal{X}$ respectively.
    Obviously, there exists a constant $M$ such that $\forall\lambda' \in B\(\lambda, r\)$, $\Vert\lambda'\Vert \leq M$.
    Therefore, for all $\(\lambda', g'\) \in B_{\lambda, g, r}$, $v\(\lambda', g'\)$ is upper bounded.
    Now, we can show that $v\(\lambda', g'\)$ is also locally lower bounded near $\(\lambda, g\)$.
    Let $z = g + \rho\(g_{0} - g\)$ such that
    $g_{0} \in \mathcal{X}$, $z \in B\(g, r\)$ and $\rho > 1$.
    It is not hard to see such a point $z$ exists.
    For arbitrarily chosen $\lambda' \in B\(\lambda, r\)$, consider
    \bs
    \mathcal{V} = \left\{v: v = \(1 - \frac{1}{\rho}\)g + \frac{1}{\rho}z, g \in B\(g, r\)\right\}.
    \end{split}\end{align}
    $\mathcal{V}$ is equal to the ball in $\mathcal{X}$ with center $g_{1} = \(1 - \frac{1}{\rho}\)g + \frac{1}{\rho}z$ and radius $\(1 - \frac{1}{\rho}\)r$.
    For all $v \in \mathcal{V}$, $\(2g_{1} - v\) \in \mathcal{V}$,  so there is
    \bs
    \gamma\(\lambda', g_{0}\) \leq \frac{1}{2}\gamma\(\lambda', v\) + \frac{1}{2}\gamma\(\lambda', 2g_{1} - v\)
    \end{split}\end{align}
    which implies that
    \bs
    \gamma\(\lambda', v\) \geq 2\gamma\(\lambda', g_{1}\) - \gamma\(\lambda', 2g_{1} - v\) \geq 2\gamma\(\lambda', g_{1}\) - M.
    \end{split}\end{align}
    Since $\gamma\(\lambda', g_{0}\)$ is convex with respect to $\lambda'$ on $\mathbb{R}^{J + 1}$, it is locally Lipschitz with respect to $\lambda'$, 
    see for example Lemma 14.26 in \cite{villani2009optimal}.
    So we get that $\gamma\(\lambda', g\)$ is locally bounded on a
    neighborhood of $\(\lambda, g\)$. We may still denoted the bound as $M$. \\
    \textbf{step} 4 Now, we can prove that $\gamma\(\cdot, \cdot\)$ is jointly
    locally Lipschitz.
    Let $g_{1}, g_{2} \in B\(g, r'\)$, $g_{1} \neq g_{2}$ where $r'$ is taken small enough such that for all $\lambda' \in B\(\lambda, r\)$, $\left\vert\gamma\(\lambda', g'\)\right\vert \leq M$ for all $g \in B\(g, 2r'\)$.
    Next, let $w = g_{2} + \frac{r'}{d}\(g_{2} - g_{1}\)$ where $d = \left\Vert g_{1} - g_{2}\right\Vert_{\mathcal{X}}$.
    Obviously, $w \in B\(g, 2r'\)$ and $g_{2} = \frac{r'}{r' + d}g_{1} + \frac{d}{r' + d}w$.
    By the partial convexity of $\gamma\(\cdot, \cdot\)$ and let $\lambda' \in B\(\lambda, r\)$,
    \bs
    \gamma\(\lambda', g_{2}\) \leq \frac{r'}{r' + d}\gamma\(\lambda', g_{1}\) + \frac{d}{r' + d}\gamma\(\lambda', w\).
    \end{split}\end{align}
    Then, we have
    \bs
    \gamma\(\lambda', g_{2}\) - \gamma\(\lambda', g_{1}\) & \leq
    \frac{d}{r' + d}\(\gamma\(\lambda', w\) - \gamma\(\lambda', g_{1}\)\) \leq \frac{d}{r'}\(\gamma\(\lambda', w\) - \gamma\(\lambda', g_{1}\)\) \\
    & \leq \frac{2d}{r'}M = \frac{2M}{r'}\left\Vert g_{2} - g_{3}\right\Vert_{\mathcal{X}}.
    \end{split}\end{align}
    Therefore, there exists a constant $M_{1}$, a neighborhood $\mathcal{N}_{g}$ of $g$ and a neighborhood $\mathcal{N}_{\lambda}$ of $\lambda$ such that for all $\lambda' \in \mathcal{N}_{\lambda}$, $\gamma\(\lambda', \cdot\)$ is Lipschitz with the constant $M_{1}$ on $\mathcal{N}_{g}$.
    Similar, $\gamma\(\cdot, g\)$ is also locally Lipschitz in the above uniform sense with a constant $M_{2}$.
    Now, in a neighborhood of $\(\lambda, g\)$, we have
    \bs
    \left\vert \gamma\(\lambda', g'\) - \gamma\(\lambda'', g''\)\right\vert & \leq
    \left\vert \gamma\(\lambda', g'\) - \gamma\(\lambda', g''\)\right\vert +
    \left\vert \gamma\(\lambda', g''\) - \gamma\(\lambda'', g''\)\right\vert \\
    & \leq \(M_{1} + M_{2}\)\left\Vert \(\lambda', g'\) -
    \(\lambda'', g''\)\right\Vert_{\mathbb{R}^{J + 1} \times \mathcal{X}}.
    \end{split}\end{align}
    Now, we can directly evoke the Theorem 2.5 of \cite{preiss1990differentiability} to get that $\gamma\(\cdot, \cdot\)$ is Fr\'{e}chet differentiable at least on a dense subset of $\mathbb{R}^{J + 1} \times \mathcal{X}$. \\
\textbf{step} 5 In order to calculate the Fr\'{e}chet derivative of $\gamma\(\lambda, g\)$ at those Fr\'{e}chet differentiable points, we use the auxiliary function defined in \eqref{auxiliary function}.
    $\mathbb{E}\[\sum_{j}\lambda_{j}'\phi^{*}_{j}g_{j}'\]$ is continuous linear by Proposition 3.1.11 in \cite{niculescu2018convex} and thus Fr\'{e}chet differentiable with respect to $\(\lambda, g\)$ with the derivate in the right hand side of \eqref{Frechet derivative}.
    Obviously, $e\(\lambda', g'; \phi^{*}\)$ is partially concave and achieves its maximum at $\(\lambda, g\)$.
    Therefore, by Theorem 3.6.11 of \cite{niculescu2018convex}, we have $0 \in \nabla^{+}e\(\lambda', g'; \phi^{*}\)\vert_{\(\lambda', g'\) = \(\lambda, g\)}$.
    So, if $\gamma\(\lambda', g'\)$ is Fr\'{e}chet differentiable at $\(\lambda, g\)$, by Propsoition 3.6.9 in \cite{niculescu2018convex},
    \bs
    \nabla^{+}\(\lambda', g'; \phi^{*}\)\vert_{\(\lambda', g'\) = \(\lambda, g\)} = \left\{e'\(\lambda', g'; \phi^{*}\vert_{\(\lambda', g'\) = \(\lambda, g\)}\)\right\}
    = \{0\}.
    \end{split}\end{align}
Now, the proof ends with direct calculation.
\end{proof}

Theorem \ref{theorem envelope Theorem for social welfare potential function} essentially shows that,
under mild conditions, a generic point\footnote{ A generic property is one that holds
on a residual set.}
	can support
	taking the Fr\'{e}chet derivative of $\gamma\(\lambda, g\)$.
The conditions for
Fr\'{e}chet differentiability 
originate from the partial convexity of $\gamma\(\lambda, g\)$ in either $\lambda$ or $g$ and the joint local Lipschitz property
of $\gamma\(\lambda, g\)$ in $\(\lambda, g\)$.
In \cite{chernozhukov2018sorted},
Hadamard differentiability is verified by direct calculations with great details of the
partial effect quantile function (called sorted effect therein). We have demonstrated how to reformulate their approach
in section \ref{Weighted sorted effect}. Compared with \cite{chernozhukov2018sorted}, 
the assumptions in
Theorem \ref{theorem envelope Theorem for social welfare potential function}  
are not more restrictive in several senses.
First, Hadamard differentiability implies continuity (See for example \cite{averbuh1968razlivcnye}) which is assumed
in Theorem \ref{theorem envelope Theorem for social welfare potential function}. Second, it is
well known that for a measure space $\(\Omega, \mathcal F, \mu\)$ such that $\mathcal F$ is separable or countably generated,
and $\mu$ is $\sigma$-finite, $L^p\(\Omega\)$ is separable for $1 < p < \infty$, which implies that
$L^p\(\Omega\)$ is separable and reflexive when $\Omega$ is a Polish space, $\mathcal F$ is the Borel $\sigma$-algebra,
and $\mu$ is a probability measure.
Therefore $L^p\(\Omega\)$ for $1 < p < \infty$ satisfies the second result of Theorem \ref{theorem envelope Theorem for social welfare potential function}.
Furthermore, when $\Omega$ is an open subset of $\mathbb{R}^n$ and $\mathcal F$ is the
$\sigma$-algebra of Lebesgue measure, and $\mu = \mathcal{L}_{n}$ the $n$-dimensional Lebesgue measure,
both $L^p\(\Omega\)$ and the Sobolev space $W^{k,p}\(\Omega\)$ are also separable and reflexive. Hence
they also satisfy the second result of	Theorem \ref{theorem envelope Theorem for social welfare potential function}.
Third, we do not directly make assumptions about the underlying distribution $\mu$, such as requiring compact support or having a continuous density function.

\subsection{Parameters used in section \ref{empirical analysis}}\label{parameters}

The covariate vector $X$ in the synthetic data analysis in section \ref{empirical analysis} is drawn from a Gaussian distribution: 
\begin{align}\begin{split}\nonumber
& X_{i} \sim N\(\mu, \Sigma\), \mu = \(0, 1, 0, -1, 0\)^{T},
\begingroup
\renewcommand*{\arraystretch}{1}
\Sigma = \begin{varmatrix}[delim = p, size = \normalsize, sep = 0.5\arraycolsep]
	    1 & 0.3 & 0 & 0.1 & 0.5 \\
	    0.3 & 1 & 0.2 & 0 & -0.5 \\
	    0 & 0.2 & 1 & 0.4 & -0.1 \\
	    0.1 & 0 & 0.4 & 1 & 0.25 \\
	    0.5 & -0.5 & -0.1 & 0.25 & 1
	    \end{varmatrix}
\endgroup.
\end{split}\end{align}
The neural network weights $W_{1}$, $W_{2}$, $b_{1}$ and $b_{2}$ used to generate $Y$ are randomly set to be
\begin{align}\begin{split}\nonumber
& W_{1} =
\begingroup
\renewcommand*{\arraystretch}{1}
\begin{varmatrix}[delim = p, size = \normalsize, sep = 0.5\arraycolsep]
	    -0.8557 & -0.3984 & -0.2444 & 0.4402 & 0.5185 & -0.8261 \\
	    -0.4382 & 0.9245 & -0.1277 & -2.1882 & -1.7211 & -0.5367 \\
	    1.5279 & 1.0487 & 0.7238 & -1.8313 & 0.0273 & -0.6663 \\
	    1.7191 & 0.4149 & -1.7796 & -0.9590 & 1.1279 & -1.1755 \\
	    -2.2010 & -0.7008 & 1.0949 & -0.7687 & -1.3397 & -1.0222 \\
	    3.0063 & 1.0774 & -0.3373 & -1.0279 & 0.5106 & -0.8310 \\
	    -0.0059 & 1.1252 & 0.7112 & -1.3946 & -2.1588 & -0.3042 \\
	    -0.6861 & 0.4985 & -0.2876 & -0.9819 & -1.1552 & 0.1938
	\end{varmatrix}, \;
b_{1} = \begin{varmatrix}[delim = p, size = \normalsize, sep = 0.5\arraycolsep]
0.8990 \\
-0.1547 \\
-1.3834 \\
-0.8214 \\
1.4931 \\
-0.6971 \\
-0.4002 \\
-1.1119
\end{varmatrix}
\endgroup, \\
& W_{2} =
\begingroup
\renewcommand*{\arraystretch}{1}
\begin{varmatrix}[delim = p, size = \normalsize, sep = 0.5\arraycolsep]
	    -0.6777 & 1.6212 & 1.1582 & -1.5582 & 1.0625 & 0.0431 & -1.1843 & -0.0145 \\
	    2.5586 & -2.1053 & 0.6113 & 1.0777 & 1.1351 & -1.3304 & -0.0080 & 0.4248
	    \end{varmatrix}, \;
b_{2} = \begin{varmatrix}[delim = p, size = \normalsize, sep = 0.5\arraycolsep]
1.8659 \\
1.0406
\end{varmatrix}
\endgroup.
\end{split}\end{align}

\subsection{Supplementary discussion}

\paragraph{Origination of area and coarea formulas}
To the best of our knowledge, a result similar to Theorem
\ref{area formula classic} was first published in
\cite{federer1944surface}.
The coarea formula in the form of  \eqref{coarea formula integral} was first
published in \cite{federer1959curvature}, one of the most important papers by
Federer.
Results similar to Theorem
\ref{coarea formula classic} were presented earlier by
Aleksandr Semyonovich Kronrod, Herbert Federer, Laurence Chisholm Young (L.C.Young), Ennio De Giorgi and maybe even others.
The statements of the area and coarea formulas in Theorems
\ref{area formula classic} and \ref{coarea formula classic} are closer to
\cite{federer2014geometric} and the more readable \cite{evans2018measure}.
Roughly speaking, generalizations of both Theorem \ref{area formula classic} and Theorem \ref{coarea formula classic} can still work when the Euclidean spaces therein are replaced by certain kind of \textit{surface}, e.g. Riemannian manifolds and rectifiable sets of Euclidean spaces. 
See 3.2.20-3.2.22 and 3.2.46 in \cite{federer2014geometric} for classical developments.
In the Technical addendum \ref{technical addendum}, 
we provide readable proofs to both Theorem \ref{area formula classic} and Theorem \ref{coarea formula classic} together with a complete set of preliminary results needed in their derivation.

The area and coarea formulas in Theorems \ref{area formula classic}
and \ref{coarea formula classic} are closely related to another important
result called Morse-Sard Theorem, which essentially states that the set
of critical values is a null set. See for example Theorem
\ref{morse-sard theorem}. Corollary \ref{Sard type lemma}, which
follows directly from Theorem \ref{area formula classic},
is indeed Theorem \ref{morse-sard theorem} when $n=m$. In economics, one of the most well-known applications of the Sard type theorem is
\cite{debreu1970economies} which established that except for a null set of
\textit{economies}, every \textit{economy} has a finite set of equilibria.
For a relatively simple derivation of the Morse-Sard theorem,
see \cite{figalli2008simple}.
The basic ideas in
\cite{figalli2008simple} is to make use of the Morrey inequality, which is then combined with
Whitney extension theorem \ref{Whitney extension theorem} to prove the Morse-Sard theorem in a stronger
Sobolev sense, which implies the $C^r$ case.
\cite{evans2018measure}
shows in detail how to use the coarea formula
\eqref{coarea formula integral} to develop the Morrey inequality.

\paragraph{Concepts of smallness of sets} Under the conditions of Theorem \ref{theorem envelope Theorem for social welfare potential function}, generic Fr\'{e}chet  differentiability
follows from
convexity and local Lipschitzness. Consequently,
Fr\'{e}chet  differentiability
(which is stronger than Hadamard differentiability)
is a reasonable assumption.
The generic Fr\'{e}chet differentiability property here provides a conceptual justification of the functional differential methodology in \cite{chen2003estimation} and \cite{chernozhukov2015valid}.
Theorem \ref{theorem envelope Theorem for social welfare potential function} also differs from the envelope theorems in
\cite{milgrom2002envelope}. \cite{milgrom2002envelope} characterized possibly non-convex problems for finite dimensional
relevant parameters, while in this work,
Theorem \ref{theorem envelope Theorem for social welfare potential function}
focuses on a generic property for larger relevant parameter spaces.

By the Tikhonov theorem, $T = \[0,1\]^{\mathbb{N}}$ is a compact space. It can also be metrized by
\bs
\rho\(x,y\) = \sum_{n=1}^\infty \frac{
	\rho_n\(x_n, y_n\)
	}{
		2^n \(1+\rho_n\(x_n, y_n\)\)
	}
\end{split}\end{align}
where $\rho_n$ is some distance on $\[0,1\]$.

By the product $\sigma$-algebra on the space $\prod_i E_i$, we mean the $\sigma$-algebra generated by all the coordinate projections.  The product measure then should satisfy
\bs
\mu\[
	\prod_i B_i, B_ i\in \mathcal E_i, \ \text{and}\ B_i \neq E_i, \forall i \in \mathcal{I}, B_{i} = E_{i}, \forall i \in \mathbb{N} \backslash \mathcal{I}, \ \text{where} \ \mathcal{I} \subset \mathbb{N}, 0 < \vert I \vert < \infty
\] = \prod_{i \in \mathcal{I}} \mu_i\(B_i\),
\end{split}\end{align}
where $\left\{\(E_i, \mathcal E_i, \mu_i\)\right\}$ is a family of measure spaces.
The existence and uniqueness of the  countable product Lebesgue measure is guaranteed by the Kolmogorov extension theorem (see for example chapter 1 and 3 in \cite{dellacherie1979probabilities}).

\cite{namioka1975banach} attributed a space $\mathcal X$ such that
every separable subspace $\mathcal Y$ of it has	 a separable 
			dual space $\mathcal Y^*$ to
\cite{asplund1968frechet} as an Asplund space.
Condition 2 in Theorem \ref{theorem envelope Theorem for social welfare potential function}
requires that $\mathcal X^*$ is separable, which is more restrictive than the Asplund space in condition 1.
Theorem 2.4 in \cite{lindenstrauss2003frechet} shows that on $\mathbb{R}^n$, a $\Gamma$-null set is equivalent to a set
with Lebesgue zero. By the Rademacher Theorem again, given $g$, $\gamma\(\lambda,g\)$ is
$\Gamma$-almost everywhere differentiable with respect to $\lambda$.

For a Banach space with a separable dual space,
part 1 of Theorem
\ref{theorem envelope Theorem for social welfare potential function} shows that the set of
non Fr\'{e}chet-differentiable points is contained in a meager $F_\sigma$ set, while
part 2 of Theorem
\ref{theorem envelope Theorem for social welfare potential function}  shows that it is also
a $\Gamma$-null set. In other words,	the set of
non Fr\'{e}chet differentiable points is a $\Gamma$-null set, and is contained in a
meager $F_\sigma$ set. To the best of our knowledge, there is no general ordering
between $\Gamma$-null sets and
meager $F_\sigma$ sets.
There are interesting examples even in the simplest case $\mathbb{R}$. 
On the one hand, the Smith-Volterra-Cantor set (or called fat Cantor set) is closed and nowhere dense and thus
meager $F_{\sigma}$ but of positive Lebesgue measure. On the other hand, it is also possible to construct
a dense $G_\delta$ subset $B$ of $\mathbb{R}$
with zero Lebesgue measure that is not a meager $F_\sigma$ set (by the Baire theorem). To construct $B$, enumerate the rational numbers as
$\mathbb{Q}=\{q_1, q_2, \ldots\}$, and put
$P_n = \bigcup_{j=1}^\infty B\(q_j, 2^{-j} / n\)$ for positive integers $n$. Define
$B = \bigcap_{n=1}^\infty P_n$. Since $P_n$ is open for every positive integer $n$, we have $B$ is a $G_\delta$ set. The set $B$ has zero
Lebesgue measure because
\bs
\mu\(B\) \leq \mu\(P_n\) \leq \sum_{j=1}^\infty	 \mu\(B\(q_j, 2^{-j} / n\)\)
= \sum_{j=1}^\infty \frac{2^{-j+1}}{n} = \frac{2}{n} \to 0\ \text{as}\ n \to \infty.
\end{split}\end{align}

	It is worth noting that the concepts of meager $F_\sigma$ sets (the complement of dense $G_\delta$ sets) and $\Gamma$-null sets are among a collection of related
	definitions used to describe the smallness of sets.
Part 1 of Theorem \ref{theorem envelope Theorem for social welfare potential function} describes
a pure topological property of the set of $g\in \mathcal X$ at which $\gamma\(\lambda, g\)$ is
Fr\'{e}chet differentiable,   while part 2 of
Theorem \ref{theorem envelope Theorem for social welfare potential function} is both topological
and measure theoretic in nature.
Under the conditions of part 2 of Theorem \ref{theorem envelope Theorem for social welfare potential function},
the set of $g$ at which $\gamma\(\lambda, g\)$ is not
Fr\'{e}chet differentiable satisfies other smallness properties.
For example, on the one hand, by Theorem 1 in
	\cite{preiss1984frechet}, if $\mathcal X$ is a Banach space with
	separable dual space, then a continuous convex function $f$ on $\mathcal X$ is Fr\'{e}chet differentiable outside a $\sigma$-porous set.
 By Theorem 10.4.1 in \cite{lindenstrauss2012frechet},
	a $\sigma$-porous set is
	also a	$\Gamma_1$-null set.
On the other hand, by Proposition 5.4.3 in \cite{lindenstrauss2012frechet},
if $A\subset \mathcal X$ is a $\Gamma$-null and $F_\sigma$ set, then it is $\Gamma_n$ null for every $n$. To
recapitulate, $\sigma$-porous is a metric space concept, while $\Gamma_n$-null is a topological and
measure theoretic concept. These relations are illustrated in Figure
\ref{figure generic differentiability}.

\begin{definition}{(Porous and $\sigma$-porous)}
A set $A$ in a Banach space $\mathcal X$ is called $\sigma$-porous if it is
a countable union of porous sets. A porous set $E \subset \mathcal X$
is such that there exists $0 < c < 1$, and for all
$x\in \mathcal X$, and for all $\epsilon > 0$, there is a $y \in \mathcal X$  satisfying $0 < d\(x,y\) < \epsilon$ and $B\(y, c\cdot d\(x,y\)\) \cap E = \emptyset$.
\end{definition}

\begin{definition}{($\Gamma_{n}$-null)}
Consider $\Gamma_n\(\mathcal X\)
	= C^1\(\[0,1\]^n, \mathcal X\)$
 equipped with the $\Vert\cdot\Vert_{\leq n}$ norm. Then a Borel set $A \subset \mathcal X$ is called $\Gamma_n$-null if
\bs
	L^n\{ t\in \[0,1\]^n: \gamma\(t\) \in A\} = 0.
\end{split}\end{align}
for residually many $\gamma \in \Gamma_n\(\mathcal X\)$. If $A$ is not Borel, it is also called $\Gamma_n$-null if it is contained in a $\Gamma_n$-null Borel set.
\end{definition}

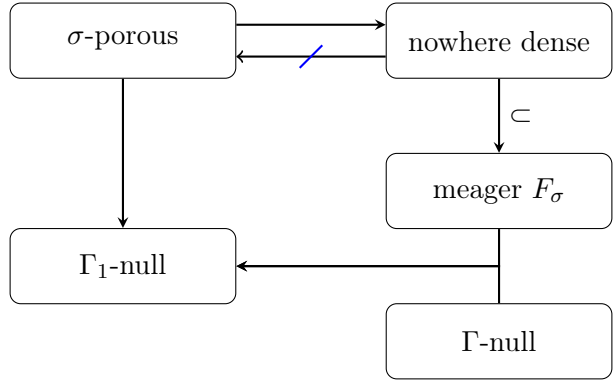
\begin{figure}[htbp]
\center
\begin{tikzpicture}[node distance=2cm]
\node (porous) [basicnode] {$\sigma$-porous};
\node (nowheredense) [basicnode, right of=porous,  xshift=3cm] {nowhere dense};
\node (gamma1null) [basicnode, below of=porous,	 yshift=-1cm] {$\Gamma_1$-null};
\draw[arrow] (porous.8) -- (nowheredense.172);
	\draw [->, thick, strike thru arrow] (nowheredense.188) -- (porous.-8);
\draw[arrow] (porous) -- (gamma1null);
\node (meagerfsigset) [basicnode, below of=nowheredense,  yshift=0cm] {meager $F_\sigma$};
\node (gammanull) [basicnode, below of=nowheredense,  yshift=-2cm] {$\Gamma$-null};
\draw[arrow] (nowheredense) -- node[anchor=west] {$\subset$}  (meagerfsigset);
	\draw[arrow] (meagerfsigset) |-	 (gamma1null);
	\draw[arrow] (gammanull) |-  (gamma1null);
\end{tikzpicture}
\caption{Relation between different concepts of small sets}
\label{figure generic differentiability}
\end{figure}

\newpage
\setstretch{1.25}
\renewcommand{\thesection}{T}
{\centering{\section{Technical addendum}\label{technical addendum}}}

\setcounter{equation}{0}
\counterwithin*{equation}{section}
\renewcommand\theequation{\thesection.\arabic{equation}}

\setcounter{figure}{0}
\counterwithin*{figure}{section}
\renewcommand\thefigure{\thesection.\arabic{figure}}

\subsection{Supplementary definitions and results}

\begin{definition}{(Outer measure)}\label{basic outer measure}
For a set $O$, an outer measure is a function
\begin{align}\begin{split}\nonumber
\mu^{*}: 2^{O} \to \[0, \infty\]\quad\text{such that}
\end{split}\end{align}
\begin{enumerate}[(a)]
  \item $\mu^{*}\(\emptyset\) = 0$.
  \item For arbitrary subsets $A, B$,
  $A \subset B \subset O$, $\mu^{*}\(A\) \leq \mu^{*}\(B\)$.
  \item For arbitrary subsets $B_{1}, B_{2}, \ldots$
  of $O$,
  \begin{align}\begin{split}\nonumber
  \mu^{*}\(\bigcup^{\infty}_{i = 1}B_{i}\) \leq
  \sum^{\infty}_{i = 1}\mu^{*}\(B_{i}\).
  \end{split}\end{align}
\end{enumerate}
If an outer measure $\mu^{*}$ is defined
on a metric space $\(O, d\)$ and satisfies
\begin{align}\begin{split}\nonumber
d\(A, B\) > 0 \Rightarrow
\mu^{*}\(A \cup B\) =
\mu^{*}\(A\) + \mu^{*}\(B\), \quad
\forall A, B \subset O.
\end{split}\end{align}
Then $\mu^{*}$ is called a
metric outer measure.
\end{definition}

\begin{theorem}
\label{Caratheodory's extension theorem}
{Carath\'{e}odory criterion} \\
Let $\mu^{*}$ be an outer measure on
a set $O$. A subset $E \subset O$ is said to be
$\mu^{*}$-measurable if
\begin{align}\begin{split}\nonumber
\mu^{*}\(T\) = \mu^{*}\(T \cap E\) +
\mu^{*}\(T \cap E^{c}\) \quad
\forall T \subset O.
\end{split}\end{align}
Let $\mathcal{M}$ be the collection of all
$\mu^{*}$-measurable sets.
Then $\mathcal{M}$ is a $\sigma$-algebra,
and the restriction of $\mu^{*}$ to $\mathcal{M}$:
$\mu = \mu^{*}\vert_{\mathcal{M}},
\mu\(E\) = \mu^{*}\(E\), E \in \mathcal{M}$
satisfies:
\begin{enumerate}[(a)]
  \item $\(O, \mathcal{M}, \mu\)$ is a
  complete measure space.
  \item If $E \subset O$ and $\mu^{*}\(E\) = 0$,
  then $E \in \mathcal{M}$ and thus $\mu\(E\) = 0$.
\end{enumerate}
From now on,
we may also call a $\mu^*$-measurable set a $\mu$-measurable set.
If $\mu^{*}$ is a metric outer measure,
then
\begin{enumerate}[(c)]
  \item $\mathcal{B}\(O\) \subset \mathcal{M}$, i.e.
  $\mathcal{M}$ contains all Borel sets of $O$ and
  thus $\(O, \mathcal{B}\(O\), \mu\)$ where $\mu$ is
implicitly further restricted to $\mathcal{B}\(O\)$ is a
  Borel measure space.
\end{enumerate}
\end{theorem}

\begin{theorem}
\label{Rademacher theorem}
    {Rademacher theorem} \\
    Let $m, n \in \{1, 2, \ldots\}$,
    $E \subset \mathbb{R}^{n}$ be an open set,
    $f: E \to \mathbb{R}^{m}$ be a
    Lipschitz function, then $f$ is
    differentiable $\mathcal{L}_{n} \; a.e.$
    and the gradient $\nabla f$
   is a measurable function.
\end{theorem}

\begin{theorem}
\label{Whitney extension theorem}
    Whitney extension theorem \\
    Let $m \in \{1, 2, \ldots\}$,
    $E \subset \mathbb{R}^{m}$ be a closed set,
    $f: E \to \mathbb{R}$,
    $g: E \to \mathbb{R}^{m}$
    be continuous functions. Denote
    \begin{align}\begin{split}\nonumber
    R\(x, a\) = \frac{f\(x\) - f\(a\) -
    g\(a\)\(x - a\)}
    {\vert x - a\vert}, \;
    x, a \in E, x \neq a.
    \end{split}\end{align}
    If for all compact set $C \subset E$,
    \begin{align}\begin{split}\nonumber
    \sup\left\{\left\vert R\(x, a\)\right\vert \vert
    0 < \left\vert x - a\right\vert \leq \delta,
    x, a \in C\right\} \to 0,
    \end{split}\end{align}
    as $\delta \downarrow 0$. Then there exists
    a $C^{1}$ function
    $\overline{f}: \mathbb{R}^{m} \to
    \mathbb{R}$
    such that
    \begin{align}\begin{split}\nonumber
    \overline{f}\vert_{E} = f, \quad
    \overline{f}'\vert_{E} = f'.
    \end{split}\end{align}
\end{theorem}

\begin{theorem}
\label{Lusin theroem}
Lusin theorem \\
    Let $\mu$ be a Borel regular measure
    over a metric space $X$,
    $m \in \{1, 2, \ldots\}$,
    $f: X \to \mathbb{R}^{m}$ be a $\mu$
    measurable function,
    $E \subset X$ be a $\mu$ measurable set,
    $\mu\(E\) < \infty$.
    Then for arbitrary $\epsilon > 0$,
    there exists a compact set
    $C \subset E$ such that
    $\mu\(E \backslash C\) < \epsilon$ and
    $f\vert_{C}$ is continuous.
\end{theorem}

\begin{theorem}
\label{Egoroff theorem}
    Egoroff theorem \\
    Let $\mu$ be a Borel regular measure
    over a metric space $X$,
    $m \in \{1, 2, \ldots\}$,
    $\{f_{n}\}$ be a
    sequence of $\mu$ measurable functions
    $f_{n}: X \to \mathbb{R}^{m}$,
    $f: X \to \mathbb{R}^{m}$ be a
    $\mu$ measurable function. If
    \begin{align}\begin{split}\nonumber
    f_{n}\(x\) \to f\(x\), \;
    \mu \; a.e. \; x \in E,
    \end{split}\end{align}
    where $E \subset X$ is $\mu$ measurable,
    $\mu\(E\) < \infty$.
    Then for arbitrary $\epsilon > 0$,
    there exists a $\mu$ measurable set
    $S \subset E$ such that
    $\mu\(E \backslash S\) < \epsilon$ and
    \begin{align}\begin{split}\nonumber
    f_{n} \to f, \text{ uniformly on } S.
    \end{split}\end{align}
\end{theorem}

\begin{definition}\label{Vitali cover}
{(Vitali cover)} Let $E \subset \mathbb{R}^{m}$,
  $m \in \{1, 2, \ldots\}$.
  If $\mathcal{V}$ is a collection of
  closed balls or closed cubes
  in $\mathbb{R}^{m}$ such that
  for all $x \in E$ and arbitrary
  $\epsilon > 0$,
  there exists $B \in \mathcal{V}$ such that
  $x \in B$ and $\mathrm{diam}B < \epsilon$,
  then $\mathcal{V}$ is called a Vitali cover of $E$.
\end{definition}

\begin{theorem}\label{Vitali covering theorem}
Vitali covering theorem\\
Let $E \subset \mathbb{R}^{m}$, $\mathcal{V}$ is a
Vitali cover of $E$. Then,
there exists an at most countable disjoint subset
$\{B_{j}\} \subset \mathcal{V}$, such that
\begin{align}\begin{split}\nonumber
\mathcal{L}^{*}_{m}\(E \backslash \bigcup_{j}B_{j}\) = 0.
\end{split}\end{align}
\end{theorem}

\begin{theorem}\label{Isodiametric inequality}
    Isodiametric inequality\\
 For all
    set $E \subset \mathbb{R}^{m}$,
    $m \in \{1, 2, \ldots\}$,
    \begin{align}\begin{split}\nonumber
    \mathcal{L}^{*}\(E\) \leq
    \alpha_{m}\(\frac{\mathrm{diam}\ E}{2}\)^{m}.
    \end{split}\end{align}
\end{theorem}

\begin{theorem}\label{Spherical Hausdorff outer measure}
Coincidence between Spherical Hausdorff and Hausdorff outer measures\\
    For all set $E \subset \mathbb{R}^{m}$,
    $m \in \{1, 2, \ldots\}$,
    \begin{align}\begin{split}\nonumber
    \mathcal{H}^{S*}_{m} = \mathcal{H}^{*}_{m},
    \end{split}\end{align}
    where $\mathcal{H}^{S*}_{m}$, the
    spherical Hausdorff outer measure
    is defined as
    \begin{align}\begin{split}\nonumber
    \mathcal{H}^{S*}_{m} = \lim_{\delta \downarrow 0}
    \inf\left\{\sum_{j \geq 1}\alpha_{m}
    \(\frac{\mathrm{diam}\ B_{j}}{2}\)^{m} :
    E \subset \bigcup_{j \geq 1}B_{j},
    \mathrm{diam}\ B_{j} \leq \delta,
    B_{j} \text{ is a closed ball}\right\}.
    \end{split}\end{align}
\end{theorem}

\begin{definition}{(Clarke subgradients)}\label{definition 3}
	Consider a locally Lipschitz function $f:\Omega \to \mathbb{R}$, where $\Omega \subset \mathbb{R}^n$ is an open subset. For each $x\in \Omega$, define
	\bs
	f^\circ\(x;v\) &\coloneqq \limsup_{\substack{y\to x\\ \lambda\downarrow 0}}
	\frac{
		f\(y + \lambda v\) -f\(y\)
	}{
\lambda
	}\\
	&=\lim_{\substack{\varepsilon\to 0\\ \epsilon\downarrow 0}}
	\sup\left\{
		\frac{
			f\(y+\lambda v\) - f\(y\)
		}{
\lambda
		}: y \in \Omega \cap B\(x, \varepsilon\), \lambda \in \(0, \epsilon\)
	\right\}, \forall v \in \mathbb{R}^n, y +\lambda v \in \Omega,
\end{split}\end{align}
and the Clarke subgradient\footnote{This is also called Clarke subdifferential in some literatures, or generalized gradient and
generalized directional derivative by Francis Clarke.
} of $f$ at $x$:
\bs
\partial f\(x\) \coloneqq \left\{
	\xi \in \mathbb{R}^n: f^\circ\(x; v\) \geq \left\langle v, \xi\right\rangle, \forall v \in \mathbb{R}^n
\right\}.
\end{split}\end{align}
\end{definition}

\begin{theorem}\label{clarke subgradients alternative}{Characterization of Clarke subgradients in $\mathbb{R}^n$} \\
Let $f: \Omega \to R$ be a locally Lipschitz function, where $\Omega \subset \mathbb{R}^{n}$ is an open set, then
\bs
\partial f\(x\) = \text{conv}\left\{
	\lim_{k \to \infty} \nabla f\(x_k\): x_k \to x, \nabla f\(x_k\)\ \text{exists}
\right\}.
\end{split}\end{align}
\end{theorem}

\begin{definition}{(Clarke Jacobian)}\label{Clarke Jacobian}
Let $F: \Omega \to \mathbb{R}^m$ be a locally Lipschitz function, where
$\Omega \subset \mathbb{R}^n$ is an open set and $m > 1$.
The Clarke Jacobian of
$F$ at $x\in \Omega$, denoted as $J_c F\(x\)$, is
\bs
J_c F\(x\) \coloneqq \text{conv}
\left\{
\lim_{k\to\infty} J F\(x_k\): x_k \to x, JF\(x_k\)\ \text{exists}
\right\}.
   \end{split}\end{align}
\end{definition}

\begin{theorem}\label{Nonsmooth implicit function theorem}
Nonsmooth implicit function theorem\\
Let $F: \Omega \to \mathbb{R}^m$ be a locally Lipschitz ($C^{1}$) function,
where $\Omega \subset \mathbb{R}^{n+m}$ is an open set. Assume that
$\(x_0, y_0\)$ is such that
\begin{enumerate}
\item $F\(x_0, y_0\) = 0$.
\item $J_{c,y} F\(x_0, y_0\)$ is full rank in the sense that all matrices in
$J_{c,y} F\(x_0, y_0\)$ is full rank,
where $J_{c,y} F\(x_0, y_0\)$
consists of all $m \times m$ component matrices in
$J_{c} F\(x_0, y_0\)$ written as $\[A_{m \times n}, B_{m \times m}\]$.
\end{enumerate}
Then there exists an $\(n+m\)$-dimensional interval $I = I_x^n \times I_y^m
\subset \Omega$, where for some positive vectors $\alpha$ and $\beta$,
\bs
I_x = \left\{
x \in \mathbb{R}^n: \vert x - x_0\vert < \alpha
\right\},\quad
I_y = \left\{
y \in \mathbb{R}^m: \vert y - y_0\vert < \beta
\right\},
   \end{split}\end{align}
where $\vert x - x_0\vert < \alpha$ means that
$\vert x_i - x_{0,i}\vert < \alpha_i$  for $\alpha = \(\alpha_1, \ldots, \alpha_n\)$, and a Lipschitz ($C^{1}$) function $\xi: I_x^n \to I_y^m$ such that
for all $(x, y) \in I_x^n \times I_y^m$,
\bs
F\(x,y\) = 0 \Leftrightarrow y = \xi\(x\).
   \end{split}\end{align}
\end{theorem}

\begin{definition}{(Sub (super) differentiability)}\label{sub sup differentiability}
Let $\Omega$ be an open set of $\mathbb{R}^n$, and $f: \Omega \to \mathbb{R}$ a function. Then
		$f$ is said to be subdifferentiable at $x$, with subgradient $p$, if
		\bs
		f\(x'\) \geq f\(x\) + \left\langle p, x' - x\right\rangle + o\(\Vert x' - x\Vert\).
\end{split}\end{align}
The convex set of all subgradients $p$ at $x$ will be denoted by $\nabla^- f\(x\)$.
If $\(-f\)$ is subdifferentiable, then $f$ is said to be superdifferentiable, and the convex set of the negated subgradients for $\(-f\)$ at $x$ is denoted as $\nabla^+ f$. 
\end{definition}

\begin{theorem}\label{morse-sard theorem}
Morse-Sard theorem\\
Let $f \in C^r\(\Omega, \mathbb{R}^m\)$, where $\Omega \in \mathbb{R}^n$ is an open set,
$r > \max\(n - m, 0\)$, then the set of critical values of $f$ is of zero Lebesgue measure and is meager.
\end{theorem}

\begin{theorem}\label{morse-sard sobolev version}
Morse-Sard theorem in Sobolev spaces \\
Let $f \in W^{n - m + 1, p}_{loc}\(\Omega, \mathbb{R}^m\)$, where $\Omega \in \mathbb{R}^n$ is an open set, $p > n \geq m$.
Then the set of critical values of $f$ is of zero Lebesgue measure.
\end{theorem}

\begin{theorem}\label{partition of unity}
Partition of unity\\
Let $E_1, \ldots, E_k$ be open sets in $\mathbb{R}^n$ and $K$ a compact subset of $\bigcup_{j}	K_j$.
Then one can find $\phi_j \in C_0^\infty\(E_j\)$ so that $\phi_j \geq 0$ and
$\sum_1^k \phi_j \leq 1$ with equality in a neighborhood of $K$.
\end{theorem}

\begin{definition}{(Continuous convergence)}\label{continuous convergence}
Let $\mathcal X, \mathcal Y$ be two metric spaces,
and $\{f_{n}\}$ be a sequence of mappings
$f_{n}: \mathcal X \to \mathcal Y$. Given $f: \mathcal X \to \mathcal Y$ then $f_{n}$
convergence to $f$ continuously if $f_{n}\(x_{n}\) \to f\(x\)$
whenever $\{x_{n}\} \subset \mathcal X, x \in \mathcal X, x_{n} \to x$.
\end{definition}

\begin{lemma}\label{UConvergence and CConvergence}
Equivalence between uniform convergence and continuous convergence \\
Let $\mathcal X$ and $\mathcal Y$ be two metric spaces, and
$\{f_{n}\}, f$ be mappings
from $\mathcal X$  to $\mathcal Y$.
\begin{enumerate}
  \item If $\mathcal X$ is compact and $f$ is continuous then
  $f_{n} \to f$ continuously if and only if
  $f_{n} \to f$ uniformly in $\mathcal X$.
  \item If $f_{n} \to f$ continuously then $f$ is continuous.
\end{enumerate}
As a consequence, if $\mathcal X$ is compact then $f_{n} \to f$ continuously
if and only if $f_{n} \to f$ uniformly in $\mathcal X$ and
$f$ is continuous.

\end{lemma}

\begin{proof}[Proof of Lemma \ref{UConvergence and CConvergence}]
\\
  \textit{1.} \textit{If part:} Let $d$ be the metric on $\mathcal Y$.
  For $f_{n} \to f$ uniformly and $x_{n} \to x$ we have
  \begin{align}\begin{split}\nonumber
  d\(f_{n}\(x_{n}\), f\(x\)\) & \leq
  d\(f_{n}\(x_{n}\), f\(x_{n}\)\) +
  d\(f\(x_{n}\), f\(x\)\) \\
  & \leq \sup_{x \in X} d\(f_{n}\(x_{n}\), f\(x_{n}\)\) +
  d\(f\(x_{n}\), f\(x\)\).
  \end{split}\end{align}
  As $n \to \infty$, the first term goes to $0$ by
  uniform convergence and the second term goes to $0$ by continuity. \\
  \textit{Only if part:} Prove by contradiction.
  If $f_{n} \to f$ continuously but not uniformly.
  Then there is a subsequence $\{n_{k}\}$ and $\epsilon > 0$
  such that for all $n_{k}$
  \begin{align}\begin{split}\nonumber
  \sup_{x \in X}d\(f_{n_{k}}\(x\), f\(x\)\) > 2\epsilon.
  \end{split}\end{align}
  By the definition of $\sup$ there is a sequence $\{x_{k}\}$
  such that
  \begin{align}\begin{split}\label{contradiction UC CC}
  d\(f_{n_{k}}\(x_{k}\), f\(x_{k}\)\) > \epsilon.
  \end{split}\end{align}
  Since $\mathcal X$ is a compact metric space there is a
  convergent subsequence $\{x_{k'}\}$ of $\{x_{k}\}$ with
  $x_{k'} \to x_{0}$. Continuous convergence and continuity of
  $f$ require
  \begin{align}\begin{split}\nonumber
  d\(f_{n_{k'}}\(x_{k'}\), f\(x_{k'}\)\) \leq
  d\(f_{n_{k'}}\(x_{k'}\), f\(x_{0}\)\) +
  d\(f\(x_{0}\), f\(x_{k'}\)\) \to 0,
  \end{split}\end{align}
  which violates \eqref{contradiction UC CC}.
\\
  \textit{2.} Let $d_{\mathcal X}$ be the metric on $\mathcal X$.
  Suppose $\{x_{n}\}$ is an
  arbitrary sequence that converges to $x$.
  For all $k$, consider sequence
  $\{f_{m}\(x_{k, m}\)\}$,
  where $d_{\mathcal X}\(x_{k, m}, x_{k}\) < \frac{1}{m}$.
  By continuous convergence,
  \begin{align}\begin{split}\nonumber
  \lim_{m \to \infty} f_{m}\(x_{k, m}\) = f\(x_{k}\).
  \end{split}\end{align}
  Therefore there exist $n_{k} \in \mathbb{N}$,
  such that for all $n > n_{k}$,
  \begin{align}\begin{split}\nonumber
  d\(f_{n}\(x_{k, n}\), f\(x_{k}\)\) < \frac{1}{k}.
  \end{split}\end{align}
  Next consider a sequence $\{y_{n}\}$,
  \begin{align}\begin{split}\nonumber
  \{y_{n}\} = \left\{\underbrace{x_{1, n_{1}}}_{n_{1}-terms},
  \underbrace{x_{2, n_{1} + n_{2}}}_{n_{2}-terms}, \ldots,
  \underbrace{x_{k, \sum^{k}_{j = 1}n_{j}}}_{n_{k}-terms},
  \ldots\right\}.
  \end{split}\end{align}
  Since $y_{n} \to x$ as $n \to \infty$ by
  continuous convergence $f_{n}\(y_{n}\) \to f\(x\)$.
  Write $\sum^{k}_{j = 1}n_{j}$ as $\sum^{k}n_{j}$.
  Note that $y_{\sum^{k}n_{j}} = x_{k, \sum^{k}n_{j}}$,
  $f_{\sum^{k}n_{j}}\(y_{\sum^{k}n_{j}}\) =
  f_{\sum^{k}n_{j}}\(x_{k, \sum^{k}n_{j}}\)$, we have
  \begin{align}\begin{split}\nonumber
  \lim_{k \to \infty} f_{\sum^{k}n_{j}}
  \(x_{k, \sum^{k}n_{j}}\) = f\(x\).
  \end{split}\end{align}
  Then by
  \begin{align}\begin{split}\nonumber
  d\(f\(x_{k}\), f\(x\)\) & \leq
  d\(f\(x_{k}\), f_{\sum^{k}n_{j}}
  \(x_{k, \sum^{k}n_{j}}\)\) +
  d\(f_{\sum^{k}n_{j}}
  \(x_{k, \sum^{k}n_{j}}\), f\(x\)\) \\
  & \leq \frac{1}{k} + o\(1\),
  \end{split}\end{align}
  $f$ is continuous.
\end{proof}

\begin{definition}{(Lusin property (N))}
Let $\(X, \mathcal{F}, \mu\)$ and $\(Y, \mathcal{E}, \nu\)$ be two measure spaces.
A function $f: X \to Y$ has the Lusin Property (N) if for all $N \subset X$ such that $\mu\(N\) = 0$, there holds $\nu\(f\(N\)\) = 0$.
\end{definition}

\begin{definition}\label{Dini derivative}
{(Dini derivative)}
Let $f: \mathbb{R} \to \mathbb{R}$, denote
\begin{align}\begin{split}\nonumber
D^{+}f\(x\) = \limsup_{h \downarrow 0}\frac{f\(x + h\) - f\(x\)}{h}, & \quad
D_{+}f\(x\) = \liminf_{h \downarrow 0}\frac{f\(x + h\) - f\(x\)}{h}, \\
D^{-}f\(x\) = \limsup_{h \uparrow 0}\frac{f\(x + h\) - f\(x\)}{h}, & \quad
D_{-}f\(x\) = \liminf_{h \uparrow 0}\frac{f\(x + h\) - f\(x\)}{h}.
\end{split}\end{align}
They are called upper right, lower right, upper left and lower left Dini derivatives.
\end{definition}

\begin{theorem}\label{Dini derivative Lusin}
Let $E \subset \mathbb{R}$ and $f: E \to \mathbb{R}$.
If $D^{+}f\(x\)$ is finite for all $x \in E$, then $f$ has Lusin property (N) on $E$.
\end{theorem}

\begin{theorem}\label{Lusin + BV = AC}
Let $f: \[a, b\] \to \mathbb{R}$ be a continuous function of bounded variation, where $a < b$, $a, b \in \mathbb{R}$. Then $f$ is absolutely continuous if and only if it fulfills Lusin property (N) on $\[a, b\]$.
\end{theorem}

For the Rademacher theorem
\ref{Rademacher theorem}, see
  Theorem 2.10.43 in \cite{federer2014geometric}.
  For the Whitney extension theorem \ref{Whitney extension theorem}, see
  Theorem 6.10 in
  \cite{evans2018measure} or
  Theorem 3.1.14 in
  \cite{federer2014geometric}. 
  For the form of
  the Lusin theorem
\ref{Lusin theroem}
and the Egoroff theorem \ref{Egoroff theorem}
  used here, see Theorem 2.3.5 and Theorem 2.3.7
  in \cite{federer2014geometric}, respectively.
  For Vitali covering theorem \ref{Vitali covering theorem}, see
  Theorem 1.10 in \cite{falconer1986geometry}.
  For the isodiametric inequality \ref{Isodiametric inequality},
  see Theorem 2.4 in \cite{evans2018measure} or Corollary 2.10.33 of
\cite{federer2014geometric}.
  For the relationship between the
  spherical Hausdorff outer measure and
  the Hausdorff outer measure, see
  2.10.6 in \cite{federer2014geometric}.
For the nonsmooth implicit function theorem \ref{Nonsmooth implicit function theorem}, see \cite{clarke1976inverse} and Theorem 11
of \cite{hiriart1979tangent}.
The Morse-Sard theorem in Sobolev spaces \ref{morse-sard sobolev version} originates from \cite{de2001morse}.
For more information about the Morse-Sard type theorem 
we refer
to \cite{figalli2008simple} and its renowned infinite dimensional version Sard-Smale theorem from \cite{smale1965infinite}.
The form of the partition of unity theorem \ref{partition of unity} is taken from
Theorem 2.2 in \cite{grigoryan2009heat}. See also Theorem 3.5 of
\cite{grigoryan2009heat} for a manifold version.
Most of these results have
been rewritten to comply with the convention and styles used
in \cite{evans2018measure} and \cite{xuguanglulecturenote2018}.
The proof of part 1 of Lemma \ref{UConvergence and CConvergence}
comes from \cite{resnick2008extreme},
  we remove the separable space
  requirement in the premise.
  Actually, if $\mathcal X$ is a compact metric space,
  then $\mathcal X$ is separable, a countable dense subset of
  $\mathcal X$ can be constructed by 
the total boundedness
  of $\mathcal X$.
Theorem \ref{Dini derivative Lusin} with respect to the upper right Dini derivative is a weaker form of Theorem \Rmnum{9}.4.6 of \cite{saks1937theory}.
For the characterization of absolutely continuous function theorem \ref{Lusin + BV = AC}, see Theorem \Rmnum{7}.6.7 of \cite{saks1937theory}.
Some elementary information about the Dini derivatives can be found in \cite{hagood2006recovering}.

\subsection{Geometric interpretation of Hausdorff measure}
In the following proposition, we 
explain intuitively how
Hausdorff measure
generalizes the usual conceptions of
area and volume.

\begin{proposition}\label{L and H measure}
  Let $k \in \{1, 2, \ldots, n\}$,
  matrix $A \in \mathbb{R}^{n \times k}$, then
  \begin{enumerate}[(i)]
  \item if $\mathrm{rank}\(A\) < k$,
  that is $\det\(A^{T}A\) = 0$,
  then $\mathcal{H}_{k}\(A\(E\)\) = 0$ for all
  $E \subset \mathbb{R}^{k}$.
  \item if $\mathrm{rank}\(A\) = k$,
  that is $\det\(A^{T}A\) > 0$,
  then for all $E \subset \mathbb{R}^{k}$,
  \begin{align}\begin{split}\nonumber
  E \text{ is Lebesgue measurable}
  \Leftrightarrow A\(E\)
  \text{ is } \mathcal{H}_{k} \text{ measurable}.
  \end{split}\end{align}
  When either of the sides holds,
  $\mathcal{H}_{k}\(A\(E\)\) =
  \sqrt{\det\(A^{T}A\)}\mathcal{L}_{k}\(E\)$.
  \end{enumerate}
\end{proposition}

\begin{proposition}\label{coarea linear}
Let $k \in \{1, 2, \ldots, n\}$,
matrix $A \in \mathbb{R}^{k \times n}$, then
\begin{enumerate}[(i)]
  \item if $\mathrm{rank}\(A\) < k$,
  that is $JA = \sqrt{\det\(AA^{T}\)} = 0$,
  then $\mathcal{H}_{n - k}\(E \cap A^{-1}\(y\)\) = 0$
  for $\mathcal{L}_{k} \; a.e. \; y \in \mathbb{R}^{k}$
  for all
  $E \subset \mathbb{R}^{n}$.
  \item if $\mathrm{rank}\(A\) = k$,
  that is $\det\(AA^{T}\) > 0$,
  then for all
  measurable $E \subset \mathbb{R}^{n}$,
  $y \in \mathbb{R}^{k}
  \mapsto \mathcal{H}_{n - k}\(E \cap A^{-1}\(y\)\)$ is
  $\mathcal{L}_{k}$ measurable, and
  \begin{align}\begin{split}\nonumber
  JA \cdot \mathcal{L}_{n}\(E\) =
  \int_{E}JA\(x\)d\mathcal{L}_{n}x =
  \int_{\mathbb{R}^{k}}
  \mathcal{H}_{n - k}\(E \cap A^{-1}\(y\)\)d\mathcal{L}_{k}y.
  \end{split}\end{align}
\end{enumerate}
\end{proposition}

Proposition \ref{L and H measure}
provides the geometric meaning of the
Hausdorff measure.
Consider the set
$A\(\[0, 1\]^{k}\), A \in \mathbb{R}^{n \times k},
k \in \{0, 1, \ldots, n\}$,
denote the column vectors of $A$ as
$a_{1}, a_{2}, \ldots, a_{k}$.
Let $e = \(e_{1}, e_{2}, \ldots, e_{k}\)^{T}
\in E \subset \mathbb{R}^{k}$, then
$Ae = \sum^{k}_{i = 1}e_{i}a_{i}$, therefore
\begin{align}\begin{split}\nonumber
A\(\[0, 1\]^{k}\) =
\left\{\sum^{k}_{i = 1}e_{i}a_{i} \vert
e = \(e_{1}, e_{2}, \ldots, e_{k}\)^{T}
\in \[0, 1\]^{k}\right\}
\end{split}\end{align}
is the parallelepiped spanned by vectors
$a_{1}, a_{2}, \ldots, a_{k}$.
By Proposition \ref{L and H measure},
\begin{align}\begin{split}\nonumber
\mathcal{H}_{k}\(A\(\[0, 1\]^{k}\)\) =
\sqrt{\det\(A^{T}A\)}.
\end{split}\end{align}

We point out that
$\sqrt{\det(A^{T}A)}$ is the $k$-dimensional
volume (in $n$-dimensional space) of
the parallelepiped spanned by
the column vectors of matrix $A$.
To see this, first consider the $k = n$ scenario.
In this case
$\sqrt{\det(A^{T}A)} = \vert\det\(A\)\vert$.
Denote the volume of the parallelepiped
spanned by $A$ as $\text{vol}\(A\)$,
or more transparently as
$\text{Vol}\(a_{1}, a_{2}, \ldots, a_{k}\)$ where
$\{a_{i}\}^{k}_{i = 1}$ are the columns of $A$.
The volume of a parallelepiped
is the ``area'' of its base, times its height.
A base is the parallelepiped determined
by arbitrarily chosen $k - 1$ vectors from
$\{a_{i}\}^{k}_{i = 1}$, and the height corresponding to
this base is the perpendicular distance of the
remaining vector from the base.
Denote the remaining vector as $a_{i}$.
If $a_{i}$ is scaled by
a factor of $c$,
then the perpendicular distance of $a_{i}$
from the base and thus the volume will be scaled by
a factor of $\vert c\vert$.
If $a_{i}$ is translated to
$a'_{i} = a_{i} + \omega a_{j}, i \neq j$,
since $a_{j}$ is parallel to the base,
the height and thus the volume
will not change\footnote{See
\url{https://textbooks.math.gatech.edu/ila/determinants-volumes.html}
for a visualization.}, i.e.
\begin{align}\begin{split}\label{volume homogeneity}
\V\(a_{1}, \ldots, ca_{i} + \omega a_{j},
\ldots, a_{k}\) =
\vert c\vert \V\(a_{1}, \ldots, a_{i}, \ldots,
a_{k}\).
\end{split}\end{align}
Swapping two columns of $A$ just reorders
the vectors $\{a_{i}\}^{k}_{i = 1}$
and will not change the volume, 
\begin{align}\begin{split}\label{volume swap}
\V\(a_{1}, \ldots, a_{i}, \ldots, a_{j}, \ldots, a_{k}\) =
\V\(a_{1}, \ldots, a_{j}, \ldots, a_{i}, \ldots, a_{k}\).
\end{split}\end{align}
Since 
$\vert\det\(A\)\vert$ can also be characterized
by the properties \eqref{volume homogeneity} and
\eqref{volume swap}
and
$\vert\det\(I_{k}\)\vert = \V\(I_{k}\) = 1$ where $I_{k}$
is the $k \times k$ identity matrix,
we have $\sqrt{\det\(A^{T}A\)} =
\vert\det\(A\)\vert = \V\(A\)$.

When $k < n$, using the
Singular Value Decomposition
$A = U\Sigma V$ where $U, V$ are orthogonal matrixes
and $\Sigma$ a rectangular diagonal matrix with
non-negative real diagonal:
\begin{align}\begin{split}\nonumber
\sqrt{\det\(A^{T}A\)} =
\sqrt{\det\(\(U\Sigma V\)^{T}U\Sigma V\)} =
\sqrt{\det\(\Sigma^{T}\Sigma\)}
\end{split}\end{align}
where $\Sigma^{T}\Sigma$ is a $k \times k$
diagonal matrix.
We claim that the $k$-dimensional volume
(in $n$-dimensional space) of $A$ is also
$\sqrt{\det\(\Sigma^{T}\Sigma\)}$.
In particular, since orthogonal transformations
preserves inner products and thus
lengths (norm) and angles,
\begin{align}\begin{split}\nonumber
\V\(A\) = \V\(U\Sigma V\) = \V\(\Sigma V\).
\end{split}\end{align}
Note that, if we choose a orientation for a
base and allow for signed height and volume, then we have
\begin{align}\begin{split}\nonumber
\V_{s}\(a_{1}, \ldots, a_{i} + \Delta a_{i}, \ldots, a_{k}\) =
\V_{s}\(a_{1}, \ldots, a_{i}, \ldots, a_{k}\) +
\V_{s}\(a_{1}, \ldots, \Delta a_{i}, \ldots, a_{k}\)
\end{split}\end{align}
and $\V\(A\) = \vert \V_{s}\(A\)\vert$\footnote{We do not explicitly distinguish
$k$-dimensional volume and $n$-dimensional volume
in $n$-dimensional space, they are both denoted as
$\V$. The specific meaning is clear based on the context.}.
Signed volume also satisfies homogeneity
similar to \eqref{volume homogeneity}, with
the scale factor changing from
$\vert c\vert$ to $c$. Then we have
\begin{align}\begin{split}\nonumber
\V_{s}\(a_{1}, a_{2}, \ldots, a_{k}\) & =
\V_{s}\(\sum^{n}_{i_{1} = 1}a_{i_{1}, 1}e_{i_{1}},
\sum^{n}_{i_{2} = 1}a_{i_{2}, 2}e_{i_{2}}, \ldots,
\sum^{n}_{i_{k} = 1}a_{i_{k}, k}e_{i_{k}}\) \\
& = \sum^{n}_{i_{1} = 1}\sum^{n}_{i_{2} = 1}\cdots
\sum^{n}_{i_{k} = 1}\prod^{k}_{j = 1}a_{i_{j}, j}
\V_{s}\(e_{i_{1}}, e_{i_{2}}, \ldots, e_{i_{k}}\),
\end{split}\end{align}
where each $e_{i_j}$ is a $n\times 1$ vector with $1$ in the
$i_j$th position and zeros otherwise.
When $\{i_{j}\}^{k}_{j = 1}$ contains repeated values,
$\textrm{rank}\(e_{i_{1}}, e_{i_{2}}, \ldots, e_{i_{k}}\) <
k$ and $\V_{s}\(e_{i_{1}}, e_{i_{2}}, \ldots, e_{i_{k}}\) = 0$.
Therefore, scaling a row of $n \times k$ matrix by
a factor of $d$ also scales the volume by
the factor of $\vert d\vert$, implying that
$\V\(\Sigma V\) = \prod^{k}_{i = 1}\sigma_{i}\V\(V\)$,
where $\sigma_{i}$ are the diagonal elements of $\Sigma$.
Intuitively, $\V\(\[0, 1\]^{k}\)=1$ since volume is base times height.
Thus, the volume of $\V\(V\)$ is equal to $1$.

For an matrix
$A \in \mathbb{R}^{k \times n}$,
consider the 
geometric meaning of $JA$.
We already know that the
$k$-dimensional volume of
$A^{T}\(\[0, 1\]^{k}\)$
is $\sqrt{\det\(AA^{T}\)}$.
For arbitrary $P \in \mathbb{R}^{n \times k}$,
the $k$-dimensional volume of
$AP\(\[0, 1\]^{k}\)$ is equal to the
$k$-dimensional volume of the
$AP^{\bot}\(\[0, 1\]^{k}\)$, where
$P^{\bot}$ is the orthogonal projection of
$P$ to the
orthogonal complement of
the null space of $A$ (which, by definition,
is the span of the columns of $A^{T}$).
Note that
$\det\(\(P^{\bot}\)^{T}P^{\bot}\)
\leq \det\(P^{T}P\)$.
Therefore,
\begin{align}\begin{split}\nonumber
JA = \sup_{P}\frac{\V\(AP\)}
{\V\(P\)} = \sup_{P}
\frac{\mathcal{H}_{k}\(AP
\(\[0, 1\]^{k}\)\)}{\mathcal{H}_{k}\(P
\(\[0, 1\]^{k}\)\)},
\end{split}\end{align}
where supremum is taken over all
$k$-dimensional nondegenerate
parallelepiped $P$.

\subsection{Proofs of area and coarea formulas}
\label{Proof sketches of area and coarea formulas}

Figure \ref{relation between area and coarea formula proofs}
illustrates the ordering of the proofs for area and coarea formulas.
Besides	 \cite{federer2014geometric} and \cite{evans2018measure},
we also 
borrow lots of material from
\cite{robertjerrard} and \cite{xuguanglulecturenote2018}, especially from the latter. Corollaries
\ref{Sard type lemma} and
\ref{another Sard lemma} are used to illustrate the applicability of
the area and coarea formulas Theorem \ref{area formula classic}
and Theorem \ref{coarea formula classic}.

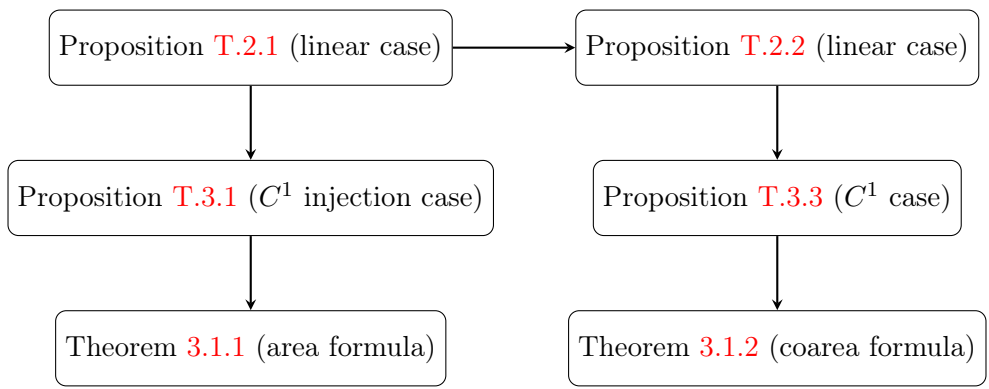
\begin{figure}[!h]
\center
	\caption{Ordering of proofs for area and coarea formulas}
\begin{tikzpicture}[node distance=2cm]
\label{relation between area and coarea formula proofs}
\node (lineararea) [basicnode] {Proposition \ref{L and H measure} (linear case)};
\node (linearcoarea) [basicnode, right of=lineararea,  xshift=5cm] {Proposition \ref{coarea linear} (linear case)};
\node (areasimple) [basicnode, below of=lineararea,  yshift=0cm] {Proposition \ref{area formula simple} ($C^1$ injection case)};
\draw[arrow] (lineararea) -- (areasimple);
\node (coareasimple) [basicnode, below of=linearcoarea,	 yshift=0cm] {Proposition \ref{coarea formula simple} ($C^1$ case)};
\draw[arrow] (linearcoarea) -- (coareasimple);
\node (areaformula) [basicnode, below of=areasimple,  yshift=0cm] {Theorem \ref{area formula classic} (area formula)};
\draw[arrow] (areasimple) -- (areaformula);
\draw[arrow] (lineararea) -- (linearcoarea);
\node (coareaformula) [basicnode, below of=coareasimple,  yshift=0cm] {Theorem \ref{coarea formula classic} (coarea formula)};
\draw[arrow] (coareasimple) -- (coareaformula);
\end{tikzpicture}
\end{figure}

\begin{proof}[Proof Outline of Proposition \ref{L and H measure}] \\
  \textbf{step} 1 The Hausdorff outer measure
  $\mathcal{H}^{*}_{k}$ is invariant under
  orthogonal transformations, i.e.
  if $A$ is a orthogonal transformations then
  \begin{align}\begin{split}\nonumber
  \mathcal{H}^{*}_{k}\(A\(E\)\) = \mathcal{H}^{*}_{k}\(E\),
  \forall E \subset \mathbb{R}^{k}
  \end{split}\end{align}
  \textbf{step} 2 Let
  $A = \(I_{k}, 0\)^{T}$,
  prove that
  \begin{align}\begin{split}\nonumber
  \mathcal{H}^{*}_{k}\(A\(E\)\) =
  \mathcal{H}^{*}_{k}\(E \times \{0\}\) =
  \mathcal{H}^{*}_{k}\(E\) =
  \mathcal{L}^{*}_{k}\(E\)
  \end{split}\end{align}
  where $\mathcal{L}^{*}_{k}$ is the $k$-dimensional
  Lebesgue outer measure. \\

  \noindent \textbf{step} 3
There exists
orthogonal $T \in \mathbb{R}^{n \times n}$,
  \begin{align}\begin{split}\nonumber
  \mathcal{H}^{*}_{k}\(A\(E\)\) =
\mathcal{H}^{*}_{k}\(T A\(E\)\) =
  \mathcal{H}^{*}_{k}\(\(A_{1}\(E\)\) \times \{0\}\) =
  \mathcal{L}^{*}_{k}\(A_{1}\(E\)\),
  \end{split}\end{align}
  where $A_{1} \in \mathbb{R}^{k \times k}$,
$A \in \mathbb{R}^{n \times k}$,
 and $T A = \(A_{1}^{T}, 0\)^{T}$. \\
  \textbf{step} 4 Proof of
  \begin{align}\begin{split}\nonumber
  \mathcal{L}^{*}_{k}\(A_{1}\(E\)\) =
  \vert\det\(A_{1}\)\vert \mathcal{L}^{*}_{k}\(E\) =
  \sqrt{\det\(A^{T}A\)}\mathcal{L}^{*}_{k}\(E\).
  \end{split}\end{align}
  \textbf{step} 5 For all $E \subset \mathbb{R}^{k}$,
  $E \text{ is Lebesgue measurable}
  \Leftrightarrow A\(E\)
  \text{ is } \mathcal{H}_{k} \text{ measurable}$.
\end{proof}

\begin{remark} We provide two supplementary details to assist with
understanding the proof outlines of Proposition \ref{L and H measure}
and the following results.
\begin{enumerate}
\item
Hausdorff outer measure $\mathcal{H}^{*}_{k}$ has the
invariance under the orthogonal transformations. This is a direct
result of the fact that
Hasudroff outer measure keeps distance inequality, i.e.
if $k, l, m, n \in \mathbb{Z}^{+}$,
$E \subset \mathbb{R}^{l}$,
map $f: E \to \mathbb{R}^{n}$ and map
$g: E \to \mathbb{R}^{m}$ satisfies
\begin{align}\begin{split}\nonumber
\vert f\(x\) - f\(y\)\vert \leq
C\vert g\(x\) - g\(y\)\vert \quad
\forall x, y \in E,
\end{split}\end{align}
then
$\mathcal{H}^{*}_{k}\(f\(E\)\) \leq C^{k}\mathcal{H}^{*}_{k}\(g\(E\)\)$.
This property will also be used in the following
discussion.

\item
A cube is a
subset of $\mathbb{R}^{k}$ in the form of
\begin{align}\begin{split}\nonumber
\prod^{k}_{i = 1}\(a_{i}, b_{i}\), \quad
\prod^{k}_{i = 1}\(a_{i}, b_{i}\], \quad
\prod^{k}_{i = 1}\[a_{i}, b_{i}\), \quad
\prod^{k}_{i = 1}\[a_{i}, b_{i}\] \quad
a_{i} < b_{i},
\forall i \in \{1, 2, \ldots, k\}.
\end{split}\end{align}
Let $E \in \mathbb{R}^n$ be a nonempty open set, then there exists a sequence of
disjoint left open and right closed cubes $\{Q_k\}_{k=1}^\infty$, such that
\bs
E = \bigcup_{k=1}^\infty Q_k = \bigcup_{k=1}^\infty \bar Q_k.
\end{split}\end{align}
\end{enumerate}
\end{remark}

\begin{proposition}\label{area formula simple}
  Let $E \subset \mathbb{R}^{k}$ be an open set,
  $k \in \{1, 2, \ldots, n\}$,
  $\psi: E \to \mathbb{R}^{n}$ be a $C^{1}$
  injection and
  $\det\(\psi'\(x\)^{T}\psi'\(x\)\) > 0$ for all
  $x \in E$,
  then for all $D \subset \psi\(E\)$,
  \begin{align}\begin{split}\nonumber
  D \text{ is } \mathcal{H}_{k} \text { measurable}
  \Leftrightarrow \psi^{-1}\(D\)
  \text{ is Lebesgue measurable}.
  \end{split}\end{align}
  When $D$ is $\mathcal{H}_{k}$ measurable,
  \begin{align}\begin{split}\nonumber
  \mathcal{H}_{k}\(D\) = \int_{\psi^{-1}\(D\)}
  \sqrt{\det\(\psi'\(x\)^{T}\psi'\(x\)\)}
  d\mathcal{L}_{k}x.
  \end{split}\end{align}
\end{proposition}

Proposition
\ref{area formula simple}
plays the most essential role in a proof
of more general area formula Theorem
\ref{area formula classic}.
Proposition \ref{area formula simple}
can also be powerful when used alone.\\

\begin{proof}[Proof Outline of Proposition \ref{area formula simple}] \\
  \textbf{step} 1 Estimates of
  $\vert\psi\(x\) - \psi\(y\)\vert$:
  Let $k \in \{1, 2, \ldots, n\}$,
  $E \subset \mathbb{R}^{k}$ be an open set,
  $\psi: E \to \mathbb{R}^{n}$ be a
  $C^{1}$ injection
  such that $\det\(\psi'\(x\)^{T}\psi'\(x\)\) > 0$
  for all $x \in E$,
  then:
  \begin{enumerate}
    \item if $K \subset E$ is convex and compact,
    then there exists $0 < c < C$ such that
    \begin{align}\begin{split}\label{estimate convex compact set}
    c\vert x - y\vert \leq
    \vert\psi\(x\) - \psi\(y\)\vert \leq
    C\vert x - y\vert
    \end{split}\end{align}
    for all $x, y \in K$.
    \item For arbitrary $x_{0} \in E$ and
    for all $0 < \epsilon < 1$,
    there exists $\delta > 0$, such that
    open ball
    $B\(x_{0}, \delta\) \subset E$ and
    \begin{align}\begin{split}\label{estimate local}
    \(1 - \epsilon\)\vert\psi'\(x_{0}\)
    \(x - y\)\vert \leq \vert\psi\(x\) -
    \psi\(y\)\vert \leq \(1 + \epsilon\)\vert
    \psi'\(x_{0}\)\(x - y\)\vert
    \end{split}\end{align}
    for all $x, y \in B\(x_{0}, \delta\)$.
  \end{enumerate}
  \textbf{step} 2 Prove that, for all $D \subset \psi\(E\)$,
  \begin{align}\begin{split}\nonumber
  D \text{ is } \mathcal{H}_{k} \text { measurable}
  \Leftrightarrow \psi^{-1}\(D\)
  \text{ is Lebesgue measurable},
  \end{split}\end{align}
  by estimates
  \eqref{estimate convex compact set},
  the fact that Hausdorff outer measure keeps
  distance inequality, and the fact that
  if $D$ is a $\mathcal{H}_{k}$ measurable set with
  $\mathcal{H}_{k}\(D\) < \infty$, then there exist
  a Borel set $P$ and a $\mathcal{H}_{k}$ zero measure set
  $Z$ such that $D = P \cup Z$\footnote{This Borel set, zero measure set construction of $\mathcal{H}_{k}$ measurable set is also a common useful result.
  Actually, the complete result states that there exist Borel sets $P_{1}, P_{2}$ and
  $\mathcal{H}_{k}$ zero measure sets $Z_{1}, Z_{2}$,
  such that
  $D = P_{1} \cup Z_{1} = P_{2} \backslash Z_{2}$.}. \\
  \textbf{step} 3 By \eqref{estimate local}
  and Proposition \ref{L and H measure},
  closed cube
  $Q \subset E$ satisfies:
  for all $x_{0} \in E$,
  for all $0 < \epsilon < 1$,
  there exists $\delta > 0$, such that,
  if $\mathrm{diam}\ \(Q\) < \delta$ and
  $x_{0} \in Q$ then
  \begin{align}\begin{split}\label{L and H local estimate}
  \(1 - \epsilon\)^{k}\sqrt{\det\(
  \psi'\(x_{0}\)^{T}\psi'\(x_{0}\)\)}
  \mathcal{L}_{k}\(Q\) & \leq
  \mathcal{H}_{k}\(\psi\(Q\)\) \\
  & \leq \(1 + \epsilon\)^{k}\sqrt{\det\(
  \psi'\(x_{0}\)^{T}\psi'\(x_{0}\)\)}
  \mathcal{L}_{k}\(Q\)
  \end{split}\end{align}
  \textbf{step} 4 Prove that, for all closed cube $Q \subset E$,
  \begin{align}\begin{split}\nonumber
  \mathcal{H}_{k}\(\psi\(Q\)\) = \int_{Q}\sqrt{\det\(
  \psi'\(x\)^{T}\psi'\(x\)\)}
  d\mathcal{L}_{k}x
  \end{split}\end{align}
  by \eqref{L and H local estimate}. \\
  \textbf{step} 5 Prove that for all bounded open set $E_{b}$
  such that $\overline{E_{b}} \subset E$,
  if $O \subset E_{b}$ is $\mathcal{L}_{k}$ measurable then
  \begin{align}\begin{split}\nonumber
  \mathcal{H}_{k}\(\psi\(O\)\) = \int_{O}\sqrt{\det\(
  \psi'\(x\)^{T}\psi'\(x\)\)}
  d\mathcal{L}_{k}x.
  \end{split}\end{align}
  Conclude using the fact that any open set
  $E \subset \mathbb{R}^{k}$ can be decomposed
  to a countable disjoint union of bounded cube.
\end{proof}

\begin{proof}[Proof of Theorem \ref{area formula
classic}] \\
A classical proof based on Proposition \ref{area formula simple} can be
separated into three fundamental parts. \\
  \textbf{part} 1 
  In case that $\psi$ is not
  necessarily bijective, while still
  requiring that $J\psi\(x\) > 0$
  for all $x \in E$,
  by the implicit function theorem,
  for all $x \in E$ there exist
  a neighborhood $U$ such that $\psi$
  is bijective in $U$.
  Take a Vitali cover $\mathcal{V}$
  of $E$ such that $\psi$ is bijective in
  every closed ball $B \in \mathcal{V}$.
  Then by the Vitali covering theorem,
  there exists an at most countable
  disjoint subset
  $\{B_{j}\} \subset \mathcal{V}$,
  such that
  $\mathcal{L}^{*}_{k}\(E \backslash \bigcup_{j}
  B_{j}\) = 0$.
  From the definition of $\mathcal{H}_{0}$,
  \begin{align}\begin{split}\nonumber
  \sum_{j}1\(y \in \psi\(S \cap B_{j}\)\)
  = \mathcal{H}_{0}\(\bigcup_{j}\(S \cap B_{j}\)
  \cap \psi^{-1}\(y\)\).
  \end{split}\end{align}
  By the property of Lebesgue integral and Proposition \ref{area formula simple},
  \begin{align}\begin{split}\nonumber
  \int_{S}J\psi\(x\)d\mathcal{L}_{k}x
   = \sum_{j}\int_{\psi\(S\)}
  1\(y \in \psi\(S \cap B_{j}\)\)
  d\mathcal{H}_{k}y
   = \int_{\psi\(S\)}
  \mathcal{H}_{0}\(\bigcup_{j}\(S \cap B_{j}\) \cap
  \psi^{-1}\(y\)\)d\mathcal{H}_{k}y.
  \end{split}\end{align}
  Then note that
  \begin{align}\begin{split}\nonumber
  \int_{\psi\(S\)} & \mathcal{H}_{0}\(S \cap
  \psi^{-1}\(y\)\)d\mathcal{H}_{k}y \\
  = \int_{\psi\(S\)} & \mathcal{H}_{0} \(
  \bigcup_{j}\(S \cap B_{j}\) \cap
  \psi^{-1}\(y\)\)d\mathcal{H}_{k}y +
  \int_{\psi\(S\) \backslash \psi\(
  S \backslash \bigcup_{j}B_{j}\)}
  \mathcal{H}_{0}\(\(
  S \backslash \bigcup_{j}B_{j}\) \cap
  \psi^{-1}\(y\)\)d\mathcal{H}_{k}y \\
  & +
  \int_{\psi\(S \backslash \bigcup_{j}
  B_{j}\)} \mathcal{H}_{0}\(\(
  S \backslash \bigcup_{j}B_{j}\) \cap
  \psi^{-1}\(y\)\)d\mathcal{H}_{k}y \\
  = \int_{\psi\(S\)} & \mathcal{H}_{0}
  \(\bigcup_{j}\(S \cap B_{j}\) \cap
  \psi^{-1}\(y\)\)d\mathcal{H}_{k}y + 0 + 0.
  \end{split}\end{align}
  When $J\psi\(x\) = 0$,
  let $\epsilon > 0$, define $\psi_{\epsilon}: E
  \to \mathbb{R}^{k + n}$ as
  \begin{align}\begin{split}\nonumber
  x \mapsto \(\epsilon x, \psi\(x\)\),
  \end{split}\end{align}
  $Crit\(\psi\) = \{x \in E \vert J\psi\(x\) = 0\}$.
  Note that $J\psi_{\epsilon}\(x\) > 0$ for all
  $x \in E$,
  \begin{align}\begin{split}\nonumber
  \int_{Crit\(\psi\)}J\psi_{\epsilon}\(x\)
  d\mathcal{L}_{k}x & =
  \int_{\psi_{\epsilon}\(Crit\(\psi\)\)}
  \mathcal{H}_{0}\(Crit(\psi) \cap
  \psi^{-1}_{\epsilon}\(y\)\)d\mathcal{H}_{k}y \\
  & \geq
  \int_{\psi_{\epsilon}\(Crit\(\psi\)\)}
  d\mathcal{H}_{k}y =
  \mathcal{H}_{k}\(\psi_{\epsilon}\(Crit\(\psi\)\)\).
  \end{split}\end{align}
  By the fact that Hausdorff outer measure keeps distance inequality,
  \begin{align}\begin{split}\nonumber
  \mathcal{H}_{k}\(\psi\(Crit\(\psi\)\)\) \leq
  \mathcal{H}_{k}\(\psi_{\epsilon}\(Crit\(\psi\)\)\),
  \end{split}\end{align}
  since the coordinate projection from
  $\psi_{\epsilon}\(Crit\(\psi\)\)$ to
  $\psi\(Crit\(\psi\)\)$ satisfies
  a distance inequality with $C = 1$.
  Therefore,
  \begin{align}\begin{split}\nonumber
  \mathcal{H}_{k}\(\psi\(Crit\(\psi\)\)\) \leq
  \int_{Crit\(\psi\)}J\psi_{\epsilon}\(x\)d\mathcal{L}_{k}x,
  \end{split}\end{align}
  the right hand side converges to $0$
  as $\epsilon \downarrow 0$ if
  $E$ is bounded. Then conclude that
  $\mathcal{H}_{k}\(\psi\(Crit\(\psi\)\)\) = 0$ for (not
  necessarily bounded) open set
  $E \subset \mathbb{R}^{k}$ by the
  cube decomposition in the \textbf{step} 5 of the sketch of
  the proof of
  Proposition \ref{area formula simple}.
  Now, one can conclude that
  \eqref{area formula area} is true for $C^{1}$
  function $\psi$. \\
  \textbf{part} 2 To verify \eqref{area formula area}
  when $\psi$ is a Lipschitz function but
  does not necessarily belong to $C^{1}$,
  a Lusin type approximation of $\psi$
  can be used. To continue,
  we use Rademacher theorem \ref{Rademacher theorem},
  Whitney extension theorem \ref{Whitney extension theorem},
  Lusin theroem \ref{Lusin theroem}  and Egoroff theorem \ref{Egoroff theorem}.

 Assume first $E$ is bounded, by
  Rademacher theorem \ref{Rademacher theorem}, $\psi$ is
  differentiable $\mathcal{L}_{k} \; a.e.$ and
  the gradient
  $\psi' \leq \text{Lip} \psi$ is measurable,
  where
  \begin{align}\begin{split}\nonumber
  \text{Lip} \psi = \sup\left\{\frac{\vert
  \psi\(x_{1}\) - \psi\(x_{2}\)\vert}
  {\vert x_{1} - x_{2}\vert} \vert
  x_{1}, x_{2} \in E,
  x_{1} \neq x_{2}\right\}.
  \end{split}\end{align}
  Apply Lusin theorem to $\psi'$,
  there exists a compact set
  $C \subset E$ such that
  $\mathcal{L}_{k}\(E \backslash C\) <
  \frac{1}{2}\epsilon$ and $\psi'\vert_{C}$
  is continuous. Let
  \begin{align}\begin{split}\nonumber
  R\(x, a\) = \frac{\psi\(x\) -
  \psi\(a\) - \psi'\(a\)\(x - a\)}
  {\vert x - a\vert}, \; x, a \in C,
  x \neq a,
  \end{split}\end{align}
  since $\psi$ is differentiable,
  for all $a \in C$,
  \begin{align}\begin{split}\nonumber
  R\(a\) = \sup\{\vert R\(x, a\)\vert \vert
  0 < \vert x - a\vert \leq \delta,
  x \in C\} \to 0,
  \end{split}\end{align}
  as $\delta \downarrow 0$. Then by
  Egoroff theorem and regularity
  of Lebesgue measure, there exists
  a compact set $C' \subset C$ such that
  $\mathcal{L}_{k}\(C \setminus C'\) <
  \frac{1}{2}\epsilon$ and
  $R\(a\)$ converge to $0$ uniformly
  on $C'$. Now, we can apply
  Whitney extension theorem to
  $\psi$ and $\psi'$ (actually,
  to each component function of $\psi$ and its
  gradient), i.e. there exists
  a $C^{1}$ function $\overline{\psi}$
  such that
  \begin{align}\begin{split}\label{C1 extension epsilon}
  \overline{\psi}\vert_{C'}\(x\) =
  \psi\vert_{C'}\(x\), & \quad
  \overline{\psi}'\vert_{C'}\(x\) =
  \psi'\vert_{C'}\(x\), \\
  \mathcal{L}_{k}\(\left\{x \vert
  \overline{\psi}\vert_{E}\(x\) \neq
  \psi\(x\)\right\}\) < \epsilon, & \quad
  \mathcal{L}_{k}\(\left\{x \vert
  \overline{\psi}'\vert_{E}\(x\) \neq
  \psi'\(x\)\right\}\) < \epsilon.
  \end{split}\end{align}
  Now, one can conclude that
  \eqref{area formula area} is true for
  Lipschitz function $\psi$ and set
  $S \cap C'$. \\
  \textbf{part} 3 The final step is to verify a
  Lusin property (N) for
  \begin{align}\begin{split}\nonumber
  \int_{\psi\(S\)}\mathcal{H}_{0}\(S \cap
  \psi^{-1}\(y\)\)d\mathcal{H}_{k}y.
  \end{split}\end{align}
  Specifically, for arbitrary measurable
  $S \subset E$,
  \begin{align}\begin{split}\nonumber
  \int_{\psi\(S\)}\mathcal{H}_{0}\(S \cap
  \psi^{-1}\(y\)\)d\mathcal{H}_{k}y \leq
  \(\text{Lip}\psi\)^{n}\mathcal{L}_{k}\(S\).
  \end{split}\end{align}
  To see this, let
  \begin{align}\begin{split}\nonumber
  \mathcal{Q}_{m} = \left\{Q \vert Q =
  \prod^{k}_{i = 1}\(a_{i}, b_{i}\],
  a_{i} = \frac{c_{i}}{m},
  b_{i} = \frac{c_{i} + 1}{m},
  c_{i} \in \mathbb{Z}\right\},
  \end{split}\end{align}
  and
  \begin{align}\begin{split}\nonumber
  g_{m}\(y\) = \sum_{Q \in \mathcal{Q}_{m}}
  1\(y \in \psi\(S \cap Q\)\).
  \end{split}\end{align}
  Since $\mathbb{R}^{k} = \bigcup_{Q \in
  \mathcal{Q}_{m}}$, and $g_{m}\(y\)$
  is the number of cubes $Q \in
  \mathcal{Q}_{m}$ such that
  \begin{align}\begin{split}\nonumber
  \mathcal{H}_{0}\(S \cap Q \cap \psi^{-1}\(y\)\) > 0.
  \end{split}\end{align}
  Therefore, for all $y \in \mathbb{R}^{n}$,
  \begin{align}\begin{split}\nonumber
  g_{m}\(y\) \uparrow \mathcal{H}_{0}\(S \cap
  \psi^{-1}\(y\)\),
  \end{split}\end{align}
  as $m \to \infty$.
  Then by the 
  monotone convergence theorem
  \begin{align}\begin{split}\nonumber
  \int_{\psi\(S\)}\mathcal{H}_{0}\(S \cap
  \psi^{-1}\(y\)\)d\mathcal{H}_{k}y & =
  \lim_{m \to \infty}\int_{\psi\(S\)}
  g_{m}\(y\)d\mathcal{H}_{k}y \\
  & = \lim_{m \to \infty}
  \sum_{Q \in \mathcal{Q}_{m}}
  \mathcal{H}_{k}\(\psi\(S \cap Q\)\) \\
  & \leq \lim_{m \to \infty}
  \sum_{Q \in \mathcal{Q}_{m}}
  \(\text{Lip}\psi\)^{k}\mathcal{L}_{k}
  \(S \cap Q\) \\
  & = \(\text{Lip}\psi\)^{k}\mathcal{L}_{k}\(S\).
  \end{split}\end{align}
  Now note that,
  \begin{align}\begin{split}\nonumber
  \int_{S \cap C'}J\psi\(x\)d\mathcal{L}_{k}x \leq
  \int_{\psi\(S\)}\mathcal{H}_{0}\(S \cap
  \psi^{-1}\(y\)\)d\mathcal{H}_{k}y
  \leq
  \int_{S \cap C'}J\psi\(x\)d\mathcal{L}_{k}x +
  \(\text{Lip}\psi\)^{k}\mathcal{L}_{k}
  \(S \backslash C'\),
  \end{split}\end{align}
  where $\mathcal{L}_{k}\(S \backslash C'\) <
  \epsilon$.
  Note that $J\psi$ is bounded on $E$ due to
  the Lipschitz continuity of $\psi$ (By Rademacher theorem, $J\psi$ exists
  for $\mathcal{L}_{k} \; a.e.\; x \in E$), i.e.
  there exists a constant $M$, such that
  \begin{align}\begin{split}\nonumber
  \int_{S \backslash C'}J\psi\(x\)d\mathcal{L}_{k}x \leq
  M\mathcal{L}_{k}\(S \backslash C'\).
  \end{split}\end{align}
  Therefore,
  \begin{align}\begin{split}\nonumber
  \int_{S \cap C'}J\psi\(x\)d\mathcal{L}_{k}x & \leq
  \int_{S}J\psi\(x\)d\mathcal{L}_{k}x \\
  & \leq
  \int_{S \cap C'}J\psi\(x\)d\mathcal{L}_{k}x +
  M\mathcal{L}_{k}\(S \backslash C'\),
  \end{split}\end{align}
  and \eqref{area formula area} follows
  from the arbitrariness of $\epsilon$.
  The case when open set
  $E \subset \mathbb{R}^{k}$ is unbounded
  follows from the cube decomposition.
\end{proof}

\begin{cor}\label{Sard type lemma}
{A Sard type lemma} \\
Let $k \in \{1, 2, \ldots, n\}$,
$E \subset \mathbb{R}^{k}$ be an open set,
$\psi: E \to \mathbb{R}^{n}$ be a $C^{1}$ function
and $\text{Crit}\(\psi\) = \{
x \in E \vert J\psi\(x\) = 0\}$,
then $\mathcal{H}_{k}\(\psi\(\text{Crit}\(\psi\)\)\) = 0$.
\end{cor}

\begin{proof}
Since $\psi$ is a $C^{1}$ function,
$J\psi: E \to \mathbb{R}$ is continuous.
Therefore,
$Crit\(\psi\) =
\{x \in E \vert J\psi\(x\) \geq 0\} \backslash
\{x \in E \vert J\psi\(x\) > 0\}$
is $\mathcal{L}_{k}$ measurable.
By \eqref{area formula area},
\begin{align}\begin{split}\nonumber
\int_{Crit(\psi)}J\psi\(x\)d\mathcal{L}_{k}x & =
\int_{\psi\(Crit\(\psi\)\)}
\mathcal{H}_{0}\(Crit(\psi) \cap \psi^{-1}\(y\)\)d\mathcal{H}_{k}y
\\ & \geq \int_{\psi\(Crit\(\psi\)\)}d\mathcal{H}_{k}y =
\mathcal{H}_{k}\(\psi\(Crit\(\psi\)\)\).
\end{split}\end{align}
The integral in the left hand side of
the first equality is $0$.
\end{proof}

\begin{proof}[Proof Outline of Proposition \ref{coarea linear}] \\
  \textbf{step} 1 To prove \textit{(\rmnum{1})},
  note that
  \begin{align}\begin{split}\nonumber
  \mathcal{L}^{*}_{k}\(A\(\mathbb{R}^{n}\)\) =
  \mathcal{L}^{*}_{k}\(U \circ \Sigma \circ V\(\mathbb{R}^{n}\)\) =
  \mathcal{L}^{*}_{k}\(U \circ \Sigma\(\mathbb{R}^{n}\)\) =
  \mathcal{L}^{*}_{k}\(\Sigma\(\mathbb{R}^{n}\)\),
  \end{split}\end{align}
  by SVD and the fact that Lebesgue outer measure is
  invariant under orthogonal transformation.
  Since $\mathrm{rank}\(A\) < k$,
  $\mathrm{rank}\(\Sigma\) < k$, therefore,
  \bs
\mathcal{L}^{*}_{k}\(\Sigma\(\mathbb{R}^{n}\)\) =
  \mathcal{L}^{*}_{k}\(\mathbb{R}^{\mathrm{rank}\(\Sigma\)}\) = 0.
\end{split}\end{align}
\textbf{step} 2 By SVD / PD, $A = WPV$,
  where $V \in \mathbb{R}^{n \times n}$ is
  orthogonal, $P: \mathbb{R}^{n} \to
  \mathbb{R}^{k}$
  is the coordinate projection of the
  first $k$ dimensions, and
  $W \in \mathbb{R}^{k \times k}$ is
  symmetric\footnote{We do not distinguish
  between linear transformation and its matrix.}. \\
\textbf{step} 3
  By the Fubini-Tonelli theorem,
  $y \in \mathbb{R}^{k} \mapsto
  \mathcal{L}_{n - k}\(\left\{\(x_{1}, \ldots,
  x_{n - k}\) : x \in
  V\(E\) \cap \(P\)^{-1}\(y\)\right\}\)$
  is $\mathcal{L}_{k}$ measurable and
  \begin{align}\begin{split}\nonumber
  \mathcal{L}_{n}\(E\) = \mathcal{L}_{n}\(V\(E\)\) =
  \int_{\mathbb{R}^{k}}
  \mathcal{L}_{n - k}\(\left\{\(x_{1}, \ldots,
  x_{n - k}\) : x \in
  V\(E\) \cap \(P\)^{-1}\(y\)\right\}\)
  d\mathcal{L}_{k}y.
  \end{split}\end{align}
  Then, note that
  \begin{align}\begin{split}\nonumber
  \mathcal{L}_{n - k}\(\left\{\(x_{1}, \ldots,
  x_{n - k}\) : x
  \in V\(E\) \cap \(P\)^{-1}\(y\)\right\}\)
  & = \mathcal{H}_{n - k}\(V\(E\) \cap
  \(P\)^{-1}\(y\)\) \\
  & = \mathcal{H}_{n - k}\(E \cap A^{-1} \circ W\(y\)\)
  \end{split}\end{align}
  by $A^{-1} = V^{-1} \circ
  \(P\)^{-1} \circ W^{-1}$. \\
\textbf{step} 4 By 
Proposition \ref{L and H measure} and a standard approximation procedure,
  \begin{align}\begin{split}\nonumber
  JM \cdot \int_{\mathbb{R}^{k}}f\(M\(x\)\)d\mathcal{L}_{k}x =
  \int_{\mathbb{R}^{k}}f\(y\)d\mathcal{L}_{k}y.
  \end{split}\end{align}
  for all $M \in \mathbb{R}^{k \times k}$
  invertible and $f$ nonnegative
  $\mathcal{L}_{k}$ measurable. Therefore,
  \begin{align}\begin{split}\nonumber
  JW \cdot \int_{\mathbb{R}^{k}}\mathcal{H}_{n - k}
  \(E \cap A^{-1} \circ W\(y\)\)d\mathcal{L}_{k}y =
  \int_{\mathbb{R}^{k}}\mathcal{H}_{n - k}
  \(E \cap A^{-1}\(y\)\)d\mathcal{L}_{k}y.
  \end{split}\end{align}
\end{proof}

\begin{proposition}\label{coarea formula simple}
Let $E \subset \mathbb{R}^{n}$ be an open set,
$k \in \{1, 2, \ldots, n\}$,
$\varphi: E \to \mathbb{R}^{k}$ be a $C^{1}$
function,
then for all measurable $S$,
$S \cap \varphi^{-1}\(y\)$ is $\mathcal{H}_{n - k}$
measurable for
$\mathcal{H}_{k} \; a.e. \; y$,
$y \in \mathbb{R}^{k} \mapsto
\mathcal{H}_{n - k}\(S \cap \varphi^{-1}\(y\)\)$
is $\mathcal{H}_{k}$ measurable, and
\begin{align}\begin{split}\nonumber
\int_{S}J\varphi\(x\)d\mathcal{L}_{n}x =
\int_{\mathbb{R}^{k}}\mathcal{H}_{n - k}\(S \cap
\varphi^{-1}\(y\)\)d\mathcal{H}_{k}y.
\end{split}\end{align}
\end{proposition}

\begin{remark}
\indent
\begin{enumerate}
  \item We should note that for open set
  $E \subset \mathbb{R}^{n}$,
  $k \in \{1, 2, \ldots, n\}$,
  $\varphi: E \to \mathbb{R}^{k}$ be at least
  continuous, $S \subset E$ measurable, then,
  $\varphi\(S\)$ \textbf{is not necessarily
  $\mathcal{H}_{k}$ measurable}.
  Actually, if $\varphi: E \subset \mathbb{R}^{n}
  \to \mathbb{R}^{m}, m, n \in \{1, 2, \ldots\}$,
  $\varphi$ is continuous, then
  \begin{align}\begin{split}\nonumber
  \varphi \text{ map } \mathcal{L}_{n}
  \text{ measurable subset of } E
  \text{ to } & \mathcal{L}_{m} \text{ measurable set in } \mathbb{R}^{m}
  \Leftrightarrow \\
  \varphi \text{ map } \mathcal{L}_{n}
  \text{ zero measure subset of } E
  \text{ to } & \mathcal{L}_{m} \text{ zero measure set in } \mathbb{R}^{m}.
  \end{split}\end{align}
  Even if $\varphi$ is more smooth than
  continuous, the right hand side of above
  relationship will not be
  automatically satisfied.
  \item Although $S \cap \varphi^{-1}\(y\)$ may not
  be $\mathcal{H}_{n - k}$ measurable for all
  $y \in \mathbb{R}^{k}$,
  $E \cap \varphi^{-1}\(y\)$ is $\mathcal{H}_{n - k}$
  measurable for all $y \in \mathbb{R}^{k}$,
  since $E$ is open and $\varphi^{-1}\(y\)$
  is a Borel set.
\end{enumerate}
\end{remark}

Proposition \ref{coarea formula simple} states
one of the most essential idea of more general
coarea formula Theorem \ref{coarea formula classic}.
Besides, we should first verify that the
integrand on the right hand side of
\eqref{coarea formula area} is well defined.
A classical proof of
Proposition \ref{coarea formula simple}
can be separated
into two fundamental parts.\\

\begin{proof}[Proof of Proposition \ref{coarea formula simple}] \\
\textbf{part} 1 
  We start from verifying
  a Lusin property (N) for
  \begin{align}\begin{split}\nonumber
  \int_{\mathbb{R}^{k}}\mathcal{H}^{*}_{n - k}\(
  S \cap \varphi^{-1}\(y\)\)d\mathcal{H}_{k}y.
  \end{split}\end{align}
First of all, the outer integral of $f$ is defined as
\begin{align}\begin{split}\nonumber
  \int^{*}_{\mathbb{R}^{k}}f\(x\)d\mathcal{H}_{k}x =
  \inf\left\{\int_{\mathbb{R}^{k}}g\(x\)d\mathcal{H}_{k}x
  \vert g \text{ is } \mathcal{H}_{k} \text{ measurable },
  f \leq g \; a.e.\right\}.
  \end{split}\end{align}
  For arbitrary measurable
  $S \subset E$, the required Lusin property (N) is provided by Eilenberg inequality which states that
  \begin{align}\begin{split}\label{Eilenberg inequality}
  \int^{*}_{\mathbb{R}^{k}}\mathcal{H}^{*}_{n - k}\(
  S \cap \varphi^{-1}\(y\)\)d\mathcal{H}_{k}y \leq
  \frac{\alpha_{n - k}\alpha_{k}}{\alpha_{n}}
  \(\text{Lip}\varphi\)^{k}
  \mathcal{H}_{n}\(S\)
  \end{split}\end{align}
  holds. Besides, we can also show that
  $\mathcal{H}_{n - k}\(S \cap \varphi^{-1}\(y\)\)$
  is well defined and $\mathcal{H}_{k}$ measurable, and therefore the outer integral in \eqref{Eilenberg inequality} is actually redundant.

  To verify \eqref{Eilenberg inequality},
  we use the isodiametric inequality \ref{Isodiametric inequality}
and
  the coincidence between spherical Hausdorff outer measure
  and Hausdorff outer measure \ref{Spherical Hausdorff outer measure}.
  By the equivalence in Theorem \ref{Spherical Hausdorff outer measure},
   for all
  $l > 0$, there exists an at most countable
  collection of closed balls $\{B^{l}_{j}\}$
  such that $S \subset \bigcup_{j}B^{l}_{j}$,
  $\mathrm{diam}\ B^{l}_{j} < \frac{1}{l}$
  for all $j$, and
  \begin{align}\begin{split}\nonumber
  \sum_{i}\alpha_{n}\(\frac{\mathrm{diam}\
  B^{l}_{j}}{2}\)^{n} \leq
  \mathcal{H}_{n}\(S\) + \frac{1}{j}.
  \end{split}\end{align}
  Define
  \begin{align}\begin{split}\nonumber
  g^{l}_{j}\(y\) =
  \alpha_{n - k}\(\frac{\mathrm{diam}\
  B^{l}_{j}}{2}\)^{n - k}
  1\(y \in \varphi\(B^{l}_{j}\)\),
  \end{split}\end{align}
  since $B^{l}_{j}$ is a closed ball,
  $g^{l}_{j}$ is $\mathcal{H}^{k}$ measurable.
  Note that $\{B^{l}_{j}\}$
  covers $A \cap \varphi^{-1}\(y\)$
  for all $y$,
  \begin{align}\begin{split}\nonumber
  \mathcal{H}^{*}_{n - k, \frac{1}{l}}\(S
  \cap \varphi^{-1}\(y\)\) \leq
  \sum_{j}g^{l}_{j}\(y\).
  \end{split}\end{align}
  Then 
by the Fatou lemma,
the isodiametric inequality,
  and the fact that
  Hausdorff outer measure keeps
  distance inequality,
  \begin{align}\begin{split}\nonumber
  \int^{*}_{\mathbb{R}^{k}}
  \mathcal{H}^{*}_{n - k}\(S \cap \varphi^{-1}
  \(y\)\)d\mathcal{H}_{k}y & =
  \int^{*}_{\mathbb{R}^{k}}
  \lim_{l \to \infty}\mathcal{H}^{*}_{n - k,
  \frac{1}{l}}\(S \cap \varphi^{-1}
  \(y\)\)d\mathcal{H}_{k}y \\
  & \leq \int_{\mathbb{R}^{k}}
  \liminf_{l \to \infty}
  \sum_{j}g^{l}_{j}\(y\)d\mathcal{H}_{k}y \\
  & \leq \liminf_{l \to \infty}
  \int_{\mathbb{R}^{k}}\sum_{j}
  g^{l}_{j}\(y\)d\mathcal{H}_{k}y \\
  & \leq \liminf_{l \to \infty}
  \sum_{j}\int_{\mathbb{R}^{k}}
  g^{l}_{j}\(y\)d\mathcal{H}_{k}y \\
  & \leq \liminf_{l \to \infty}
  \sum_{j}a_{n - k}\(\frac{
  \mathrm{diam}\ B^{l}_{j}}{2}\)^{n - k}
  \mathcal{H}_{k}\(\varphi\(B^{l}_{j}\)\) \\
  & \leq \liminf_{l \to \infty}
  \sum_{j}a_{n - k}\(\frac{
  \mathrm{diam}\ B^{l}_{j}}{2}\)^{n - k}
  \alpha_{k}\(\frac{\mathrm{diam}\
  \varphi\(B^{l}_{j}\)}{2}\)^{k} \\
  & \leq \frac{\alpha_{n - k}\alpha_{k}}
  {\alpha_{n}}\(\text{Lip}\varphi\)^{k}
  \liminf_{l \to \infty}\sum_{j}
  \(\frac{\mathrm{diam}\ B^{l}_{j}}{2}\)^{n} \\
  & \leq \frac{\alpha_{n - k}\alpha_{k}}
  {\alpha_{n}}\(\text{Lip}\varphi\)^{k}
  \mathcal{H}_{n}\(S\).
  \end{split}\end{align}

  Next, we should verify that
  $S \cap \varphi^{-1}\(y\)$ is
  $\mathcal{H}_{n - k}$ measurable for
  $\mathcal{H}_{k} \; a.e. \; y$, and
  $y \in \mathbb{R}^{k} \mapsto
  \mathcal{H}_{n - k}\(S \cap \varphi^{-1}\(y\)\)$
  is $\mathcal{H}_{k}$ measurable.
  First, assuming that $S$ is compact,
 we can write
  \begin{align}\begin{split}\nonumber
  \mathcal{H}^{*}_{n - k}\(S \cap \varphi^{-1}\(y\)\) =
  \lim_{\delta \downarrow 0}
  \mathcal{H}^{*}_{n - k, \delta}\(
  S \cap \varphi^{-1}\(y\)\) =
  \sup_{\delta > 0}\mathcal{H}^{*}_{n - k, \delta}\(
  S \cap \varphi^{-1}\(y\)\).
  \end{split}\end{align}
Note that in this case,
  $S \cap \varphi^{-1}\(y\)$ is a Borel set
  and thus $\mathcal{H}_{n - k}$ measurable, therefore
  it suffices to verify that
  $y \mapsto \mathcal{H}^{*}_{n - k , \delta}
  \(S \cap \varphi^{-1}\(y\)\)$ is
  $\mathcal{H}_{k}$ measurable
  for all $\delta > 0$.
  Actually, $\mathcal{H}^{*}_{n - k , \delta}
  \(S \cap \varphi^{-1}\(y\)\)$ is
  upper semicontinuous. To see this,
  note that the spherical Hausdorff
  outer measure can also be defined
  by open balls, and thus
  for arbitrary $\epsilon > 0$,
  there exits an at most countable
  collection of open balls
  $\{B_{j}\}$ such that
  $S \cap \varphi^{-1} \subset
  \bigcup_{j}B_{j}$, for all $j$,
  $\mathrm{diam}\ B_{j} \leq \delta$, and
  \begin{align}\begin{split}\nonumber
  \sum_{j}\alpha_{n - k}
  \(\frac{\mathrm{diam}\ B_{j}}{2}\)^{n - k}
  \leq \mathcal{H}^{*}_{n - k, \delta}\(
  S \cap \varphi^{-1}\(y\)\) + \epsilon.
  \end{split}\end{align}
  The compactness of $S$ implies that
  if $\vert y - y'\vert$ small enough,
  then $S \cap \varphi^{-1}\(y'\)
  \subset \bigcup_{j}B_{j}$.
  Therefore,
  \begin{align}\begin{split}\nonumber
  \limsup_{y' \to y}\mathcal{H}^{*}_{n - k, \delta}
  \(S \cap \varphi^{-1}\(y'\)\) \leq
  \mathcal{H}^{*}_{n - k, \delta}\(
  S \cap \varphi^{-1}\(y'\)\) + \epsilon,
  \end{split}\end{align}
  then the upper semicontinuity follows
  from the arbitrariness of $\epsilon$.
  Second, let $S$ be just measurable,
  then by the regularity of Lebesgue
  measure, there exists a sequence of
  compact sets
  $C_{1} \subset C_{2} \subset \ldots$
  such that $S \backslash
  \bigcup^{\infty}_{i = 1}C_{i}$ is of
  zero Lebesgue measure.
  By the cube decomposition of open set
  and the Eilenberg inequality,
  \begin{align}\begin{split}\nonumber
  \int^{*}_{\mathbb{R}^{k}}\mathcal{H}^{*}_{n - k, \delta}
  \(\(S \backslash \bigcup^{\infty}_{i = 1}
  C_{i}\) \cap \varphi^{-1}\(y\)\)d\mathcal{H}_{k}y = 0,
  \end{split}\end{align}
  i.e. $\mathcal{H}^{*}_{n - k, \delta}
  \(\(S \backslash \bigcup^{\infty}_{i = 1}
  C_{i} \cap \varphi^{-1}\(y\)\)\) = 0$,
  $\mathcal{H}_{k} \; a.e. \; y$. The $a.e.$ measurability
  of $S \cap \varphi^{-1}\(y\)$ and measurability
  of $\mathcal{H}^{*}_{n - k}\(S \cap \varphi^{-1}\(y\)\)$
  follow from
  \begin{align}\begin{split}\nonumber
  \mathcal{H}^{*}_{n - k}\(S \cap \varphi^{-1}\(y\)\) =
  \mathcal{H}^{*}_{n- k}\(\bigcup^{\infty}_{i = 1}C_{i}
  \cap \varphi^{-1}\(y\)\) + \mathcal{H}^{*}_{n - k}
  \(\(S \backslash \bigcup^{\infty}_{i = 1}
  C_{i} \cap \varphi^{-1}\(y\)\)\),
  \end{split}\end{align}
  and now we can use $\mathcal{H}_{n - k}\(S
  \cap \varphi^{-1}\(y\)\)$ instead of
  $\mathcal{H}^{*}_{n - k}\(S
  \cap \varphi^{-1}\(y\)\)$. \\
\textbf{part} 2 Assuming that $J\varphi\(x\) > 0$
  for all $x \in E$. We use an
  estimates of
  $\vert\varphi\(x\) - \varphi\(y\)\vert$
  similar to \eqref{estimate local}:
  Let $k \in \{1, 2, \ldots, n\}$,
  $E \subset \mathbb{R}^{n}$ be an open set,
  $\varphi: E \to \mathbb{R}^{k}$ be a
  $C^{1}$ function
  such that $J\varphi\(x\) > 0$
  for all $x \in E$,
  then for arbitrary $x_{0} \in E$ and
  for all $\epsilon > 0$,
  there exists $\delta > 0$, such that open ball
  $B\(x_{0}, \delta\) \subset E$ and
  \begin{align}\begin{split}\label{estimate differentiable}
  \vert\varphi\(x\) - \varphi\(y\) -
  \varphi'\(x_{0}\)\(x - y\)\vert \leq
  \epsilon\vert x - y\vert
  \end{split}\end{align}
  for all $x, y \in B\(x_{0}, \delta\)$.

  Let $x_{0} \in E$,
  since $J\varphi\(x_{0}\) > 0$, without loss
  of generality, assuming that the
  first $k$ columns
  $\left\{\frac{\partial}{\partial x_{1}}\varphi\(x_{0}\),
  \frac{\partial}{\partial x_{2}}\varphi\(x_{0}\), \ldots,
  \frac{\partial}{\partial x_{k}}\varphi\(x_{0}\)\right\}$ are
  linear independent.
  Define $\Phi\(x\) = \(\varphi\(x\), x_{k + 1}, \ldots, x_{n}\)$,
  by the implicit function theorem,
  $\Phi$ is a bijection on a neighborhood
  of $x_{0}$. By definition,
  \begin{align}\begin{split}\nonumber
  \Phi\(x\) & = \(\Phi\(x\) -
  \Phi'\(x_{0}\)\(x - x_{0}\)\) +
  \Phi'\(x_{0}\)\(x - x_{0}\) \\
  & = \Phi'\(x_{0}\)
  \[\(\Phi'\(x_{0}\)\)^{-1}
  \(\Phi\(x\) - \Phi'\(x_{0}\)\(x - x_{0}\)\)
  + \(x - x_{0}\)\],
  \end{split}\end{align}
  denote the term in square brackets
  on the right hand side of the second equality
  as $g\(x\)$, then by estimate
  \eqref{estimate differentiable},
  for all $\epsilon > 0$ there exists
  $\delta > 0$ such that
  \begin{align}\begin{split}\label{Lip approximation}
  \(1 - \epsilon\)\vert x - y\vert \leq
  \vert g\(x\) - g\(y\)\vert \leq
  \(1 + \epsilon\)\vert x - y\vert,
  \end{split}\end{align}
  for all $x, y \in B\(x_{0}, \delta\)$.
  Therefore, $\varphi = A \circ g$ on
  $B\(x_{0}, \delta\)$,
  where $A = P\Phi'\(x_{0}\)$,
  $P: \mathbb{R}^{n} \to
  \mathbb{R}^{k}$
  is the coordinate projection of the
  first $k$ dimensions.
  By \eqref{Lip approximation},
  $1 - \epsilon \leq
  \vert g'\vert \leq 1 + \epsilon$,
  by definition,
  \begin{align}\begin{split}\nonumber
  \(J\varphi\(x\)\)^{2} =
  \det\(\varphi'\(x\)\varphi'\(x\)^{T}\) =
  \det\(A \circ g'\(x\) \circ g'\(x\)^{T}
  \circ A^{T}\).
  \end{split}\end{align}
  By SVD / PD,
  $g'\(x\) \circ g'\(x\)^{T} =
  Q^{T}CQ$, where $C$ is diagonal
  with $\(1 - \epsilon\)^{2} \leq
  c_{ii} \leq \(1 + \epsilon\)^{2}$,
  $i \in \{1, 2, \ldots, n\}$,
  $Q^{T}Q = I_{n}$;
  $A = PU^{T}$, where
  $P \in \mathbb{R}^{k \times k}$
  is symmetric and
  $U \in \mathbb{R}^{n \times k}$
  is orthogonal. Therefore,
  \begin{align}\begin{split}\nonumber
  \det\(A \circ g'\(x\) \circ g'\(x\)^{T}
  \circ A^{T}\) =
  \det\(PU^{T}Q^{T}CQUP^{T}\) =
  \(\det\(A\)\)^{2}\det\(U^{T}Q^{T}CQU\).
  \end{split}\end{align}
  Note that $QU$ is also
  orthogonal, then
  \begin{align}\begin{split}\nonumber
  \(1 - \epsilon\)^{2k} \leq
  \det\(U^{T}Q^{T}CQU\) \leq
  \(1 + \epsilon\)^{2k}.
  \end{split}\end{align}
  As a result,
  \begin{align}\begin{split}\label{coarea local estimate}
  \(1 - \epsilon\)^{k}JA \leq
  J\varphi\(x\) \leq
  \(1 + \epsilon\)^{k}JA
  \end{split}\end{align}
  on $B\(x_{0}, \delta\)$.
  Compare \eqref{coarea local estimate} and
  \eqref{L and H local estimate},
  then see that the ideas of
  the \textbf{step} 4
  and \textbf{step} 5 of the
  proof sketch of Proposition
  \ref{area formula simple}
  can be used here.

  When $J\varphi\(x\) = 0$,
  let $\epsilon > 0$, define
  $\varphi_{\epsilon}: E \times
  \mathbb{R}^{k} \to \mathbb{R}^{k}$ as
  \begin{align}\begin{split}\nonumber
  \(x, z\) \mapsto \varphi\(x\) +
  \epsilon z,
  \end{split}\end{align}
  $Crit\(\varphi\) = \{x \in E
  \vert J\varphi\(x\) = 0\}$.
  Note that $J\varphi_{\epsilon}\(x, z\) > 0$
  for all $\(x, z\) \in E \times \mathbb{R}^{k}$,
  for arbitrary $w \in \mathbb{R}^{k}$,
  \begin{align}\begin{split}\nonumber
  \int_{\mathbb{R}^{k}}\mathcal{H}_{n - k}
  \(S \cap \varphi^{-1}\(y\)\)d\mathcal{H}_{k}y & =
  \int_{\mathbb{R}^{k}}\mathcal{H}_{n - k}
  \(S \cap \varphi^{-1}\(y - \epsilon w\)\)
  d\mathcal{H}_{k}y \\
  & = \frac{1}{\alpha_{k}}\int_{B\(0, 1\)}
  \int_{\mathbb{R}^{n - k}}
  \mathcal{H}_{n - k}\(S \cap \varphi^{-1}
  \(y - \epsilon w\)\)d\mathcal{H}_{k}yd\mathcal{H}_{k}w.
  \end{split}\end{align}
  Let $P': \mathbb{R}^{n}
  \times \mathbb{R}^{k} \to \mathbb{R}^{k}$
  be the coordinate projection of the
  last $k$ dimensions,
  $S' = S \times B\(0, 1\) \subset
  \mathbb{R}^{n + k}$.
  Then note that for all
  $y \in \mathbb{R}^{k}$,
  $w \in B\(0, 1\)$,
  \begin{align}\begin{split}\nonumber
  S' \cap \varphi^{-1}_{\epsilon}\(y\)
  \cap \(P'\)^{-1}\(w\) =
  \(S \cap \varphi^{-1}
  \(y - \epsilon w\)\) \times \{w\}.
  \end{split}\end{align}
  Therefore, by the Eilenberg inequality
  and the Fubini-Tonelli theorem,
  \begin{align}\begin{split}\nonumber
  \int_{\mathbb{R}^{k}}\mathcal{H}_{n - k} &
  \(S \cap \varphi^{-1}\(y\)\)d\mathcal{H}_{k}y \\
  & = \frac{1}{\alpha_{k}}\int_{B\(0, 1\)}
  \int_{\mathbb{R}^{k}}\mathcal{H}_{n - k}
  \(S' \cap \varphi^{-1}_{\epsilon}\(y\)
  \cap \(P'\)^{-1}\(w\)\)
  d\mathcal{H}_{k}yd\mathcal{H}_{k}w \\
  & \leq \frac{\alpha_{n - k}}{\alpha_{n}}
  \int_{\mathbb{R}^{k}}
  \(\text{Lip}P'\)^{k}\mathcal{H}_{n}
  \(S' \cap \varphi^{-1}_{\epsilon}
  \(y\)\)d\mathcal{H}_{k}y \\
  & \leq \frac{\alpha_{n - k}}{\alpha_{n}}
  \int_{S'}J\varphi_{\epsilon}\(x, z\)
  d\mathcal{H}_{n + k}\(x, z\)
  \leq \frac{\alpha_{n - k}\alpha_{k}}
  {\alpha_{n}}\mathcal{L}_{n}\(S\)
  \sup_{\(x, z\) \in S'}
  J\varphi_{\epsilon}\(x, z\).
  \end{split}\end{align}
The
  last inequality above uses the fact
  that $\mathcal{H}_{n + k} = \mathcal{L}_{n + k}$ is the completion of
$\mathcal{L}_{n} \times \mathcal{L}_{k}$.
  If $E$ is bounded, the right hand side
  of the last inequality converges to $0$
  as $\epsilon \downarrow 0$.
  Then conclude for (not necessarily bounded)
  open set $E \subset \mathbb{R}^{n}$ by the
  cube decomposition.
\end{proof}

\begin{proof}[Proof of Theorem \ref{coarea formula classic}] \\
We show that
Theorem \ref{coarea formula classic}
follows from Proposition \ref{coarea formula simple},
using the Rademacher-Whitney-Lusin-Egoroff
framework as in \textbf{part} 2
of the proof sketch of the
Theorem \ref{area formula classic}.
Without loss of generality,
let $E$ be bounded.
Suppose we already find a $C^{1}$ function
$\overline{\varphi}$
and compact set $C'$, such that
\eqref{C1 extension epsilon} holds
for an $\epsilon > 0$.
Now, by the Eilenberg inequality,
\begin{align}\begin{split}\nonumber
\int_{\mathbb{R}^{k}}\mathcal{H}_{n - k}
\(S \cap \varphi^{-1}\(y\)\)d\mathcal{H}_{k}y & =
\int_{S \cap C'}J\varphi\(x\)d\mathcal{L}_{n}x +
\int_{\mathbb{R}^{k}}\mathcal{H}_{n - k}
\(\(S \backslash C'\) \cap
\varphi^{-1}\(y\)\)d\mathcal{H}_{k}y\\
& \leq \int_{S \cap C'}J\varphi\(x\)
d\mathcal{L}_{n}x + \frac{\alpha_{n - k}\alpha_{k}}
{\alpha_{n}}\(\text{Lip}\varphi\)^{k}
\mathcal{L}_{n}\(S \backslash C'\) \\
& \leq \int_{S \cap C'}J\varphi\(x\)
d\mathcal{L}_{n}x + M_{1}\epsilon,
\end{split}\end{align}
where $M_{1}$ is a constant that does
not depend on $\epsilon$.
Since $\varphi$ is Lipschitz,
there exists another constant
$M_{2}$, such that $J\varphi < M_{2}$,
and thus,
\begin{align}\begin{split}\nonumber
\int_{S \setminus C'}J\varphi\(x\)d\mathcal{L}_{n}x
\leq M_{2}\epsilon.
\end{split}\end{align}
Then conclude by the same discussion
as \textbf{part} 3
of the proof sketch Theorem
\ref{area formula classic}.
\end{proof}

\begin{cor}\label{another Sard lemma}
{Another Sard type lemma} \\
Let $k \in \{1, 2, \ldots, n\}$,
$E \subset \mathbb{R}^{n}$ be an open set,
$\varphi: E \to \mathbb{R}^{k}$ be a $C^{1}$
function. Then for
$\mathcal{L}_{k} \; a.e. \; y \in \mathbb{R}^{k}$,
\begin{align}\begin{split}\label{local flat measure zero}
\mathcal{H}_{n - k}\(Crit\(\varphi\) \cap
\varphi^{-1}\(y\)\) = 0
\end{split}\end{align}
and $\varphi^{-1}\(y\) \backslash Crit(\varphi)$
can be locally
parameterized 
by implicit functions.
\end{cor}

\begin{proof}
Since $\varphi$ is a $C^{1}$ function,
$J\varphi: E \to \mathbb{R}$ is continuous.
Therefore,
$Crit\(\varphi\) =
\{x \in E : J\varphi\(x\) \geq 0\} \backslash
\{x \in E : J\varphi\(x\) > 0\}$
is $\mathcal{L}_{n}$ measurable.
Let $y \in \mathbb{R}^{k}$,
by \eqref{coarea formula area},
\begin{align}\begin{split}\nonumber
\int_{Crit\(\varphi\)}
J\varphi\(x\)d\mathcal{L}_{n}x =
\int_{\mathbb{R}^{k}}
\mathcal{H}_{n - k}\(Crit\(\varphi\) \cap
\varphi^{-1}\(y\)\)d\mathcal{H}_{k}y.
\end{split}\end{align}
The integral in the left hand side
of the above equality is $0$,
thus by the property of Lebesgue integral
and the fact that $\mathcal{L}_{k} = \mathcal{H}_{k}$ in $\mathbb{R}^{k}$,
\eqref{local flat measure zero} holds for
$\mathcal{L}_{k} \; a.e. \; y \in \mathbb{R}^{k}$.

For arbitrary
$x \in \varphi^{-1}\(y\) \backslash
Crit(\varphi)$ where $y$ satisfies
\eqref{local flat measure zero}, $J\varphi\(x\) > 0$.
Note that 
$\varphi'\(x\)$ is row full rank.
Therefore, the implicit function theorem
can be applied in a neighborhood of $x$.
Specifically, pick $k$ linear independent
columns of $\varphi'\(x\)$,
without loss of generality, assuming that
the first $k$ columns
$\left\{\frac{\partial}{\partial x_{1}}\varphi\(x\),
\frac{\partial}{\partial x_{2}}\varphi\(x\), \ldots,
\frac{\partial}{\partial x_{k}}\varphi\(x\)\right\}$ are
linear independent.
Let $U$ be a neighborhood of $x$,
define 
\begin{align}\begin{split}\nonumber
\Psi\(x_{1}', \ldots, x_{k}',
x_{k + 1}', \ldots, x_{n}'\)
= 
\varphi\(x_{1}', \ldots, x_{k}',
x_{k + 1}', \ldots, x_{n}'\) - y,
\end{split}\end{align}
where  $y = \varphi\(x\)$,
then $\nabla_{x_{1}, x_{2}, \ldots, x_{k}}
\Psi\(x_{1}, \ldots, x_{k},
x_{k + 1}, \ldots, x_{n}\)$ is full rank,
and thus there exists a $C^{1}$ implicit function $g$
in a open subset
$B_{x_{1}, \ldots, x_{k}} \times
B_{x_{k + 1}, \ldots, x_{n}} \subset U$,
where
\begin{align}\begin{split}\nonumber
B_{x_{1}, \ldots, x_{k}} & =
\left\{\(x_{1}', \ldots, x_{k}'\)
\in \mathbb{R}^{k}: \left\vert
\(x_{1}', \ldots, x_{k}'\)
- \(x_{1}, \ldots, x_{k}\)
\right\vert < \alpha\right\}, \\
B_{x_{k + 1}, \ldots, x_{n}} & =
\left\{\(x_{k + 1}', \ldots, x_{n}'\)
\in \mathbb{R}^{n - k}: \left\vert
\(x_{k + 1}', \ldots, x_{n}'\) -
\(x_{k + 1}, \ldots, x_{n}\)
\right\vert < \beta\right\},
\end{split}\end{align}
such that $\Psi\(x_{1}', \ldots, x_{k}',
x_{k + 1}', \ldots, x_{n}'\) = 0 \Leftrightarrow
\(x_{1}', \ldots, x_{k}'\) =
g\(x_{k + 1}', \ldots, x_{n}'\)$ for all
$x' \in B_{x_{1}, \ldots, x_{k}} \times
B_{x_{k + 1}, \ldots, x_{n}}$.
Now, define
\begin{align}\begin{split}\nonumber
\psi\(x_{k + 1}', \ldots, x_{n}'\) = \(
g\(x_{k + 1}', \ldots, x_{n}'\),
x_{k + 1}', \ldots, x_{n}'\)^{T},
\end{split}\end{align}
then $\psi$ is a $C^{1}$ diffeomorphism
and thus a homeomorphism.
\end{proof}

\end{document}